\newcommand{\onlineversion}[2]{#1}{} 

\onlineversion{
\documentclass[11pt,letterpaper]{article}

\topmargin=-0.35in \topskip=0pt \headsep=15pt \oddsidemargin=0pt
\textheight=9in \textwidth=6.52in \voffset=0in
}{

}

\usepackage{amsmath,amsfonts,graphpap,amscd,mathrsfs,graphicx,lscape}
\usepackage{epsfig,amssymb,amstext,xspace}
\usepackage{mathtools}
\onlineversion{
\usepackage{theorem}
}{}
\usepackage{algorithm}
\usepackage{algorithmic}
\usepackage{subcaption}
\usepackage{paralist}
\usepackage{color}              
\usepackage{epsfig}
\usepackage{bm}
\usepackage{bbm}
\usepackage{hyperref}
\usepackage{booktabs}
\usepackage{multirow}
\usepackage[numbers]{natbib}

\onlineversion{
\newtheorem{theorem}{Theorem}[section]
\newtheorem{lemma}[theorem]{Lemma}
\newtheorem{example}[theorem]{Example}
\newtheorem{corollary}[theorem]{Corollary}
\newtheorem{remark}[theorem]{Remark}
}{}

\newtheorem{proposition}[theorem]{Proposition}

\newtheorem{definition}[theorem]{Definition}

\newtheorem{assumption}[theorem]{Assumption}
\newtheorem{conjecture}[theorem]{Conjecture}

\def\squareforqed{\hbox{\rule{2.5mm}{2.5mm}}}

\def\QED{\ifmmode\squareforqed 
  \else{\nobreak\hfil   
    \penalty50                 
    \hskip1em                  
    \null                      
    \nobreak                   
    \hfil                      
    \squareforqed              
    \parfillskip=0pt           
    \finalhyphendemerits=0     
    \endgraf}                  
  \fi}

\def\blksquare{\rule{2mm}{2mm}}
\def\qedsymbol{\blksquare}
\newcommand{\bg}[1]{\medskip\noindent{\bf #1}}
\newcommand{\ed}{{\hfill\qedsymbol}\medskip}

\onlineversion{
\newenvironment{proof}{\bg{Proof : }}{\ed}
}{}
\newenvironment{proofof}[1]{\textbf{Proof of #1 : }}{\ed}

\newcommand{\comment}[1]{}

\newcommand{\junk}[1]{}

\newlength{\tmp} \newlength{\lpsx} \newlength{\lpsy} \newlength{\upsx}
\newlength{\upsy}

\newcommand{\xhdr}[1]{\vskip 6pt \noindent {\bf #1.}}
\onlineversion{
\newcommand{\email}[1]{\texttt{#1}}
}{}

\begin{document}

\onlineversion{

\setcounter{page}{0}

\title{Mechanism Design for AI Overviews: \\ Creator Incentives and Long-Term Profit}

\author{
Yihang Wu\\
       Zhejiang University\\
       \email{yhwu\_is@zju.edu.cn}
\and
Jiajun Tang\\
       Zhejiang University\\
       \email{jiajuntang@zju.edu.cn}
\and
Jinfei Liu\thanks{Advising authors.}\\
       Zhejiang University\\
       \email{jinfeiliu@zju.edu.cn}
\and
Haifeng Xu\footnotemark[1]\\
       University of Chicago\\
       \email{haifengxu@uchicago.edu}
\and
Fan Yao\footnotemark[1]\\
       UNC-Chapel Hill\\
       \email{fanyao@unc.edu}
}
}{}

\onlineversion{
\maketitle
}{}

\begin{abstract}
  The integration of AI Overviews into search engines enhances user experience but diverts traffic from content creators, potentially discouraging high-quality content creation and causing user attrition that undermines long-term search engine profit. To address this issue, we propose a game-theoretic model of creator competition with costly effort, characterize equilibrium behavior, and design two incentive mechanisms: a \emph{citation mechanism} that references sources within an AI Overview, and a \emph{compensation mechanism} that offers monetary rewards to creators. For both cases, we provide structural insights for profit-maximizing mechanisms. Evaluations parameterized by real click data show that although AI Overviews harm long-term search engine profit, interventions based on our proposed mechanisms can increase long-term profit across a range of realistic scenarios, pointing toward a more sustainable trajectory for AI-enhanced search ecosystems.
\end{abstract}

\onlineversion{
\renewcommand{\thepage}{}
\clearpage
\pagenumbering{arabic}
}{
}

\section{Introduction} \label{sec:intro}

The paradigm of web search is undergoing a profound transformation driven by generative AI (GenAI). Mainstream search engines, including Google \cite{google_ai_overview}, Microsoft Bing \cite{bing_ai_overview}, and Baidu \cite{li2025towards}, are increasingly integrating AI-generated summaries (a.k.a. AI Overviews) directly into their results pages. By consolidating information from multiple sources into an immediate response, AI Overviews offer significant user benefits by reducing search time and simplifying information synthesis.


The rollout of AI Overviews has driven considerable revenue growth for search engines. Alphabet’s Q1 2025 report shows Google Search revenue of US \$50.7B (up 10\% YoY) and attributes part of the growth to the engagement with AI Overviews, which reaches 1.5B monthly users \cite{google_q1_2025_report}. However, behind the impressive figure lies a growing concern: AI Overviews are shifting traffic away from web content creators. Chegg, the leading student-first connected learning platform, reported that ``AI Overview is retaining traffic that historically had come to Chegg'' and that its non-subscriber traffic fell 49\% in January 2025 \cite{chegg_q4_2024_report}. Publishers have also sued Google over unlicensed summaries and diverted traffic \cite{gizmodo_rolling_stone_sue}. These shifts may weaken incentives for high-quality content creation, degrading user experience over time, and ultimately harming search engine profit. Motivated by these pressing issues, we study the following questions: 

\emph{Do AI Overviews benefit search engines in the long run? If not, how can search engines design mechanisms to align creator incentives with their profit objective in the presence of AI Overviews?}


To answer these questions, we employ a game-theoretic model, grounded on the classic \emph{position-based model (PBM)} \cite{craswell2008experimental}, to study the competition among content creators. In our model, a number of $n$ creators (e.g., news sites, blogs, knowledge bases) compete for a keyword query. The traffic of a page depends on its rank-based position bias and costly creation effort, and the search engine ranks pages by effort. We characterize mixed Nash equilibria under both homogeneous and heterogeneous costs of creation effort. Specifically, our theoretical analysis focuses on homogeneous and binary-type creator costs to balance tractability and realism.

Building on the equilibrium characterization, we then study two incentive mechanisms: (1) a \emph{citation mechanism} that cites selected organic results (i.e., web pages created by human creators) within AI Overviews to increase their position biases and address copyright concerns. This mechanism has been widely adopted by existing systems \cite{google_ai_overview,bing_ai_overview,perplexity}, but its design has not been theoretically examined; (2) a \emph{compensation mechanism} that pays creators based on creation effort, which also essentially raises position biases in our model. We solve mechanism design problem for the search engine’s profit objective and show sharp structural simplifications: optimal compensation mechanisms and effective citation mechanisms can be taken to increase position biases in a piecewise-constant manner by creator cost type, which enlightens simple, tractable, and effective practical design.

We ran a real-world user click experiment to estimate how AI Overviews change position biases. Parameterizing our PBM-based model with these estimates, numerical experiments show that introducing AI Overviews without incentives reduces long-term search engine profit (i.e., after creators adjust to the new equilibrium under the changed position biases), whereas adding our citation and compensation mechanisms increases long-term profit across a range of realistic settings, thus answering our motivating questions.

Our contributions lie in four aspects: (a) \textbf{Modeling}: We build a PBM-based game-theoretic model of creator competition in search engines, and formalize both AI Overviews and incentive mechanisms as interventions that alter the position biases; (b) \textbf{Conceptual}: Our results sharply narrow the mechanism design space and thus provide actionable guidance on how to improve long-term profit for search engines, helping sustain a healthy creator ecosystem in the GenAI era; (c) \textbf{Technical}: We extend equilibrium analysis techniques to accommodate multiplicative payoffs (rank-based position bias times the effort of creators) and heterogeneous cost structures, revealing novel equilibrium properties that are distinctive to our model; (d) \textbf{Empirical}: We collect real user click data that quantifies how AI Overviews affect position biases, providing an empirical foundation that supports future research on AI-enhanced search ecosystems.

\section{Related Works}

\xhdr{Content Creator Competition} Creator competition for user attention is well studied in recommendation settings \cite{ben2018game,ben2020content,hron2023modeling,yao2023bad,yao2023rethinking,jagadeesan2023supply}. Recent studies have further investigated the interplay between human creators and GenAI, including human-GenAI competition/collaboration \cite{yao2024human}, dynamic strategic behavior of human creators \cite{esmaeili2024strategize}, long-term ecosystem effects \cite{taitler2025braess}, and selective response strategies for GenAI \cite{taitler2025selective}, as well as data sharing between human and GenAI \cite{keinan2025strategic,taitler2025data} and how GenAI diversity shapes competition \cite{raghavan2024competition}. Differently, we study creator competition \emph{in search engines}, where AI Overviews reshape incentives, and we design mechanisms to improve search engine profit. Relatedly, \cite{madmon2025search} studies search engine creator competition but focuses on the effects of ranking rules on equilibrium stability and user welfare.

\xhdr{Contest Design} Contest design literature primarily addresses equilibrium characterization and optimal reward design, with reward structures typically classified as rank-order (e.g., all-pay auctions \cite{krishna2009auction}) or smooth allocation (e.g., Tullock contests \cite{tullock1980efficient}). See \cite{vojnovic2015contest,sisak2009multiple} for comprehensive surveys. We focus on rank-order allocation contests. On equilibrium characterization, classic results include single- and multi-prize all-pay auctions \cite{baye1996all,barut1998symmetric}, extensions with heterogeneous linear costs and restricted prize sequences \cite{bulow2006matching}, and nonlinear costs with homogeneous prizes \cite{siegel2009all}. Most related is \cite{xiao2016asymmetric}, which allows general prize structures with nonlinear costs and provides uniqueness under structural conditions on reward sequence. Our key difference is \emph{multiplicative} payoffs: in our model, a creator’s payoff equals the rank-based position bias times her own effort level, whereas standard rank-order contests award a fixed rank-based prize independent of effort. This difference necessitates new equilibrium analysis techniques. On reward design, prior work optimizes rank-order prizes to induce effort \cite{lazear1981rank,glazer1988optimal,moldovanu2001optimal}. Most related is \cite{golrezaei2025contest}, which we extend to the multiplicative payoff and heterogeneous cost setting.
\section{Model} \label{sec:model}

There are $n \geqslant 2$ web page creators competing for a single keyword query. Search engine users aim to obtain satisfactory answers via keyword search. We denote $[n] = \{1, 2, \dots, n\}$.


\xhdr{User browsing behavior} We adopt the \emph{position-based model (PBM)} to characterize user browsing behavior. In PBM, a user clicks a page if and only if they examine it and are attracted by its content. Let $p_j$ denote the probability that a user examines the page at rank $j \in [n]$, referred to as \emph{position bias} of rank $j$. \cite{joachims2005accurately} shows that examination depends strongly on rank and typically decreases with rank, so we assume $1 \geqslant p_1 > p_2 > \cdots > p_n > 0$. Let $x_i \in [0, 1]$ be the probability that a user is attracted to creator $i$'s page conditional on examination. Creators exert effort to increase the attractiveness, so we treat $x_i$ as the strategy (effort level) of creator $i$. Hence, if creator $i$'s page is ranked at $j$, its click probability is $x_i p_j$.


\xhdr{Web page creation} Given the position bias vector $\bm{p} = (p_1, p_2, \ldots, p_n)$, each creator $i \in [n]$ chooses an effort level $x_i \in [0, 1]$ to create a web page. A vector $\bm{x} = (x_1, x_2, \ldots, x_n)$ consisting of deterministic effort choices is referred to as a pure strategy profile, $\bm{x}_{-i}$ denotes the profile excluding creator $i$. Naturally, exerting a higher effort level incurs a higher cost. Formally, creator $i$ incurs cost $g_i(x) = \gamma_i g(x)$ when exerting effort $x$, where $g: [0,1] \to \mathbb{R}_{\geqslant 0}$ is strictly increasing, differentiable, and strictly convex with $g(0) = 0$, and $\gamma_i > 0$ is a cost parameter that reflects creator $i$'s capability in web page creation. Here, $g$ captures the common cost structure shared by all creators, while $\gamma_i$ accounts for individual heterogeneity. For analytical convenience, we sometimes assume that $g(x) = x^{\beta}$ for a constant $\beta > 1$, which is commonly adopted in the literature \cite{yao2024human,golrezaei2025contest}.


We assume each click generates unit revenue without loss of generality \footnote{If a click brings $r_i$ units of revenue, we can rescale the utility by a factor $1 / r_i$, which does not alter the strategic behavior of creator $i$.}. If creator $i$ chooses $x_i$ and her web page is displayed at rank $j$, her expected revenue is $x_i p_j$. The expected utility is revenue minus cost: 
$u_i(x_i, \bm{x}_{-i}) = x_i p_j - \gamma_i g(x_i)$. The search engine ranks pages in descending order of effort levels \cite{robertson1977probability}. Ties are broken uniformly at random. Let $x_{(1)} \geqslant x_{(2)} \geqslant \cdots \geqslant x_{(n)}$ denote the sorted effort levels, where $x_{(j)}$ is the $j$-th highest effort. We assume that exerting the maximum effort $1$ is prohibitively costly, i.e., $u_i(1, \bm{x}_{-i}) < 0$ for all $i \in [n]$.




\begin{remark}[Ranking by effort] \label{remark:rank-by-effort}
    Since our model links effort to click-through rates (\cite{ai2025genai} adopts a similar assumption), our effort variable is designed to model the composite measure for a specific query across quality, timeliness, etc. In addition, by linking effort to click outcomes, effort becomes quantifiable rather than a latent variable. Besides, using a single value for effort is aligned with a commonly used modeling practice in game-theoretic modeling of content creation competition (\cite{jagadeesan2023supply,yao2024human,keinan2025strategic,golrezaei2025contest}). We provide further justification for our modeling choices in Appendix~\ref{sec:ap-model}.
\end{remark}

In addition to pure strategies, each creator may choose a mixed strategy, which is a probability distribution over the effort levels in $[0, 1]$. We denote by $F_i$ the cumulative distribution function (CDF) of creator $i$’s mixed strategy, and by $f_i$ the corresponding probability density function (PDF) if it exists. Let $\bm{F} = (F_1, \dots, F_n)$ be the mixed strategy profile and $\bm{F}_{-i}$ excludes creator $i$. Then $\bm{F}^* = (F_1^*, \ldots, F_n^*)$ is a \emph{mixed Nash equilibrium} (MNE) \cite{nash1950equilibrium} if for every creator $i \in [n]$ and any CDF $F_i$, $\mathbb{E}_{x \sim \bm{F}^*} [u_i(x_i, \bm{x}_{-i})] \geqslant \mathbb{E}_{x_i \sim F_i, \bm{x}_{-i} \sim \bm{F}_{-i}^*} [u_i(x_i, \bm{x}_{-i})]$. If $F_i^* = F^*$ for all $i \in [n]$, we say MNE is \emph{symmetric}.

\xhdr{Search engine's mechanisms} Deploying an AI Overview may reduce creator effort, lowering user welfare, and, through user attrition, ultimately harm search engine profit. We study two mechanism formats to address this: \emph{citation mechanism design} and \emph{compensation mechanism design}.

In the citation mechanism, the search engine cites in its AI Overview the organic search results as references. Suppose the AI Overview cites $m$ pages. Let $p_1^C, \ldots, p_n^C$ be the position biases of the $n$ organic results when they are not cited ($C$ corresponds to Type C in Table~\ref{tab:position-bias-no-ai-overview}), and let $q_1, \ldots, q_m$ be the position biases of the $m$ positions (or slots) within the AI Overview. Slot $j \in [m]$ cites the $i$-th organic result with probability $r_{ji} \in [0, 1]$, where $\sum_{i = 1}^n r_{ji} = 1$. Let $\bm{p} = (p_1, \ldots, p_n)$ as the position bias vector  for organic results after applying the citation mechanism, then for $i \in [n]$, we have $p_i = p_i^C + \sum_{j = 1}^m r_{ji} q_j$. We require the citation mechanism to maintain the monotonicity $p_1 \geqslant \cdots \geqslant p_n$. Example~\ref{ex:position-bias-monotone} provides justification for this requirement. Ignoring ties, the utility of the creator at rank $i$ is $x_{(i)} p_i - g_{(i)}(x_{(i)})$. Denote the feasible set of $\bm{p}$ as $\mathcal{P}$. Indeed, determining this set requires specifying $mn$ parameters, which must satisfy not only basic probability axioms but also the monotonicity condition of $\bm{p}$. Consequently, the structure $\mathcal{P}$ is highly complex. Note that $p_i^C$ and $q_j$ are estimated from user behavior, while the platform controls the citation probabilities $r_{ji}$ and thereby induces a position bias vector $\bm{p}$ within the feasible set $\mathcal{P}$.

In the compensation mechanism, the search engine rewards creators according to the attractiveness of their pages (equivalently, their exerted effort) \cite{tiktok2025creator}. Essentially, the compensation redesigns the position bias. Let $c_i$ be the compensation per effort level for rank $i$, with $c_1 \geqslant \cdots \geqslant c_n \geqslant 0$, and let $\bm{c} = (c_1, \dots, c_n)$. Ignoring ties, the creator at rank $i$ has utility $x_{(i)}(p_i + c_i) - g_{(i)}(x_{(i)})$.


\xhdr{Search engine's objective} The search engine aims to maximize its profit, defined as the revenue generated from user traffic minus the total compensation paid to creators. Since user traffic is positively correlated with user welfare, we model revenue as proportional to user welfare. Welfare has two parts: (a) \emph{Web pages}: users examine rank $i$ with probability $p_i$ and obtain utility proportional to quality $x_{(i)}$, giving welfare $\sum_{i = 1}^n \mathbb{E}[x_{(i)}]p_i$, where the expectation is taken over the mixed strategies of effort levels; (b) \emph{AI Overview}: let $p_0$ be the position bias of the AI Overview, then users examine it with probability $p_0$ and obtain utility proportional to its quality, which depends on average expected effort of human creators. Let $\mu = \frac{1}{n} \sum_{i = 1}^n \mathbb{E}[x_i]$. The AI Overview quality is $h(\mu)$, where $h: [0, 1] \to [0, 1]$ is continuously differentiable, strictly increasing, and concave. The concavity assumption aligns with the idea of scaling laws~\cite{kaplan2020scaling} and is widely adopted in recent related works \cite{ai2025genai,yao2024human,esmaeili2024strategize}.

Note that the profit should subtract the total compensation cost $\sum_{i = 1}^n \mathbb{E}[x_{(i)}] c_i$. Then the citation and compensation mechanism design problem can be formalized as the following optimization problem: 
\begin{equation} \label{eq:mechanism-design}
    \begin{aligned}
        \max_{\bm{p} \in \mathcal{P}, \bm{c}} \quad & \alpha\left[\sum_{i=1}^n \mathbb{E}[x_{(i)}] p_i +  h \left(\mu\right) p_0\right] - \sum_{i=1}^n \mathbb{E}[x_{(i)}] c_i \\
        \text{s.t.} \quad & p_1 \geqslant p_2 \geqslant \dots \geqslant p_n > 0, \\
        & c_1 \geqslant c_2 \geqslant \dots \geqslant c_n \geqslant 0.
    \end{aligned}
\end{equation}
Here, $\alpha > 0$ is a constant that transforms user welfare into search engine revenue, capturing search engine profitability. We denote the objective function hereafter as $W(\bm{p}, \bm{c})$.

In the following sections, we first analyze Problem~\eqref{eq:mechanism-design} under the symmetric setting where all creators share the same cost parameter (Section~\ref{sec:symmetric}), then extend the analysis to the asymmetric setting with heterogeneous cost parameters, in particular focusing on the case of binary distinct cost types (Section~\ref{sec:asymmetric}). For notational clarity, we denote the game with symmetric costs by $\mathcal{G}^{(1)}$, the game with binary cost types by $\mathcal{G}^{(2)}$, the general game with $n$ cost types by $\mathcal{G}^{(n)}$.
\section{Mechanisms for Symmetric Creators} \label{sec:symmetric}

We start our theoretical discussion with a simple yet fundamental setting, i.e., competition under symmetric creator costs (game $\mathcal{G}^{(1)}$). Specifically, assume $\gamma_1 = \gamma_2 = \cdots = \gamma_n \equiv \gamma$. This scenario commonly arises in highly competitive domains, such as news content creation, where major news outlets often face similar or identical production costs.

\subsection{Equilibrium Properties} \label{subsec:sym-equilibrium}

We first characterize equilibrium properties. The main theorem establishes the uniqueness and symmetry of the equilibrium. Proofs are in Appendix~\ref{sec:ap-symmetric}.

\begin{theorem} \label{thm:sym-equ-property}
    Game $\mathcal{G}^{(1)}$ admits a unique symmetric mixed Nash equilibrium.
\end{theorem}

Let $F$ be the common equilibrium CDF and $u$ the equilibrium utility. Every effort level in the support of $F$ must yield the same expected utility:
\begin{equation} \label{eq:indifference-sym}
    x \sum_{i=1}^n \binom{n-1}{i-1} [F(x)]^{n-i}[1-F(x)]^{i-1}p_i
    - \gamma g(x) = u,
    \qquad x \in \operatorname{supp}(F).
\end{equation}
Here, the binomial term is the probability of ranking $i$-th when the other $n-1$ creators independently draw their effort from $F$.

The main idea of the proof is as follows. We first show that MNE exists in Appendix~\ref{sec:ap-existence} based on the seminal results by \cite{reny1999existence}. Then we show that all creators attain the same equilibrium utility, and prove that the equilibrium CDF is continuous with no mass points, thus ruling out the existence of pure Nash equilibrium. Next, we prove that the support and the CDF are identical across all creators in any MNE, therefore establishing the symmetry. Finally, we prove the uniqueness follows from the indifference condition \eqref{eq:indifference-sym}. In addition to Theorem~\ref{thm:sym-equ-property}, we provide more equilibrium properties in Appendix~\ref{sec:ap-symmetric} and describe a numerical procedure to compute MNE in Appendix~\ref{subsubsec:binary-type-algorithm}.

Theorem~\ref{thm:sym-equ-property} is consistent with real-world creator competition in search engines. For example, major news outlets (e.g., CNN, BBC) typically have similar production costs. They distribute effort differentially across topics, so their ranks differ across queries, as predicted by the symmetric MNE.

Note that we prove in Theorem~\ref{thm:sym-equ-property} that MNE must be symmetric, whereas prior works on contests with rank-based rewards only establish the existence of symmetric equilibria without ruling out asymmetric ones \cite{jagadeesan2023supply,golrezaei2025contest}. In fact, multiple equilibria have been found in some seminal works \cite{barut1998symmetric,xiao2016asymmetric}. The key difference is that our rank rewards are strictly decreasing due to the position bias effect, while prior works may assign equal rewards to multiple ranks. Specifically, when the least reward is shared by multiple ranks, some creators may find it optimal to exert minimal effort to tie with others for the least reward, leading to the existence of asymmetric equilibria.

\subsection{Profit-Maximizing Mechanisms} \label{subsec:md-sym}

In this section, we solve the mechanism design problem \eqref{eq:mechanism-design} for maximizing search engine's profit. Due to the complexity of joint optimization over both $\bm{p}$ and $\bm{c}$, and the intricate design space $\mathcal{P}$, we first optimize over $\bm{c}$ while treating $\bm{p}$ as given. The proofs in this subsection are in Appendix~\ref{sec:ap-sym-proofs}.

Throughout this subsection we assume $g(x) = x^{\beta}$ with $\beta \geqslant 2$, which not only guarantees the convexity of $g(x) / x$, a property essential for our analysis, but also reflects practical intuition: as $x$ increases, indicating a stronger user preference for the content, the creator’s marginal cost should rise at a sufficiently rapid rate. The following theorem characterizes the structure of optimal $\bm{c}$ for problem \eqref{eq:mechanism-design}. 



\begin{theorem} \label{thm:compensation-md-sym}
   In the symmetric setting, if $c_n \leqslant \alpha p_n$, then there exists an optimal solution to Problem~\eqref{eq:mechanism-design} such that $c_1 = c_2 = \cdots = c_{n - 1}$.
\end{theorem}

We remind readers that $\alpha > 0$ here is the amount of revenue the platform gains from one unit of user welfare. In Appendix~\ref{subsec:optimal-cn}, we demonstrate that the optimal value of $c_n$ is zero across a wide range of plausible parameter settings, thereby justifying the imposed condition $c_n \leqslant \alpha p_n$.

Theorem \ref{thm:compensation-md-sym} generalizes a key characterization result of \cite{golrezaei2025contest}, which studies contests in recommender systems. Specifically, if we set $u = 0$, $\gamma = 1$ and $g(x) = x^{\beta + 1}$, then this special case corresponds to the setting of \cite{golrezaei2025contest}.  Like theirs, our proof also first reformulates the problem via Lemma~\ref{lem:compensation-md-reform}, and then shows the Schur-concavity (Definition~\ref{def:schur-concavity}) of the objective function with respect to the first $n - 1$ components of $\bm{c}$, implying that equalizing $c_1, \dots, c_{n - 1}$ improves the objective. However, the extra challenge in our more general setting is the difficulty of establishing  the Schur-concavity since the presence of $u > 0$ makes the analysis of the function $(u + x^{\beta}) / x$ more delicate, without a closed-form expression of its inverse  function any more, which necessitates  more involved analysis. In contrast, establishing Schur-concavity in \cite{golrezaei2025contest} is relatively straightforward.

Theorem~\ref{thm:compensation-md-sym} significantly narrows the design space of compensation mechanisms for search engines. Appendix~\ref{subsec:grid-search-details} describes how this theorem can be leveraged to solve the compensation mechanism design problem efficiently via a one-dimensional search on the value of $c^* := c_1 = \cdots = c_{n - 1}$.

Now we consider the problem of jointly optimizing over both $\bm{p}$ and $\bm{c}$. Unfortunately, the intricate design space $\mathcal{P}$ introduces significant analytical challenges. We therefore propose a feasible citation mechanism and prove it strictly improves the objective.

Let $p_1^B, \ldots, p_n^B$ be the position biases of the $n$ organic results when an AI Overview is present but includes no citations ($B$ corresponds to Type B in Table~\ref{tab:position-bias-no-ai-overview}). Let $\Delta p_i^B$ be the increase in position bias when the $i$-th ranked page is cited and then ranked uniformly at random within the AI Overview. Here $\Delta p_i^B$ captures both the reduction in organic result position bias due to the introduction of citations in the AI Overview and the gain from being cited.

Experiments in Section~\ref{subsec:data-collection-position-bias} indicate that $\Delta p_i^B > 0$ for all $i \in [n]$ in realistic scenarios. This motivates the \textbf{Uniform-But-Last (UBL)} citation mechanism, which equalizes the expected position bias increase across the first $n - 1$ pages by controlling citation probabilities. Let $q_i$ be the probability that the $i$-th ranked page is cited, and assume that cited pages are ranked uniformly at random within the overview. UBL is defined by
\begin{equation} \label{eq:UBL}
    \begin{cases}
        q_1 \Delta p_1^B = q_2 \Delta p_2^B = \cdots = q_{n - 1} \Delta p_{n - 1}^B, \\
        \sum_{i = 1}^{n - 1} q_i = m, q_1, q_2, \ldots, q_{n - 1} \in [0, 1], q_n = 0,
    \end{cases}
\end{equation}

where $m$ is the number of slots for citations in the AI Overview. This system clearly admits a unique solution, yielding $p \in \mathcal{P}$ that satisfies $p_1 \geqslant p_2 \geqslant \cdots \geqslant p_n > 0$. Recall $W(\bm{p}, \bm{c})$ denote the objective function in Problem \eqref{eq:mechanism-design}. The following theorem shows that UBL improves the objective value.


\begin{proposition} \label{prop:UBL-inc-sym} 
    Denote $\bm{p}^{\text{UBL}}$ as the position bias vector induced by UBL, and $\bm{p}^B$ as the position bias vector when there is no citation in the AI Overview. Then $\max_{\bm{c}} W(\bm{p}^{\text{UBL}}, \bm{c}) > \max_{\bm{c}} W(\bm{p}^B, \bm{c})$.
\end{proposition}

The intuition behind this proposition is natural: while both citation and compensation mechanisms adjust position biases, citations strengthen incentives without direct compensation payments. Appendix~\ref{subsec:ap-exp-p} provides numerical evidence that UBL consistently performs close to the optimal mechanism under plausible parameter settings, offering a simple yet effective practical solution.

\section{Mechanisms for Asymmetric Creators with Binary Distinct Types} \label{sec:asymmetric}

In this section, we consider the asymmetric case with binary creator types (game $\mathcal{G}^{(2)}$): high-ability (type $H$) creators with cost parameter $\gamma_H$ and low-ability (type $L$) creators with cost parameter $\gamma_L$, where $\gamma_H < \gamma_L$. Let $n_H$ and $n_L$ be the numbers of type-$H$ and type-$L$ creators with $n_H + n_L = n$. Without loss of generality, we label creators $1, \dots, n_H$ as type $H$ and creators $n_H + 1, \dots, n$ as type $L$. We focus on binary types since it captures common settings where experts and non-experts coexist (e.g., specialized knowledge creation) and yield different cost structures, while allowing analysis that remains tractable with valuable insights. Extending to more cost types becomes prohibitively complex (see Appendix~\ref{subsec:ap-asym-three}).

\subsection{Equilibrium Properties} \label{subsec:asym-two}

Let $[\underline{x}_H, \overline{x}_H]$ and $[\underline{x}_L, \overline{x}_L]$ be the equilibrium supports of type-$H$ and type-$L$ creators. We call an equilibrium \emph{hybrid} if the supports of the two types overlap, and \emph{separated} otherwise. Appendix~\ref{subsec:ap-asym-general} establishes general properties under arbitrary cost structures (the game $\mathcal{G}^{(n)}$) and provides a high-level characterization for game $\mathcal{G}^{(2)}$: (i) creators of the same type use identical strategies; (ii) type-$H$ creators typically exert higher effort due to their lower cost; (iii) the support of type-$H$ creators is the continuous interval $[\underline{x}_H, \overline{x}_H]$. We next summarize the main equilibrium properties of game $\mathcal{G}^{(2)}$ (Theorem~\ref{thm:asym-equ}). We first define the \emph{pseudo strategy} \cite{xiao2016asymmetric}.




\begin{definition}[pseudo strategy]
    A function $F_H^u: \mathbb{R} \to [0,1]$ is called the pseudo strategy of type-$H$ creator yielding utility $u$ if it satisfies the following equation:
    \begin{equation} \label{eq:pseudo-strategy}
        x \left(\sum_{i = 1}^{n_H} \binom{n_H - 1}{i - 1} [F_H^u(x)]^{n_H - i} [1 - F_H^u(x)]^{i - 1} p_i\right) - \gamma_H g(x) = u.
    \end{equation}
\end{definition}

Intuitively, the pseudo strategy $F_H^u$ is the strategy type-$H$ creators would use to attain utility $u$ (not necessarily the equilibrium utility) if type-$L$ creators were absent. Let $\left[\underline{x}_H^u, \overline{x}_H^u\right]$ be the support of $F_H^u$, $u_L^u$ be the utility type-$L$ creators obtain by the best response against $F_H^u$ on $\left[\underline{x}_H^u, \overline{x}_H^u\right]$, and $\overline{u}_H = \max_x x p_{n_H} - \gamma_H g(x)$ be the maximum utility type-$H$ creators can achieve. We can then state the main theorem.

\begin{theorem} \label{thm:asym-equ}
    The following properties hold in game $\mathcal{G}^{(2)}$:
    \begin{enumerate}
        \item The equilibrium utilities of type-$L$ and type-$H$ creators ($u_L$ and $u_H$) are unique.
        \item A separated equilibrium exists if and only if one of the following two conditions holds: (1) $u_L^{\overline{u}_H} \leqslant u_L$; (2) there exists a $u$ such that $u_L^u = u_L$ and $u_L^u$ is attained at $\underline{x}_H^u$; moreover, if a separated equilibrium exists, it is the unique equilibrium.
    \end{enumerate}
\end{theorem}

Proofs are in Appendix~\ref{subsec:ap-asym-two}. The uniqueness of $u_L$ extends the symmetric case; then we show that $u_L^u$ is strictly increasing in $u$ and $u_L^{u_H} = u_L$, which implies the uniqueness of $u_H$ and enables computing $u_H$ via binary search. The separation conditions in the second part of the theorem formalize that separation occurs precisely when a type-$L$ creator cannot gain by deviating to the type-$H$ support due to the higher cost. Appendix~\ref{subsec:ap-asym-two} also refines the structure of separated and hybrid equilibria and provides an equilibrium computation algorithm. Our analysis generalizes techniques from \cite{xiao2016asymmetric} and identifies the separating equilibrium specific to our multiplicative payoff structure.

In practice, separated equilibria often arise in specialized domains with large cost gaps (e.g., advanced math queries where Wikipedia and Math Stack Exchange dominate). Hybrid equilibria are consistent with general search contexts (e.g., queries on daily life), where high-ability creators cover a broad range of topics while low-ability creators concentrate on a subset, yielding overlapping supports.

\subsection{Profit-Maximizing Mechanisms} \label{subsec:asym-mechanism}

We then study problem~\eqref{eq:mechanism-design} in the binary-type setting. We first consider the case where the equilibrium remains separated both before and after implementing the mechanism. In this case, the design reduces to applying the symmetric analysis (Section~\ref{subsec:md-sym}) to each type separately. Consistent with Section~\ref{subsec:md-sym}, we assume $g(x) = x^{\beta}$ with $\beta \geqslant 2$ throughout this subsection.

We begin with the compensation mechanism. Because the strategies of type-$H$ and type-$L$ do not overlap, the compensation design decouples across types. We obtain the following structural characterization. Proofs of this subsection are in Appendix~\ref{sec:ap-asym-proofs}.

\begin{proposition} \label{prop:asym-separated-compensation-md}
    In the binary-type setting, assume that the equilibrium is separated both before and after implementing the mechanism. If $c_n \leqslant \alpha p_n$ and $c_{n_H} \leqslant \alpha p_{n_H}$, then there exists an optimal solution to Problem~\eqref{eq:mechanism-design} such that $c_1 = \cdots = c_{n_H - 1}$ and $c_{n_H + 1} = \cdots = c_{n - 1}$.
\end{proposition}

This follows directly from Theorem~\ref{thm:compensation-md-sym} and the separated equilibrium structure. Moreover, as discussed in Section~\ref{subsec:optimal-cn}, if we ignore the monotonicity constraint, the optimal compensation (under plausible parameters) satisfies $c_1 = \cdots = c_{n_H - 1} \geqslant c_{n_H} = 0$ and $c_{n_H + 1}  = \cdots = c_{n - 1} \geqslant c_n = 0$; enforcing monotonicity requires increasing $c_{n_H}$ to the smallest feasible value $c_{n_H} = c_{n_H + 1}$. Accordingly, we can parameterize any potential optimal compensation vector as
\begin{equation} \label{eq:asym-optimal-compensation-vector}
    \bm{c} = \bigl(\underbrace{c_H, \ldots, c_H}_{n_H - 1},
           \underbrace{c_L, \ldots, c_L}_{n_L}, 0\bigr).
\end{equation}

\begin{remark} \label{remark:asym-multi-type}
    This argument extends to any number of types: if equilibrium stays separated before and after the mechanism, an optimal compensation vector can be taken piecewise-constant by type.
\end{remark}

Then we consider the citation mechanism. Motivated by \eqref{eq:asym-optimal-compensation-vector} and the UBL mechanism (Section~\ref{subsec:md-sym}), we propose the \textbf{Two-Part-But-Last (TPBL)} citation mechanism for any compensation vector of the form \eqref{eq:asym-optimal-compensation-vector}. Formally, given $\bm{c}$, TPBL is defined by ($\Delta p_i^B$ and $q_i$ are defined as equation~\eqref{eq:UBL}):
\[ \begin{cases}
    q_1 \Delta p_1^B = \cdots = q_{n_H - 1} \Delta p_{n_H - 1}^B, q_{n_H} \Delta p_{n_H}^B = \cdots = q_{n - 1} \Delta p_{n - 1}^B, \\
    (q_1 \Delta p_1^B) / (q_{n_H} \Delta p_{n_H}^B) = c_H / c_L, \sum_{i = 1}^{n - 1} q_i = m, q_1, q_2, \ldots, q_{n - 1} \in [0, 1], q_n = 0.
\end{cases} \]

This system admits a unique solution, yielding $\bm{p} \in \mathcal{P}$ that satisfies $p_1 \geqslant \cdots \geqslant p_n > 0$. Then the following proposition shows TPBL improves the objective.

\begin{proposition} \label{prop:TPBL-inc-asym}
    Denote $\bm{p}^{\text{TPBL}}$ as the position bias vector induced by TPBL, and $\bm{p}^B$ as the position bias vector when there is no citation in the AI Overview. If the equilibrium is separated before and after the mechanism, then $\max_{\bm{c}} W(\bm{p}^{\text{TPBL}}, \bm{c}) > \max_{\bm{c}} W(\bm{p}^B, \bm{c})$.
\end{proposition}

Indeed, TPBL is not intrinsically tied to the ratio $c_H/c_L$: it assigns position bias in a piecewise-constant manner by type, and the ratio can be further optimized. Experiments in Appendix~\ref{subsec:ap-exp-p} show that TPBL performs close to the optimal mechanism under plausible parameter settings.

If the equilibrium is hybrid either before or after implementing the mechanism, analysis becomes considerably more complex due to overlapping supports. This case lies between the symmetric case ($c_H = c_L$) and the separated case (large $c_H - c_L$). Theorem~\ref{thm:compensation-md-sym} and Proposition~\ref{prop:asym-separated-compensation-md} characterize these boundary cases, motivating the following conjecture for the hybrid regime.

\begin{conjecture} \label{conj:asym-hybrid-compensation-md}
    In the binary-type setting, suppose the equilibrium is hybrid either before or after implementing the mechanism. If $c_n \leqslant \alpha p_n$ and $c_{n_H} \leqslant \alpha p_{n_H}$, then there exists an optimal solution to Problem~\eqref{eq:mechanism-design} such that $c_1 = \cdots = c_{n_H - 1}$ and $c_{n_H + 1} = \cdots = c_{n - 1}$.
\end{conjecture}

We evaluate this conjecture experimentally in Section~\ref{subsec:ap-exp-c} under a variety of plausible parameters. Under the conjectured compensation structure, TPBL can be applied analogously in the hybrid case.

\begin{remark}[Separation before and after TPBL]
    Let $v = p + c$ and let $g(x) = x^\beta$, then the equilibrium is separated whenever $\underline{x}_H v_{n_H} - \gamma_L\underline{x}_H^\beta \leqslant u_L$. This provides a sufficient condition for determining whether TPBL preserves the equilibrium type without explicitly computing the equilibrium. Equivalently, letting $r = \gamma_L / \gamma_H$, the sufficient condition can be written as $(\beta - r)r^{\frac{1}{\beta - 1}} \left(\frac{v_{n_H}}{v_n}\right)^{\frac{\beta}{\beta - 1}} \leqslant \beta - 1$. In particular, the simple condition $\gamma_L \geqslant \beta\gamma_H$ guarantees a separated equilibrium without any ex-post equilibrium check.
\end{remark}


\begin{remark}[Binary creator types] \label{remark:binary-type}
    The binary-type model should be viewed as a coarse-grained representation of heterogeneous creator costs, which is natural in search environments: platforms typically group creators into a small number of categories, such as high-ability and low-ability groups, rather than estimate precise individual costs. Our setting also provides a useful bridge from homogeneous creator models to richer heterogeneous environments: the piecewise-constant compensation structure extends directly to multiple cost categories, as noted in Remark~\ref{remark:asym-multi-type}. In addition, the binary-type analysis identifies the equilibrium structure and type-dependent incentive structure that are absent from single-type models~\cite{hron2023modeling,jagadeesan2023supply,golrezaei2025contest}. Overall, the two-type model captures the central heterogeneity feature needed for our analysis, while keeping the equilibrium and mechanism design problems analytically tractable.
\end{remark}

\begin{remark}[Implications for creators]
    When $p_n + c_n$ remains unchanged, the mechanisms leave equilibrium utility unchanged for symmetric creators and for type-$L$ creators in the binary-type setting, since their utility depends only on this reward and their cost function (Propositions~\ref{prop:symmetric-equilibrium-properties} and~\ref{prop:asym-unique-uH}). For equilibria that remain strictly separated, $u_H = \max_x\{xv_{n_H} - \gamma_H g(x)\}$ increases with $v_{n_H}$. For equilibria that remain hybrid, the common upper support endpoint $\bar{x}$ satisfies $u_L = \bar{x}v_1 - \gamma_L g(\bar{x})$. Increasing $v_1$ raises $\bar{x}$ and hence $u_H = u_L + (\gamma_L - \gamma_H)g(\bar{x})$.
\end{remark}

\section{Experiments} \label{sec:experiments}

In this section, we conducted a real-world user click experiment to quantify how AI Overviews and citations affect position biases, and then we analyze the profit impact of introducing AI Overviews based on the real click data.

\subsection{Data Collection and Position Bias Estimation} \label{subsec:data-collection-position-bias}

\begin{table}[t]
\centering
\small
\setlength{\tabcolsep}{5pt}
\caption{Estimated position biases of organic search results (95\% CI)}
\begin{tabular}{lccccc}
\hline
Type & page1 & page2 & page3 & page4 & page5 \\
\hline
A & $0.9911\pm0.0088$ & $0.9473\pm0.0214$ & $0.7173\pm0.0434$ & $0.6792\pm0.0446$ & $0.4170\pm0.0478$ \\
B & $0.7557\pm0.0417$ & $0.6502\pm0.0471$ & $0.3974\pm0.0477$ & $0.3661\pm0.0462$ & $0.2280\pm0.0405$ \\
C & $0.4801\pm0.0494$ & $0.4328\pm0.0487$ & $0.3307\pm0.0465$ & $0.3209\pm0.0458$ & $0.1605\pm0.0355$ \\
\end{tabular}
\begin{tabular}{lccccc}
\hline
Type & page6 & page7 & page8 & page9 & page10 \\
\hline
A & $0.2294\pm0.0412$ & $0.2116\pm0.0398$ & $0.1742\pm0.0360$ & $0.1485\pm0.0342$ & $0.1421\pm0.0337$ \\
B & $0.1572\pm0.0354$ & $0.1519\pm0.0351$ & $0.1430\pm0.0347$ & $0.1392\pm0.0336$ & $0.1158\pm0.0309$ \\
C & $0.1273\pm0.0335$ & $0.1203\pm0.0321$ & $0.0787\pm0.0265$ & $0.0762\pm0.0259$ & $0.0529\pm0.0217$ \\
\hline
\end{tabular}
\label{tab:position-bias-no-ai-overview}
\end{table}

\begin{table}[t]
\centering
\small
\caption{Estimated position biases of pages included in the AI Overview (95\% CI)}
\begin{tabular}{lccccc}
\hline
AI Overview & page1 & page2 & page3 & page4 \\
\hline
$0.9888\pm0.0111$ & $0.7289\pm0.0429$ & $0.4417\pm0.0473$ & $0.3221\pm0.0460$ & $0.1480\pm0.0354$ \\
\hline
\end{tabular}
\label{tab:position-bias-ai-overview}
\end{table}

We built a custom platform simulating a search engine and recruited 50 participants to complete search tasks. \footnote{The code of our user click experiment platform is available at \href{https://anonymous.4open.science/r/search-ai-2404}{https://anonymous.4open.science/r/search-ai-2404}.} We designed 24 search queries with different topics and search intents (see Appendix~\ref{subsec:ap-search-queries}). Result pages and AI Overview answers were constructed from real Google results. The 24 queries were randomly and evenly assigned to three types of interfaces (see examples in Appendix~\ref{subsec:ap-search-interfaces}): (A) no AI Overview; (B) AI Overview without citations; and (C) AI Overview citing four relevant pages per query. Each participant completed all 24 tasks (8 per interface). We logged dwell-time and click behavior, then estimated position biases (with 95\% confidence intervals) for organic pages (Table~\ref{tab:position-bias-no-ai-overview}) and the AI Overview and its cited pages (Table~\ref{tab:position-bias-ai-overview}) following the methodology in \cite{chuklin2015click}. Details for computing position bias are in Appendix~\ref{subsec:ap-position-bias-calculation}. We also provide additional justification for our experimental design in Appendix~\ref{sec:ap-model}.

We focus on aggregate trends of position bias here and report additional analyses in Appendix~\ref{subsec:ap-more-data-analysis}. Comparing lines 1-2 of Table~\ref{tab:position-bias-no-ai-overview} shows that introducing an AI Overview substantially reduces the position biases of organic pages. Table~\ref{tab:position-bias-ai-overview} further indicates that the AI Overview itself attains a very high position bias, comparable to the top-ranked page in Type A, suggesting it captures a large share of user attention. Line 3 of Table~\ref{tab:position-bias-no-ai-overview} shows that adding citations further lowers the position bias of organic pages. Summing position biases yields 4.6577 (Type A), 3.1045 (Type B), and 3.8211 (Type C, including cited pages), implying that citations shift attention toward cited pages and increase the overall probability that users examine a page. These findings also provides empirical support for the positive position bias increments that motivate the UBL (Section~\ref{subsec:md-sym}) and TPBL (Section~\ref{subsec:asym-mechanism}) mechanisms. Two additional observations are noteworthy: (i) the fourth cited page already has a low position bias comparable to the lowest-ranked page without an AI Overview, suggesting that citing more than four pages may not further increase total position bias; (ii) even with citations, total position bias remains below the level without an AI Overview, confirming that the AI Overview continues to divert attention away from result pages.

\subsection{Analysis of Search Engine Profit} \label{subsec:experiment-profit}

We now examine how introducing an AI Overview affects search engine profit in the long term, and how our mechanisms affect long-term profit. We present the asymmetric case; the symmetric case is the special case $n_H = 10$ and is deferred to Appendix~\ref{subsec:ap-exp-profit}. Experimental parameters and practical interpretations are summarized in Appendix~\ref{subsec:ap-exp-para}.


\xhdr{Long-term effects without mechanism design} Without compensation, profit coincides with user welfare. We sweep $\beta \in \{2, 3, 4, 5, 6, 7, 8\}$ and $h(x) \in \{x, \sqrt{x}, \sqrt[3]{x}, \sqrt[4]{x}, \sqrt[5]{x}\}$, and fix $\gamma_H = 1, \gamma_L = 2, n_H = 7$ since they have limited impact on the trend (see Appendix~\ref{subsec:ap-exp-profit}). We compare welfare without an AI Overview ($W_1$) to welfare with an AI Overview that cites all $10$ pages with equal probability and ranks cited pages uniformly at random within the overview ($W_2$), which serves as a realistic baseline, since AI Overviews typically include citations now. Figure~\ref{fig:long-asym} reports the ratio $W_2 / W_1$ across tested parameters with colored shaded areas standing for standard error induced by uncertainty in the estimated position biases. It shows that $W_2 / W_1 < 1$ for all tested configurations, confirming that introducing an AI Overview without incentive mechanisms reduces long-term profit. Note that $W_2 / W_1$ rises as $h$ becomes more concave and as $\beta$ increases. The former directly increases the AI Overview utility, while the latter makes $g$ downward and be more convex on $[0, 1]$, diminishing the influence of the decrease in position biases on creation efforts, thus elevating $W_2 / W_1$.

\begin{figure*}[t]
    \centering
    \begin{minipage}[t]{0.27\textwidth}
        \centering
        \includegraphics[width=0.861\linewidth]{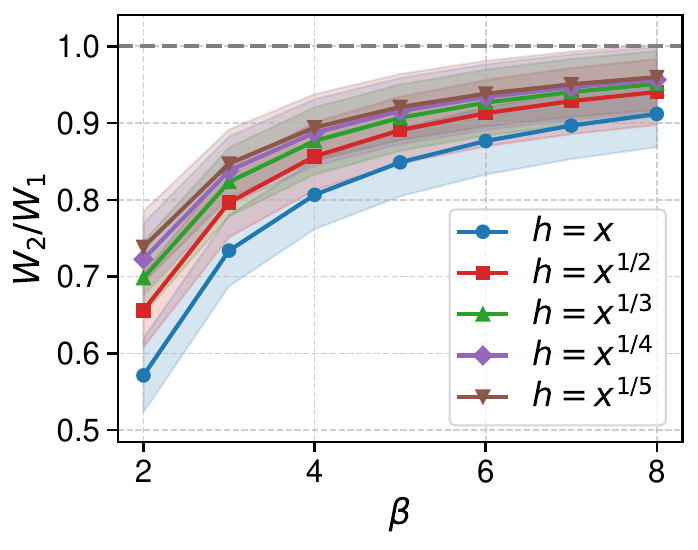}
        \captionof{figure}{Long-term profit without mechanism design}
        \label{fig:long-asym}
    \end{minipage}
    \begin{minipage}[t]{0.72\textwidth}
        \centering
        \includegraphics[width=0.323\linewidth]{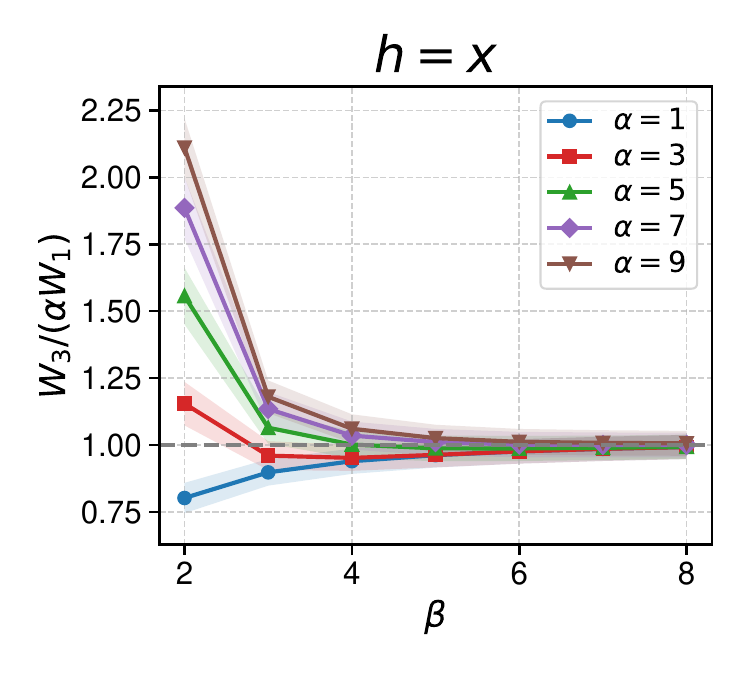}\hspace{-0.01\linewidth}%
        \includegraphics[width=0.323\linewidth]{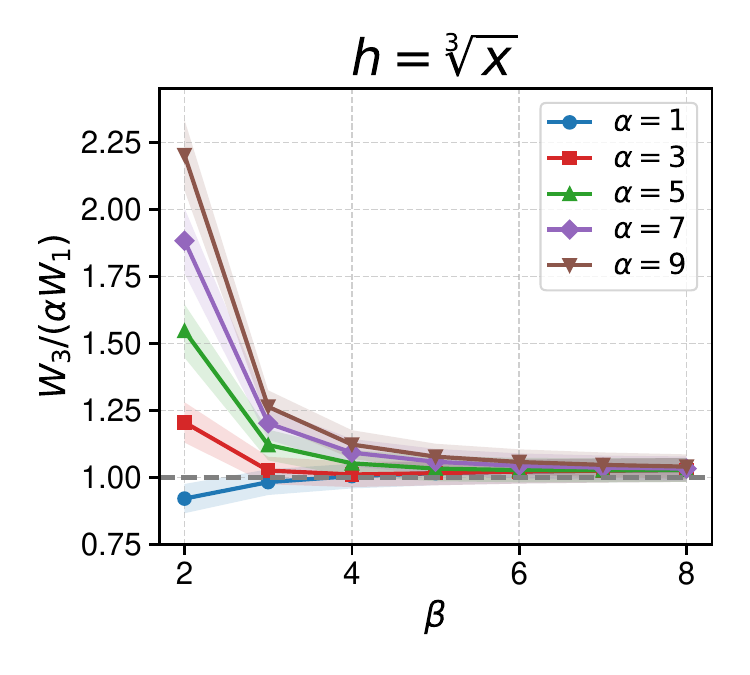}\hspace{-0.01\linewidth}%
        \includegraphics[width=0.323\linewidth]{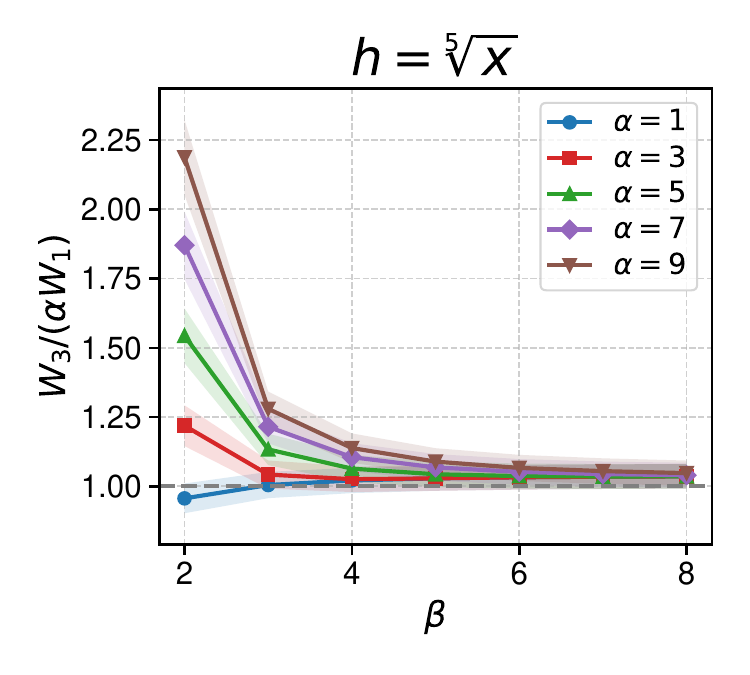}
        \captionof{figure}{Long-term profit under mechanism design}
        \label{fig:TPBL-fig-ratio}
    \end{minipage}
\end{figure*}

\xhdr{Long-term effects with mechanism design} We then introduce the citation mechanism (TPBL) and the compensation mechanism, and sweep $\beta \in \{2, 3, 4, 5, 6, 7, 8\}$, $\alpha \in \{1, 3, 5, 7, 9\}$, and $h(x) \in \{x, \sqrt[3]{x}, \sqrt[5]{x}\}$. Let $W_3$ be profit under the mechanisms. Figure~\ref{fig:TPBL-fig-ratio} reports $W_3 / (\alpha W_1)$ across tested parameters (dividing by $\alpha$ since $W_1$ omits this factor). Comparing Figure~\ref{fig:long-asym} and Figure~\ref{fig:TPBL-fig-ratio}, profit improves for all parameter settings after introducing the mechanisms, and $W_3 / (\alpha W_1) > 1$ holds in most configurations. The overall trend is weakly sensitive to $h$ but strongly shaped by $\alpha$ and $\beta$: larger $\alpha$ increases $W_3 / (\alpha W_1)$ because higher profitability enhances the return on monetary compensation, thereby improving the efficacy of the compensation mechanism. This also explains why $W_3 / (\alpha W_1)$ falls below 1 when $\alpha$ is small. When $\beta$ is large, $W_3 / (\alpha W_1)$ stays close to $W_2 / W_1$ because creators already exert high effort, which diminishes the influence of mechanism design on creator behavior. In particular, the optimal compensation becomes very small, limiting the effectiveness of the compensation mechanism. Conversely, smaller $\beta$ makes incentives more effective, yielding larger profit gains.

Overall, we can answer the main question of this study: deploying AI Overviews reduces long-term search engine profit due to diminished creator incentives. However, the proposed citation and compensation mechanisms enable search engines to increase long-term profit after deploying AI Overviews, especially when profitability $\alpha$ is high and creators face higher creation cost (lower $\beta$).

Practically, search engines should prioritize strong incentives for high creation cost queries (e.g., specialized or academic topics): under high $\alpha$, deploying such mechanisms is particularly valuable; under low $\alpha$, introducing AI Overviews for these queries should be approached cautiously to avoid long-term profit loss. For low creation cost queries (e.g., lifestyle or entertainment), compensation has limited marginal value, making improvements to the intrinsic quality of AI Overviews and the citation mechanism more effective.

\section{Conclusion and Discussion} \label{sec:conclusion}

In this study, we develop a creator competition model under keyword search, characterize the mixed Nash equilibria in both symmetric and binary-type asymmetric settings, design optimal compensation mechanisms and citation mechanisms that have provable improvement guarantees and perform close to the optimal mechanism across plausible parameters with piecewise-constant structural insights for search engine profit. Evaluations parameterized by real user click data indicate that although AI Overviews harm long-term search engine profit, appropriately designed mechanisms can increase long-term profit across realistic regimes. Our results provide a principled, tractable blueprint for sustaining a healthy search ecosystem in the GenAI era.

Two directions for future work are promising. From a theoretical perspective, it is interesting to further characterize equilibria under more general cost structures and derive the optimal structure of compensation mechanisms beyond the boundary cases examined in this paper. The modeling of AI Overviews' capability can be further refined, e.g., by considering the capability as a function of the quality of top-$K$ sources. Besides, a more general model could introduce query or topic relevance and allow creators to allocate effort across different queries or topics. For empirical aspects, it is urgent to study how the capability of AI Overviews affects position biases, enabling a more fine-grained mechanism design problem.

\onlineversion{
\bibliographystyle{plain}
}{
}
\bibliography{aioverview}

@misc{google_ai_overview,
  author = {Google},
  howpublished = {https://blog.google/products-and-platforms/products/search/generative-ai-google-search-may-2024/},
  title = {Generative AI in Search: Let Google do the searching for you},
  year = {2024}
}

@misc{bing_ai_overview,
  author = {Microsoft Bing},
  howpublished = {https://support.microsoft.com/en-us/topic/understanding-ai-overview-in-search-results-6b9497e1-bcda-42e4-879b-7b9710361651},
  title = {Understanding AI Overview in search results},
  year = {2025}
}

@article{li2025towards,
  title={Towards AI Search Paradigm},
  author={Li, Yuchen and Cai, Hengyi and Kong, Rui and Chen, Xinran and Chen, Jiamin and Yang, Jun and Zhang, Haojie and Li, Jiayi and Wu, Jiayi and Chen, Yiqun and others},
  journal={arXiv preprint arXiv:2506.17188},
  year={2025}
}

@misc{perplexity,
  author = {Perplexity},
  howpublished = {https://www.perplexity.ai/hub/blog/getting-started-with-perplexity},
  title = {Perplexity AI},
  year = {2024}
}

@misc{google_q1_2025_report,
  author = {Alphabet},
  howpublished = {http://s206.q4cdn.com/479360582/files/doc\_financi\\als/2025/q1/2025q1-alphabet-earnings-release.pdf},
  title = {Alphabet Q1 2025 Earnings Report},
  year = {2025}
}

@misc{chegg_q4_2024_report,
  author = {Chegg},
  howpublished = {https://investor.chegg.com/Press-Releases/press-release-details/2025/Chegg-Reports-2024-Fourth-Quarter-and-Full-Year-Financial-Results/default.aspx},
  title = {Chegg Reports 2024 Fourth Quarter and Full Year Financial Results},
  year = {2025}
}

@misc{gizmodo_rolling_stone_sue,
  author = {Gizmodo},
  howpublished = {https://gizmodo.com/rolling-stone-publisher-sues-google-over-ai-overview-summaries-2000658662},
  title = {Rolling Stone Publisher Sues Google Over AI Overview Summaries},
  year = {2025}
}

@inproceedings{craswell2008experimental,
  title={An experimental comparison of click position-bias models},
  author={Craswell, Nick and Zoeter, Onno and Taylor, Michael and Ramsey, Bill},
  booktitle={Proceedings of the 2008 international conference on web search and data mining},
  pages={87--94},
  year={2008}
}

@inproceedings{joachims2005accurately,
  title={Accurately interpreting clickthrough data as implicit feedback},
  author={Joachims, Thorsten and Granka, Laura and Pan, Bing and Hembrooke, Helene and Gay, Geri},
  booktitle={28th Annual International ACM SIGIR Conference on Research and Development in Information Retrieval, SIGIR 2005},
  pages={154--161},
  year={2005}
}

@book{chuklin2015click,
  title={Click models for web search},
  author={Chuklin, Aleksandr and Markov, Ilya and De Rijke, Maarten},
  year={2015},
  publisher={Springer Nature}
}

@article{robertson1977probability,
  title={The probability ranking principle in IR},
  author={Robertson, Stephen E},
  journal={Journal of documentation},
  volume={33},
  number={4},
  pages={294--304},
  year={1977},
  publisher={MCB UP Ltd}
}

@misc{tiktok2025creator,
  author = {TikTok},
  howpublished = {https://www.tiktok.com/creator-academy/en/article/RPM-understanding-the-four-key-factors},
  title = {Finding Creator Rewards Program success},
  year = {2025}
}

@article{kaplan2020scaling,
  title={Scaling laws for neural language models},
  author={Kaplan, Jared and McCandlish, Sam and Henighan, Tom and Brown, Tom B and Chess, Benjamin and Child, Rewon and Gray, Scott and Radford, Alec and Wu, Jeffrey and Amodei, Dario},
  journal={arXiv preprint arXiv:2001.08361},
  year={2020}
}

@article{nash1950equilibrium,
  title={Equilibrium points in n-person games},
  author={Nash Jr, John F},
  journal={Proceedings of the national academy of sciences},
  volume={36},
  number={1},
  pages={48--49},
  year={1950},
  publisher={national academy of sciences}
}

@article{reny1999existence,
  title={On the existence of pure and mixed strategy Nash equilibria in discontinuous games},
  author={Reny, Philip J},
  journal={Econometrica},
  volume={67},
  number={5},
  pages={1029--1056},
  year={1999},
  publisher={Wiley Online Library}
}

@article{ben2018game,
  title={A game-theoretic approach to recommendation systems with strategic content providers},
  author={Ben-Porat, Omer and Tennenholtz, Moshe},
  journal={Advances in Neural Information Processing Systems},
  volume={31},
  year={2018}
}

@article{ben2020content,
  title={Content provider dynamics and coordination in recommendation ecosystems},
  author={Ben-Porat, Omer and Rosenberg, Itay and Tennenholtz, Moshe},
  journal={Advances in Neural Information Processing Systems},
  volume={33},
  pages={18931--18941},
  year={2020}
}

@inproceedings{hron2023modeling,
  title={Modeling content creator incentives on algorithm-curated platforms},
  author={Hron, Jiri and Krauth, Karl and Jordan, Michael and Kilbertus, Niki and Dean, Sarah},
  booktitle={The Eleventh International Conference on Learning Representations},
  year={2023}
}

@inproceedings{yao2023bad,
  title={How Bad is Top-$ K $ Recommendation under Competing Content Creators?},
  author={Yao, Fan and Li, Chuanhao and Nekipelov, Denis and Wang, Hongning and Xu, Haifeng},
  booktitle={International Conference on Machine Learning},
  pages={39674--39701},
  year={2023},
  organization={PMLR}
}

@article{yao2023rethinking,
  title={Rethinking incentives in recommender systems: are monotone rewards always beneficial?},
  author={Yao, Fan and Li, Chuanhao and Sankararaman, Karthik Abinav and Liao, Yiming and Zhu, Yan and Wang, Qifan and Wang, Hongning and Xu, Haifeng},
  journal={Advances in Neural Information Processing Systems},
  volume={36},
  pages={74582--74601},
  year={2023}
}

@article{jagadeesan2023supply,
  title={Supply-side equilibria in recommender systems},
  author={Jagadeesan, Meena and Garg, Nikhil and Steinhardt, Jacob},
  journal={Advances in Neural Information Processing Systems},
  volume={36},
  pages={14597--14608},
  year={2023}
}

@inproceedings{yao2024human,
  title={Human vs. Generative AI in Content Creation Competition: Symbiosis or Conflict?},
  author={Yao, Fan and Li, Chuanhao and Nekipelov, Denis and Wang, Hongning and Xu, Haifeng},
  booktitle={International Conference on Machine Learning},
  pages={56885--56913},
  year={2024},
  organization={PMLR}
}

@article{esmaeili2024strategize,
  title={How to Strategize Human Content Creation in the Era of GenAI?},
  author={Esmaeili, Seyed A and Lim, Kevin and Bhawalkar, Kshipra and Feng, Zhe and Wang, Di and Xu, Haifeng},
  journal={arXiv preprint arXiv:2406.05187},
  year={2024}
}

@inproceedings{taitler2025braess,
  title={Braess’s paradox of generative ai},
  author={Taitler, Boaz and Ben-Porat, Omer},
  booktitle={Proceedings of the AAAI Conference on Artificial Intelligence},
  volume={39},
  pages={14139--14147},
  year={2025}
}

@inproceedings{taitler2025selective,
  title={Selective Response Strategies for GenAI},
  author={Taitler, Boaz and Ben-Porat, Omer},
  booktitle={Forty-second International Conference on Machine Learning},
  year={2025}
}

@article{keinan2025strategic,
  title={Strategic Content Creation in the Age of GenAI: To Share or Not to Share?},
  author={Keinan, Gur and Ben-Porat, Omer},
  journal={arXiv preprint arXiv:2505.16358},
  year={2025}
}

@article{taitler2025data,
  title={Data Sharing with a Generative AI Competitor},
  author={Taitler, Boaz and Madmon, Omer and Tennenholtz, Moshe and Ben-Porat, Omer},
  journal={arXiv preprint arXiv:2505.12386},
  year={2025}
}

@article{raghavan2024competition,
  title={Competition and diversity in generative ai},
  author={Raghavan, Manish},
  journal={arXiv preprint arXiv:2412.08610},
  year={2024}
}

@article{madmon2025search,
  title={The search for stability: Learning dynamics of strategic publishers with initial documents},
  author={Madmon, Omer and Pipano, Idan and Reinman, Itamar and Tennenholtz, Moshe},
  journal={Journal of Artificial Intelligence Research},
  volume={83},
  year={2025}
}

@article{ai2025genai,
  title={GenAI vs. Human Creators: Procurement Mechanism Design in Two-/Three-Layer Markets},
  author={Ai, Rui and Simchi-Levi, David and Xu, Haifeng},
  journal={arXiv preprint arXiv:2511.06559},
  year={2025}
}

@article{dechenaux2015survey,
  title={A survey of experimental research on contests, all-pay auctions and tournaments},
  author={Dechenaux, Emmanuel and Kovenock, Dan and Sheremeta, Roman M},
  journal={Experimental Economics},
  volume={18},
  number={4},
  pages={609--669},
  year={2015},
  publisher={Cambridge University Press \& Assessment}
}

@book{krishna2009auction,
  title={Auction theory},
  author={Krishna, Vijay},
  year={2009},
  publisher={Academic press}
}

@article{tullock1980efficient,
  title={Efficient rent seeking},
  author={Tullock, Gordon and others},
  journal={Toward a theory of the rent-seeking society},
  volume={97},
  pages={112},
  year={1980},
  publisher={Springer}
}

@book{vojnovic2015contest,
  title={Contest theory: Incentive mechanisms and ranking methods},
  author={Vojnovic, Milan},
  year={2015},
  publisher={Cambridge University Press}
}

@article{sisak2009multiple,
  title={Multiple-prize contests--the optimal allocation of prizes},
  author={Sisak, Dana},
  journal={Journal of Economic Surveys},
  volume={23},
  number={1},
  pages={82--114},
  year={2009},
  publisher={Wiley Online Library}
}

@article{baye1996all,
  title={The all-pay auction with complete information},
  author={Baye, Michael R and Kovenock, Dan and De Vries, Casper G},
  journal={Economic Theory},
  volume={8},
  number={2},
  pages={291--305},
  year={1996},
  publisher={Springer}
}

@article{barut1998symmetric,
  title={The symmetric multiple prize all-pay auction with complete information},
  author={Barut, Yasar and Kovenock, Dan},
  journal={European Journal of Political Economy},
  volume={14},
  number={4},
  pages={627--644},
  year={1998},
  publisher={Elsevier}
}

@article{bulow2006matching,
  title={Matching and price competition},
  author={Bulow, Jeremy and Levin, Jonathan},
  journal={American Economic Review},
  volume={96},
  number={3},
  pages={652--668},
  year={2006},
  publisher={American Economic Association}
}

@article{siegel2009all,
  title={All-pay contests},
  author={Siegel, Ron},
  journal={Econometrica},
  volume={77},
  number={1},
  pages={71--92},
  year={2009},
  publisher={Wiley Online Library}
}

@article{xiao2016asymmetric,
  title={Asymmetric all-pay contests with heterogeneous prizes},
  author={Xiao, Jun},
  journal={Journal of Economic Theory},
  volume={163},
  pages={178--221},
  year={2016},
  publisher={Elsevier}
}

@article{lazear1981rank,
  title={Rank-order tournaments as optimum labor contracts},
  author={Lazear, Edward P and Rosen, Sherwin},
  journal={Journal of political Economy},
  volume={89},
  number={5},
  pages={841--864},
  year={1981},
  publisher={The University of Chicago Press}
}

@article{glazer1988optimal,
  title={Optimal contests},
  author={Glazer, Amihai and Hassin, Refael},
  journal={Economic Inquiry},
  volume={26},
  number={1},
  pages={133--143},
  year={1988},
  publisher={Wiley Online Library}
}

@article{moldovanu2001optimal,
  title={The optimal allocation of prizes in contests},
  author={Moldovanu, Benny and Sela, Aner},
  journal={American Economic Review},
  volume={91},
  number={3},
  pages={542--558},
  year={2001},
  publisher={American Economic Association}
}

@article{golrezaei2025contest,
  title={The Contest Behind the Feed: Optimal Contest for Recommender Systems},
  author={Golrezaei, Negin and Hajiaghayi, MohammadTaghi and Shin, Suho},
  journal={Available at SSRN 5258940},
  year={2025}
}

@article{ando1989majorization,
  title={Majorization, doubly stochastic matrices, and comparison of eigenvalues},
  author={Ando, Tsuyoshi},
  journal={Linear Algebra and its Applications},
  volume={118},
  pages={163--248},
  year={1989},
  publisher={Elsevier}
}

@book{arnold1992ordinary,
  title={Ordinary differential equations},
  author={Arnold, Vladimir I},
  year={1992},
  publisher={Springer Science \& Business Media}
}

@book{rowley2018organizing,
  title={Organizing knowledge: Introduction to access to information: Introduction to access to information},
  author={Rowley, Jennifer E and Farrow, John},
  year={2018},
  publisher={Routledge}
}

\clearpage

\onlineversion{
\appendix
}{
\elecappendix

\medskip
}
\section{Additional Model and Experiment Justification} \label{sec:ap-model}

We provide additional justification for our model assumptions and experimental design choices.

\xhdr{Rank by Effort} Remark~\ref{remark:rank-by-effort} interprets effort as a query-specific composite measures that drive clicks, including quality, timeliness, and relevance. This choice is standard in creator competition models \cite{jagadeesan2023supply,yao2024human,keinan2025strategic,golrezaei2025contest}, and it is also operational: it ties effort directly to click outcomes (\cite{ai2025genai} adopts a similar assumption) rather than treating it as a latent object. Besides, higher-dimensional effort modeling does not provide more insight and renders the theoretical analysis intractable. More broadly, this modeling choice inherits from the contest design literature, where substantial empirical work shows that the effort-based abstractions match real behavior well \cite{dechenaux2015survey}.

The connection to industrial ranking pipelines is also close. Learning-to-rank systems combine many signals, while our effort variable represents those composite measures. Introducing stochastic ranking (such as learning-to-rank) would smooth expected payoffs, but as long as higher effective effort raises expected rank or click probability, noise mainly smooth incentives rather than reversing their direction. Therefore, it need not overturn the main qualitative conclusions. Besides, the equilibrium behaviors generated by our model admit natural real-world interpretations in both the symmetric and asymmetric settings; see Sections~\ref{subsec:sym-equilibrium} and~\ref{subsec:asym-two}. Therefore, the ``rank by effort'' assumption should be viewed as a deliberate modeling compromise between realism and tractability.

\xhdr{Maximum effort} In our model, the maximum effort level means that the created content attracts users with probability one, and achieving such certainty is naturally prohibitively costly. Formally, our model is equivalent to the following formulation. Creator $i$ chooses an underlying investment level $e_i \in [0, +\infty)$, which induces a click attractiveness $q(e_i)$, corresponding to $x_i \in [0,1]$ in our model. Here, $q: [0, +\infty) \to [0, 1)$ is continuous and strictly increasing, with $q(0) = 0$ and $\lim_{e_i \to +\infty}q(e_i) = 1$. Creators are ranked according to $q(e_i)$, and the utility of creator $i$ at rank $j$ is $q(e_i)p_j - \gamma_i g(q(e_i))$, then the change of variables $x_i = q(e_i)$ yields exactly the same mixed-strategy equilibria on $x_i \in [0,1)$ since $q$ is strictly increasing. The endpoint $x = 1$ therefore corresponds to the limiting case in which the underlying investment tends to infinity, and it is natural to assume that no creator makes an infinite investment. Under this interpretation, if the maximum effort level is prohibitively costly, effort levels sufficiently close to the maximum are also prohibitively costly.

\xhdr{Binary Creator Types} We analyze the asymmetric case under a binary-type model in Section~\ref{sec:asymmetric} since the binary-type model is a tractable and actionable way to capture asymmetry across creators (see Remark~\ref{remark:binary-type}).

\xhdr{Modeling of AI Overview Quality} The modeling of AI Overview quality as a concave function aligns with the idea of scaling laws~\cite{kaplan2020scaling} and is widely adopted in recent related works \cite{ai2025genai,yao2024human,esmaeili2024strategize}. The use of average expected effort $\mu$ as an aggregator captures the idea that AI Overview quality depends on the overall source ecosystem. Besides, if $h(\mu)$ were replaced by a function of the top-$K$ sources, the overall trend of impact of AI Overviews and our mechanisms under different parameters would remain. Thus, using $\mu$ provides a natural baseline that balances realism and tractability.

\xhdr{Interpretation of $\alpha$ and $\beta$} The parameters $\alpha$ and $\beta$ are query-level primitives rather than fixed profitability and cost measure. Search engines can estimate these parameters for different types of queries in real time from large-scale data and adjust compensation policies accordingly. This is why our experiments analyze strategy implications under different values of $\alpha$ and $\beta$ (e.g., when $\alpha$ is small, reducing AI Overview displays may be preferable). Thus, parameter heterogeneity across queries is not outside the scope of the model; it is one of the main reasons why the model is useful for platform design.

\xhdr{Sensitivity to the profit specification} First, the computational cost is independent of $\bm{p}$ and $\bm{c}$, therefore does not affect the optimal mechanism. Second, advertising revenue and dwell time are captured through parameter $\alpha$: higher user welfare can increase engagement, dwell time, and advertising exposure, which in turn raises platform revenue. One approximation in our model is the use of a constant $\alpha$ to describe the relationship between user welfare and advertising or other revenue. We believe this marginal conversion rate is likely to remain approximately stable over a relevant operating range. More generally, suppose that revenue is modeled as a nonlinear function of user welfare. Under diminishing marginal returns, this function would naturally be increasing and concave. Replacing the linear conversion by an increasing concave function preserves the Schur-concavity of the objective with respect to $\bm{c}$, and therefore does not change the structural conclusions of our mechanism design. Moreover, global concavity is stronger than necessary: it is sufficient that the resulting composite objective remain Schur-concave. Thus, moderate departures from concavity can also be accommodated.

\xhdr{Experimental Scale and Stability} The experiment is designed to identify the stable trend of how AI Overview and our mechanisms affect search engine profit under different parameters, rather than precise decrease or increase ratios. With 50 participants and 24 tasks per participant, we obtain data for 1200 queries, including 400 for each interface type, which is sufficient to estimate the overall trend. Moderate sampling noise does not materially affect the stable trend and the impact of different parameters on mechanism effectiveness (this can also be seen in our experimental results with error bars). A larger sample would refine the effect further, but the central comparative patterns and the implications for mechanism effectiveness are already clear and intuitive at the current scale.
\section{Existence of Mixed Nash Equilibrium} \label{sec:ap-existence}

The purpose of this section is to prove Theorem~\ref{thm:mix-existence}, which states that game $\mathcal{G}^{(n)}$ defined in Section~\ref{sec:model} admits a mixed Nash equilibrium.

\begin{theorem} \label{thm:mix-existence}
    Game $\mathcal{G}^{(n)}$ admits a mixed Nash equilibrium.
\end{theorem}

To prove the existence of mixed Nash equilibrium, we use the results from \cite{reny1999existence}. We first introduce some definitions.

A game $G = (X_i, u_i)_{i = 1}^n$ consists of a set of players $[n]$, each with a pure strategy space $X_i$ and a payoff function $u_i : X \to \mathbb{R}$, where $X = \times_{i=1}^n X_i$. If each $X_i$ is a compact Hausdorff space, then $G$ is called a compact Hausdorff game. The mixed extension of $G$ is defined as $\overline{G} = (M_i, u_i)_{i=1}^n$, where $M_i$ is the set of probability measures on the Borel subsets of $X_i$, and $u_i : M \to \mathbb{R}$ is defined as $u_i(\bm{\mu}) = \int_X u_i(\bm{x}) \textup{d}\bm{\mu}$, where $\bm{\mu} \in M = \times_{i=1}^n M_i$. If each $X_i$ is a metric space, then each $M_i$ can be endowed with the weak* topology, which can be metrized by the Lévy-Prokhorov metric.

The notion of payoff security is crucial in establishing the existence of mixed Nash equilibrium. To this end, we first introduce the definition of securing a payoff.

\begin{definition}[secure a payoff]
    Player $i$ can secure a payoff of $v$ at $\bm{\mu} \in M$ if there exists $\overline{\mu}_i \in M_i$, such that $u_i(\overline{\mu}_i, \bm{\mu}_{-i}') \geqslant v$ for all $\bm{\mu}_{-i}'$ in some open neighborhood of $\bm{\mu}_{-i}$.
\end{definition}

The following definition formalizes the notion of payoff security for a game.

\begin{definition}[payoff secure]
    Let $\overline{G} = (M_i, u_i)_{i = 1}^n$ be the mixed extension of $G = (X_i, u_i)_{i = 1}^n$. Then $\overline{G}$ is payoff secure if for every $\bm{\mu} \in M$ and every $\varepsilon > 0$, each player $i$ can secure a payoff of $u_i(\bm{\mu}) - \varepsilon$ at $\bm{\mu}$.
\end{definition}

Thus, payoff security requires that each player can approximately secure her payoff at any given strategy profile, even when the other players slightly deviate from their strategies. The seminal result by \cite{reny1999existence} establishes the existence of mixed Nash equilibrium in compact Hausdorff games that satisfy certain semi-continuity and payoff security conditions.

\begin{theorem} \label{thm:mix-existence-reny}
    (Corollary~5.2 of \cite{reny1999existence})
    Suppose that game $G = (X_i, u_i)_{i = 1}^n$ is a compact Hausdorff game. Then $G$ possesses a mixed Nash equilibrium if $\sum_{i = 1}^n u_i(\bm{x})$ is upper semicontinuous in $x$ on $X$ and the mixed extension of $G$ is payoff secure.
\end{theorem}

With the above theorem, we can now prove the existence of mixed Nash equilibrium in game $\mathcal{G}^{(n)}$. We first show that game $\mathcal{G}^{(n)}$ is compact and Hausdorff.

\begin{lemma} \label{lem:game-compact-hausdorff}
    Game $\mathcal{G}^{(n)}$ is compact and Hausdorff.
\end{lemma}

\begin{proof}
    The pure strategy space of each player is a subset of $[0, 1]$, which is a compact Hausdorff space. Thus, game $\mathcal{G}^{(n)}$ is a compact Hausdorff game.
\end{proof}

Next, we show that the sum of payoffs is upper semicontinuous.

\begin{lemma} \label{lem:sum-payoff-usc}
    In game $\mathcal{G}^{(n)}$, the sum of payoffs $\sum_{i=1}^n u_i(\bm{x})$ is upper semicontinuous in $x$ on $X$.
\end{lemma}

\begin{proof}
    When all creators adopt a pure strategy, since ties are broken uniformly at random, their payoffs can be aggregated without altering the total payoff by assigning creators with the same strategies to distinct ranks within their corresponding ranking range. For example, if creators $i_1, i_2, \ldots, i_k$ choose the same strategy $x_{(j)}$, we may assign them to ranks $j, j + 1, \ldots, j + k - 1$, respectively, without changing the sum of their payoffs. Based on this assignment rule, we obtain
    \[ \sum_{i = 1}^n u_i(\bm{x}) = \sum_{i = 1}^n x_{(i)} p_i - \sum_{i = 1}^n \gamma_i g(x_i). \]
    Because $g$ is continuous, $\sum_{i=1}^n u_i(\bm{x})$ is continuous and therefore upper semicontinuous on $X$.
\end{proof}

Finally, we show that the mixed extension of game $\mathcal{G}^{(n)}$ is payoff secure. To this end, we first introduce several definitions and lemmas. For any subset of creators $S \subseteq [n]$, we define a function $W_S: [0, 1]^{|S| - 1} \times [0, 1]^{|S|} \to \mathbb{R}$ as
\begin{equation} \label{eq:W-def}
    W_S(\bm{y}_{-j}, \bm{p}) = \sum_{k = 1}^{|S|} \left(p^{(k)} \sum_{\substack{S_{k - 1} \subseteq S \setminus \{j\} \\ |S_{k - 1}| = k - 1}} \left(\prod_{i \in S_{k - 1}}(1 - y_i) \prod_{i' \in (S \setminus \{j\}) \setminus S_{k - 1}} y_{i'}\right)\right),
\end{equation}
where $j \in S$, $\bm{y}_{-j} = (y_i)_{i \in S \setminus \{j\}}$, and $p^{(k)}$ is the $k$-th highest element in $\bm{p}$. Let $\bm{F}_{-j} = (F_i)_{i \in S \setminus \{j\}}$ be the mixed strategy profile of creators in $S \setminus \{j\}$, then $W_{[n]}(\bm{F}_{-j}(x), \bm{p})$ represents the expected position bias of creator $j$ when she exerts effort level $x$ and no other creator in $S$ exerts effort level equals to $x$ with positive probability. The next lemma illustrates some properties of function $W_S$.

\begin{lemma} \label{lem:W-monotone}
    (Lemma~1 of \cite{xiao2016asymmetric})
    For any subset of creators $S \subseteq [n]$, function $W_S(\bm{y}_{-j}, \bm{p})$ is linear in $\bm{p}$, and it is strictly increasing in each $y_i$ for $i \in S \setminus \{j\}$.
\end{lemma}

\begin{proof}
    The linearity in $\bm{p}$ is straightforward from the definition of $W_S$. To see the monotonicity, note that
    \[ W_S(\bm{y}_{-j}, \bm{p}) = y_i W_{S \setminus \{i\}}(\bm{y}_{-\{i, j\}}, \overline{\bm{p}}) + (1 - y_i) W_{S \setminus \{i\}}(\bm{y}_{-\{i, j\}}, \underline{\bm{p}}), \]
    where $\overline{\bm{p}}$ and $\underline{\bm{p}}$ are obtained by removing the lowest and highest elements from $\bm{p}$, respectively. Differentiating with respect to $y_i$ on both sides yields
    \[ \frac{\partial W_S(\bm{y}_{-j}, \bm{p})}{\partial y_i} = W_{S \setminus \{i\}}(\bm{y}_{-\{i, j\}}, \overline{\bm{p}}) - W_{S \setminus \{i\}}(\bm{y}_{-\{i, j\}}, \underline{\bm{p}}) = W_{S \setminus \{i\}}(\bm{y}_{-\{i, j\}}, \overline{\bm{p}} - \underline{\bm{p}}), \]
    where the second equality follows from the linearity of $W_{S \setminus \{i\}}$ in $\bm{p}$. Since $\overline{\bm{p}} - \underline{\bm{p}}$ is a vector with positive elements, by the definition of $W_{S \setminus \{i\}}$, we have $\frac{\partial W_S(\bm{y}_{-j}, \bm{p})}{\partial y_i} > 0$. Thus, $W_S$ is strictly increasing in each $y_i$ for $i \in S \setminus \{j\}$.
\end{proof}

In the following discussion, the subscript of $W$ will be omitted when there is no ambiguity. Based on the above definitions and lemmas, we can now prove that the mixed extension of game $\mathcal{G}^{(n)}$ is payoff secure. Similar techniques have been used in \cite{jagadeesan2023supply,golrezaei2025contest}.

\begin{lemma} \label{lem:payoff-secure}
    The mixed extension of game $\mathcal{G}^{(n)}$ is payoff secure.
\end{lemma}

\begin{proof}
    Consider any mixed-strategy profile $\bm{\mu} \in M$ and any $\varepsilon > 0$. By the definition of expected payoff under mixed strategies, there exists a pure strategy $x_i^{\text{dev}}$ such that
    \[ u_i(x_i^{\text{dev}}, \bm{\mu}_{-i}) \geqslant u_i(\bm{\mu}). \]
    Otherwise, using the definition of expected payoff under mixed strategies, we would have
    \[ u_i(\bm{\mu}) = \int_{X_i} u_i(x_i, \bm{\mu}_{-i}) \textup{d}\bm{\mu}_i < u_i(\bm{\mu}), \]
    a contradiction. To satisfy the requirement of equation~\eqref{eq:W-def}, we adjust $x_i^{\text{dev}}$ slightly to obtain $x_i^{\text{sec}}$ such that $(\mu_j)_{j \in [n] \setminus \{i\}}$ have no atom in $[x_i^{\text{sec}} - \delta', x_i^{\text{sec}}]$ for a small $\delta' > 0$. Let $F_j$ denote the CDF corresponding to $\mu_j$. By equation~\eqref{eq:W-def}, 
    \[ u_i(x_i^{\text{sec}}, \bm{\mu}_{-i}) = x_i^{\text{sec}} W(\bm{F}_{-i}(x_i^{\text{sec}}), \bm{p}) - \gamma_i g(x_i^{\text{sec}}). \]
    On the other hand, by the continuity of $F_j$ in $[x_i^{\text{sec}} - \delta', x_i^{\text{sec}}]$, we can choose a sufficiently small $\delta < \delta'$ such that
    \begin{equation} \label{eq:secure-payoff-1}
        x_i^{\text{sec}} W(\bm{F}_{-i}(x_i^{\text{sec}} - \delta), \bm{p}) - \gamma_i g(x_i^{\text{sec}}) > u_i(\bm{\mu}) - \frac{\varepsilon}{2}.
    \end{equation}
    Define $A = \{x_i \mid x_i < x_i^{\text{sec}}\}$, $A^{\delta} = \{x_i \mid x_i < x_i^{\text{sec}} - \delta\}$. Then $F_j(x_i^{\text{sec}}) = \mu_j(A)$ and $F_j(x_i^{\text{sec}} - \delta) = \mu_j(A^{\delta})$. Equation~\eqref{eq:secure-payoff-1} can be rewritten as
    \[ \Delta_1 := x_i^{\text{sec}} W(\bm{\mu}_{-i}(A^{\delta}), \bm{p}) - \gamma_i g(x_i^{\text{sec}}) > u_i(\bm{\mu}) - \frac{\varepsilon}{2}. \]
    Let $A' = \{x_i \mid \exists x_i' \in A^{\delta}\ \text{s.t.}\ |x_i' - x_i| < \delta\}$. By the definition of the Lévy–Prokhorov metric, for every $\mu_j' \in B_{\delta}(\mu_j)$,
    \[ \mu_j(A^{\delta}) \leqslant \mu_j'(A') + \delta. \]
    Since every $x_i \in A'$ satisfies $x_i < x_i^{\text{sec}}$, we have
    \[ \mu_j'(A) \geqslant \mu_j'(A') \geqslant \mu_j(A^{\delta}) - \delta. \]
    Consequently, for all $\mu_j' \in B_{\delta}(\mu_j)$ with $j \in [n] \setminus \{i\}$,
    \begin{align*}
        u_i(x_i^{\text{sec}}, \bm{\mu}_{-i}') & = x_i^{\text{sec}} W(\bm{\mu}_{-i}'(A), \bm{p}) - \gamma_i g(x_i^{\text{sec}}) \\
        & \geqslant x_i^{\text{sec}} W(\bm{\mu}_{-i}(A^{\delta}) - \delta, \bm{p}) - \gamma_i g(x_i^{\text{sec}}) \\
        & = \Delta_1 + \Delta_2,
    \end{align*}
    where the inequality follows from Lemma~\ref{lem:W-monotone} and
    \[ \Delta_2 = x_i^{\text{sec}} \left(W(\bm{\mu}_{-i}(A^{\delta}) - \delta, \bm{p}) - W(\bm{\mu}_{-i}(A^{\delta}), \bm{p})\right). \]
    Because $\Delta_2 \to 0$ as $\delta \to 0$, we can choose $\delta$ small enough so that $\Delta_2 \geqslant -\frac{\varepsilon}{2}$. Combining the above results yields
    \[ u_i(x_i^{\text{sec}}, \bm{\mu}_{-i}') \geqslant \Delta_1 + \Delta_2 > u_i(\bm{\mu}) - \varepsilon. \]
    This shows that player $i$ can secure a payoff of at least $u_i(\bm{\mu}) - \varepsilon$ at $\bm{\mu}$. Therefore, the mixed extension is payoff-secure.
\end{proof}

\begin{proofof}{Theorem~\ref{thm:mix-existence}}
    The conclusion immediately follows from Theorem~\ref{thm:mix-existence-reny}, Lemma~\ref{lem:game-compact-hausdorff}, Lemma~\ref{lem:sum-payoff-usc}, and Lemma~\ref{lem:payoff-secure}.
\end{proofof}
\section{Properties of Symmetric Equilibrium} \label{sec:ap-symmetric}

The purpose of this subsection is to prove several important properties of symmetric equilibrium in game $\mathcal{G}^{(1)}$. For notational simplicity, we do not explicitly introduce the citation and compensation mechanism in the discussion of equilibrium properties, and treat the original position bias vector $\bm{p}$ as the parameter; this is without loss of generality since both mechanisms operate by adjusting position biases.

We first prove the following lemma, which establishes that any equilibrium must be a symmetric mixed Nash equilibrium in game $\mathcal{G}^{(1)}$.

\begin{lemma} \label{lem:symmetric-equilibrium}
    In game $\mathcal{G}^{(1)}$, the equilibrium must be a symmetric mixed Nash equilibrium.
\end{lemma}

\begin{proof}
    Let $(F_1, \ldots, F_n)$ denote the strategy profile with support $\textup{supp}(F_i) = [\underline{x}_i, \overline{x}_i]$, and $u_i$ the equilibrium utility for creator $i$. Our proof can be divided into the following six steps:
    
    \textbf{Step 1: In equilibrium, all creators obtain the same utility.}

    Suppose, for a contradiction, that $\max_i u_i > \min_j u_j$. Let $k = \arg\max_i u_i, l = \arg\min_j u_j$, and thus $u_k > u_l$. If $l$ deviates by placing all mass on  $\overline{x}_k + \varepsilon$ (where $\varepsilon \to 0^+$ is an infinitesimal perturbation, hereafter we omit such $\varepsilon$), that is, $F_l(x) = 0$ for $x < \overline{x}_k$ and $F_l(x) = 1$ for $x \geqslant \overline{x}_k$. Then the utility of $l$ after deviation becomes:
    \[ u_l' = \overline{x}_k W(\bm{F}_{-l}(\overline{x}_k), \bm{p}) - \gamma g(\overline{x}_k) \]
    where the expected position bias $W$ is defined in equation~\eqref{eq:W-def}. Note that $F_k(\overline{x}_k) = 1$, $F_l(\overline{x}_k) \leqslant 1$, and thus we have
    \begin{align*}
        W(\bm{F}_{-l}(\overline{x}_k), \bm{p}) &= F_k(\overline{x}_k) W(\bm{F}_{-\{k, l\}}(\overline{x}_k), \overline{\bm{p}}) + (1 - F_k(\overline{x}_k)) W(\bm{F}_{-\{k, l\}}(\overline{x}_k), \underline{\bm{p}}) \\
        &= W(\bm{F}_{-\{k, l\}}(\overline{x}_k), \underline{\bm{p}}) + F_k(\overline{x}_k) W(\bm{F}_{-\{k, l\}}(\overline{x}_k), \overline{\bm{p}} - \underline{\bm{p}}) \\
        &\geqslant W(\bm{F}_{-\{k, l\}}(\overline{x}_k), \underline{\bm{p}}) + F_l(\overline{x}_k) W(\bm{F}_{-\{k, l\}}(\overline{x}_k), \overline{\bm{p}} - \underline{\bm{p}}) \\
        &= W(\bm{F}_{-k}(\overline{x}_k), \bm{p}).
    \end{align*}
    where $\overline{\bm{p}}$ and $\underline{\bm{p}}$ are obtained by removing the lowest and highest elements from $\bm{p}$, respectively. Therefore, we have
    \[ u_k = \overline{x}_k W(\bm{F}_{-k}(\overline{x}_k), \bm{p}) - \gamma g(\overline{x}_k) \leqslant \overline{x}_k W(\bm{F}_{-l}(\overline{x}_k), \bm{p}) - \gamma g(\overline{x}_k) = u_l', \]
    which yields $u_l < u_k \leqslant u_l'$, a profitable deviation for creator $l$. This contradicts equilibrium.

    Note that although the definition of $W$ needs no other creator assigns positive probability to $\overline{x}_k$, the above inequality still holds because whether other creators have a positive probability density on $\overline{x}_k$ does not affect the proof above.

    \textbf{Step 2: $\min_i \underline{x}_i = \arg\max_x x p_n - \gamma g(x)$}

    Denote $j = \arg\min_i \underline{x}_i$, $\underline{x} = \arg\max_x x p_n - \gamma g(x)$, we prove by contradiction. If $\underline{x}_j > \underline{x}$, note that $x p_n - \gamma g(x)$ is concave as a function of $x$, thus has a unique maximum point, thus $\underline{x}_j p_n - \gamma g(\underline{x}_j) < \underline{x} p_n - \gamma g(\underline{x})$. Hence, creator $j$ could profitably deviate by shifting mass downward to $\underline{x}$, which violates equilibrium stability requirement.

    If $\underline{x}_j < \underline{x}$, consider the case where creator $j$ deviates his infimum strategy to $\underline{x}$, then the expected position bias on $\underline{x}$ (denoted by $q_j$) satisfies $q_j > p_n$. Note that $\underline{x}$ is the only maximum point of $x p_n - \gamma g(x)$, thus we have $\underline{x} q_j - \gamma g(\underline{x}) > \underline{x} p_n - \gamma g(\underline{x}) > \underline{x}_j p_n - \gamma g(\underline{x}_j)$. Again, a profitable deviation arises, contradicting equilibrium.
    
    \textbf{Step 3: For all $i \in [n]$, $F_i$ is continuous, that is, there is no point mass in the probability distribution.}

    Prove by contradiction. Assume the strategy of creator $i$ has a point mass on $x_i$, suppose the expected position bias of creator $i$ on $x_i$ is $q_i$. Let $x_{\max} = \arg\max_x x q_i - \gamma g(x)$, we claim that $x_i = x_{\max}$. If $x_i < x_{\max}$, let $\varepsilon > 0$ be such that $x_i + \varepsilon < x_{\max}$, then the expected position bias of creator $i$ on $x_i + \varepsilon$ (denoted by $q_i'$) satisfies $q_i' \geqslant q_i$, thus $(x_i + \varepsilon)q_i' - \gamma g(x_i + \varepsilon) \geqslant (x_i + \varepsilon)q_i - \gamma g(x_i + \varepsilon) > x_i q_i - \gamma g(x_i)$, the last inequality stems from the concavity of $x q - \gamma g(x)$ as a function of $x$. So, creator $i$ has the incentive to deviate his strategy on $x_i$ to $x_i + \varepsilon$, which violates equilibrium stability requirement.

    If $x_i > x_{\max}$, we claim that there exists a $\delta > 0$ such that no other creator assigns positive probability to $[x_i - \delta, x_i]$, then creator $i$ has the incentive to deviate to $x_i - \delta$. We prove the claim. Assume the expected position bias for any creator $j \neq i$ on $x_i - \delta$ is $q_j$, thus the utility of $j$ when choosing $x_i - \delta$ is $u_{j_1} = (x_i - \delta)q_j - \gamma g(x_i - \delta)$. But if creator $j$ deviates to $x_i + \varepsilon$ (where $\varepsilon \to 0^+$ is an infinitesimal perturbation, hereafter we omit such $\varepsilon$), the expected position bias $q_j'$ satisfies $q_j' > q_j$ because of the point mass of creator $i$ on $x_i$, thus the utility of creator $j$ becomes $u_{j_2} = x_i q_j' - \gamma g(x_i)$. Because $g$ is a continuous function and $x_i q_j' > (x_i - \delta)q_j$, there exists a $\delta > 0$ such that $u_{j_2} > u_{j_1}$, thus no other creator assigns positive probability to $[x_i - \delta, x_i]$.

    Thus $x_i = x_{\max}$. Recall the result of Step~2, we can derive $q_i = p_n$, otherwise, $\max_x x q_i - \gamma g(x) > \max_x x p_n - \gamma g(x)$, which contradicts the result of Step~1. Thus, the point mass can only be on $\min_k \underline{x}_k$. Assume $i = \arg\min_k \underline{x}_k$, if creator $j \neq i$ assigns positive probability to $\underline{x}_i + \varepsilon$ (where $\varepsilon \to 0^+$ is an infinitesimal perturbation, hereafter we omit such $\varepsilon$), then $u_i = \underline{x}_i p_n - \gamma g(\underline{x}_i)$, $u_j' = \underline{x}_i((1 - F_i(\underline{x}_i))p_n +F_i(\underline{x}_i) p_{n - 1}) - \gamma g(\underline{x}_i) > u_i$, which contradicts the conclusion of Step~1. So, there is no point mass in the probability distribution, and thus the equilibrium cannot be a pure strategy equilibrium.

    \textbf{Step 4: $\overline{x}_1 = \cdots = \overline{x}_n$.}

    Prove by contradiction. Suppose that there exist creators $i$ and $j$ such that $\overline{x}_i < \overline{x}_j$. Consider the case where creator $j$ takes the strategy $\overline{x}_i$, then the utility of creator $j$ is
    \[ u_j' = \overline{x}_i W(\bm{F}_{-j}(\overline{x}_i), \bm{p}) - \gamma g(\overline{x}_i). \]
    Note that $\overline{x}_i < \overline{x}_j$, so we have $F_i(\overline{x}_i) = 1$, $F_j(\overline{x}_i) < 1$, thus
    \begin{equation} \label{eq:symmetric-equilibrium-1}
        \begin{aligned}
            W(\bm{F}_{-j}(\overline{x}_i), \bm{p}) &= W(\bm{F}_{-\{i, j\}}(\overline{x}_i), \underline{\bm{p}}) + F_i(\overline{x}_i) W(\bm{F}_{-\{i, j\}}(\overline{x}_i), \overline{\bm{p}} - \underline{\bm{p}}) \\
            &> W(\bm{F}_{-\{i, j\}}(\overline{x}_i), \underline{\bm{p}}) + F_j(\overline{x}_i) W(\bm{F}_{-\{i, j\}}(\overline{x}_i), \overline{\bm{p}} - \underline{\bm{p}}) \\
            &= W(\bm{F}_{-i}(\overline{x}_i), \bm{p}).
        \end{aligned}
    \end{equation}
    Therefore, we have 
    \[ u_i = \overline{x}_i W(\bm{F}_{-i}(\overline{x}_i), \bm{p}) - \gamma g(\overline{x}_i) < \overline{x}_i W(\bm{F}_{-j}(\overline{x}_i), \bm{p}) - \gamma g(\overline{x}_i) = u_j', \]
    so $u_j' > u_i = u_j$, which violates equilibrium stability requirement, as it induces a profitable unilateral deviation of creator $j$.

    \textbf{Step 5: $\underline{x}_1 = \cdots = \underline{x}_n$.}
    
    Prove by contradiction. Let $i = \arg\min_k \underline{x}_k$, $j = \arg\max_k \underline{x}_k$. Assume that $\underline{x}_i < \underline{x}_j$. Thus, $F_i(\underline{x}_j) > 0$, $F_j(\underline{x}_j) = 0$. Similar to the derivation in equation~\eqref{eq:symmetric-equilibrium-1}, we have $W(\bm{F}_{-j}(\underline{x}_j), \bm{p}) > W(\bm{F}_{-i}(\underline{x}_j), \bm{p})$. Therefore, we have
    \[ u_j = \underline{x}_j W(\bm{F}_{-j}(\underline{x}_j), \bm{p}) - \gamma g(\underline{x}_j) > \underline{x}_j W(\bm{F}_{-i}(\underline{x}_j), \bm{p}) - \gamma g(\underline{x}_j) = u_i', \]
    where $u_i'$ denotes the utility of creator $i$ when choosing $\underline{x}_j$. According to the result of Step~1, Step~3 and the continuity of $g$, there must exist a $\delta > 0$ such that no creator has a positive probability density on $[\underline{x}_j - \delta, \underline{x}_j]$.

    Let $x_j = \arg\max_x x W(\bm{F}_{-j}(\underline{x}_j), \bm{p}) - \gamma g(x)$, then we claim that $\underline{x}_j = x_j$. Otherwise, if $\underline{x}_j > x_j$, let $\delta' = \underline{x}_j - x_j$, $\delta'' = \min\{\delta, \delta'\}$, then creator $j$ has the incentive to deviate from $\underline{x}_j$ to $\underline{x}_j - \delta''$; if $\underline{x}_j < x_j$, then $W(\bm{F}_{-j}(\underline{x}_j), \bm{p}) \leqslant W(\bm{F}_{-j}(x_j), \bm{p})$, thus
    \[ x_j W(\bm{F}_{-j}(x_j), \bm{p}) - \gamma g(x_j) \geqslant x_j W(\bm{F}_{-j}(\underline{x}_j), \bm{p}) - \gamma g(x_j) > \underline{x}_j W(\bm{F}_{-j}(\underline{x}_j), \bm{p}) - \gamma g(\underline{x}_j), \]
    then creator $j$ has the incentive to deviate from $\underline{x}_j$ to $x_j$. But if $\underline{x}_j = x_j$, then
    \[ x_j W(\bm{F}_{-j}(\underline{x}_j), \bm{p}) - \gamma g(x_j) = \max_x x W(\bm{F}_{-j}(\underline{x}_j), \bm{p}) - \gamma g(x) > \max_x x p_n - \gamma g(x), \]
    which contradicts the result of Step~1. So, there must be $\underline{x}_1 = \cdots = \underline{x}_n$ in equilibrium.

    \textbf{Step 6: $F_1 = \cdots = F_n$.}

    We first show that $\text{Supp}(F_1) \cup \cdots \cup \text{Supp}(F_n) = [\underline{x}, \overline{x}]$, where $\underline{x} = \underline{x}_1 = \cdots = \underline{x}_n$ and $\overline{x} = \overline{x}_1 = \cdots = \overline{x}_n$. Prove by contradiction. Suppose that there exists a gap $(x_1, x_2)$ in the union of supports, then similar to Step~2 and Step~5, we can derive $x_2 = \arg\max_x x W(\bm{F}_{-i}(x_2), \bm{p}) - \gamma g(x)$, where $i$ is any creator with $x_2$ in her support. Note that $\max_x x p_n - \gamma g(x) < \max_x x W(\bm{F}_{-i}(x_2), \bm{p}) - \gamma g(x)$, thus the utility of any creator $i$ when choosing $x_2$ is strictly greater than the equilibrium utility, which contradicts the equilibrium definition. Thus, there is no gap in the union of supports.

    Then consider any two creators $i$ and $j$ ($i \neq j$). For any $x \in \text{Supp}(F_i)$, we have
    \begin{align*}
        u_i &= x W(\bm{F}_{-i}(x), \bm{p}) - \gamma g(x), \\
            &= x \left(W(\bm{F}_{-\{i, j\}}(x), \underline{\bm{p}}) + F_j(x) W(\bm{F}_{-\{i, j\}}(x), \overline{\bm{p}} - \underline{\bm{p}})\right) - \gamma g(x),
    \end{align*}
    for any $x \in \text{Supp}(F_j)$, we have
    \begin{align*}
        u_j &= x W(\bm{F}_{-j}(x), \bm{p}) - \gamma g(x), \\
            &= x \left(W(\bm{F}_{-\{i, j\}}(x), \underline{\bm{p}}) + F_i(x) W(\bm{F}_{-\{i, j\}}(x), \overline{\bm{p}} - \underline{\bm{p}})\right) - \gamma g(x),
    \end{align*}
    Recall that Step~1 gives $u_i = u_j$, which forces $F_i(x) = F_j(x)$ for any $x \in \text{Supp}(F_i) \cap \text{Supp}(F_j)$. By Step~5 and the continuity of any $F_i$, we can choose the largest $\varepsilon > 0$ such that any $F_i$ has a positive density on $[\underline{x}, \underline{x} + \varepsilon]$. If $\underline{x} + \varepsilon < \overline{x}$, then there exists a subset of creators $P \subset [n]$ whose CDFs assign zero density after $\underline{x} + \varepsilon$. Let $x'$ be the infimum of the supports of the creators in $P$ that lies strictly above $\underline{x} + \varepsilon$ (this point belongs to the support of some creator $k \in P$). Because $\text{Supp}(F_1) \cup \cdots \cup \text{Supp}(F_n) = [\underline{x}, \overline{x}]$, there must be another creator $l \notin P$ whose support contains $x'$. However, at this point, the CDFs satisfy $F_k(x') \neq F_j(x')$, a contradiction. Hence, we must have $\underline{x} + \varepsilon = \overline{x}$, and consequently $F_1 = \cdots = F_n$ on the whole interval $[\underline{x}, \overline{x}]$.
\end{proof}

Let $F$ be the equilibrium CDF and let $\underline{x}$ and $\overline{x}$ be the infimum and supremum of its support. Step~6 in the proof of Lemma~\ref{lem:symmetric-equilibrium} implies $\text{supp}(F) = [\underline{x}, \overline{x}]$. We formally state this result in the following corollary.

\begin{corollary} \label{cor:symmetric-no-gap}
    In game $\mathcal{G}^{(1)}$, the support of $F$ is a continuous interval, i.e., $\textup{supp}(F) = [\underline{x}, \overline{x}]$.
\end{corollary}

The following proposition summarizes several important properties of the symmetric equilibrium.

\begin{proposition} \label{prop:symmetric-equilibrium-properties}
    In game $\mathcal{G}^{(1)}$, the equilibrium has the following properties:
    \begin{enumerate}
        \item The equilibrium PDF $f$ exists almost everywhere on the support $[\underline{x}, \overline{x}]$.
        \item The equilibrium PDF $f$ is positive on $(\underline{x}, \overline{x}]$ and $f(\underline{x}) = 0$.
        \item $\underline{x} = \arg\max_x x p_n - \gamma g(x)$, and the equilibrium utility is $u = \underline{x} p_n - \gamma g(\underline{x})$.
        \item $\overline{x}$ is the largest root of equation $x p_1 - \gamma g(x) = u$, where $u$ is the equilibrium utility.
    \end{enumerate}
\end{proposition}

The third property indicates that $\underline{x}$ and the equilibrium utility is uniquely determined by $p_n$, while the fourth property shows that $\overline{x}$ depends only on $p_1$ and $p_n$. Note that the third property is specific to our model. The reason is that in our model, the utility function $x p_n - \gamma g(x)$ attains its maximum at a positive effort level, while in standard contests with rank-based rewards, the utility function is simply a constant reward minus the cost, so if the infimum of support were positive, a player could always deviate to $0$ and increase her utility by avoiding the cost while still receiving the same reward.


\begin{proof}
    \begin{enumerate}
        \item Since $F$ is monotone and atomless (from Step~3 in the proof of Lemma~\ref{lem:symmetric-equilibrium}), PDF $f$ exists almost everywhere on the support $[\underline{x}, \overline{x}]$.
        \item Corollary~\ref{cor:symmetric-no-gap} shows that the support of $F$ is a continuous interval $[\underline{x}, \overline{x}]$. Based on this result, we can refer to equation~\eqref{eq:chain-rule-KJ}, Lemma~\ref{lem:K-monotone} and Lemma~\ref{lem:KJ-monotone} to deduce $f(\underline{x}) = 0$ and $f(x) > 0$ for $x \in (\underline{x}, \overline{x}]$.
        \item Step~2 in the proof of Lemma~\ref{lem:symmetric-equilibrium} shows that $\underline{x} = \arg\max_x x p_n - \gamma g(x)$. The indifference condition of mixed Nash equilibrium implies that the expected utility must be equal to the equilibrium utility $u$ at $\underline{x}$. Thus, we have $u = \max_x x p_n - \gamma g(x) = \underline{x} p_n - \gamma g(\underline{x})$.
        \item According to the indifference condition of mixed Nash equilibrium, at the supremum $\overline{x}$ of the support, the expected utility must be equal to the equilibrium utility $u$. From part~3 above, we know that $u = \max_x x p_n - \gamma g(x) = \underline{x} p_n - \gamma g(\underline{x})$.
        
        Let $x^*$ be the unique solution to $p_1 = \gamma g'(x^*)$. Because $g$ is strictly convex, the equation $u = x p_1 - \gamma g(x)$ has two solutions: one less than $x^*$ and the other greater than $x^*$. The lower one cannot be the supremum of the support of the mixed Nash equilibrium, since at that point any creator would have an incentive to deviate to $x^*$ to achieve strictly higher utility, which contradicts the indifference condition of mixed Nash equilibrium. Hence, the supremum of the support must be the largest solution to the equation $u = x p_1 - \gamma g(x)$.
    \end{enumerate}
\end{proof}

Using the indifference condition~\eqref{eq:indifference-sym}, we establish the uniqueness of the symmetric mixed Nash equilibrium.

\begin{lemma} \label{lem:unique-symmetric-equilibrium}
    In game $\mathcal{G}^{(1)}$, the symmetric equilibrium is unique.
\end{lemma}

\begin{proof}
    We first reformulate equation~\eqref{eq:indifference-sym} as
    \[ \sum_{i = 1}^n \binom{n - 1}{i - 1} [F(x)]^{n - i} [1 - F(x)]^{i - 1} p_i = \frac{u + \gamma g(x)}{x}. \]
    Lemma~\ref{lem:KJ-monotone} shows that the left-hand side is a strictly increasing function of $F(x)$, thus for any $x \in \textup{supp}(F)$, there exists a unique $F(x)$ that satisfies the above equation. Therefore, the CDF $F$ of the symmetric mixed Nash equilibrium is unique.
\end{proof}

According to Theorem~\ref{thm:mix-existence}, Lemma~\ref{lem:symmetric-equilibrium}, Lemma~\ref{lem:unique-symmetric-equilibrium}, we can conclude that the symmetric mixed Nash equilibrium is the unique equilibrium in the symmetric case. Thus, Theorem~\ref{thm:sym-equ-property} is proved.

We provide an example to illustrate how to solve the unique symmetric mixed Nash equilibrium when $n = 2$.

\begin{example} \label{ex:sym-n2}
    If $n = 2$, equation \eqref{eq:indifference-sym} becomes
    \[ x (F(x) p_1 + (1 - F(x)) p_2) - \gamma g(x) = u. \]
    
    Solving the above equation directly, we can derive the unique equilibrium CDF:
    \begin{equation} \label{eq:sym-equ-cdf-n2}
        F(x) = \frac{u + \gamma g(x)}{x (p_1 - p_2)} - \frac{p_2}{p_1 - p_2}.
    \end{equation}
    
    The following three equations determine $\overline{x}$, $\underline{x}$, and $u$ (thus $F$) in equilibrium, where $\underline{x}$ is the unique solution to the first equation, $u$ is determined by the second equation, and $\overline{x}$ is the largest solution of the last equation:
    \[ \begin{cases}
        p_2 - \gamma g'(\underline{x}) = 0, \\
        \underline{x} p_2 - \gamma g(\underline{x}) = u, \\
        \overline{x} p_1 - \gamma g(\overline{x}) = u.
    \end{cases} \]
\end{example}

We then present two additional properties of the symmetric equilibrium. In what follows, the tuple $(m, \bm{p}, \gamma, g(\cdot))$ denotes a symmetric setting with $m$ creators sharing the same cost parameter $\gamma$ and cost function $g(\cdot)$, and compete under position bias vector $\bm{p}$.

\begin{proposition} \label{prop:asym-optimal-ch-cl-1}
    Configurations $(m, \bm{p}, \gamma, g(\cdot))$ (denoted as $C_1$) and $(m, s \bm{p}, s \gamma, g(\cdot))$ (denoted as $C_2$) yield the same equilibrium strategy for any scalar $s > 0$.
\end{proposition}

\begin{proof}
    Let $F_1$ and $F_2$ denote the equilibrium strategies under configurations $C_1$ and $C_2$, respectively, with corresponding equilibrium utilities $u_1$ and $u_2$ and supports $[\underline{x}_1,\overline{x}_1]$ and $[\underline{x}_2,\overline{x}_2]$. By the third part of Proposition~\ref{prop:symmetric-equilibrium-properties}, 
    \[ \underline{x}_1 = \arg\max_x x p_m - \gamma g(x) = \arg\max_x x s p_m - s \gamma g(x) = \underline{x}_2, \]
    thus we have
    \[ u_2 = \underline{x}_2 s p_m - s \gamma g(\underline{x}_2) = s (\underline{x}_1 p_m - \gamma g(\underline{x}_1)) = s u_1. \]
    
    Note that
    \begin{align*}
        \sum_{i = 1}^{m} \binom{m - 1}{i - 1} [F_1(x)]^{m - i} [1 - F_1(x)]^{i - 1} p_i &= \frac{u_1 + \gamma g(x)}{x}, \\
        \sum_{i = 1}^{m} \binom{m - 1}{i - 1} [F_2(x)]^{m - i} [1 - F_2(x)]^{i - 1} s p_i &= \frac{u_2 + s \gamma g(x)}{x} = s \cdot \frac{u_1 + \gamma g(x)}{x}.
    \end{align*}
    
    Therefore, if $F_1(x_1) = F_2(x_2)$, then $x_1 = x_2$ is a feasible solution to both equations above. Since this feasible solution constitutes an equilibrium strategy and Theorem~\ref{thm:sym-equ-property} ensures that the equilibrium strategy is unique, we conclude that $F_1 = F_2$.
\end{proof}

\begin{proposition} \label{prop:asym-optimal-ch-cl-2}
    Let $C_1 = (m, \bm{p}, \gamma_1, g(\cdot))$ and $C_2 = (m, \bm{p}, \gamma_2, g(\cdot))$, where $g(x) = x^{\beta}$. Let $F_1$ and $F_2$ denote the equilibrium strategies under configurations $C_1$ and $C_2$, respectively. If $F_1(x_1) = F_2(x_2)$, then $x_1 = k x_2$, where $k = \left(\frac{\gamma_2}{\gamma_1}\right)^{\frac{1}{\beta - 1}}$.
\end{proposition}

\begin{proof}
    Let $F_1$ and $F_2$ denote the equilibrium strategies under configurations $C_1$ and $C_2$, respectively, with corresponding equilibrium utilities $u_1$ and $u_2$ and supports $[\underline{x}_1,\overline{x}_1]$ and $[\underline{x}_2,\overline{x}_2]$. By the third part of Proposition~\ref{prop:symmetric-equilibrium-properties},
    \[ \underline{x}_1 = \arg\max_x x p_m - \gamma_1 x^{\beta} = \left(\frac{p_m}{\gamma_1 \beta}\right)^{\frac{1}{\beta - 1}}, \quad \underline{x}_2 = \arg\max_x x p_m - \gamma_2 x^{\beta} = \left(\frac{p_m}{\gamma_2 \beta}\right)^{\frac{1}{\beta - 1}}, \]
    thus we have
    \[ \frac{u_1}{u_2} = \frac{\underline{x}_1 p_m - \gamma_1 \underline{x}_1^{\beta}}{\underline{x}_2 p_m - \gamma_2 \underline{x}_2^{\beta}} = \left(\frac{\gamma_2}{\gamma_1}\right)^{\frac{1}{\beta - 1}} := k. \]
    Following a line of argument similar to the proof of Proposition~\ref{prop:asym-optimal-ch-cl-1}, we can show that if $F_1(x_1) = F_2(x_2)$, then $x_1 = k x_2$.
\end{proof}

\section{An Example of Non-decreasing Position Bias} \label{sec:position-bias-monotone}

In this section, we provide an illustrative example in which an increasing position bias vector leads to a sharp drop in equilibrium effort levels, thus illustrating the importance of the assumption that the position bias vector is non-increasing in our Problem~\eqref{eq:mechanism-design}.

\begin{example} \label{ex:position-bias-monotone}
    Consider a setting with $n = 2$, where the position bias vector is $\bm{p} = (p_1, p_2)$ with $p_1 < p_2$ and both creators share the same cost function $\gamma g(x)$. We first prove that no pure Nash equilibrium exists. We then demonstrate that, under the unique mixed Nash equilibrium, creators exert substantially lower effort compared to the case where the bias vector is reversed, i.e., $\bm{p} = (p_2, p_1)$.

    We first show the non-existence of pure Nash equilibrium. First, there cannot be an equilibrium profile $(x_1, x_2)$ with $x_1 = x_2 := x$. If $x \neq 0$, then under such a profile, either creator could profitably deviate by reducing her effort by an arbitrarily small $\varepsilon \to 0^+$, thus raising her payoff from $x\bigl((p_1 + p_2)/2\bigr) - \gamma g(x)$ to $(x - \varepsilon) p_2 - \gamma g(x - \varepsilon)$. If $x = 0$, then either creator could deviate to $\arg\max_x x p_1 - \gamma g(x)$, which yields strictly positive utility. Hence, any pure Nash equilibrium must satisfy $x_1 \neq x_2$. Without loss of generality, assume $x_1 > x_2$. In a pure Nash equilibrium, the first-order conditions give  
    \[ p_2 - \gamma g'(x_2) = 0 \quad \text{and} \quad p_1 - \gamma g'(x_1) = 0. \]  
    Because $g'(\cdot)$ is strictly increasing and $p_1 < p_2$, the two equations cannot be held simultaneously. Therefore, no pure Nash equilibrium exists.

    Then we consider the mixed Nash equilibrium. Following Step~2 of the proof of Lemma~\ref{lem:symmetric-equilibrium}, we can show that the supremum of the support of the mixed Nash equilibrium is
    \[ \overline{x} = \arg\max_x x p_1 - \gamma g(x). \]  
    However, if we reverse the position bias vector to $\bm{p} = (p_2, p_1)$, then by Lemma~\ref{lem:symmetric-equilibrium}, the infimum of the support of the unique mixed Nash equilibrium is  
    \[ \underline{x}' = \arg\max_x x p_1 - \gamma g(x). \]
    Thus, we have $\underline{x}' = \overline{x}$. This means that when the position bias vector is non-decreasing, creators exert substantially lower effort in equilibrium compared to the case where the position bias vector is non-increasing.
\end{example}
\section{Omitted Proofs and Discussions in Section~\ref{subsec:md-sym}} \label{sec:ap-sym-proofs}

\subsection{Reformulation of Problem~\eqref{eq:mechanism-design}} \label{subsec:ap-sym-comp-reform}

Observe that the terms $\mathbb{E}[x_{(i)}]$ and $\mu$ in the objective function of Problem \eqref{eq:mechanism-design} depend on the equilibrium CDF $F$, so they are related to the optimization variable $\bm{c}$. However, since this relationship is not explicit, direct analysis becomes challenging. Adopting an approach analogous to \cite{golrezaei2025contest}, we reformulate the objective function to depend purely on $\bm{c}$.

With compensation, the indifference condition \eqref{eq:indifference-sym} becomes
\[ x \left(\sum_{i = 1}^n \binom{n - 1}{i - 1} [F(x)]^{n - i} [1 - F(x)]^{i - 1} (p_i + c_i)\right) - \gamma g(x) = u. \]
We rearrange to obtain:
\begin{equation} \label{eq:indifference-sym-2}
    \sum_{i = 1}^n \binom{n - 1}{i - 1} [F(x)]^{n - i} [1 - F(x)]^{i - 1} (p_i + c_i) = \frac{u + \gamma g(x)}{x}.
\end{equation}
Define $K(x) = \frac{u + \gamma g(x)}{x}$. Lemma~\ref{lem:K-monotone} shows that $K(x)$ is strictly increasing over the support of the equilibrium.

\begin{lemma} \label{lem:K-monotone}
    $K'(\underline{x}) = 0$, and $K'(x) > 0$ for every $x \in (\underline{x}, \overline{x}]$; consequently, $K(x)$ is strictly increasing on $[\underline{x}, \overline{x}]$.
\end{lemma}

\begin{proof}
    Differentiating $K(x)$ yields
    \[ K'(x) = \frac{\gamma (g'(x) x - g(x)) - u}{x^2}. \]
    Define $H(x) = \gamma \left(g'(x) x - g(x)\right) - u$. Then:
    \[ H'(x) = \gamma g''(x) x > 0, \]
    where the inequality follows from the strict convexity of $g$. Hence, $H(x)$ is strictly increasing on $[\underline{x}, \overline{x}]$. Evaluating at $\underline{x}$:
    \begin{align*}
        H(\underline{x}) &= \gamma (g'(\underline{x}) \underline{x} - g(\underline{x})) - u \\
        &= \gamma (g'(\underline{x}) \underline{x} - g(\underline{x})) - (\underline{x} (p_n + c_n) - \gamma g(\underline{x})) \\
        &= \underline{x} (\gamma g'(\underline{x}) - (p_n + c_n)) = 0.
    \end{align*}
    Therefore, $H(x) > 0$ for all $x \in (\underline{x}, \overline{x}]$, implying $K'(x) > 0$ on that interval. Thus, we conclude that $K(x)$ is strictly increasing on $[\underline{x}, \overline{x}]$.
\end{proof}

Since $K(x)$ is strictly increasing, its inverse $K^{-1}$ exists on $[\underline{x}, \overline{x}]$. We may thus define:

\begin{equation} \label{eq:J-def}
    x = K^{-1}\left(\sum_{i = 1}^n \binom{n - 1}{i - 1} [F(x)]^{n - i} [1 - F(x)]^{i - 1} (p_i + c_i)\right) := J(F(x), \bm{p} + \bm{c}).
\end{equation}

Therefore, we have

\[ K(J(F(x), \bm{p} + \bm{c})) = \sum_{i = 1}^n \binom{n - 1}{i - 1} [F(x)]^{n - i} [1 - F(x)]^{i - 1} (p_i + c_i). \]

The following lemma expresses the expectation of a non-negative random variable in terms of its CDF, which will be useful in the subsequent analysis.

\begin{lemma} \label{lem:expectation-cdf}
    For a non-negative random variable $X \sim F$, the expectation can be expressed as
    \[ \mathbb{E}[X] = \int_0^{\infty} (1 - F(t)) \, \textup{d} t. \]
\end{lemma}

\begin{proof}
    By the definition of expectation, we have
    \begin{align*}
        \mathbb{E}[X] = \int_0^{\infty} x \, \textup{d}F(x) = \int_0^{\infty} \int_0^x \, \textup{d}t \, \textup{d}F(x) = \int_0^{\infty} \int_t^{\infty} \, \textup{d}F(x) \, \textup{d}t = \int_0^{\infty} (1 - F(t)) \, \textup{d}t,
    \end{align*}
    where the third equality follows from Fubini’s theorem.
\end{proof}

Using the definitions and lemmas above, the following lemma provides a reformulation of the original problem. For compensation mechanism design, where the position bias vector $\bm{p}$ is fixed, we denote the corresponding objective function in Problem~\eqref{eq:mechanism-design} by $W(\bm{c})$.

\begin{lemma} \label{lem:compensation-md-reform}
    When treating $\bm{p}$ as given, Problem \eqref{eq:mechanism-design} can be reformulated as
    \begin{equation} \label{eq:compensation-md-sym}
        \begin{aligned}
            \max_{\bm{c}} \ & W(\bm{c}) = n \int_0^1 K(J(y, \alpha\bm{p} - \bm{c})) J(y, \bm{p} + \bm{c}) \, \textup{d}y + \alpha h\left(\int_0^1 J(y, \bm{p} + \bm{c}) \, \textup{d}y\right) p_0 \\
            \textup{s.t.} \quad & c_1 \geqslant c_2 \geqslant \dots \geqslant c_n \geqslant 0.
        \end{aligned}
    \end{equation}
\end{lemma}

\begin{proof}
    This reformulation applies specifically to the objective function, while all constraints remain unchanged. We first set aside the factor $\alpha$ and the brackets, and begin by analyzing the second term of the objective function. In symmetric equilibrium, all creators adopt effort levels drawn from the same distribution $F$, implying identical expectations, denoted by $\mathbb{E}[x]$. Thus, the second term simplifies as follows:
    \begin{align*}
        h\left(\frac{1}{n} \cdot n \mathbb{E}[x]\right) p_0 &= h\left(\mathbb{E}[x]\right) p_0 \\ 
        &= h\left(\int_{\underline{x}}^{\overline{x}} J(F(x), \bm{p} + \bm{c}) \, \textup{d}F(x)\right) p_0 = h\left(\int_0^1 J(y, \bm{p} + \bm{c}) \, \textup{d}y\right) p_0,
    \end{align*}
    where the final substitution eliminates the dependence on $F$.
    
    We now turn to the first term of the objective. Let $F_{(i)}$ denote the CDF of the $i$-th order statistic, and define the auxiliary function:
    \[ H(y, \bm{p}) = \sum_{i = 1}^n \left(1 - \sum_{j = 1}^i \binom{n}{j - 1} y^{n - j + 1} (1 - y)^{j - 1}\right) p_i, \]
    then we have
    \begin{align*}
        \sum_{i=1}^n \mathbb{E}[x_{(i)}] p_i &= \sum_{i=1}^n \left(\underline{x} + \int_{\underline{x}}^{\overline{x}} (1 - F_{(i)}(x)) \, \textup{d}x\right) \cdot p_i \\
        &= \underline{x} \left(\sum_{i=1}^n p_i\right) + \int_{\underline{x}}^{\overline{x}} \sum_{i = 1}^n \left(1 - \sum_{j = 1}^i \binom{n}{j - 1} [F(x)]^{n - j + 1} [1 - F(x)]^{j - 1}\right) p_i \, \textup{d}x \\
        &= \underline{x} \left(\sum_{i=1}^n p_i\right) + \int_{\underline{x}}^{\overline{x}} \sum_{i = 1}^n \left(1 - \sum_{j = 1}^i \binom{n}{j - 1} [F(x)]^{n - j + 1} [1 - F(x)]^{j - 1}\right) p_i \frac{1}{f(x)} \, \textup{d}F(x) \\
        &= \underline{x} \left(\sum_{i=1}^n p_i\right) + \int_{\underline{x}}^{\overline{x}} H(F(x), \bm{p}) \frac{1}{f(x)} \, \textup{d}F(x).
    \end{align*}
    Here, the first equality follows from Lemma~\ref{lem:expectation-cdf}, the second from the distribution function of order statistics, and the third holds since Proposition~\ref{prop:symmetric-equilibrium-properties} ensures that 
    the probability density function $f$ exists and is positive almost everywhere on $[\underline{x}, \overline{x}]$.
    
    Differentiating both sides of equation \eqref{eq:indifference-sym-2} with respect to $x$ and applying the chain rule together with equation \eqref{eq:J-def}, we derive:
    \begin{equation} \label{eq:chain-rule-KJ}
        \frac{\text{d}[K(J(F(x), \bm{p} + \bm{c}))]}{\text{d}[F(x)]} f(x) = K'(x) = K'(J(F(x), \bm{p} + \bm{c})).
    \end{equation}
    Consequently,
    \begin{align*}
        \int_{\underline{x}}^{\overline{x}} H(F(x), \bm{p}) \frac{1}{f(x)} \, \textup{d}F(x) &= \int_{\underline{x}}^{\overline{x}} H(F(x), \bm{p}) \frac{\text{d}[K(J(F(x), \bm{p} + \bm{c}))]}{\text{d}[F(x)]} \frac{1}{K'(J(F(x), \bm{p} + \bm{c}))} \, \textup{d}F(x) \\
        &= \int_0^1 H(y, \bm{p}) \frac{\text{d}[K(J(y, \bm{p} + \bm{c}))]}{\text{d}y} \frac{1}{K'(J(y, \bm{p} + \bm{c}))} \, \textup{d}y \\
        &= \int_0^1 H(y, \bm{p}) J'(y, \bm{p} + \bm{c}) \, \textup{d}y.
    \end{align*}
    where the final equality follows from another application of the chain rule. Note that $H(0, \bm{p}) = \sum_{i=1}^n p_i, H(1, \bm{p}) = 0, J(0, \bm{p} + \bm{c}) = \underline{x}, J(1, \bm{p} + \bm{c}) = \overline{x}$, we apply integration by parts:
    \begin{align*}
        \sum_{i=1}^n \mathbb{E}[x_{(i)}] p_i &= \underline{x} \left(\sum_{i=1}^n p_i\right) + H(y, \bm{p}) J(y, \bm{p} + \bm{c}) \Big|_0^1 - \int_0^1 H'(y, \bm{p}) J(y, \bm{p} + \bm{c}) \, \textup{d}y \\
        &= n \int_0^1 K(J(y, \bm{p})) J(y, \bm{p} + \bm{c}) \, \textup{d}y.
    \end{align*}
    The final equality uses Lemma~14 of \cite{golrezaei2025contest}, which establishes $H'(y, \bm{p}) = -n K(J(y, \bm{p}))$. 
    
    Finally, we consider the third term of the objective function. By following a similar derivation as above, the total compensation cost can be reformulated as
    \[ \sum_{i=1}^n \mathbb{E}[x_{(i)}] c_i = n \int_0^1 K(J(y, \bm{c})) J(y, \bm{p} + \bm{c}) \, \textup{d}y. \]
    Note that the function $K \circ J$ is linear in its second argument, thus combining the above derivations completes the reformulation.
\end{proof}

\subsection{Proof of Theorem~\ref{thm:compensation-md-sym}} \label{subsec:ap-sym-comp-theorem}

To prove Theorem~\ref{thm:compensation-md-sym}, we first introduce some definitions.

\begin{definition}[symmetric function]
    A function $f: \mathbb{R}^n \to \mathbb{R}$ is symmetric if for any permutation $\pi$ of $[n]$ and $\bm{p} \in \mathbb{R}^n$, we have
    \[ f(p_1, p_2, \ldots, p_n) = f(p_{\pi(1)}, p_{\pi(2)}, \ldots, p_{\pi(n)}). \]
\end{definition}

We rewrite $W(\bm{c})$ (defined in Lemma~\ref{lem:compensation-md-reform}) as a symmetric function $W^{\text{sym}}: \mathbb{R}^{n} \to \mathbb{R}$, defined without constraint $c_1 \geqslant c_2 \geqslant \dots \geqslant c_n$. Let $c_{(1)} \geqslant c_{(2)} \geqslant \dots \geqslant c_{(n)}$ be the descending order of the elements in $\bm{c}$, and let $\bm{c_{(\cdot)}} = (c_{(1)}, c_{(2)}, \ldots, c_{(n)})$. Then we define
\[ W^{\text{sym}}(\bm{c}) = n \int_0^1 K(J(y, \alpha\bm{p} - \bm{c}_{(\cdot)})) J(y, \bm{p} + \bm{c}_{(\cdot)}) \, \textup{d}y +\alpha h\left(\int_0^1 J(y, \bm{p} + \bm{c}_{(\cdot)}) \, \textup{d}y\right) p_0, \]

It is clear that $W^{\text{sym}}(\bm{c})$ is symmetric. Next, we define Schur-concave functions.

\begin{definition}[Schur-concavity] \label{def:schur-concavity}
    Given two vectors $\bm{p}, \bm{y} \in \mathbb{R}^n$, we say $\bm{p}$ majorizes $\bm{y}$, denoted as $\bm{p} \succ \bm{y}$, if
    \begin{align*}
        \sum_{i = 1}^k p_{(i)} &\geqslant \sum_{i = 1}^k y_{(i)} \quad \text{for } k = 1, 2, \ldots, n - 1, \\
        \sum_{i = 1}^n p_{(i)} &= \sum_{i = 1}^n y_{(i)},
    \end{align*}
    
    A symmetric function $f: \mathbb{R}^n \to \mathbb{R}$ is Schur-concave if for any $\bm{p}, \bm{y} \in \mathbb{R}^n$ such that $\bm{p} \succ \bm{y}$, we have $f(\bm{p}) \leqslant f(\bm{y})$.
\end{definition}

The following lemma provides a useful criterion for Schur-concavity, which we will use in the proof.

\begin{lemma} \label{lem:schur-concave} \cite{ando1989majorization} 
    A concave function $f: \mathbb{R}^n \to \mathbb{R}$ is Schur-concave if and only if it is symmetric.
\end{lemma}

Then we establish two auxiliary lemmas.

\begin{lemma} \label{lem:KJ-monotone}
    Function $K(J(y, \bm{p}))$ is strictly increasing in $y \in [0, 1]$.
\end{lemma}

\begin{proof}
    Let $I(y, \bm{p})$ denote the derivative of $K(J(y, \bm{p}))$ with respect to $y$. According to Lemma~17 of \cite{golrezaei2025contest}, we have
    \begin{equation} \label{eq:KJ-derivative}
        I(y, \bm{p}) = (n - 1) \sum_{i = 1}^{n - 1} \binom{n - 2}{i - 1} y^{n - i - 1} (1 - y)^{i - 1} (p_i - p_{i + 1}).
    \end{equation}
    Recall that $p_1 > p_2 > \cdots > p_n$, hence for all $y \in (0, 1)$, $I(y, \bm{p}) > 0$, thus $K(J(y, \bm{p}))$ is strictly increasing in $y$ on $[0, 1]$.
\end{proof}

\begin{lemma} \label{lem:md-sym-aux}
    If $g(x) = x^{\beta}$ with $\beta \geqslant 2$, then the function $\frac{2(K'(x))^2}{K''(x)} - K(x)$ is strictly increasing on $[\underline{x}, \overline{x}]$.
\end{lemma}

\begin{proof}
    Define $L(x) = \frac{2(K'(x))^2}{K''(x)} - K(x)$. Differentiation yields:
    \[ L'(x) = \frac{K'(x)}{(K''(x))^2} \left(3(K''(x))^2 - 2K'(x) K'''(x)\right). \]
    The derivatives of $K(x)$ are computed as follows:
    \begin{align*}
        K'(x) &= \gamma (\beta - 1) x^{\beta - 2} - \frac{u}{x^2}, \\
        K''(x) &= \gamma (\beta - 1)(\beta - 2) x^{\beta - 3} + \frac{2u}{x^3}, \\
        K'''(x) &= \gamma (\beta - 1)(\beta - 2)(\beta - 3) x^{\beta - 4} - \frac{6u}{x^4}.
    \end{align*}
    Substituting these into the expression for $L'(x)$ and simplifying, we obtain:
    \[ L'(x) = \frac{K'(x)}{(K''(x))^2} \cdot \gamma\beta(\beta - 1)\left[2u(\beta + 1)x^{\beta - 6} + \gamma(\beta - 1)(\beta - 2)x^{2\beta - 6}\right]. \]
    By Lemma~\ref{lem:K-monotone}, $K'(x) > 0$ on $(\underline{x}, \overline{x}]$. Given $\beta \geqslant 2$, all terms within the bracket are non-negative, and the leading coefficient is positive. Therefore, $L'(x) > 0$ on $(\underline{x}, \overline{x}]$, which implies that $L(x)$ is strictly increasing on $[\underline{x}, \overline{x}]$.
\end{proof}

Let $b_i(y) = \binom{n - 1}{i - 1} y^{n - i} (1 - y)^{i - 1}$ for $i \in [n]$. Note that $b_i(y) \geqslant 0$ for all $y \in [0, 1]$. Using the above lemmas, we now prove Theorem~\ref{thm:compensation-md-sym}.

\begin{proofof}{Theorem~\ref{thm:compensation-md-sym}}
    Consider the two terms in $W^{\text{sym}}(\bm{c})$ separately. Let $\pi$ be a permutation such that $c_{\pi(1)} \geqslant c_{\pi(2)} \geqslant \dots \geqslant c_{\pi(n)}$, and let $\bm{c}_{-n} = (c_1, c_2, \ldots, c_{n - 1})$. Consider the second term of $W^{\text{sym}}(\bm{c})$. The derivative of $J(y, \bm{p} + \bm{c}_{(\cdot)})$ with respect to $c_i$ for $i = 1, \dots, n - 1$ is:
    \[ \frac{\partial J(y, \bm{p} + \bm{c}_{(\cdot)})}{\partial c_i} = \frac{b_{\pi^{-1}(i)}(y)}{K'(J(y, \bm{p} + \bm{c}_{(\cdot)}))}, \]
    Differentiating again with respect to $c_j$ for $j = 1, \dots, n - 1$ gives:
    \[ \frac{\partial^2 J(y, \bm{p} + \bm{c}_{(\cdot)})}{\partial c_i \partial c_j} = -\frac{b_{\pi^{-1}(i)}(y) b_{\pi^{-1}(j)}(y) K''(J(y, \bm{p} + \bm{c}_{(\cdot)}))}{(K'(J(y, \bm{p} + \bm{c}_{(\cdot)})))^3}, \]
    
    Since $b_i(y) \geqslant 0$, $K'(x) > 0$, and $K''(x) > 0$, the above expression is non-positive. Hence, $J(y, \bm{p} + \bm{c}_{(\cdot)})$ is concave in $\bm{c}_{-n}$, and so is $\int_0^1 J(y, \bm{p} + \bm{c}_{(\cdot)}) \textup{d}y$. Since $h$ is increasing and concave, $\alpha > 0$, $\alpha h\left(\int_0^1 J(y, \bm{p} + \bm{c}_{(\cdot)}) \textup{d}y\right) p_0$ is also concave in $\bm{c}_{-n}$.

    Now consider the first term of $W^{\text{sym}}(\bm{c})$. The derivative of the integrand with respect to $c_i$ for $i = 1, \dots, n - 1$ is:
    \[ -b_{\pi^{-1}(i)}(y) J(y, \bm{p} + \bm{c}_{(\cdot)}) + \frac{K(J(x, \alpha\bm{p} - \bm{c}_{(\cdot)})) b_{\pi^{-1}(i)}(y)}{K'(J(y, \bm{p} + \bm{c}_{(\cdot)}))}, \]
    Differentiating again with respect to $c_j$ for $j = 1, \dots, n - 1$ gives:
    \begin{align*}
        &\quad H_{ij}(y, \bm{c}) \\
        &= -\frac{2b_{\pi^{-1}(i)}(y) b_{\pi^{-1}(j)}(y)}{K'(J(y, \bm{p} + \bm{c}_{(\cdot)}))} - \frac{b_{\pi^{-1}(i)}(y) b_{\pi^{-1}(j)}(y) K''(J(y, \bm{p} + \bm{c}_{(\cdot)})) K(J(y, \alpha\bm{p} - \bm{c}_{(\cdot)}))}{(K'(J(y, \bm{p} + \bm{c}_{(\cdot)})))^3} \\
        &= -\left( \frac{2(K'(J(y, \bm{p} + \bm{c}_{(\cdot)})))^2}{K''(J(y, \bm{p} + \bm{c}_{(\cdot)}))} + K(J(y, \alpha\bm{p} - \bm{c}_{(\cdot)})) \right) \frac{b_{\pi^{-1}(i)}(y) b_{\pi^{-1}(j)}(y) K''(J(y, \bm{p} + \bm{c}_{(\cdot)}))}{(K'(J(y, \bm{p} + \bm{c}_{(\cdot)})))^3} \\
        &= \left( K(J(y, \bm{p} + \bm{c}_{(\cdot)})) - \frac{2(K'(J(y, \bm{p} + \bm{c}_{(\cdot)})))^2}{K''(J(y, \bm{p} + \bm{c}_{(\cdot)}))} - (\alpha + 1)K(J(y, \bm{p})) \right) \\ 
        &\quad \cdot \frac{b_{\pi^{-1}(i)}(y) b_{\pi^{-1}(j)}(y) K''(J(y, \bm{p} + \bm{c}_{(\cdot)}))}{(K'(J(y, \bm{p} + \bm{c}_{(\cdot)})))^3}.
    \end{align*}
    
    By Lemma~\ref{lem:md-sym-aux}, the first two terms in the parentheses decrease with $y$, and by Lemma~\ref{lem:KJ-monotone}, the third term increases with $y$. Thus, the expression in parentheses decreases with $y$. Since the factor outside the parentheses is positive, if the parenthesized expression is non-positive at $y = 0$, then $H_{ij}(y, \bm{c}) \leqslant 0$ for any $y \in [0, 1]$. At $y = 0$, the parenthesized expression equals:
    \[ p_n + c_n - (\alpha + 1)p_n = c_n - \alpha p_n, \]
    Therefore, if $c_n \leqslant \alpha p_n$, then $H_{ij}(y, \bm{c}) \leqslant 0$. Let $H(y, \bm{c})$ be the matrix $(H_{ij}(y, \bm{c}))_{i, j = 1}^{n - 1}$, and define
    \[ L(y, \bm{c}) = \left( K(J(y, \bm{p} + \bm{c}_{(\cdot)})) - \frac{2(K'(J(y, \bm{p} + \bm{c}_{(\cdot)})))^2}{K''(J(y, \bm{p} + \bm{c}_{(\cdot)}))} - (\alpha + 1)K(J(y, \bm{p})) \right) \frac{K''(J(y, \bm{p} + \bm{c}_{(\cdot)}))}{(K'(J(y, \bm{p} + \bm{c}_{(\cdot)})))^3}, \]
    Then for any $\bm{z} \in \mathbb{R}^{n - 1}$,
    \begin{align*}
        \bm{z}^\top H(y, \bm{c}) \bm{z} = L(y, \bm{c}) \cdot \sum_{i = 1}^{n - 1} \sum_{j = 1}^{n - 1} z_i z_j b_{\pi^{-1}(i)}(y) b_{\pi^{-1}(j)}(y) = L(y, \bm{c}) \cdot \left(\sum_{i = 1}^{n - 1} z_i b_{\pi^{-1}(i)}(y)\right)^2 \leqslant 0.
    \end{align*}
    Hence, $H(y, \bm{c})$ is negative semidefinite, and the integral in the first term of $W^{\text{sym}}(\bm{c})$ is concave in $\bm{c}_{-n}$. Combining this with the concavity of the second term, by Lemma~\ref{lem:schur-concave}, $W^{\text{sym}}(\bm{c})$ is Schur-concave in $\bm{c}_{-n}$ when $c_n \leqslant \alpha p_n$. Therefore, there exists an optimal solution with $c_1 = c_2 = \cdots = c_{n - 1}$ when $c_n \leqslant \alpha p_n$.
\end{proofof}

\subsection{Discussion on the optimal value of \texorpdfstring{$c_n$}{c\_n}} \label{subsec:optimal-cn}

In this subsection, we discuss the condition on $c_n$ in Theorem~\ref{thm:compensation-md-sym} and analyze its optimal value in realistic scenarios.

The condition $c_n \leqslant \alpha p_n$ in Theorem~\ref{thm:compensation-md-sym} does not appear in previous literature on rank-based contests \cite{glazer1988optimal,barut1998symmetric,golrezaei2025contest}. In those models, the utility function is simply reward minus cost, which implies that the infimum equilibrium effort is always zero (as explained in the discussion below Proposition~\ref{prop:symmetric-equilibrium-properties}). Consequently, the optimal compensation for the last position would intuitively also be zero, because setting a positive $c_n$ would raise compensation costs without eliciting higher effort, and therefore cannot be optimal. In our setting, however, raising $c_n$ increases the infimum effort $\underline{x}$. Hence, it is not immediate that the optimal $c_n$ should be zero. In fact, we can give an example in which the optimal $c_n$ is not zero.

\begin{example} \label{ex:cn-increasing}
    If $n = 2$, $p_1 = 0.7$, $p_2 = 0.1$, $\alpha = 1$, $g(x) = x^2$, $\gamma = 1$, and $h(x) = \sqrt{x}$, then Figure~\ref{fig:cn-increasing} plots the optimal value of $W(\bm{c})$ for each value of $c_2$. The plot shows that the optimal $c_2$ is approximately $0.025 > 0$.
\end{example}

\begin{figure}[t]
    \centering
    \includegraphics[width=0.4\textwidth]{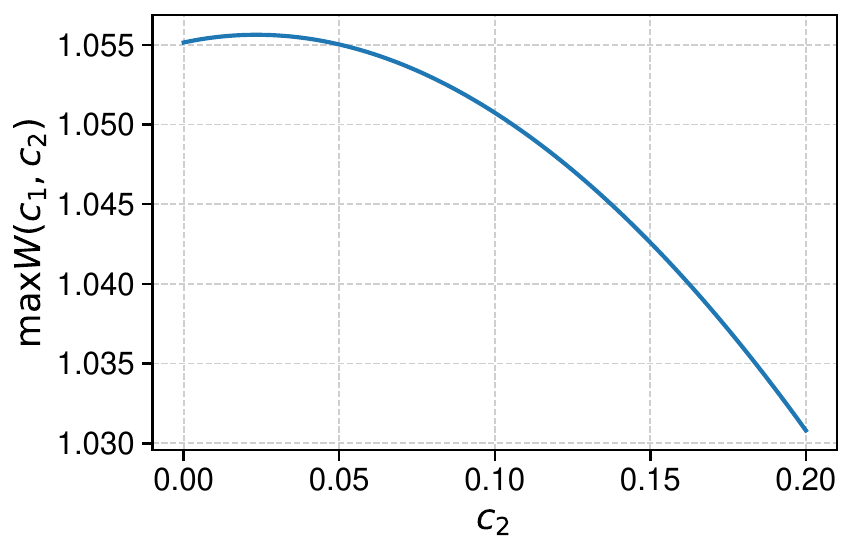} 
    \caption{Plot of $\max \ W(\bm{c})$ against $c_n$ in Example~\ref{ex:cn-increasing}}
    \label{fig:cn-increasing}
\end{figure}

The numerical experiments below demonstrate that the scenario illustrated in Example~\ref{ex:cn-increasing} is uncommon in practice. Taking the second line of Table~\ref{tab:position-bias-no-ai-overview} as a realistic position bias vector for compensation mechanism design, we conduct numerical experiments across a range of parameter settings. These experiments consistently show that the optimal value of $c_n$ is usually zero. We also examine the reasons why $c_n > 0$ can arise in Example~\ref{ex:cn-increasing}.

Figures~\ref{fig:cn-experiments-1}–\ref{fig:cn-experiments-3} examine the optimal value of $c_n$ under different choices of $g, \alpha, h$. We do not vary $\gamma$ since Proposition~\ref{prop:asym-optimal-ch-cl-2} shows that the value of $\gamma$ does not affect the structure of the equilibrium. For each $c_n$, we compute the optimal objective value $W(\bm{c})$ via grid search; this is what the vertical axis $\max\ W(\bm{c})$ in the figures represents. In all three experiments, the position bias vector is divided equally into two halves: the first five positions and the last five positions, corresponding to the setting $n_H = n_L = 5$ in Section~\ref{sec:asymmetric}. Sub-figure (a) in each panel displays results for the first five positions, while sub-figure (b) shows results for the last five. Parameters that are not varied are held at their default values: $\alpha = 2, g(x) = x^4, \gamma = 1, h(x) = \sqrt{x}$.

The plots show that the optimal value of $c_n$ is always $0$. The overall trend is hardly affected by changes in $\alpha$ and $h$, but the choice of $g$ has a visible effect: a lower power in $g$ corresponds to a more gradual decline. This can be understood intuitively: a smaller power in $g$ shifts the equilibrium towards lower effort levels (a smaller exponent makes $g$ larger on $[0, 1]$), which enhances the positive incentive effect of increasing $c_n$ and hence slows the decline. In Example~\ref{ex:cn-increasing}, where the power of $g$ is $2$, the observed trend is consistent with this explanation.

\begin{figure}[t]
    \centering
    \captionsetup[subfigure]{justification=centering}
    \begin{subfigure}[b]{0.36\linewidth}
        \centering
        \includegraphics[width=\linewidth]{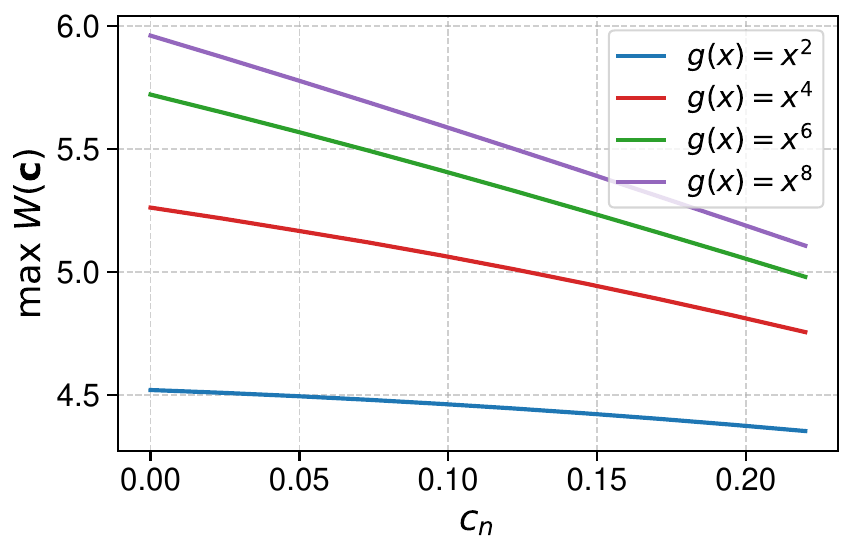}
        \caption{First group of position biases}
        \label{fig:cn-g1}
    \end{subfigure}
    \begin{subfigure}[b]{0.36\linewidth}
        \centering
        \includegraphics[width=\linewidth]{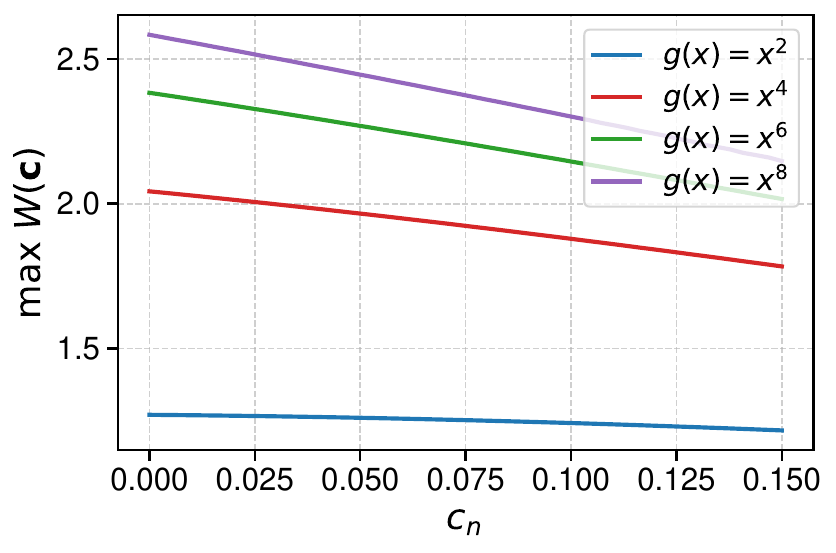}
        \caption{Second group of position biases}
        \label{fig:cn-g2}
    \end{subfigure}
    \caption{Plots of $\max W(\bm{c})$ against $c_n$ with different $g$}
    \label{fig:cn-experiments-1}
\end{figure}

\begin{figure}[t]
    \centering
    \captionsetup[subfigure]{justification=centering}
    \begin{subfigure}[b]{0.36\linewidth}
        \centering
        \includegraphics[width=\linewidth]{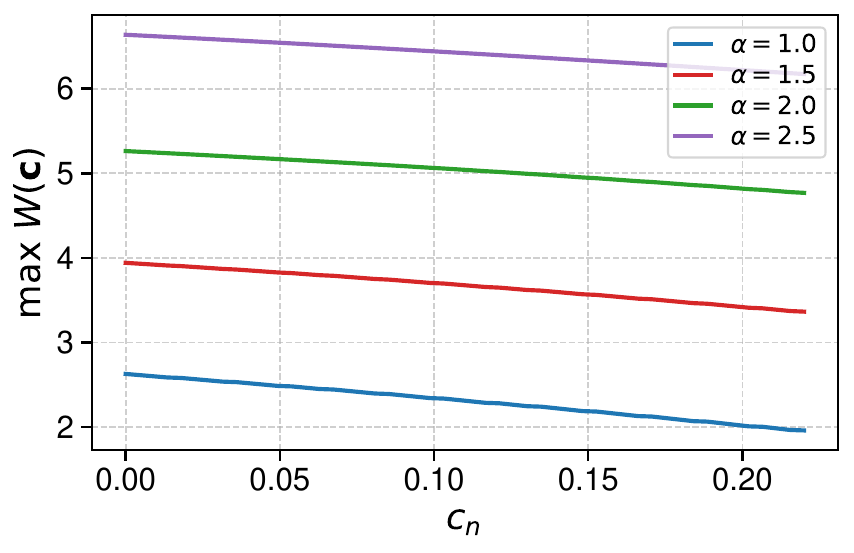}
        \caption{First group of position biases}
        \label{fig:cn-alpha1}
    \end{subfigure}
    \begin{subfigure}[b]{0.36\linewidth}
        \centering
        \includegraphics[width=\linewidth]{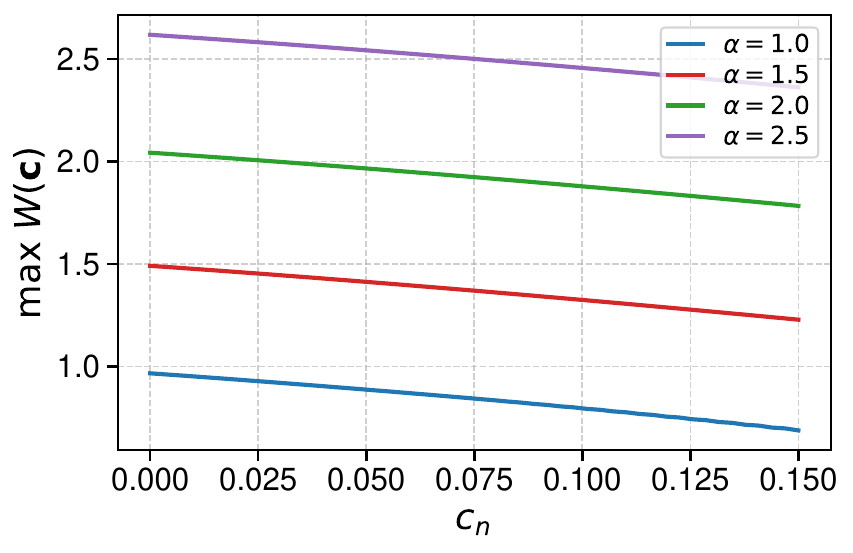}
        \caption{Second group of position biases}
        \label{fig:cn-alpha2}
    \end{subfigure}
    \caption{Plots of $\max W(\bm{c})$ against $c_n$ with different $\alpha$}
    \label{fig:cn-experiments-2}
\end{figure}

\begin{figure}[t]
    \centering
    \captionsetup[subfigure]{justification=centering}
    \begin{subfigure}[b]{0.36\linewidth}
        \centering
        \includegraphics[width=\linewidth]{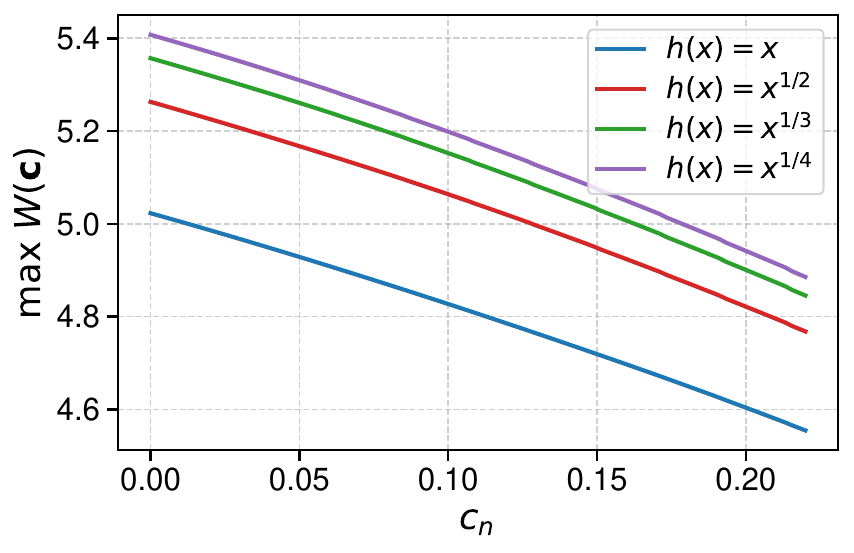}
        \caption{First group of position biases}
        \label{fig:cn-h1}
    \end{subfigure}
    \begin{subfigure}[b]{0.36\linewidth}
        \centering
        \includegraphics[width=\linewidth]{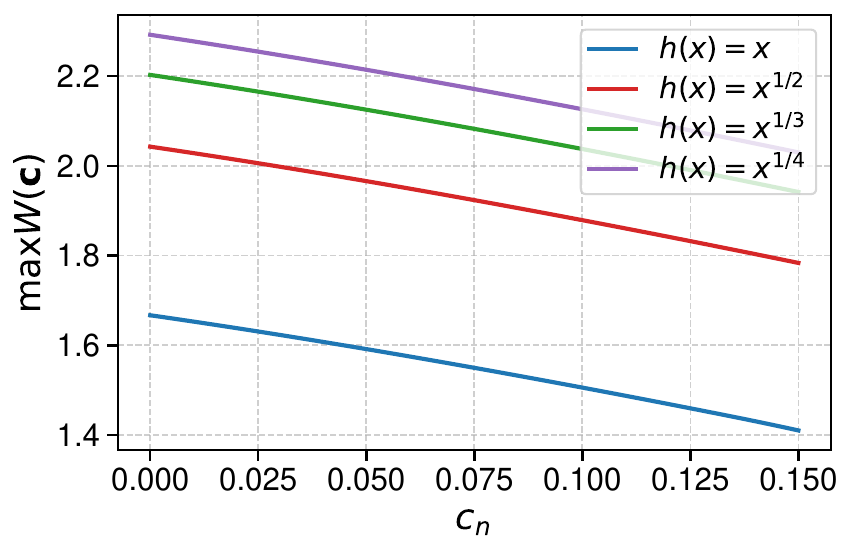}
        \caption{Second group of position biases}
        \label{fig:cn-h2}
    \end{subfigure}
    \caption{Plots of $\max W(\bm{c})$ against $c_n$ with different $h$}
    \label{fig:cn-experiments-3}
\end{figure}

\begin{figure}[t]
    \centering
    \captionsetup[subfigure]{justification=centering}
    \begin{subfigure}[b]{0.4\linewidth}
        \centering
        \includegraphics[width=\linewidth]{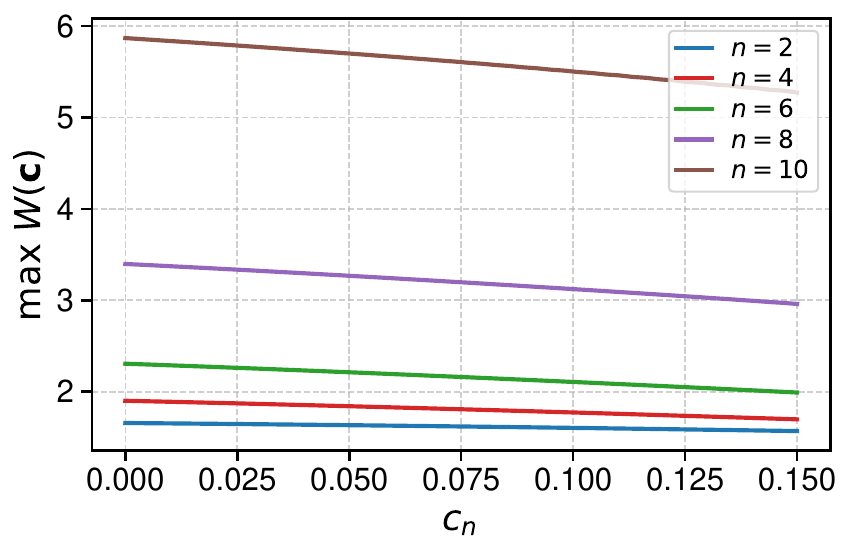}
        \caption{First group of position biases}
        \label{fig:cn-player1}
    \end{subfigure}
    \begin{subfigure}[b]{0.4\linewidth}
        \centering
        \includegraphics[width=\linewidth]{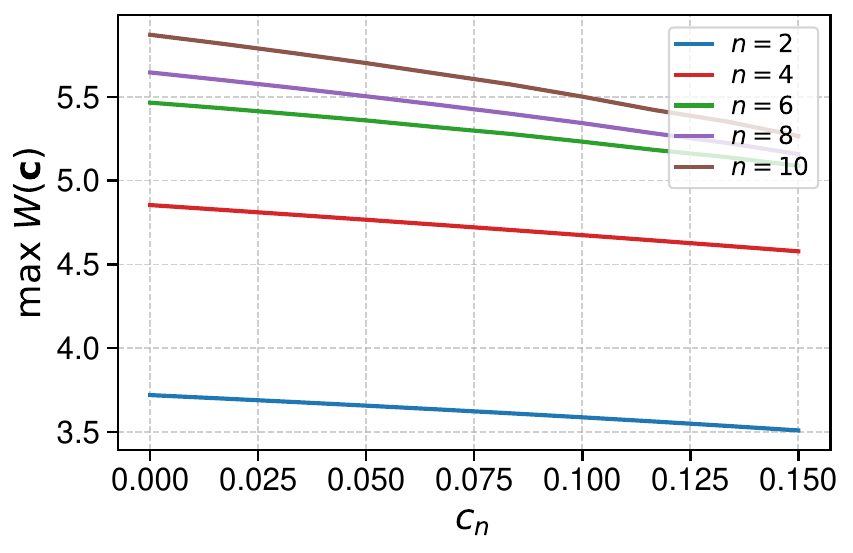}
        \caption{Second group of position biases}
        \label{fig:cn-player2}
    \end{subfigure}
    \caption{Plots of $\max W(\bm{c})$ against $c_n$ with different number of creators}
    \label{fig:cn-experiments-4}
\end{figure}

Figure~\ref{fig:cn-experiments-4} examines the effect of the number of creators $n$, while other parameters remain at their default values: $\alpha = 2, g(x) = x^4, \gamma = 1, h(x) = \sqrt{x}$. In Figure~\ref{fig:cn-player1}, $n$ corresponds to $n_L$ in Section~\ref{sec:asymmetric}, thus for $n = 2$, the position bias vector corresponds to the last two elements of the entire position bias vector; for $n = 4$, it corresponds to the last four elements, and so on. Similarly, in Figure~\ref{fig:cn-player2}, $n$ corresponds to $n_H$ in Section~\ref{sec:asymmetric}; thus for $n = 2$, it corresponds to the first two elements of the position bias vector, and so on. 

The experimental results show that when $n$ is small, the decline is gradual, a pattern consistent with Example~\ref{ex:cn-increasing}. However, in this experiment, we do not observe cases where the optimal $c_n > 0$. Intuitively, this is because in Example~\ref{ex:cn-increasing}, the differences between adjacent position biases are large, making the positive incentive effect of increasing $c_n$ for lower ranks significant, while the negative incentive effect for higher ranks is minimal, leading to $c_n > 0$. In contrast, in our experiment, the differences between adjacent positions in the position bias vector are relatively small; therefore we do not observe cases where $c_n > 0$.


\subsection{Grid Search for Compensation Mechanism Design under Symmetric Case} \label{subsec:grid-search-details}

Setting $c^* := c_1 = \cdots = c_{n - 1}$ and recalling from Appendix~\ref{subsec:optimal-cn} that $c_n = 0$ is typically optimal in realistic settings, Theorem~\ref{thm:compensation-md-sym} reduces the compensation design problem to the choice of a single key parameter: $c^*$. An approximate solution can be efficiently obtained via a grid search in the feasible domain.

A remaining issue, however, is that the search for $c^*$ is a priori unbounded. The proof of Theorem~\ref{thm:compensation-md-sym} implies that $\frac{\partial W(\bm{c})}{\partial c_i}$ decreases in $c_j$ for any $i, j \in [n - 1]$. Consequently, according to the chain rule, the total derivative of $W(\bm{c})$ with respect to $c^*$ is also decreasing. Hence, a natural heuristic is to gradually increase $c^*$ until $W(\bm{c})$ starts to decrease, since further increases of $c^*$ beyond this point would only lower $W(\bm{c})$. The following proposition shows that this process terminates, making the approach feasible.

\begin{proposition} \label{prop:compensation-grid-search-valid}
    There exists $\overline{c} > 0$ such that for all $c^* \geqslant \overline{c}$ and all $i \in [n - 1]$, $\frac{\partial W(\bm{c^*})}{\partial c_i} < 0$.
\end{proposition}

\begin{proof}
    For any $i \in [n - 1]$, the expression for $\frac{\partial W(\bm{c})}{\partial c_i}$ is:
    \begin{equation} \label{eq:W-derivative-c_i}
        \begin{aligned}
            \frac{\partial W(\bm{c})}{\partial c_i} &= n \int_0^1 -b_i(y) J(y, \bm{p} + \bm{c}) + K(J(y, \alpha\bm{p} - \bm{c})) \cdot \frac{b_i(y)}{K'(J(y, \bm{p} + \bm{c}))} \, \textup{d}y \\
            &\quad + \alpha h'\left(\int_0^1 J(y, \bm{p} + \bm{c}) \, \textup{d}y\right) p_0 \cdot \int_0^1 \frac{b_i(y)}{K'(J(y, \bm{p} + \bm{c}))} \, \textup{d}y.
        \end{aligned}
    \end{equation}
    Since $b_i(y) = 0$ and $K'(J(y, \bm{p} + \bm{c})) = 0$ at $y = 0$, we evaluate the limit of $b_i(y) / K'(J(y, \bm{p} + \bm{c}))$ as $y \to 0$ using L’Hôpital’s rule:
    \[ \lim_{y \to 0} \frac{b_i (y)}{K'(J(y, \bm{p} + \bm{c}))} = \lim_{y \to 0} \frac{b_i' (y)}{[K'(J(y, \bm{p} + \bm{c}))]'} = \lim_{y \to 0} \frac{b_i' (y) K'(J(y, \bm{p} + \bm{c}))}{K''(J(y, \bm{p} + \bm{c})) [K(J(y, \bm{p} + \bm{c}))]'}. \]
    If $i \neq n - 1$, then $b_i'(0) = 0$; if $i = n - 1$, then $b_i'(0) = n - 1$. Moreover, $K'(J(0, \bm{p} + \bm{c})) = K'(\underline{x}) = 0$, and by strict convexity of $K$, $K''(J(0, \bm{p} + \bm{c})) = K''(\underline{x}) > 0$. Also, from equation \eqref{eq:KJ-derivative}, $[K(J(0, \bm{p} + \bm{c}))]' = I(0, \bm{p} + \bm{c}) = (n - 1)((p_{n - 1} + c_{n - 1}) - (p_n + c_n)) > 0$. Hence, $\lim_{y \to 0} \frac{b_i (y)}{K'(J(y, \bm{p} + \bm{c}))} = 0$ for all $i = 1, 2, \ldots, n - 1$, implying that $b_i(y) / K'(J(y, \bm{p} + \bm{c}))$ is continuous on $[0, 1]$. This justifies differentiation under the integral, so the expression for $\partial W / \partial c_i$ is valid.
    
    For $t \geqslant 0$, define $\bm{c}(t) := (t, \ldots, t, 0), J_t(y) := J(y, \bm{p} + \bm{c}(t)), \zeta(y) := \sum_{k = 1}^{n - 1} b_k(y) = 1 - (1 - y)^{n - 1}$. Therefore, by \eqref{eq:W-derivative-c_i}, for every $i \in [n - 1]$,
    \begin{align*}
        \frac{\partial W(\bm{c}(t))}{\partial c_i}
        &= n \int_0^1 \left(
            -b_i(y) J_t(y)
            + \left(\sum_{k = 1}^n \alpha b_k(y)p_k - \zeta(y)t\right)
            \frac{b_i(y)}{K'(J_t(y))}
        \right)\textup{d}y \\
        &\quad + \alpha p_0 h'\left(\int_0^1 J_t(y)\,\textup{d}y\right)
        \int_0^1 \frac{b_i(y)}{K'(J_t(y))}\,\textup{d}y.
    \end{align*}

    Since $\sum_{k = 1}^n b_k(y)p_k \leqslant p_1$, $\zeta(y)t \geqslant 0$, and $h'$ is decreasing while $J_t(y) \geqslant J_t(0) = \underline{x}$ for all $y \in [0,1]$, it follows that
    \begin{align}
        \frac{\partial W(\bm{c}(t))}{\partial c_i}
        &\leqslant
        -n \int_0^1 b_i(y) J_t(y)\,\textup{d}y
        + \left(\alpha n p_1 + \alpha p_0 h'(\underline{x})\right)
        \int_0^1 \frac{b_i(y)}{K'(J_t(y))}\,\textup{d}y.
        \label{eq:grid-search-proof-upper-bound}
    \end{align}

    We now bound the two terms on the right-hand side. First, $J_t(y)$ is increasing in $t$ for every $y \in [0,1]$, because $K(J_t(y)) = \sum_{k = 1}^n b_k(y)p_k + \zeta(y)t$ is increasing in $t$ and $K$ is strictly increasing. Hence,
    \[
        \int_0^1 \frac{b_i(y)}{K'(J_t(y))}\,\textup{d}y
        \leqslant
        \int_0^1 \frac{b_i(y)}{K'(J_0(y))}\,\textup{d}y
        =: C_i < \infty.
    \]
    Thus the second term on the right-hand side of \eqref{eq:grid-search-proof-upper-bound} is uniformly bounded. Second, fix any closed interval $I_i \subset (0,1)$. Since $b_i(y) > 0$ and $\zeta(y) > 0$ for all $y \in (0,1)$, there exist constants $m_i := \min_{y \in I_i} b_i(y) > 0, \eta_i := \min_{y \in I_i} \zeta(y) > 0$. Also, $\sum_{k = 1}^n b_k(y)p_k \geqslant p_n$, so for every $y \in I_i$, $K(J_t(y)) \geqslant p_n + \eta_i t$. Since $K$ is increasing, $J_t(y) \geqslant K^{-1}(p_n + \eta_i t)$. Therefore,
    \[
        \int_0^1 b_i(y)J_t(y)\,\textup{d}y
        \geqslant
        \int_{I_i} b_i(y)J_t(y)\,\textup{d}y
        \geqslant
        |I_i|\,m_i\,K^{-1}(p_n + \eta_i t).
    \]
    Because $g(x)=x^\beta$ with $\beta \geqslant 2$, we have $K(x) \to \infty$ as $x \to \infty$, so $K^{-1}(s) \to \infty$ as $s \to \infty$. Hence $\int_0^1 b_i(y)J_t(y)\,\textup{d}y \to \infty$ as $t \to \infty$. Combining this with the fact that the second term on the right-hand side of \eqref{eq:grid-search-proof-upper-bound} is uniformly bounded, we have $\frac{\partial W(\bm{c}(t))}{\partial c_i} \to -\infty$ as $t \to \infty$ for every $i \in [n - 1]$. Therefore, for each $i \in [n - 1]$, there exists $\overline{c}_i > 0$ such that $\frac{\partial W(\bm{c}(t))}{\partial c_i} < 0, \forall t \geqslant \overline{c}_i$. Then we take $\overline{c} := \max_{i \in [n - 1]} \overline{c}_i$, and the proof is complete.
\end{proof}

Based on the above discussion, Algorithm~\ref{alg:compensation-grid} outlines the grid search procedure to approximate the solution to the compensation mechanism design problem.

\begin{algorithm}[t]
    \caption{Grid Search for Compensation Mechanism Design}
    \label{alg:compensation-grid}
    \begin{algorithmic}[1]
        \REQUIRE Number of creators $n$, step sizes $\Delta$, position bias vector $\bm{p}$, $p_0$
        \ENSURE Approximate optimal compensation vector $\bm{c}$
        \STATE Initialize best objective value $best\_obj \gets -\infty$
        \STATE Initialize best compensation vector $best\_c \gets \text{null}$
        \FOR{$c^*$ from $0$ with step size $\Delta$}
            \STATE $\bm{c} \gets (c^*, c^*, \ldots, c^*, 0)$
            \STATE Compute objective value $obj$ using reformulated problem \eqref{eq:compensation-md-sym}
            \IF{$obj > best\_obj$}
                \STATE $best\_obj \gets obj$
                \STATE $best\_c \gets \bm{c}$
            \ELSE
                \STATE \textbf{break}
            \ENDIF
        \ENDFOR
        \RETURN $best\_c$
    \end{algorithmic}
\end{algorithm}

\subsection{Proof of Proposition~\ref{prop:UBL-inc-sym}}

\begin{proofof}{Proposition~\ref{prop:UBL-inc-sym}}
    First, following a derivation analogous to that of \eqref{eq:W-derivative-c_i}, the partial derivative of $W(\bm{p}, \bm{c})$ with respect to $p_i$ is  
    \begin{equation} \label{eq:W-derivative-p_i}
        \begin{aligned}
            \frac{\partial W(\bm{p}, \bm{c})}{\partial p_i} &= n \int_0^1 \alpha b_i (y) J(y, \bm{p} + \bm{c}) + K(J(y, \alpha\bm{p} - \bm{c})) \cdot \frac{b_i (y)}{K'(J(y, \bm{p} + \bm{c}))}\textup{d}y \\ 
            &\quad + \alpha h'\left(\int_0^1 J(y, \bm{p} + \bm{c}) \textup{d}y\right) p_0 \cdot \int_0^1 \frac{b_i (y)}{K'(J(y, \bm{p} + \bm{c}))} \textup{d}y.
        \end{aligned}
    \end{equation}
    Consequently,
    \[ \frac{\partial W(\bm{p}, \bm{c})}{\partial p_i} - \frac{\partial W(\bm{p}, \bm{c})}{\partial c_i} = (\alpha + 1) n \int_0^1 b_i(y) J(y, \bm{p} + \bm{c}) \, \textup{d}y > 0. \]

    Thus, if we keep $\bm{p} + \bm{c}$ fixed, then for any pair $(p_i, c_i)\;(i \in [n - 1])$, increasing $p_i$ while decreasing $c_i$ by the same amount raises $W(\bm{p}, \bm{c})$. Let $s_i = q_i \Delta p_i^B$ for $i \in [n]$ and $\bm{s} = (s_1, \ldots, s_n)$. Define $r_i = \min\{s_i, c_i\}$ and $\bm{r} = (r_1, \dots, r_n)$. Then
    \begin{equation} \label{eq:UBL-inc-sym-1}
        W(\bm{p} + \bm{r}, \bm{c} - \bm{r}) > W(\bm{p}, \bm{c}).
    \end{equation}

    Under the UBL mechanism, all $s_i$ are equal for $i \in [n - 1]$ and all $c_i$ are equal as well, while $s_n = c_n = 0$. Hence $\bm{r}$ can take two possible forms.
    \begin{enumerate}
        \item Case 1: $\bm{r} = \bm{s}$. This occurs when $s_i \leqslant c_i$ for every $i \in [n - 1]$. In this case, equation~\eqref{eq:UBL-inc-sym-1} directly yields the statement of the theorem.
        \item Case 2: $\bm{r} = \bm{c}$. This occurs when $c_i < s_i$ for every $i \in [n - 1]$. Then $W(\bm{p} + \bm{r}, \bm{c} - \bm{r}) = W(\bm{p} + \bm{c}, \bm{0})$. From equation~\eqref{eq:W-derivative-p_i} we see that $W(\bm{p} + \bm{c}, \bm{0})$ is increasing in each $p_i\;(i \in [n - 1])$. Therefore,
        \[ W(\bm{p} + \bm{s}, \bm{0}) > W(\bm{p} + \bm{c}, \bm{0}) = W(\bm{p} + \bm{r}, \bm{c} - \bm{r}). \]
        Combining this inequality with \eqref{eq:UBL-inc-sym-1} completes the proof of the theorem.
    \end{enumerate}
\end{proofof}

\section{Properties of Asymmetric Equilibrium} \label{sec:ap-asymmetric}

\subsection{General Results for Arbitrary Cost Structures} \label{subsec:ap-asym-general}

We analyze the equilibrium properties of game $\mathcal{G}^{(n)}$ in this subsection. Without loss of generality, we assume that the cost parameters satisfy $\gamma_1 \leqslant \gamma_2 \leqslant \cdots \leqslant \gamma_n$. The next proposition shows that, in any equilibrium, creators with lower cost parameters always achieve higher utilities.

\begin{proposition} \label{prop:asym-utility-order}
    In game $\mathcal{G}^{(n)}$, if $\gamma_i < \gamma_j$, then $u_i > u_j$ in any equilibrium.
\end{proposition}

\begin{proof}
    Proof by contradiction. Suppose $\gamma_i < \gamma_j$ but $u_i \leqslant u_j$. Consider the supremum of the support of $j$, denoted by $\overline{x}_j$, we have $F_j(\overline{x}_j) = 1$ and $F_i(\overline{x}_j) \leqslant 1$, thus
    \begin{align*}
        W(\bm{F}_{-i}(\overline{x}_j), \bm{p}) &= W(\bm{F}_{-\{i,j\}}(\overline{x}_j), \underline{\bm{p}}) + F_j(\overline{x}_j) W(\bm{F}_{-\{i,j\}}(\overline{x}_j), \overline{\bm{p}} - \underline{\bm{p}}) \\
        &\geqslant W(\bm{F}_{-\{i,j\}}(\overline{x}_j), \underline{\bm{p}}) + F_i(\overline{x}_j) W(\bm{F}_{-\{i,j\}}(\overline{x}_j), \overline{\bm{p}} - \underline{\bm{p}}) \\
        &= W(\bm{F}_{-j}(\overline{x}_j), \bm{p}).
    \end{align*}
    Together with the assumption $\gamma_i < \gamma_j$, we have
    \[ \overline{x}_j W(\bm{F}_{-i}(\overline{x}_j), \bm{p}) - \gamma_i g(\overline{x}_j) > \overline{x}_j W(\bm{F}_{-j}(\overline{x}_j), \bm{p}) - \gamma_j g(\overline{x}_j) = u_j \geqslant u_i. \]
    Hence, creator $i$ can increase her utility by deviating to $\overline{x}_j$, contradicting the assumption that the strategy profile is an equilibrium. Therefore, we must have $u_i > u_j$.
\end{proof}

Let $(F_1, F_2, \ldots, F_n)$ denote the mixed strategy profile in equilibrium. The next result generalizes Theorem~\ref{thm:sym-equ-property}. It shows that whenever a subset of creators share the same cost parameter, their equilibrium strategies must be symmetric.

\begin{proposition} \label{prop:not-all-symmetric-equilibrium}
    In game $\mathcal{G}^{(n)}$, if there exists a subset of creators $S \subseteq [n]$ such that for any $i, j \in S$, $\gamma_i = \gamma_j$, then the equilibrium must be symmetric 
    among all creators with the same cost parameter, i.e., for any $i, j \in S$, $F_i = F_j$.
\end{proposition}

\begin{proof}
    It is not hard to see that Step~1, Step~4, Step~5 and Step~6 in the proof of Lemma~\ref{lem:symmetric-equilibrium} still hold for any subset of creators with identical cost parameters. In fact, Step~2 is not important for this corollary, and Step~3 can be modified as follows. First, $x_i = x_{\max} := \arg\max_x x q_i - \gamma g(x)$ still holds for any creator $i$ in the subset. Thus, the point mass of any creator in the subset can only be on the same point $\underline{x} := \arg\max_x x q - \gamma g(x)$, otherwise $\max_x x q_i - \gamma g(x) \neq \max_x x q_j - \gamma g(x)$ for some $q_i \neq q_j$, which contradicts the result of Step~1. Then following the argument in Step~3 we can show that there is no point mass in the probability distribution for any creator in the subset. Combining the results of Step~1, Step~4, Step~5 and Step~6, we can derive that the equilibrium must be among all creators with the same cost parameter.
\end{proof}

The next proposition establishes stochastic dominance among the equilibrium strategies of different creators.

\begin{proposition} \label{prop:asym-stochastic-dominance}
    In game $\mathcal{G}^{(n)}$, if $i < j$, then $F_i(x) \leqslant F_j(x)$ for any $x \in \textup{supp}(F_i)$.
\end{proposition}

\begin{proof}
    If $\gamma_i = \gamma_j$, Proposition~\ref{prop:not-all-symmetric-equilibrium} implies $F_i = F_j$. Hence, it suffices to consider the case $\gamma_i < \gamma_j$. Consider the supremum of the support of $j$, denoted by $\overline{x}_j$, choose any point $x_i$ in the support of $i$ such that $x_i < \overline{x}_j$ (the proposition is trivial if no such point exists). Using the notation introduced in equation \eqref{eq:W-def}, we have
    \begin{align*}
        &x_i W(\bm{F}_{-i}(x_i), \bm{p}) - \gamma_i g(x_i) = u_i, \overline{x}_j W(\bm{F}_{-i}(\overline{x}_j), \bm{p}) - \gamma_i g(\overline{x}_j) \leqslant u_i, \\
        &\overline{x}_j W(\bm{F}_{-j}(\overline{x}_j), \bm{p}) - \gamma_j g(\overline{x}_j) = u_j, x_i W(\bm{F}_{-j}(x_i), \bm{p}) - \gamma_j g(x_i) \leqslant u_j.
    \end{align*}
    Combining these inequalities yields
    \begin{equation} \label{asym-stochastic-dominance-1}
        \begin{aligned}
            &\quad x_i (W(\bm{F}_{-i}(x_i), \bm{p}) - W(\bm{F}_{-j}(x_i), \bm{p})) + (\gamma_j - \gamma_i) g(x_i) \\
            &\geqslant u_i - u_j \\
            &\geqslant \overline{x}_j (W(\bm{F}_{-i}(\overline{x}_j), \bm{p}) - W(\bm{F}_{-j}(\overline{x}_j), \bm{p})) + (\gamma_j - \gamma_i) g(\overline{x}_j).
        \end{aligned}
    \end{equation}
    
    Because $\overline{x}_j$ is the supremum of the support of $j$, we have $F_j(\overline{x}_j) = 1$ and $F_i(\overline{x}_j) \leqslant 1$, thus
    \begin{equation} \label{asym-stochastic-dominance-2}
        \begin{aligned}
            W(\bm{F}_{-i}(\overline{x}_j), \bm{p}) &= W(\bm{F}_{-\{i,j\}}(\overline{x}_j), \underline{\bm{p}}) + F_j(\overline{x}_j) W(\bm{F}_{-\{i,j\}}(\overline{x}_j), \overline{\bm{p}} - \underline{\bm{p}}) \\
            &\geqslant W(\bm{F}_{-\{i,j\}}(\overline{x}_j), \underline{\bm{p}}) + F_i(\overline{x}_j) W(\bm{F}_{-\{i,j\}}(\overline{x}_j), \overline{\bm{p}} - \underline{\bm{p}}) \\
            &= W(\bm{F}_{-j}(\overline{x}_j), \bm{p}).
        \end{aligned}
    \end{equation}
    Combining \eqref{asym-stochastic-dominance-1} with the assumption $\gamma_i < \gamma_j$ and $x_i < \overline{x}_j$ yields $W(\bm{F}_{-i}(x_i), \bm{p}) > W(\bm{F}_{-j}(x_i), \bm{p})$. A computation analogous to equation \eqref{asym-stochastic-dominance-2} then gives $F_i(x_i) < F_j(x_i)$. Since $x_i$ was arbitrary, the proposition follows.
\end{proof}

The intuition behind Proposition~\ref{prop:asym-stochastic-dominance} is straightforward: creators with lower cost parameters can afford to exert higher effort levels, leading to stochastically dominant strategies.

Moreover, the proof above implies that when $i < j$, if there exists some $x \in \textup{supp}(F_i)$ for which $F_i(x) = F_j(x)$, then either $\gamma_i = \gamma_j$ or $x = \overline{x}_i = \overline{x}_j$.

Denote the support of $F_i$ as $[\underline{x}_i, \overline{x}_i]$. The following corollary is a direct consequence of Proposition~\ref{prop:asym-stochastic-dominance}.

\begin{corollary} \label{cor:asym-support-monotone}
    In game $\mathcal{G}^{(n)}$, the supports of the equilibrium strategies satisfy the following properties:
    \begin{enumerate}
        \item If $i < j$ and $\gamma_i \neq \gamma_j$, then $\underline{x}_i > \underline{x}_j$ and $\overline{x}_i \geqslant \overline{x}_j$.
        \item $\underline{x}_n = \arg\max_x x p_n - \gamma_n g(x)$.
    \end{enumerate}
\end{corollary}

\begin{proof}
    \begin{enumerate}
        \item If $i < j$ and $\gamma_i \neq \gamma_j$, Proposition~\ref{prop:asym-stochastic-dominance} implies $F_i(x) \leqslant F_j(x)$ for any $x \in \textup{supp}(F_i) \cup \textup{supp}(F_j)$. Therefore, $\underline{x}_i \geqslant \underline{x}_j$ and $\overline{x}_i \geqslant \overline{x}_j$. To show the strict inequality $\underline{x}_i > \underline{x}_j$, suppose on the contrary that $\underline{x}_i = \underline{x}_j$. Then the proof of Proposition~\ref{prop:asym-stochastic-dominance} implies $F_i(\underline{x}_i) < F_j(\underline{x}_j) = 0$, a contradiction.
        \item The first part of the corollary implies that $\underline{x}_n = \min_{i \in [n]} \underline{x}_i$. Following the argument in Step~2 of the proof of Theorem~\ref{thm:sym-equ-property}, we can show that $\underline{x}_n = \arg\max_x x p_n - \gamma_n g(x)$.
    \end{enumerate}
\end{proof}

To state the next result, we recall the notions of participating and active players introduced in \cite{xiao2016asymmetric}. Given an equilibrium, define $\mathcal{P}^*(x) = \{i \mid \underline{x}_i \leqslant x \leqslant \overline{x}_i\}$ and $\mathcal{A}^*(x)$ as the set of players with increasing $F_i$ in a neighborhood of $x$. Therefore, $\mathcal{A}^*(x) \subseteq \mathcal{P}^*(x)$. We refer to $\mathcal{P}^*(x)$ as the set of participating players at effort level $x$, and $\mathcal{A}^*(x)$ as the set of active players at effort level $x$. The following proposition implies that the active players form a nested structure within the participating players.

\begin{proposition}(generalization of Proposition~1 in \cite{xiao2016asymmetric}) \label{prop:asym-nested-gaps}
    Given any equilibrium of game $\mathcal{G}^{(n)}$, suppose $i, j \in \mathcal{P}^*(x)$ with $i < j$. If $i \notin \mathcal{A}^*(x)$, then $j \notin \mathcal{A}^*(x)$.
\end{proposition}

\begin{proof}
    If $\gamma_i = \gamma_j$, Proposition~\ref{prop:not-all-symmetric-equilibrium} implies $F_i = F_j$, then the proposition holds trivially. Hence, it suffices to consider the case $\gamma_i < \gamma_j$. Proof by contradiction. Suppose that there exist $i, j \in \mathcal{P}^*(x)$ with $i < j$ such that $i \notin \mathcal{A}^*(x)$ but $j \in \mathcal{A}^*(x)$. Let $x_2 = \sup\{x \mid i, j \in \mathcal{P}^*(x), i \notin \mathcal{A}^*(x), j \in \mathcal{A}^*(x)\}$, $x_1 = \inf\{x < x_2 \mid F_i(x) = F_i(x_2)\ \textup{for}\ x' \in [x, x_2]\}$. Then we have
    \begin{align*}
        W(\bm{F}_{-i}(x_2), \bm{p}) - W(\bm{F}_{-i}(x_1), \bm{p}) &= (W(\bm{F}_{-\{i, j\}}(x_2), \underline{\bm{p}}) + F_j(x_2) W(\bm{F}_{-\{i,j\}}(x_2), \overline{\bm{p}} - \underline{\bm{p}})) \\
        &\quad - (W(\bm{F}_{-\{i, j\}}(x_1), \underline{\bm{p}}) + F_j(x_1) W(\bm{F}_{-\{i,j\}}(x_1), \overline{\bm{p}} - \underline{\bm{p}})) \\
        &> (W(\bm{F}_{-\{i, j\}}(x_2), \underline{\bm{p}}) + F_j(x_1) W(\bm{F}_{-\{i,j\}}(x_2), \overline{\bm{p}} - \underline{\bm{p}})) \\
        &\quad - (W(\bm{F}_{-\{i, j\}}(x_1), \underline{\bm{p}}) + F_j(x_1) W(\bm{F}_{-\{i,j\}}(x_1), \overline{\bm{p}} - \underline{\bm{p}})) \\
        &= W(\bm{F}_{-\{i, j\}}(x), \underline{\bm{p}})\big|_{x_1}^{x_2} + F_j(x_1) W(\bm{F}_{-\{i,j\}}(x), \overline{\bm{p}} - \underline{\bm{p}})\big|_{x_1}^{x_2}.
    \end{align*}
    
    Similarly,
    \[ W(\bm{F}_{-j}(x_2), \bm{p}) - W(\bm{F}_{-j}(x_1), \bm{p}) = W(\bm{F}_{-\{i, j\}}(x), \underline{\bm{p}})\big|_{x_1}^{x_2} + F_i(x_2) W(\bm{F}_{-\{i,j\}}(x), \overline{\bm{p}} - \underline{\bm{p}})\big|_{x_1}^{x_2}. \]
    
    Since $F_i(x_2) = F_i(x_1) \leqslant F_j(x_1)$, we have
    \[ W(\bm{F}_{-i}(x_2), \bm{p}) - W(\bm{F}_{-i}(x_1), \bm{p}) > W(\bm{F}_{-j}(x_2), \bm{p}) - W(\bm{F}_{-j}(x_1), \bm{p}). \]
    
    With the assumption that $\gamma_i < \gamma_j$, we have
    \[ [W(\bm{F}_{-i}(x), \bm{p}) - \gamma_i g(x)]\big|_{x_1}^{x_2} > [W(\bm{F}_{-j}(x), \bm{p}) - \gamma_j g(x)]\big|_{x_1}^{x_2}. \]

    Since $j \in \mathcal{A}^*(x)$ for any $x \in (x_1, x_2)$, we have $[W(\bm{F}_{-j}(x), \bm{p}) - \gamma_j g(x)]\big|_{x_1}^{x_2} = 0$, thus
    \[ x_2 W(\bm{F}_{-i}(x_2), \bm{p}) - \gamma_i g(x_2) > x_1 W(\bm{F}_{-i}(x_1), \bm{p}) - \gamma_i g(x_1) = u_i, \]
    which contradicts the requirement of equilibrium. Therefore, the proposition holds.
\end{proof}

Based on Proposition~\ref{prop:asym-stochastic-dominance} and Proposition~\ref{prop:asym-nested-gaps}, we can sketch the general form of a mixed Nash equilibrium in game $\mathcal{G}^{(n)}$. Higher-cost creators enter the competition at lower effort levels and also cease competing at lower effort levels; lower-cost creators concentrate their effort more heavily on high values; and higher-cost players may still choose high effort, but only when lower-cost players are active at the same effort level. Section~\ref{subsec:ap-asym-three} provides some concrete examples with three creators that illustrate the properties established in this subsection.

Specifically, for the binary-type setting (game $\mathcal{G}^{(2)}$) where $\gamma_1 = \cdots = \gamma_k = \gamma_H$ and $\gamma_{k+1} = \cdots = \gamma_n = \gamma_L$ with $\gamma_H < \gamma_L$, the support structure of the mixed Nash equilibrium can be sketched as shown in Figure~\ref{fig:2-type-support-sketch} (the ellipsis indicates that the equilibrium behavior in this interval cannot be characterized precisely until specific parameter values are given). By Proposition~\ref{prop:not-all-symmetric-equilibrium}, all creators of the same type employ identical equilibrium strategies. Let $[\underline{x}_H, \overline{x}_H]$ and $[\underline{x}_L, \overline{x}_L]$ denote the supports of type-$H$ and type-$L$ creators, respectively. Corollary~\ref{cor:asym-support-monotone} implies that $\underline{x}_H > \underline{x}_L = \arg\max_x x p_n - \gamma_L g(x)$ and $\overline{x}_H \geqslant \overline{x}_L$. Consequently, on the interval $[\underline{x}_L, \underline{x}_H)$ only type-$L$ creators compete. 

On the interval $[\underline{x}_H, \overline{x}_H]$, Proposition~\ref{prop:asym-nested-gaps} indicates that type-$L$ creators can be active only where type-$H$ creators are active. Together with Corollary~\ref{cor:symmetric-no-gap}, it follows that the support of type-$H$ creators is the continuous interval $[\underline{x}_H, \overline{x}_H]$ without any gaps. However, the equilibrium strategy of type-$L$ creators may contain gaps within $[\underline{x}_H, \overline{x}_H]$ or may even not be active on $[\underline{x}_H, \overline{x}_H]$ if $\overline{x}_L \leqslant \underline{x}_H$.

\begin{figure}[htbp]
    \centering
    \includegraphics[width=0.45\textwidth]{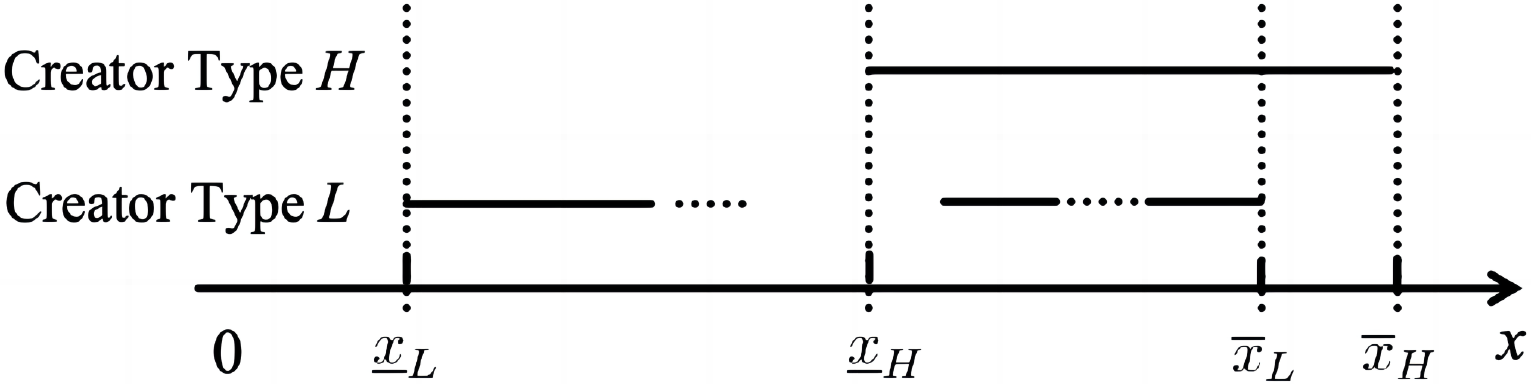} 
    \caption{Sketch of the support structure of the mixed Nash equilibrium for binary-type setting}
    \label{fig:2-type-support-sketch}
\end{figure}

It is worth noting that a pure Nash equilibrium may exist when $\gamma_1 < \gamma_2 < \cdots < \gamma_n$, which contrasts with previous results in asymmetric contests \cite{barut1998symmetric,xiao2016asymmetric}. The key distinction lies in the payoff structure: in our model, $x p-\gamma g(x)$ has a unique nonzero maximizer, whereas in standard rank-based contests the utility is simply a constant reward minus the cost, making zero the only maximizer. We refer to Appendix~\ref{subsec:ap-asym-pure} for a more detailed discussion on pure Nash equilibria in asymmetric settings.

\subsection{Results for Pure Nash Equilibrium} \label{subsec:ap-asym-pure}

The properties stated in Section~\ref{subsec:ap-asym-general} hold for any cost structure, thus also apply to the special case where $\gamma_1 < \gamma_2 < \cdots < \gamma_n$. In this subsection, we discuss additional properties that arise specifically in this scenario, especially in relation to the pure Nash equilibrium.

We begin with a proposition indicates that in any pure Nash equilibrium, creators with lower cost parameters will exert higher effort levels.

\begin{proposition} \label{prop:decrease_equilibrium}
    Assume that $\gamma_1 < \cdots < \gamma_n$ in game $\mathcal{G}^{(n)}$, if a pure Nash equilibrium exists, then the equilibrium strategy profile $(x_1, \ldots, x_n)$ must satisfies $x_1 > \cdots > x_n > 0$.
\end{proposition}

\begin{proof}
    Note that in a pure Nash equilibrium, the CDF of the strategy of each creator degenerates to a step function at her chosen effort level, hence $x_1 \geqslant \cdots \geqslant x_n$ is a direct consequence of Proposition~\ref{prop:asym-stochastic-dominance}. If there exists some $i \in [n - 1]$ such that $x_i = x_{i + 1}$, then both creators can increase her utility by slightly increasing her effort level, since $g$ is a continuous function and $x_i > x_{i + 1}$. This contradicts the assumption that the strategy profile is a pure Nash equilibrium, and thus we have $x_1 > \cdots > x_n$.
    
    Then, according to the first order condition for pure Nash equilibrium, we have
    \begin{equation} \label{eq:asym-pure-foc}
        p_i - \gamma_i g'(x_i) = 0.
    \end{equation}

    Recall that $p_i > 0$ for any $i \in [n]$, thus equation \eqref{eq:asym-pure-foc} implies $x_i > 0$ for any $i \in [n]$.
\end{proof} 

The following proposition provides a sufficient condition for the existence of pure Nash equilibrium when $\gamma_1 < \cdots < \gamma_n$.

\begin{proposition}
    Denote $x_i = \arg\max_x x p_i - \gamma_i g(x)$ for any $i \in [n]$. If $\gamma_1 < \cdots < \gamma_n$, then game $\mathcal{G}^{(n)}$ admits a pure Nash equilibrium if 
    \[ \frac{\gamma_{i + 1}}{\gamma_i} \geqslant \frac{x_i g'(x_i)}{g(x_i)}. \]
\end{proposition}

Intuitively, the larger the gap between the cost parameters of adjacent creators, the more likely a pure Nash equilibrium exists, as higher-cost creators then have a weaker incentive to deviate toward the effort levels chosen by lower-cost creators.

\begin{proof}
    Our proof can be divided into three steps:
    
    \textbf{Step 1: No creator has the incentive to unilaterally decrease her effort level.}

    Proof by contradiction. Suppose that creator $i \in [n]$ decreases her effort to $x < x_i$. There are two cases to consider:
    \begin{enumerate}
        \item Recall that $x_i$ in equilibrium maximize $x_i p_i - \gamma_i g(x_i)$, so creator $i$ does not have an incentive to decrease her effort to $x \in [x_{i + 1}, x_i)$;
        \item When creator $i$ decreases her effort to $x \in [0, x_{i + 1})$, then we have
        \[ x p_j - \gamma_i g(x) < x p_i - \gamma_i g(x) < x_i p_i - \gamma_i g(x_i) \]
        for any $j > i$, so creator $i$ does not have an incentive to decrease her effort to $x \in [0, x_{i + 1})$.
    \end{enumerate}
    
    Combining the two cases above, we can conclude that no creator has the incentive to unilaterally decrease her effort level.

    \textbf{Step 2: We only need to consider the case that for any $i \in [n]$, creator $i + 1$ seeks to increase her effort to $x_i$.}

    Note that for all $j > i, \arg\max_x x p_i - \gamma_j g(x) < \arg\max_x x p_i - \gamma_i g(x)$. Due to the concavity of $x p_i - \gamma_j g(x)$, the most profitable deviation for creator $j$ toward position $i$ is exactly $x_i$ (we ignore the infinitesimal perturbation $\varepsilon \to 0^+$). Then if creator $i + 1$ does not have an incentive to deviate to $x_i$, the above analysis implies that
    \[ x_i p_i - \gamma_{i + 1} g(x_i) \leqslant x_{i + 1} p_{i + 1} - \gamma_{i + 1} g(x_{i + 1}). \]
    
    If creator $i$ does not have an incentive to deviate to $x_{i - 1}$, similarly we have 
    \[ x_{i - 1} p_{i - 1} - \gamma_i g(x_{i - 1}) \leqslant x_i p_i - \gamma_i g(x_i). \]

    Recall that $\gamma_i < \gamma_{i + 1}$, combine the two inequalities above, we have
    \[ x_{i - 1} p_{i - 1} - \gamma_{i + 1} g(x_{i - 1}) \leqslant x_i p_i - \gamma_{i + 1} g(x_i) \leqslant x_{i + 1} p_{i + 1} - \gamma_{i + 1} g(x_{i + 1}). \]
    
    So creator $i + 1$ does not have an incentive to deviate to $x_{i - 1}$. Inductively, we can conclude that if creator $i + 1$ does not have an incentive to deviate to $x_i$, then she does not have an incentive to deviate to any $x_j$ with $j < i$. Combined with the result of Step~1, we only need to consider the case that for any $i \in [n]$, creator $i + 1$ seeks to increase her effort to $x_i$.

    \textbf{Step 3: Prove the proposition.}

    Recall the result in Proposition \ref{prop:decrease_equilibrium}, if $\gamma_1 < \ldots < \gamma_n$, then in equilibrium, the page created by creator $i \in [n]$ will occupy position $i$. So, the best response of any creator $i \in [n]$ is $x_i = \arg\max_x x p_i - \gamma_i g(x)$.

    Based on the results established in Step 2, the strategy profile constitutes a pure Nash equilibrium if and only if for every creator $ i + 1 $ with $ i \in [n] $, there exists no incentive to unilaterally deviate to the effort level $ x_i $. Formally,
    \[ x_i p_i - \gamma_{i + 1} g(x_i) \leqslant x_{i + 1} p_{i + 1} - \gamma_{i + 1} g(x_{i + 1}), \]
    associated with the first order conditions (F.O.C.)
    \[ p_i - \gamma_i g'(x_i) = 0, \quad p_{i + 1} - \gamma_{i + 1} g'(x_{i + 1}) = 0, \]
    we can derive that
    \begin{align*}
        x_i p_i - x_{i + 1} p_{i + 1} \leqslant \gamma_{i + 1} g(x_i) - \gamma_{i + 1} g(x_{i + 1}) = \frac{\gamma_{i + 1}}{\gamma_i} \frac{g(x_i)}{g'(x_i)} p_i - \frac{g(x_{i + 1})}{g'(x_{i + 1})} p_{i + 1},
    \end{align*}
    then we can derive that
    \begin{equation} \label{eq:asym-pure-sufficient-1}
        \left(x_i - \frac{\gamma_{i + 1}}{\gamma_i} \frac{g(x_i)}{g'(x_i)} \right) p_i \leqslant \left( x_{i + 1} - \frac{g(x_{i + 1})}{g'(x_{i + 1})} \right) p_{i + 1}.
    \end{equation}

    Because $g$ is a strictly convex function, we have
    \[ g(x_{i + 1}) < g(0) + g'(x_{i + 1}) x_{i + 1} = g'(x_{i + 1}) x_{i + 1}, \]
    thus $x_{i + 1} - \frac{g(x_{i + 1})}{g'(x_{i + 1})} > 0$. Therefore, if $\frac{\gamma_{i + 1}}{\gamma_i} \geqslant \frac{x_i g'(x_i)}{g(x_i)}$, the equation~\eqref{eq:asym-pure-sufficient-1} must hold, thus completing the proof.
\end{proof}

\begin{example}
    Consider the cost function $g(x) = x^2$. Then, we have $\frac{x_i g'(x_i)}{g(x_i)} = 2$. If $n = 3$, $\gamma_1 = 1$, $\gamma_2 = 2$, $\gamma_3 = 4$ and $p_1 = 0.8$, $p_2 = 0.6$, $p_3 = 0.4$, then the game admits a pure Nash equilibrium with effort levels $x_1 = 0.4$, $x_2 = 0.15$, and $x_3 = 0.05$. 
    
    We can easily verify that no creator has the incentive to unilaterally deviate from her effort level. For example, we can compute the utility of creator $3$ when she chooses $x_3$ is $0.05 \times 0.4 - 4 \times (0.05)^2 = 0.01$, while her utility when deviating to $x_2$ is $0.15 \times 0.6 - 4 \times (0.15)^2 = 0 < 0.01$, and when deviating to $x_1$ is $0.4 \times 0.8 - 4 \times (0.4)^2 = -0.32 < 0.01$.
\end{example}

A natural question is whether a pure Nash equilibrium, when it exists, is the unique equilibrium of game $\mathcal{G}^{(n)}$. The following proposition provides a sufficient condition for the pure Nash equilibrium to be unique.

\begin{proposition} \label{prop:pure-unique}
    In game $\mathcal{G}^{(n)}$, if the elasticity of $g$,
    \[ \eta_g(x) = \frac{x g'(x)}{g(x)}, \]
    is a non-decreasing function of $x$, then the pure Nash equilibrium is the unique equilibrium of game $\mathcal{G}^{(n)}$ whenever it exists.
\end{proposition}

Intuitively, the elasticity of the cost function measures its degree of convexity. For instance, the power function $g(x) = x^{\beta}$ has constant elasticity $\beta$, whereas the function $g(x) = x e^{-1 / x}$, even though it is convex, has a decreasing elasticity $1 + \frac{1}{x}$. Indeed, as $x \to +\infty$, $g(x)$ approaches the linear asymptote $y = x - 1$, which is consistent with the idea that elasticity declines when convexity weakens.

If the elasticity is non-decreasing, the convexity of cost function does not diminish with increasing effort. Consequently, a creator with a higher cost parameter encounters a sharper rise in marginal cost when deviating toward the effort level of a lower-cost competitor. This reduces the incentive to deviate from the pure-strategy profile, and thus ensures that any pure Nash equilibrium is the unique equilibrium.

\begin{proof}
    According to the first order condition for pure Nash equilibrium in equation \eqref{eq:asym-pure-foc}, since $g$ is strictly convex, the effort level of each creator in any pure Nash equilibrium is unique, thus there is only one pure Nash equilibrium. So we only need to show that no mixed strategy equilibrium exists.

    Proof by contradiction. Assume a mixed Nash equilibrium exists. First note that in a mixed Nash equilibrium, for every creator $i$ and any point $x$ in the support of her strategy, the indifference condition yields  
    \[ x W(\bm{F}_{-i}(x), \bm{p}) - \gamma_i g(x) = u_i. \]
   
    By Lemma~\ref{lem:W-monotone} and the fact that each $F_j$ is non-decreasing, the term $W(\bm{F}_{-i}(x),\bm{p})$ is non-decreasing in $x$. Differentiating both sides of the indifference condition therefore gives  
    \begin{equation} \label{eq:pure-unique-1}
        W(\bm{F}_{-i}(x), \bm{p}) \leqslant \gamma_i g'(x).
    \end{equation}

    Now suppose $\gamma_i < \gamma_j$. Because a pure Nash equilibrium exists, we have $u_j \geqslant x_i p_i - \gamma_j g(x_i)$, where $x_i = \arg\max_x x p_i - \gamma_i g(x)$. Let $\underline{x}_i$ and $\underline{x}_{i - 1}$ be the infimum of the supports of creators $i$ and $i - 1$ in the mixed Nash equilibrium, respectively. For any $x \in [\underline{x}_i, \underline{x}_{i - 1}]$, we consider the following two cases.
    \begin{enumerate}
        \item \xhdr{Case $x \geqslant x_i$} Since $W(\bm{F}_{-j}(x), \bm{p})\leqslant p_i$,
        \[ x_i p_i - \gamma_i g(x_i) > x p_i - \gamma_i g(x) \geqslant x W(\bm{F}_{-j}(x), \bm{p}) - \gamma_i g(x), \]
        i.e.  
        \[ \gamma_i (g(x) - g(x_i)) > x W(\bm{F}_{-j}(x), \bm{p}) - x_i p_i. \]
        Because $\gamma_j > \gamma_i$ and $x \geqslant x_i$,  
        \[ \gamma_j (g(x) - g(x_i)) > x W(\bm{F}_{-j}(x), \bm{p}) - x_i p_i, \]
        Hence,  
        \[ u_j \geqslant x_i p_i - \gamma_j g(x_i) > x W(\bm{F}_{-j}(x), \bm{p}) - \gamma_j g(x). \]
        So $x$ cannot belong to the support of the equilibrium strategy of creator $j$.

        \item \xhdr{Case $x < x_i$} If $x$ were in the support of the equilibrium strategy of creator $j$, then
        \begin{align*}
            u_j &= x W(\bm{F}_{-j}(x), \bm{p}) - \gamma_j g(x) < x W(\bm{F}_{-i}(x), \bm{p}) - \gamma_j g(x) \\
            &\leqslant x \gamma_i g'(x) - \gamma_j g(x) = g(x)\left(\frac{x \gamma_i g'(x)}{g(x)} - \gamma_j\right) \leqslant g(x)\left(\frac{x_i \gamma_i g'(x_i)}{g(x_i)} - \gamma_j\right) \\
            &= \frac{g(x)}{g(x_i)}(x_i p_i - \gamma_j g(x_i)) \leqslant \frac{g(x)}{g(x_i)} u_j < u_j,
        \end{align*}
        where the first strict inequality follows from $F_i(x) < F_j(x)$ (Proposition~\ref{prop:asym-stochastic-dominance}) and the derivation of equation~\eqref{eq:symmetric-equilibrium-1}, the second inequality follows from inequality \eqref{eq:pure-unique-1}, and the third inequality follows from the assumption that the elasticity $\eta_g(x)$ is a non-decreasing function.
    \end{enumerate}


    In both cases, for every $x \in [\underline{x}_i, \underline{x}_{i - 1}]$, creator $j$ fails to satisfy the indifference condition at $x$; hence, such an $x$ cannot belong to her equilibrium support. Since $i$ and $j$ were arbitrary, we first set $j = n$, then creator $n$ must choose $\arg\max_x x p_n - \gamma_n g(x)$ with probability one, which leaves the equilibrium strategies of the first $n - 1$ creators unchanged. Proceeding recursively, we successively set $j = n - 1, n - 2, \dots, 2$. In each step, each creator $j$ is forced to select $\arg\max_x x p_j - \gamma_j g(x)$ with probability one. Hence, no mixed Nash equilibrium can exist; thus, the pure Nash equilibrium (when it exists) is the unique Nash equilibrium of game $\mathcal{G}^{(n)}$.
\end{proof}

\subsection{An Example for \texorpdfstring{$n = 3$}{n = 3}} \label{subsec:ap-asym-three}

In this subsection, we present some examples of the mixed Nash equilibrium for the case where there are three creators ($n = 3$) with distinct cost parameters. The purpose of this example is to illustrate the complexity of the equilibrium structure when $n \geqslant 3$. Consider the case where $p_1 = 0.8$, $p_2 = 0.6$, $p_3 = 0.4$, and $g(x) = x^5$. Let $\gamma_1 = 1$, $\gamma_3 = 1.2$, and let $\gamma_2$ vary from $1$ to $1.2$.

Using numerical methods, we find that the equilibrium shifts smoothly from the $F_1 = F_2$ case to the $F_2 = F_3$ case as $\gamma_2$ increases from $1$ to $1.2$ with six distinct regimes. Note that the equilibria illustrated here may not be the unique equilibrium, but are intuitively reasonable and consistent with the results in Section~\ref{subsec:ap-asym-general} and Section~\ref{subsec:ap-asym-n2}.

\xhdr{Regime 1 ($\gamma_2 = 1$)} In this regime, the equilibrium satisfies $F_1 = F_2$. Numerical results show that the infimums of the supports are $\underline{x}_3 \approx 0.508$ and $\underline{x}_1 = \underline{x}_2 \approx 0.573$, and the supremum of the support is $\overline{x}_1 = \overline{x}_2 = \overline{x}_3 \approx 0.843$. Creator 1 and creator 2 are active between the infimum and supremum of their strategies, while creator 3 is only active on $(0.726, 0.843]$ (we denote $x^*$ as the point where creator 3 starts to be active) and places an atom of probability approximately $0.749$ at $\underline{x}_3$. Figure~\ref{fig:three-player-combined-1} illustrates the support structure and CDFs of the mixed Nash equilibrium in this regime. The dots in the figure indicate points where the probability mass of a creator is strictly positive; that is, the CDF has an atom at those points. As shown in the proof of Theorem~\ref{thm:sym-equ-property}, such atoms cannot occur in a symmetric setting.
\begin{figure}[t]
    \centering
    \begin{subfigure}[b]{0.48\linewidth}
        \centering
        \includegraphics[width=\linewidth]{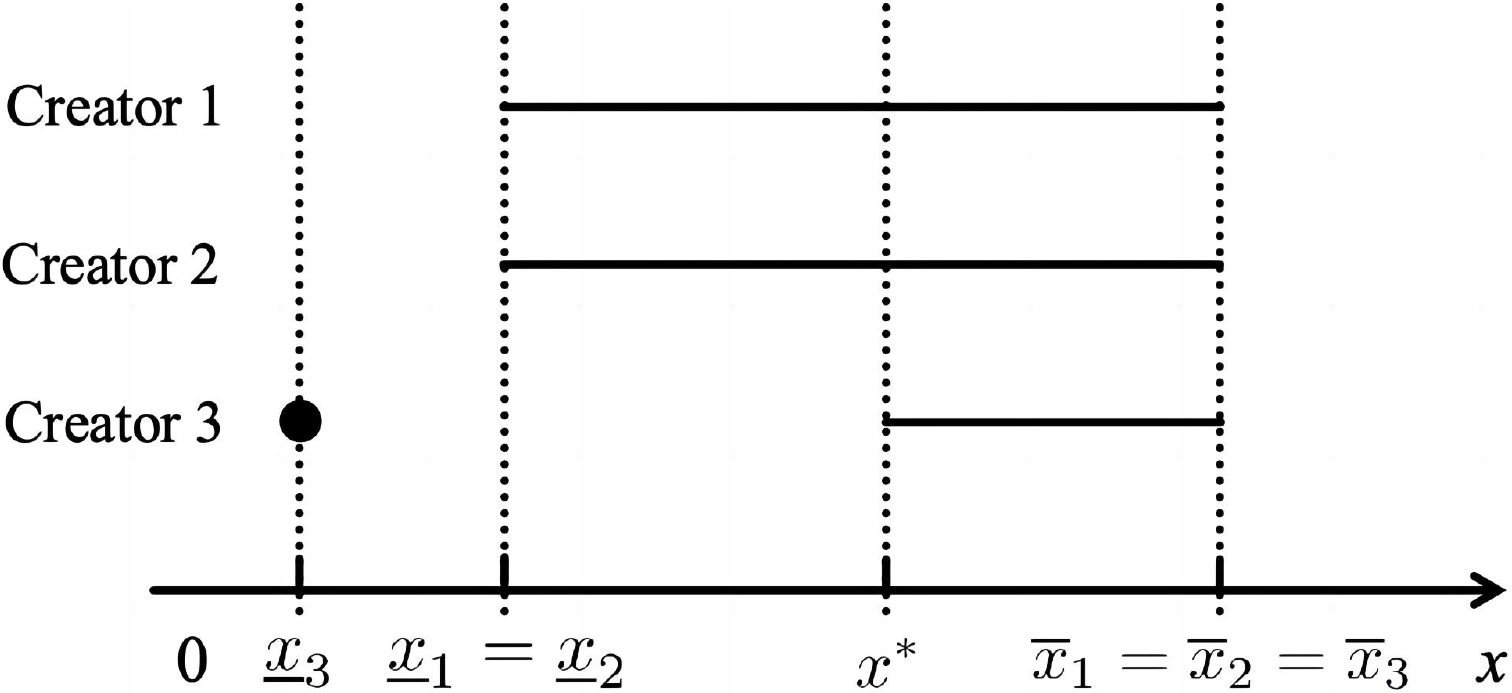}
        \caption{Support of mixed Nash equilibrium (Regime~1)}
        \label{fig:3-player-support-1}
    \end{subfigure}
    \hfill
    \begin{subfigure}[b]{0.48\linewidth}
        \centering
        \includegraphics[width=\linewidth]{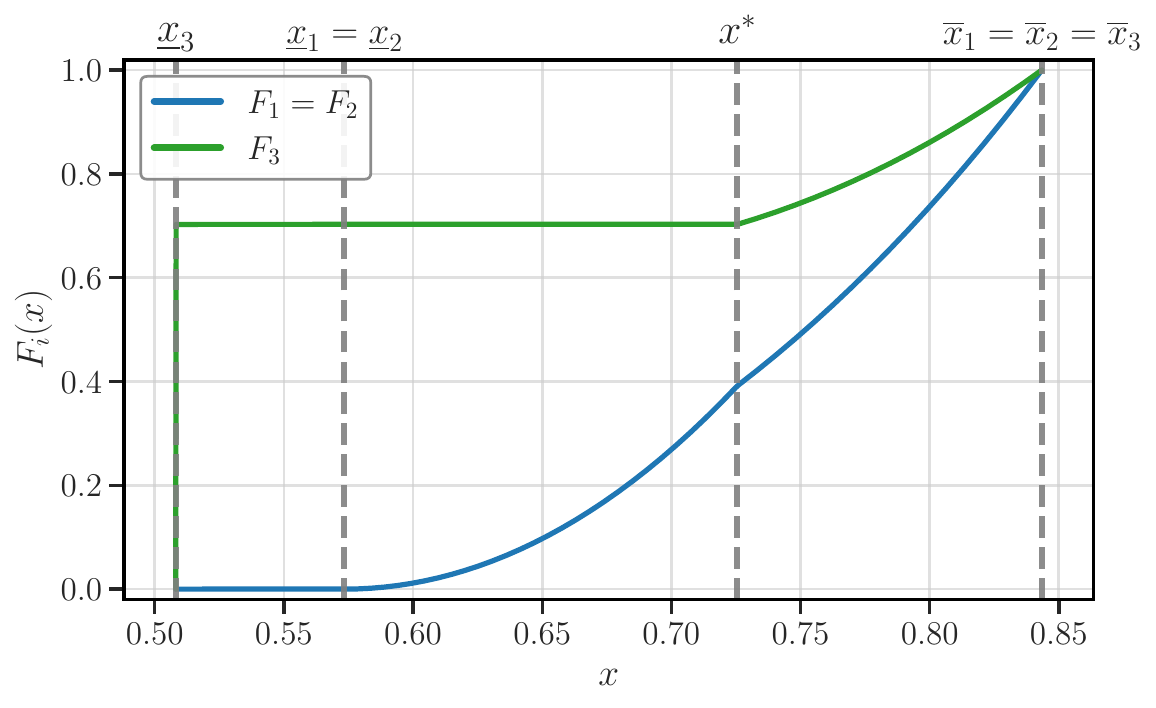}
        \caption{CDFs of mixed Nash equilibrium (Regime~1)}
        \label{fig:3-player-cdf-1}
    \end{subfigure}
    \caption{Three-player equilibrium illustrations (Regime~1)}
    \label{fig:three-player-combined-1}
\end{figure}
    
\xhdr{Regime 2 ($1 < \gamma_2 \leqslant 1.017$)} In this regime, we take $\gamma_2 = 1.015$ as an example. Numerical results show that the infimums of the supports are $\underline{x}_3 \approx 0.508$, $\underline{x}_2 \approx 0.569$, and $\underline{x}_1 \approx 0.573$, and the supremum of the support is $\overline{x}_1 = \overline{x}_2 = \overline{x}_3 \approx 0.843$. Creator 1 is active between the infimum and supremum of her strategy and places an atom of probability approximately $0.000361$ at $\underline{x}_1$. Creator 2 is active on the support of creator 1's strategy and places an atom of probability approximately $0.0481$ at $\underline{x}_2$. Creator 3 only active on $(0.698, 0.843]$ and places an atom of probability approximately $0.655$ at $\underline{x}_3$. Figure~\ref{fig:three-player-combined-2} illustrates the support structure and CDFs of the mixed Nash equilibrium in this regime.
\begin{figure}[t]
    \centering
    \begin{subfigure}[b]{0.48\linewidth}
        \centering
        \includegraphics[width=\linewidth]{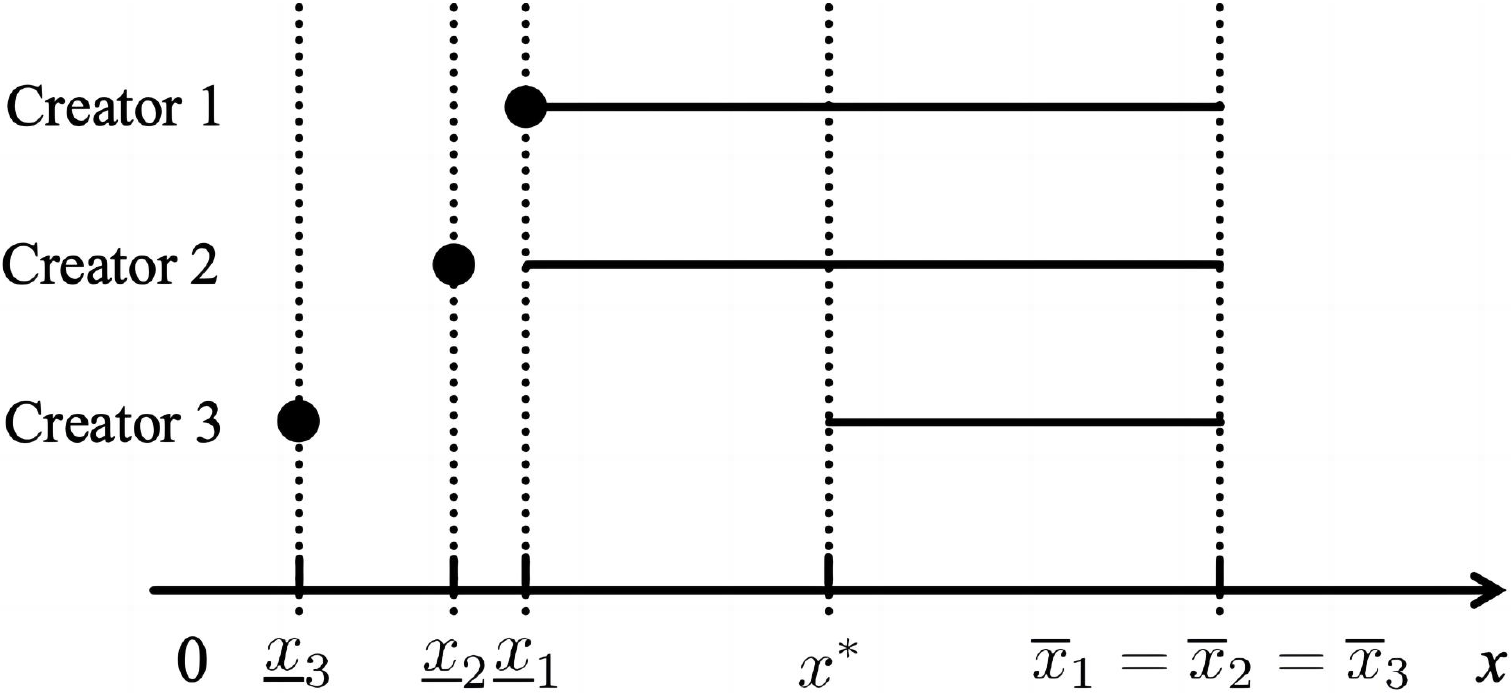}
        \caption{Support of mixed Nash equilibrium (Regime~2)}
        \label{fig:3-player-support-2}
    \end{subfigure}
    \hfill
    \begin{subfigure}[b]{0.48\linewidth}
        \centering
        \includegraphics[width=\linewidth]{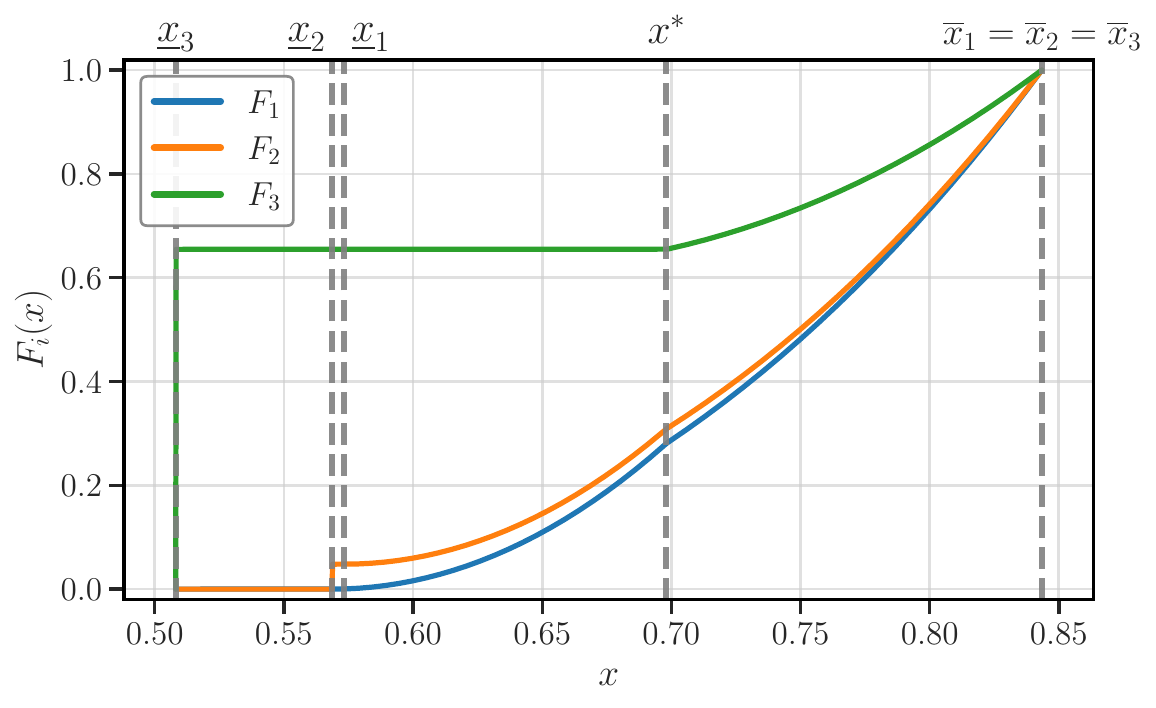}
        \caption{CDFs of mixed Nash equilibrium (Regime~2)}
        \label{fig:3-player-cdf-2}
    \end{subfigure}
    \caption{Three-player equilibrium illustrations (Regime~2)}
    \label{fig:three-player-combined-2}
\end{figure}

\xhdr{Regime 3 ($1.017 < \gamma_2 \leqslant 1.021$)} In this regime, we take $\gamma_2 = 1.02$ as an example. Numerical results show that the infimums of the supports are $\underline{x}_3 \approx 0.508$, $\underline{x}_2 \approx 0.567$, and $\underline{x}_1 \approx 0.573$, and the supremum of the support is $\overline{x}_1 = \overline{x}_2 = \overline{x}_3 \approx 0.843$. Creator 1 is active between the infimum and supremum of her strategy and places an atom of probability approximately $0.000642$ at $\underline{x}_1$. Creator 2 is active over $[0.567, 0.572]$ and $(0.573, 0.843]$, and places an atom of probability approximately $0.0543$ at $\underline{x}_2$. Creator 3 is active on $(0.567, 0.572]$ and $(0.687, 0.843]$, and places an atom of probability approximately $0.638$ at $\underline{x}_3$. Figure~\ref{fig:three-player-combined-3} illustrates the support structure (we denote $x^* = 0.572$ and $x^{**} = 0.687$) and CDFs of the mixed Nash equilibrium in this regime. Readers may note that in Figure~\ref{fig:3-player-cdf-3}, $F_3$ appears constant on the interval $[0.567, 0.572]$. This is because the probability mass placed on that interval is extremely small ($F_3$ only increases from $0.6384$ to $0.6388$), making the variation visually indistinguishable in the figure.
\begin{figure}[t]
    \centering
    \begin{subfigure}[b]{0.48\linewidth}
        \centering
        \includegraphics[width=\linewidth]{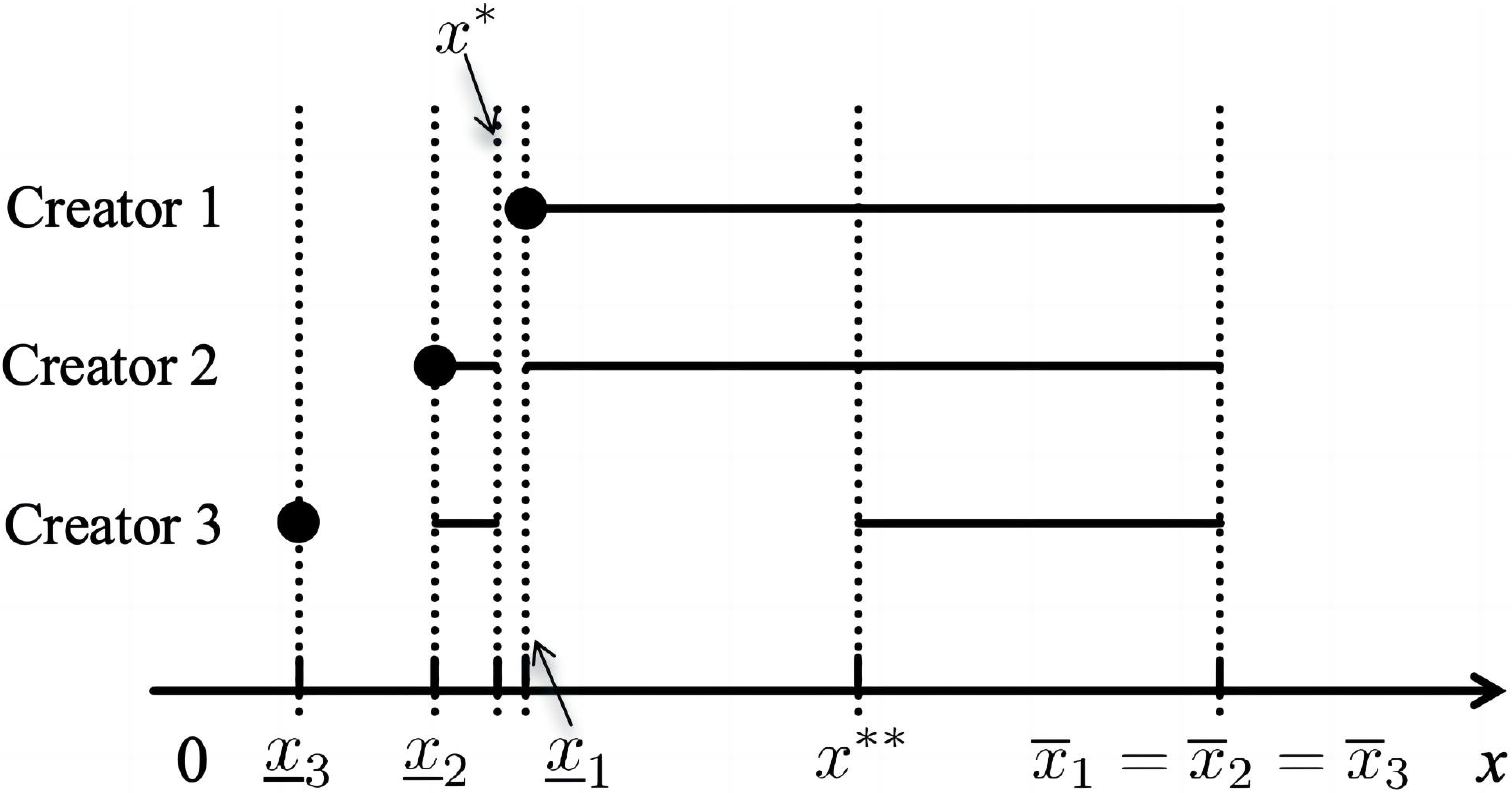}
        \caption{Support of mixed Nash equilibrium (Regime~3)}
        \label{fig:3-player-support-3}
    \end{subfigure}
    \hfill
    \begin{subfigure}[b]{0.48\linewidth}
        \centering
        \includegraphics[width=\linewidth]{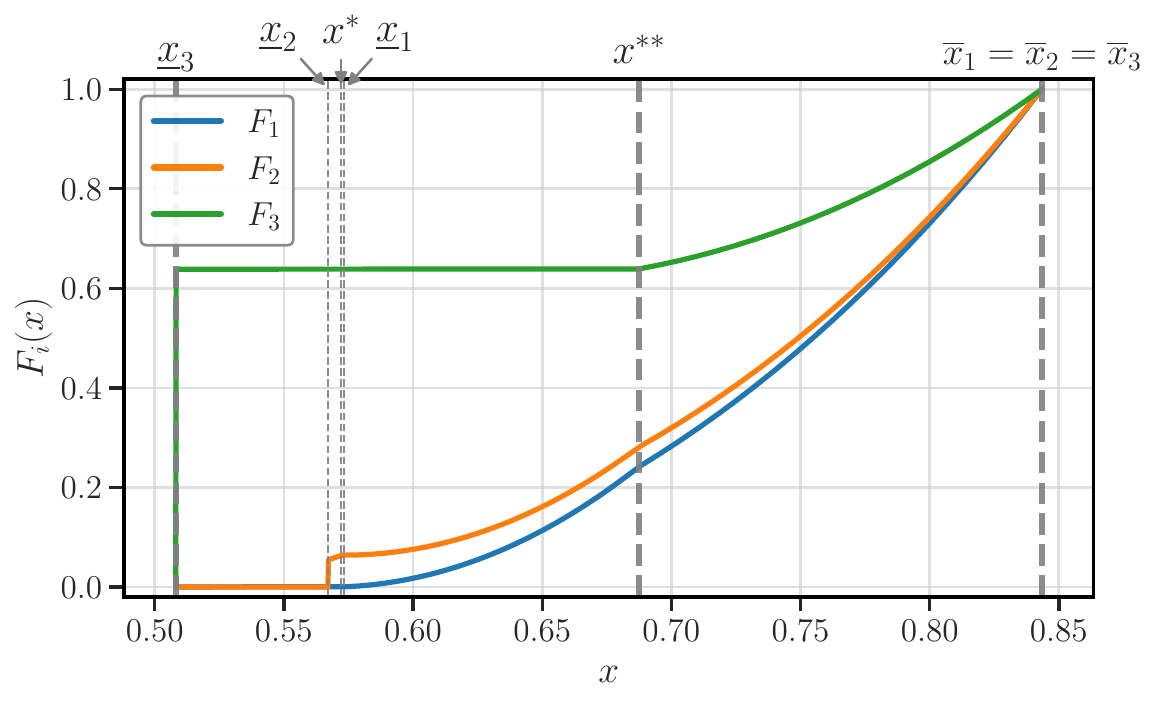}
        \caption{CDFs of mixed Nash equilibrium (Regime~3)}
        \label{fig:3-player-cdf-3}
    \end{subfigure}
    \caption{Three-player equilibrium illustrations (Regime~3)}
    \label{fig:three-player-combined-3}
\end{figure}

\xhdr{Regime 4 ($1.021 < \gamma_2 \leqslant 1.066$)} In this regime, we take $\gamma_2 = 1.05$ as an example. Numerical results show that the infimums of the supports are $\underline{x}_3 \approx 0.508$, $\underline{x}_2 \approx 0.558$, and $\underline{x}_1 \approx 0.604$, and the supremum of the support is $\overline{x}_1 = \overline{x}_2 = \overline{x}_3 \approx 0.843$. Creator 1 is active between the infimum and supremum of her strategy. Creator 2 is active between the infimum and supremum of her strategy and places an atom of probability approximately $0.0382$ at $\underline{x}_2$. Creator 3 is active on $(0.558, 0.604]$ and $(0.637, 0.843]$, and places an atom of probability approximately $0.540$ at $\underline{x}_3$. Figure~\ref{fig:three-player-combined-4} illustrates the support structure and CDFs of the mixed Nash equilibrium in this regime.
\begin{figure}[t]
    \centering
    \begin{subfigure}[b]{0.48\linewidth}
        \centering
        \includegraphics[width=\linewidth]{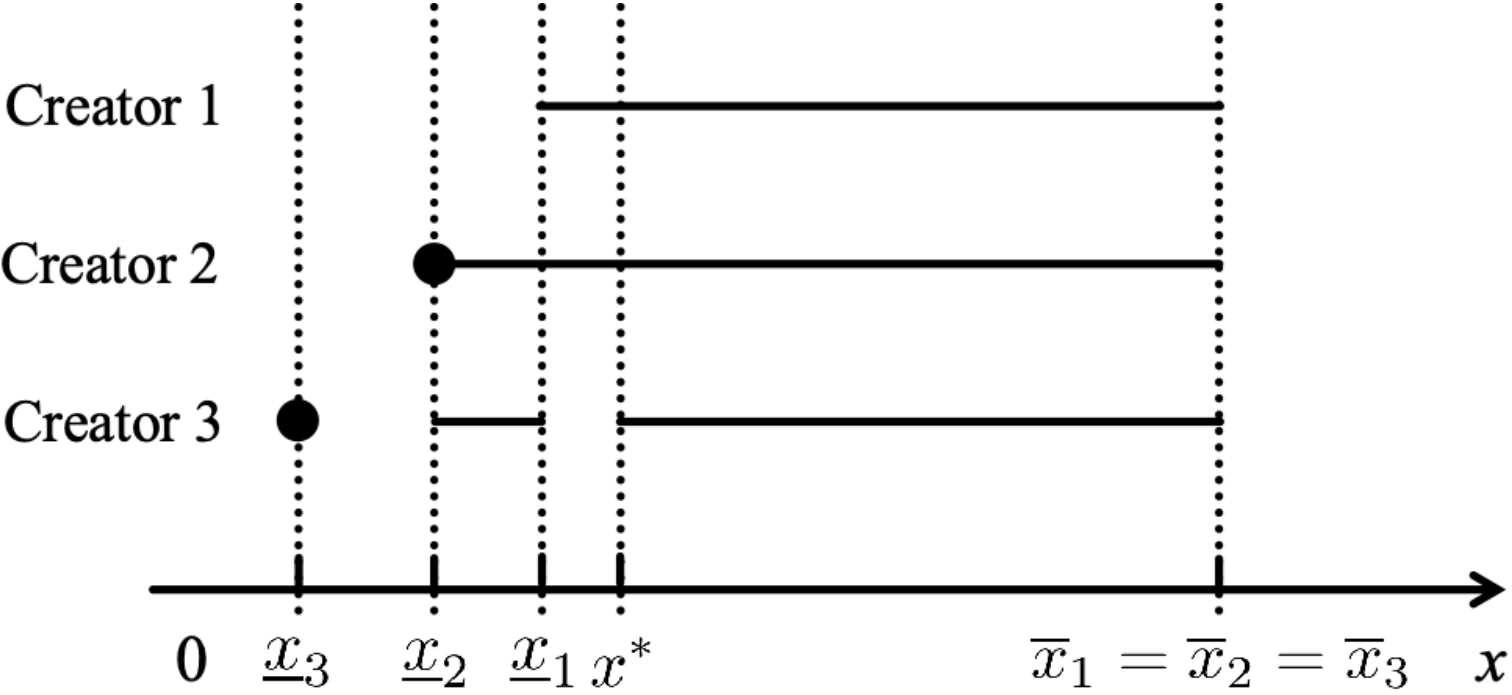}
        \caption{Support of mixed Nash equilibrium (Regime~4)}
        \label{fig:3-player-support-4}
    \end{subfigure}
    \hfill
    \begin{subfigure}[b]{0.48\linewidth}
        \centering
        \includegraphics[width=\linewidth]{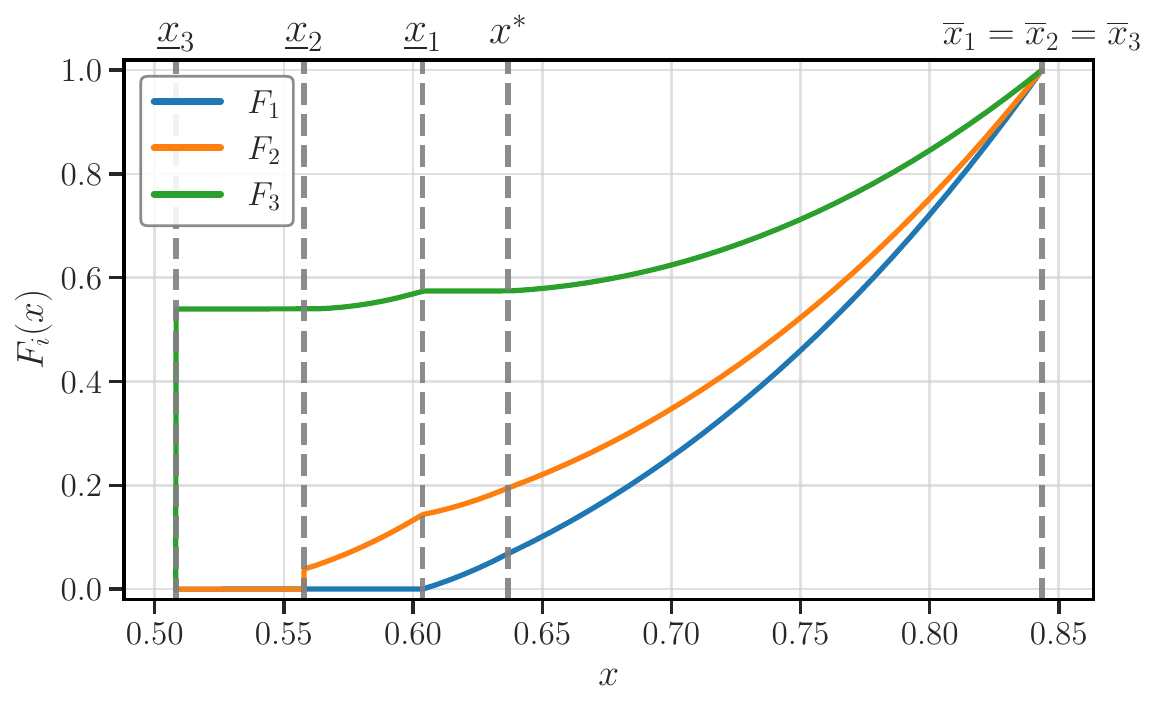}
        \caption{CDFs of mixed Nash equilibrium (Regime~4)}
        \label{fig:3-player-cdf-4}
    \end{subfigure}
    \caption{Three-player equilibrium illustrations (Regime~4)}
    \label{fig:three-player-combined-4}
\end{figure}

\xhdr{Regime 5 ($1.066 < \gamma_2 < 1.2$)} In this regime, we take $\gamma_2 = 1.10$ as an example. Numerical results show that the infimums of the supports are $\underline{x}_3 \approx 0.508$, $\underline{x}_2 \approx 0.542$, and $\underline{x}_1 \approx 0.634$, and the supremum of the support is $\overline{x}_1 = \overline{x}_2 = \overline{x}_3 \approx 0.843$. Creator 1 is active between the infimum and supremum of her strategy. Creator 2 is active between the infimum and supremum of her strategy and places an atom of probability approximately $0.0175$ at $\underline{x}_2$. Creator 3 is active on $(0.634, 0.843]$, and places an atom of probability approximately $0.368$ at $\underline{x}_3$. Figure~\ref{fig:three-player-combined-5} illustrates the support structure and CDFs of the mixed Nash equilibrium in this regime.
\begin{figure}[t]
    \centering
    \begin{subfigure}[b]{0.48\linewidth}
        \centering
        \includegraphics[width=\linewidth]{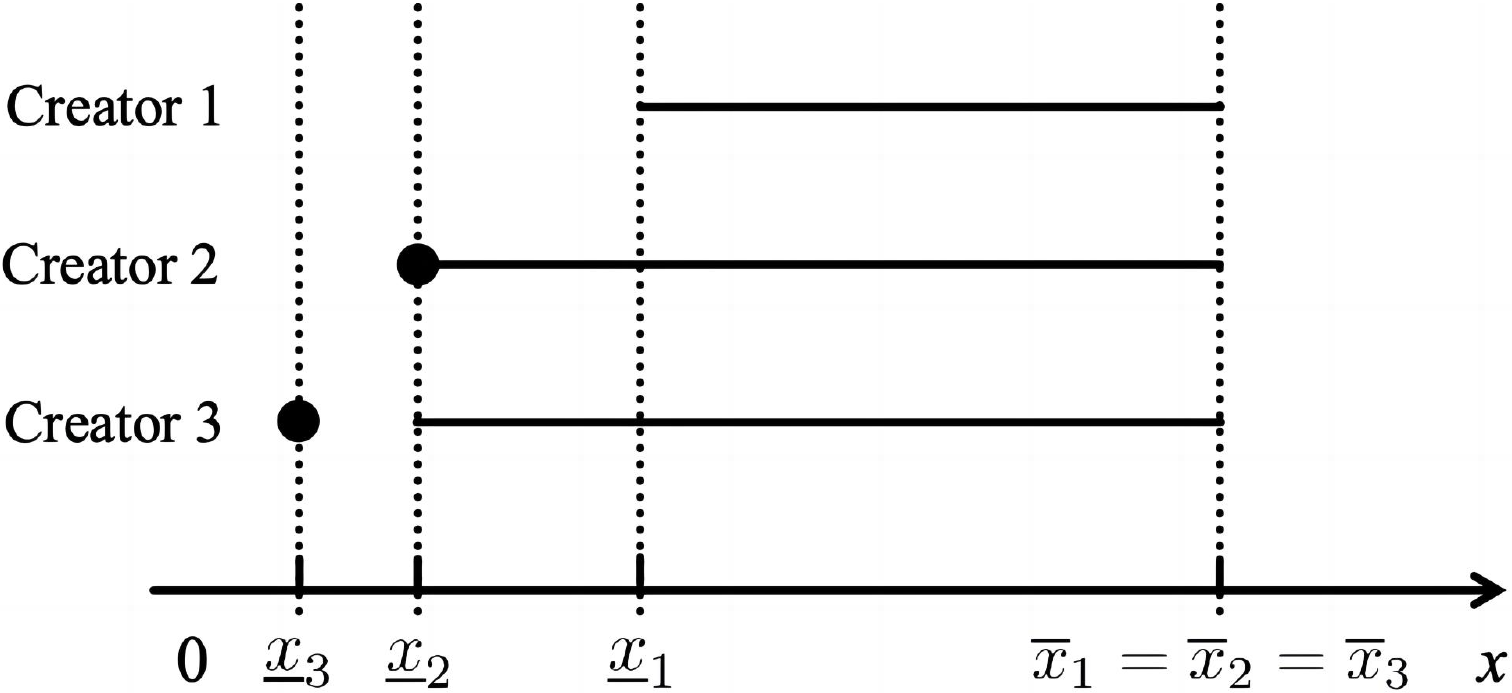}
        \caption{Support of mixed Nash equilibrium (Regime~5)}
        \label{fig:3-player-support-5}
    \end{subfigure}
    \hfill
    \begin{subfigure}[b]{0.48\linewidth}
        \centering
        \includegraphics[width=\linewidth]{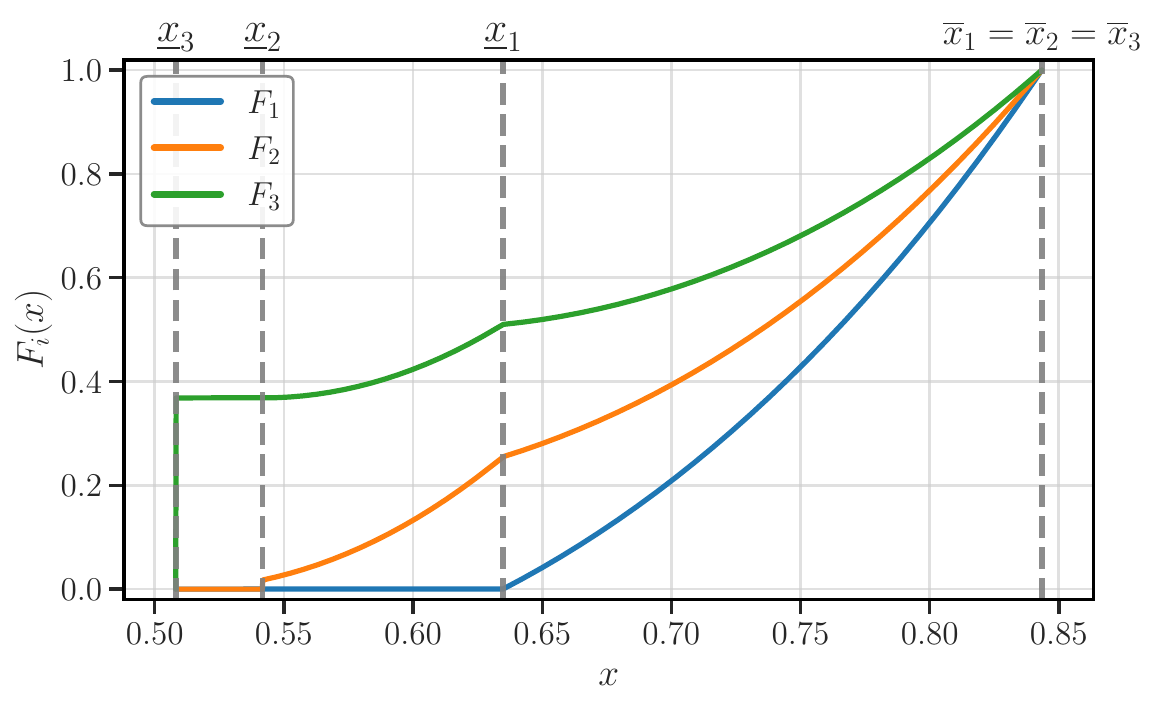}
        \caption{CDFs of mixed Nash equilibrium (Regime~5)}
        \label{fig:3-player-cdf-5}
    \end{subfigure}
    \caption{Three-player equilibrium illustrations (Regime~5)}
    \label{fig:three-player-combined-5}
\end{figure}

\xhdr{Regime 6 ($\gamma_2 = 1.2$)} In this regime, the equilibrium satisfies $F_2 = F_3$. Numerical results show that the infimums of the supports are $\underline{x}_2 = \underline{x}_3 \approx 0.508$ and $\underline{x}_1 \approx 0.672$, and the supremum of the support is $\overline{x}_1 = \overline{x}_2 = \overline{x}_3 \approx 0.843$. Three creators are all active between the infimum and supremum of their strategies, and there is no point mass in their strategies. Figure~\ref{fig:three-player-combined-6} illustrates the support structure and CDFs of the mixed Nash equilibrium in this regime.

\begin{figure}[t]
    \centering
    \begin{subfigure}[b]{0.48\linewidth}
        \centering
        \includegraphics[width=\linewidth]{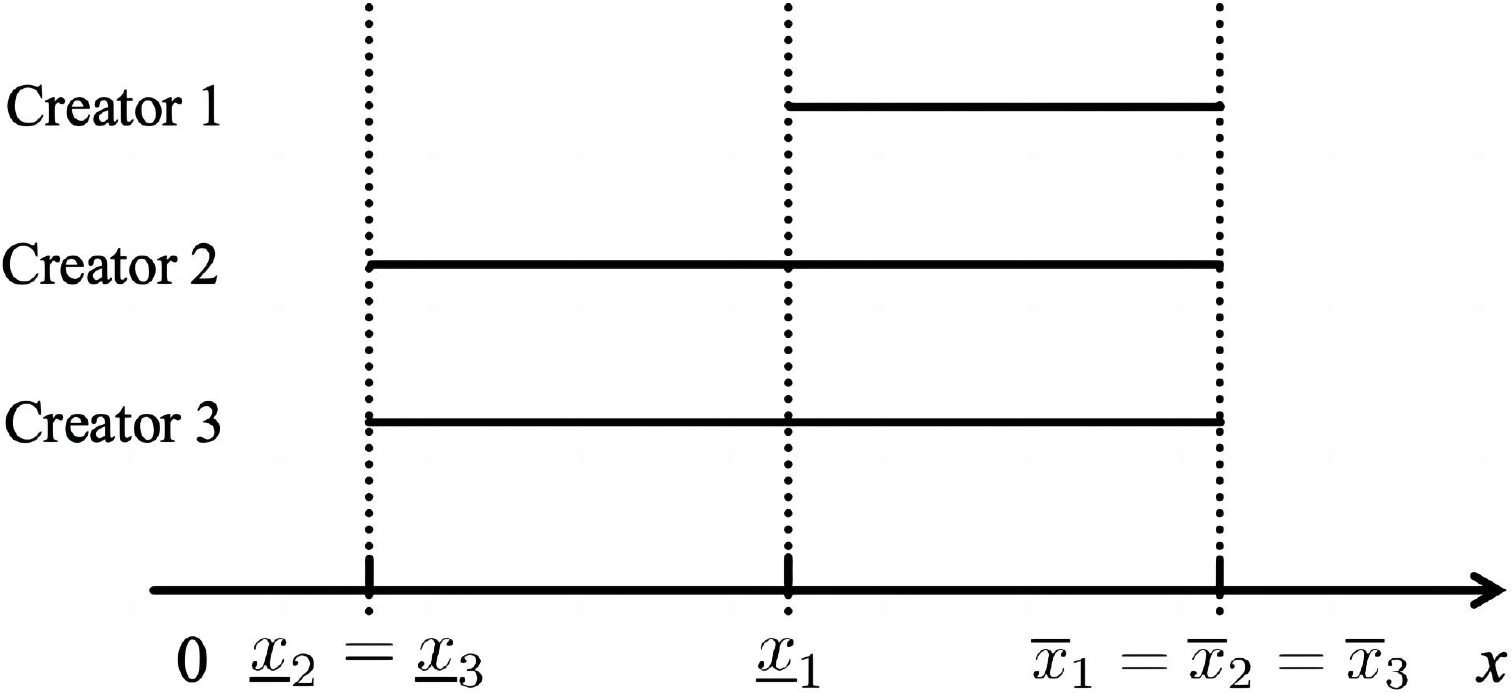}
        \caption{Support of mixed Nash equilibrium (Regime~6)}
        \label{fig:3-player-support-6}
    \end{subfigure}
    \hfill
    \begin{subfigure}[b]{0.48\linewidth}
        \centering
        \includegraphics[width=\linewidth]{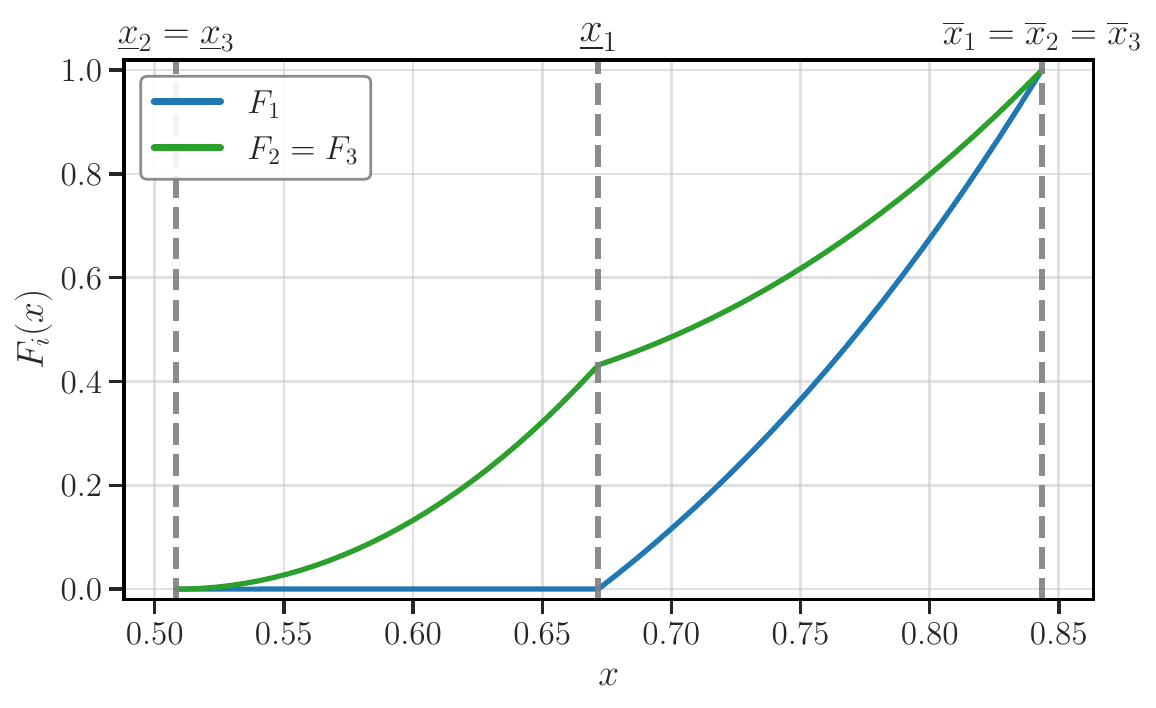}
        \caption{CDFs of mixed Nash equilibrium (Regime~6)}
        \label{fig:3-player-cdf-6}
    \end{subfigure}
    \caption{Three-player equilibrium illustrations (Regime~6)}
    \label{fig:three-player-combined-6}
\end{figure}

\subsection{Properties of Equilibrium with \texorpdfstring{$n = 2$}{n = 2}} \label{subsec:ap-asym-n2}

In this subsection, we focus on the special case where there are only two creators, i.e., $n = 2$. We provide a complete characterization of the Nash equilibrium configuration for any given cost parameters $\gamma_1, \gamma_2 > 0$. This characterization will serve as the foundation for the subsequent analysis of the binary-type setting.

Without loss of generality, consider the case where $\gamma_1 \leqslant \gamma_2$. Let $x_1 = \arg\max_x x p_1 - \gamma_1 g(x)$, $u_1 = x_1 p_1 - \gamma_1 g(x_1)$, $x_2 = \arg\max_x x p_2 - \gamma_2 g(x)$, $u_2 = x_2 p_2 - \gamma_2 g(x_2)$, $u_2' = x_1 p_1 - \gamma_2 g(x_1)$.

\begin{proposition} \label{prop:2-player-equ}
    For two creators ($n = 2$) with $\gamma_1 \leqslant \gamma_2$, the Nash equilibrium falls into three distinct regimes:
    \begin{enumerate}
        \item If $\gamma_1 = \gamma_2$, the unique equilibrium is a symmetric mixed Nash equilibrium.
        \item If $\gamma_1 < \gamma_2$ and $u_2' \leqslant u_2$, the unique equilibrium is the pure Nash equilibrium $(x_1, x_2)$.
        \item If $\gamma_1 < \gamma_2$ and $u_2' > u_2$, the equilibrium is unique and has the following structure:
        \begin{itemize}
            \item The infimum effort of Creator 2 satisfies
            \begin{equation} \label{eq:2-general-inf2}
                \underline{x}_2 = \arg\max_x x p_2 - \gamma_2 g(x),
            \end{equation}
            and is played with an atom of probability $q_2$.
            \item The infimum effort of Creator 1 satisfies 
            \begin{equation} \label{eq:2-general-inf1}
                \underline{x}_1 = \arg\max_x x (q_2 p_1 + (1 - q_2) p_2) - \gamma_1 g(x),
            \end{equation}
            and is played with an atom of probability $q_1$.
            \item Both creators share common mixed strategy support $(\underline{x}_1, \overline{x}]$, where the boundaries are determined by
            \begin{gather}
                \underline{x}_1 (q_1 p_1 + (1 - q_1) p_2) - \gamma_2 g(\underline{x}_1) = \underline{x}_2 p_2 - \gamma_2 g(\underline{x}_2) \label{eq:2-general-21}, \\
                \underline{x}_1 (q_2 p_1 + (1 - q_2) p_2) - \gamma_1 g(\underline{x}_1) = \overline{x} p_1 - \gamma_1 g(\overline{x}), \label{eq:2-general-sup1} \\
                \overline{x} p_1 - \gamma_2 g(\overline{x}) = \underline{x}_2 p_2 - \gamma_2 g(\underline{x}_2). \label{eq:2-general-sup2}
            \end{gather}
        \end{itemize}
        Let $u_1 = \max_x x (q_2 p_1 + (1 - q_2) p_2) - \gamma_1 g(x), u_2 = \max_x x p_2 - \gamma_2 g(x)$, and let the CDF of the strategy of creator $1$ and $2$ be $F_1$ and $F_2$. Then we have
        \begin{equation} \label{eq:2-general-cdf}
            \begin{aligned}
                F_1(x) = \frac{u_2 + \gamma_2 g(x)}{x (p_1 - p_2)} - \frac{p_2}{p_1 - p_2}, \\
                F_2(x) = \frac{u_1 + \gamma_1 g(x)}{x (p_1 - p_2)} - \frac{p_2}{p_1 - p_2}.
            \end{aligned}
        \end{equation}
    \end{enumerate}
\end{proposition}

Figure~\ref{fig:2-player-support} illustrates the support structure of the mixed Nash equilibrium in the third regime. Figure~\ref{fig:2-player-gamma2} illustrates the variation in the support of creator 1 strategy with $\gamma_2$ where $\gamma_1 = 1$, $g(x) = x^5$, $p_1 = 0.8$ and $p_2 = 0.6$, The upper curve traces the supremum $\overline{x}_1$ of the support, while the lower curve traces its infimum $\underline{x}_1$. As $\gamma_2$ gradually increases from $1$ to $3$, the equilibrium strategy of Creator 1 shifts smoothly from the symmetric mixed Nash equilibrium form toward the pure Nash equilibrium form, with convergence to the pure strategy occurring around $\gamma_2 \approx 2.85$. This smooth transition illustrates the continuous connection among the three equilibrium regimes described in Proposition~\ref{prop:2-player-equ}, confirming the plausibility of the proposition.

\begin{figure}[t]
    \centering
    \begin{subfigure}[b]{0.36\linewidth}
        \centering
        \includegraphics[width=\linewidth]{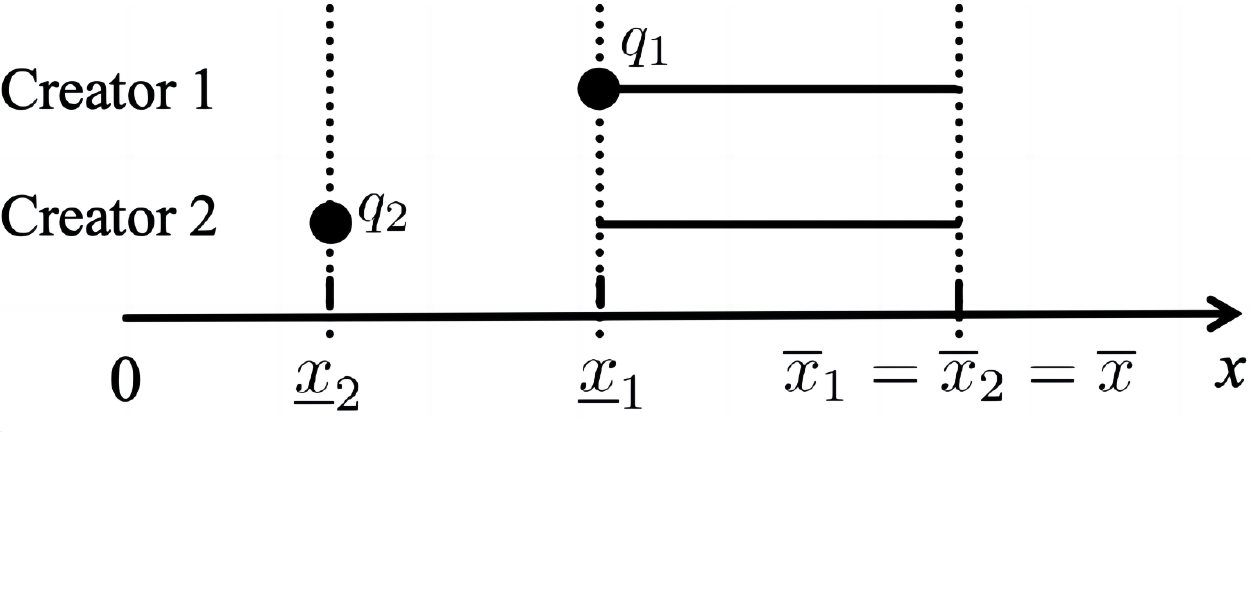}
        \caption{Support of mixed Nash equilibrium}
        \label{fig:2-player-support}
    \end{subfigure}
    \begin{subfigure}[b]{0.36\linewidth}
        \centering
        \includegraphics[width=\linewidth]{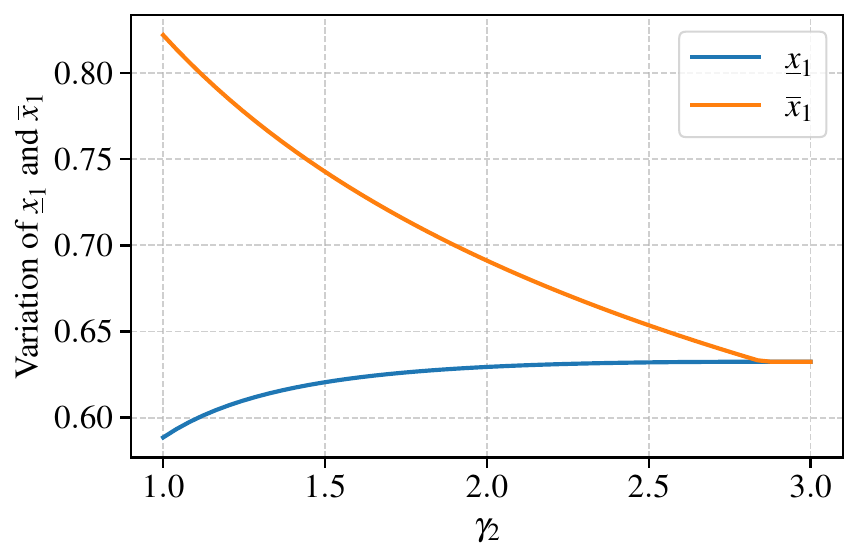}
        \caption{Variation in Creator 1's support with $\gamma_2$}
        \label{fig:2-player-gamma2}
    \end{subfigure}
    \caption{Two-player equilibrium illustrations}
    \label{fig:two-player-combined}
\end{figure}

\begin{proof}
    The first part of the proposition follows directly from Theorem~\ref{thm:sym-equ-property}. The second part is a direct consequence of Proposition~\ref{prop:decrease_equilibrium} and Proposition~\ref{prop:pure-unique}. Then, we focus on proving the third part of the proposition. The proof can be divided into five steps:

    \textbf{Step 1: The mixed strategy intervals of two creators are necessarily identical.}

    Without loss of generality, suppose that interval $(x', x'')$ is contained in the support of creator $1$ but not in that of creator $2$. Then for any $x \in (x', x'')$, the utility of creator $1$ satisfies
    \[ u_1(x) = x W(\bm{F}_{-1}(x'), \bm{p}) - \gamma_1 g(x). \]

    Due to the strict convexity of $g$, $u_1(x)$ is strictly concave in $x$, and hence is not constant on $(x', x'')$, contradicting the indifference condition for creator $1$. Therefore, the mixed strategy intervals of two creators must be identical.

    \textbf{Step 2: Equations~\eqref{eq:2-general-inf2} and \eqref{eq:2-general-inf1} hold.}

    According to Corollary~\ref{cor:asym-support-monotone}, eqution~\eqref{eq:2-general-inf2} holds. Then Step~1 implies that creator $2$ would not be active before creator $1$'s infimum strategy, thus similar to Step~2 in the proof of Theorem~\ref{thm:sym-equ-property}, we can show that equation~\eqref{eq:2-general-inf1} holds.

    \textbf{Step 3: Equations~\eqref{eq:2-general-21}-\eqref{eq:2-general-sup2} hold.}

    We first observe that after $\underline{x}_1$ neither creator can place further point masses, and there cannot be a gap before the supremum $\overline{x}$ of the joint support. Otherwise, the indifference condition would be violated, as shown in Step~3 of the proof of Theorem~\ref{thm:sym-equ-property} together with Corollary~\ref{cor:symmetric-no-gap}. Consequently, by the conclusion of Step~1, both creators must take mixed strategies over the common interval $(\underline{x}_1, \overline{x}]$.

    Then, according to the indifference condition of creator $2$ at point $\underline{x}_2$, $\underline{x}_1$ and $\overline{x}$, we have equations~\eqref{eq:2-general-21} and \eqref{eq:2-general-sup2}. Similarly, according to the indifference condition of creator $1$ at point $\underline{x}_1$ and $\overline{x}$, we have equation~\eqref{eq:2-general-sup1}.

    \textbf{Step 4: Equations~\eqref{eq:2-general-cdf} hold.}

    According to the indifference condition of creator $1$ at any point $x \in (\underline{x}_1, \overline{x}]$, we have
    \[ u_1 = x W(\bm{F}_{-1}(x), \bm{p}) - \gamma_1 g(x) = x (F_2(x) p_1 + (1 - F_2(x)) p_2) - \gamma_1 g(x). \]
    Thus, we can derive the expression of $F_2(x)$. Similarly, we can derive the expression of $F_1(x)$ according to the indifference condition of creator $2$.

    \textbf{Step 5: Equations~\eqref{eq:2-general-inf2}-\eqref{eq:2-general-sup2} admit a unique solution.}
    
    Due to the strict convexity of $g$, equation~\eqref{eq:2-general-inf2} admits a unique solution for $\underline{x}_2$. Then, given $\underline{x}_2$, combined with the fourth part of Proposition~\ref{prop:symmetric-equilibrium-properties}, the largest root of equation~\eqref{eq:2-general-sup2} gives a unique solution for $\overline{x}$.

    Next, from \eqref{eq:2-general-inf1} and the monotonicity of $g$, a larger $p_2$ leads to a larger $\underline{x}_1$. Note that the right-hand side of \eqref{eq:2-general-sup1} is constant once $\overline{x}$ is given; therefore, \eqref{eq:2-general-inf1} and \eqref{eq:2-general-sup1} together uniquely determine $\underline{x}_1$ and $q_2$.

    Finally, when $\underline{x}_1$ is known, the left-hand side of \eqref{eq:2-general-21} increases in $p_1$, while the right-hand side is constant once $\underline{x}_2$ is fixed. Therefore, \eqref{eq:2-general-21} uniquely determines the value of $q_1$.

    Combining the results of the five steps above, we complete the proof.
\end{proof}

\subsection{Results for Binary Types of Creators} \label{subsec:ap-asym-two}

We first prove the main theorem (Theorem~\ref{thm:asym-equ}) through Proposition~\ref{prop:asym-unique-uH} and Proposition~\ref{prop:asym-separated-equilibrium}. Firstly, Proposition~\ref{prop:asym-unique-uH} establishes the uniqueness of $u_H$ (Corollary~\ref{cor:asym-support-monotone} guarantees the uniqueness of $u_L$). Secondly, Proposition~\ref{prop:asym-separated-equilibrium} characterizes the condition under which the equilibrium is separated, as well as the uniqueness of the equilibrium in the separated case. Moreover, Proposition~\ref{prop:asym-separated-equilibrium} also characterizes the equilibrium in the separated case more explicitly.

\begin{proposition} \label{prop:asym-unique-uH}
    In game $\mathcal{G}^{(2)}$, let $\left[\underline{x}_H^u, \overline{x}_H^u\right]$ be the support of $F_H^u$, $u_L = \max_x x p_n - \gamma_L g(x)$ be the equilibrium utility of type-$L$ creators, $u_L^u$ denote the utility that type-$L$ creators would obtain by the best response against $F_H^u$ on $\left[\underline{x}_H^u, \overline{x}_H^u\right]$, $u_H$ be the equilibrium utility of type-$H$ creators, and $\overline{u}_H = \max_x x p_{n_H} - \gamma_H g(x)$. Then
    \begin{enumerate}
        \item $u_L^u$ is strictly increasing in $u$.
        \item If $u_L^{\overline{u}_H} \leqslant u_L$, then the equilibrium is separated with $u_H = \overline{u}_H$.
        \item If $u_L^{\overline{u}_H} > u_L$, then $u_L^u = u_L$ if and only if $u = u_H$.
        \item $u_H$ is unique.
    \end{enumerate}
\end{proposition}

Part~1 of Proposition~\ref{prop:asym-unique-uH} implies that a bisection search can be used to determine $u_H$. A lower bound of $u_H$ is $\underline{u}_H = \underline{x}_L p_n - \gamma_H g(\underline{x}_L)$ (by Lemma~\ref{lem:asym-uH-lower-bound} below); an upper bound is $\overline{u}_H = \max_x x p_{n_H} - \gamma_H g(x)$ (by Proposition~\ref{prop:symmetric-equilibrium-properties}). According to Part~2, we first compute $u_L^{\overline{u}_H}$ and compare it with $u_L$. If $u_L^{\overline{u}_H} \leqslant u_L$, the equilibrium is separated and $u_H = \overline{u}_H$. Otherwise, according to Part~3, we perform a bisection on the interval $[\underline{u}_H, \overline{u}_H]$: for each trial value $u_H$ we compute $u_L^{u_H}$ and compare it with $u_L$, thus shrinking the interval until a value satisfying $u_L^{u_H} = u_L$ is located. In practice, due to finite numerical precision, the bisection can be terminated when $u_L^{u_H}$ is sufficiently close to $u_L$. The algorithm for computing $u_H$ is summarized in Algorithm~\ref{alg:uH-computation}.

\begin{lemma} \label{lem:asym-uH-lower-bound}
    In game $\mathcal{G}^{(2)}$, the utility of type-$H$ creators (denoted as $u_H$) in any equilibrium satisfies
    \[ u_H > \underline{x}_L p_n - \gamma_H g(\underline{x}_L), \]
    where $\underline{x}_L = \arg\max_x x p_n - \gamma_L g(x)$.
\end{lemma}

\begin{proof}
    By Corollary~\ref{cor:asym-support-monotone}, the infimum of the support of type-$L$ creators is just $\underline{x}_L$. It is obvious that in any equilibrium, type-$H$ creators must achieve a utility of at least $\underline{x}_L p_n - \gamma_H g(\underline{x}_L)$; otherwise, a type-$H$ creator can deviate by choosing effort level $\underline{x}_L$ and obtain a higher utility.

    Let $x' = \arg\max_x x p_n - \gamma_H g(x)$, then $x' > \underline{x}_L$ since $\gamma_H < \gamma_L$. Consider a deviation where a type-$H$ creator (denoted as creator $h$) chooses effort level $x'$, the resulting utility is
    \[ u_H' = x' W(\bm{F}_{-h}(x'), \bm{p}) - \gamma_H g(x') \geqslant x' p_n - \gamma_H g(x') > \underline{x}_L p_n - \gamma_H g(\underline{x}_L). \]
    Thus in any Nash equilibrium, type-$H$ creators must achieve a utility strictly greater than $\underline{x}_L p_n - \gamma_H g(\underline{x}_L)$.
\end{proof}

Then we provide the proof of Proposition~\ref{prop:asym-unique-uH}.

\begin{proof}
    \begin{enumerate}
        \item First, we rewrite equation~\eqref{eq:pseudo-strategy} as
        \[ \sum_{i = 1}^{n_H} \binom{n_H - 1}{i - 1} [F_H^u(x)]^{n_H - i} [1 - F_H^u(x)]^{i - 1} p_i = \frac{u + \gamma_H g(x)}{x}. \]
        By Lemma~\ref{lem:KJ-monotone}, the left-hand side is strictly increasing in $F_H^u(x)$, while the right-hand side is clearly strictly increasing in $u$. Therefore, when $u$ increases to $u'$, we have $F_H^{u'}(x) > F_H^u(x)$. Consequently, $\underline{x}_H^{u'} < \underline{x}_H^u$ and $\overline{x}_H^{u'} < \overline{x}_H^u$.

        Let $u_L^u(x)$ denote the utility of a type-$L$ creator at effort level $x \in \left[\underline{x}_H^u, \overline{x}_H^u\right]$ when type-$H$ creators adopt the pseudo strategy $F_H^u$. Then we have
        \begin{equation} \label{eq:asym-uLu}
            u_L^u(x) = x \left(\sum_{i = 1}^{n_H + 1} \binom{n_H}{i - 1} [F_H^u(x)]^{n_H - i + 1} [1 - F_H^u(x)]^{i - 1} p_i\right) - \gamma_L g(x).
        \end{equation}
        According to Lemma~\ref{lem:KJ-monotone}, the summation term is strictly increasing in $F_H^u(x)$. Thus for any $x \in [\underline{x}_H^u, \overline{x}_H^{u'}]$, the inequality $F_H^{u'}(x) > F_H^u(x)$ implies $u_L^{u'}(x) > u_L^u(x)$.

        Now take any $x \in (\overline{x}_H^{u'}, \overline{x}_H^u]$. Define the auxiliary function $v(x) = x p_1 - \gamma_L g(x)$, we claim that $v(x)$ is strictly decreasing in $x$ over the interval $(\overline{x}_H^{u'}, \overline{x}_H^u]$. To see this, note that $\overline{u}_H = \max_x x p_n - \gamma_H g(x)$ is the highest attainable utility of type-$H$ creators in any equilibrium, then by the fourth part of Proposition~\ref{prop:symmetric-equilibrium-properties}, we have $\overline{x}_H^{u'} \geqslant \overline{x}$, where $\overline{x}$ is defined in the fourth part of Proposition~\ref{prop:symmetric-equilibrium-properties}. Since $v(x)$ is strictly decreasing for $x > \overline{x}$, it remains strictly decreasing on $(\overline{x}_H^{u'}, \overline{x}_H^u]$. Hence, for any $x \in (\overline{x}_H^{u'}, \overline{x}_H^u]$,
        \[ u_L^{u'}(\overline{x}_H^{u'}) = \overline{x}_H^{u'} p_1 - \gamma_L g(\overline{x}_H^{u'}) > x p_1 - \gamma_L g(x) \geqslant u_L^u(x). \]
        Combining the two parts above, we conclude that the maximum of $u_L^{u'}(x)$ is strictly larger than the maximum of $u_L^u(x)$. Conversely, if $u' < u$, a symmetric argument shows that the maximum of $u_L^{u'}(x)$ is strictly smaller than the maximum of $u_L^u(x)$.
        
        \item Let $F_L$ denote the strategy that satisfies
        \begin{equation} \label{eq:asym-FL}
            x \left(\sum_{i = 1}^{n_L} \binom{n_L - 1}{i - 1} [F_L(x)]^{n_L - i} [1 - F_L(x)]^{i - 1} p_{n_H + i}\right) - \gamma_L g(x) = u_L,
        \end{equation}
        then the pair where type-$L$ creators adopt $F_L$ and type-$H$ creators adopt the pseudo-strategy $F_H^{\overline{u}_H}$ constitutes an equilibrium. Indeed, by the discussion in Section~\ref{subsec:sym-equilibrium}, $F_L$ and $F_H^{\overline{u}_H}$ are the mixed Nash equilibrium strategies for type-$L$ and type-$H$ creators in their respective games. Moreover, the inequality $u_L \geqslant u_L^{\overline{u}_H}$ implies that type-$L$ creators have no incentive to deviate into the interval $[\underline{x}_H^{\overline{u}_H}, \overline{x}_H^{\overline{u}_H}]$. As for type-$H$ creators, any deviation into the support of $F_L$ would yield a utility strictly lower than $u_H = \max_x x p_{n_H} - \gamma_H g(x)$; consequently, they also have no incentive to deviate into the support of $F_L$. Thus, the profile $(F_L, F_H^{\overline{u}_H})$ forms an equilibrium.

        To complete the proof of this part, it remains to show that this equilibrium is the unique equilibrium of game $\mathcal{G}^{(2)}$. It is clearly the only separated equilibrium, so we only need to rule out the existence of any hybrid equilibrium. Suppose, for contradiction, that a hybrid equilibrium does exist. In such an equilibrium,  type-$L$ creators no longer choose effort levels strictly below the lower bound of the support of type-$H$ creators with probability 1, and therefore the equilibrium utility of type-$H$ creators would then be strictly lower than $\overline{u}_H$. By Part~1, the equilibrium utility of type-$L$ creators would then be strictly smaller than $u_L^{\overline{u}_H}$ and hence also smaller than $u_L$, This contradicts the definition of an equilibrium. Hence, no hybrid equilibrium exists.

        \item If $u_L^{\overline{u}_H} > u_L$, then Part~1 implies that $u_H < \overline{u}_H$. Denote $x_H^* = \arg\max_x x p_{n_H} - \gamma_H g(x)$, and let $[\underline{x}_H, \overline{x}_H]$ and $[\underline{x}_L, \overline{x}_L]$ denote the supports of type-$H$ and type-$L$ creators' equilibrium strategies, respectively (note that for type-$H$ creators, $\left[\underline{x}_H^u, \overline{x}_H^u\right]$ denotes the support of pseudo strategy, $[\underline{x}_H, \overline{x}_H]$ denotes the support of equilibrium strategy). By the proof of Part~1, $u_H < \overline{u}_H$ implies that $\underline{x}_H > x_H^*$. We claim that $\overline{x}_L \geqslant \underline{x}_H$, otherwise any type-$H$ creator can deviate to choose effort level $\underline{x}_H - \varepsilon$ ($\varepsilon \to 0^+$) to obtain a higher utility.
        
        With the above observations, we now prove the ``if'' part. If $u = u_H$, we first show that $\left[\underline{x}_H^{u_H}, \overline{x}_H^{u_H}\right] \subseteq [\underline{x}_H, \overline{x}_H]$. Consider the interval $[\overline{x}_L, \overline{x}_H]$. Since the support of type-$L$ creators lies entirely below this interval, the indifference condition implies that the pseudo strategy of type-$H$ creators coincides with their equilibrium strategy on this interval, hence $\overline{x}_H^{u_H} = \overline{x}_H$. Let $F_L$ be the CDF of the equilibrium strategy of type-$L$ creators, $F_L^{u_H}$ the CDF that type-$L$ creators adopt when type-$H$ creators adopt pseudo strategy $F_H^{u_H}$. By definition of the pseudo strategy, $F_L^{u_H}(x) \equiv 1$ for $x \geqslant \underline{x}_H^{u_H}$. For $x \in [\underline{x}_H^{u_H}, \overline{x}_L)$, we have $F_L(x) < 1 = F_L^{u_H}(x)$. Because the utilities of type-$L$ creators under the equilibrium and under the pseudo strategy are equal on $[\underline{x}_H^{u_H}, \overline{x}_L)$, it must be that $F_H(x) \geqslant F_H^{u_H}(x)$ (where $F_H$ is the equilibrium CDF of type-$H$ creators), and thus $\underline{x}_H \leqslant \underline{x}_H^{u_H}$. Together with $\overline{x}_H^{u_H} = \overline{x}_H$, we obtain $\left[\underline{x}_H^{u_H}, \overline{x}_H^{u_H}\right] \subseteq [\underline{x}_H, \overline{x}_H]$.
        
        Consider the interval $[\overline{x}_L, \overline{x}_H]$, as argued above, the equilibrium strategy and the pseudo strategy coincide on this interval. Therefore, at $\overline{x}_L$ the utility of type-$L$ creators is $u_L^{u_H}(\overline{x}_L) = u_L$, implying $u_L^{u_H} \geqslant u_L$. It remains to show that $u_L^{u_H}$ cannot exceed $u_L$. Take any $x \in [\underline{x}_H^{u_H}, \overline{x}_H^{u_H}]$, then $x \in [\underline{x}_H, \overline{x}_H]$ according to the above discussion. In equilibrium the utility of type-$L$ creators at $x$ is $u_L$, while the inequality $F_H(x) \geqslant F_H^{u_H}(x)$ established above yields $u_L^{u_H}(x) \leqslant u_L$. Hence $u_L^{u_H} = u_L$. 
        
        We now prove the ``only if'' part. Suppose, for contradiction, that there exists some $u \neq u_H$ such that $u_L^u = u_L$. By Part~1 and Part~2, $u_L^u$ is strictly monotone in $u$, therefore $u_L^u \neq u_L^{u_H} = u_L$, a contradiction.
        
        \item According to Part~2, when $u_L^{\overline{u}_H} \leqslant u_L$, the equilibrium is unique, and consequently the equilibrium utility $u_H$ is unique as well. If instead $u_L^{\overline{u}_H} > u_L$, there must exist some $u$ for which $u_L^u < u_L$: Indeed, the proof of Part~1 implies that there exists sufficiently small $u$ such that $\underline{x}_H^u p_1 - \gamma_L g(\underline{x}_H^u) < u_L$. Moreover, by Part~1 together with the fact that the pseudo strategy $F_H^u$ varies continuously with $u$, there exists a unique $u^*$ such that $u_L^{u^*} = u_L$. Part~3 then implies that this $u^*$ coincides with the equilibrium utility $u_H$ and is unique. In all cases, therefore, the equilibrium utility $u_H$ is unique.
    \end{enumerate}
\end{proof}

Denote $[\underline{x}_L, \overline{x}_L]$ and $[\underline{x}_H, \overline{x}_H]$ as the supports of type-$L$ and type-$H$ creators, respectively. Based on the above proposition, the characterization of the separated equilibrium is straightforward, which we summarize in the following proposition.

\begin{proposition} \label{prop:asym-separated-equilibrium}
    In game $\mathcal{G}^{(2)}$, the equilibrium is separated if and only if one of the following two conditions holds:
    \begin{enumerate}
        \item $u_L^{\overline{u}_H} \leqslant u_L$. In this case if $u_L^{\overline{u}_H} < u_L$, or $u_L^{\overline{u}_H} = u_L$ but $u_L^{\overline{u}_H}$ is attained at some $x' \in (\underline{x}_H^{\overline{u}_H}, \overline{x}_H^{\overline{u}_H}]$, then $\overline{x}_L < \underline{x}_H$.
        \item There exists a $u$ such that $u_L^u = u_L$ and $u_L^u$ is attained at $\underline{x}_H^u$. In this case $\overline{x}_L = \underline{x}_H$.
    \end{enumerate}

    Moreover, $\underline{x}_H = \underline{x}_H^{u_H}$ and $\overline{x}_H = \overline{x}_H^{u_H}$ in both cases, where $u_H$ is the equilibrium utility of type-$H$ creators determined in Proposition~\ref{prop:asym-unique-uH}. Whenever the equilibrium is separated, it is the unique equilibrium of game $\mathcal{G}^{(2)}$.
\end{proposition}

\begin{proof}
    \begin{enumerate}
        \item It follows directly from Part~2 of Proposition~\ref{prop:asym-unique-uH}.
        \item We first show that $\overline{x}_L = \underline{x}_H^u$. Otherwise, if $\overline{x}_L < \underline{x}_H^u$, denote $x_H^* = \arg\max_x x p_{n_H} - \gamma_H g(x)$. Then according to the proof of Part~1 of Proposition~\ref{prop:asym-unique-uH}, $u < \overline{u}_H$ implies that $\underline{x}_H^u > x_H^*$. Thus, any type-$H$ creator can deviate to choose effort level $\underline{x}_H - \varepsilon$ ($\varepsilon \to 0^+$) to obtain a higher utility. If instead $\overline{x}_L > \underline{x}_H^u$, since the support structure in Figure~\ref{fig:2-type-separated-support-2} can be seen as a special case of the support structure of a hybrid equilibrium shown in Figure~\ref{fig:2-type-hybrid-support-2}, by the argument in the proof of Proposition~\ref{prop:asym-hybrid-equilibrium}, the support structure cannot be the case shown in Figure~\ref{fig:2-type-hybrid-support-1}, thus no hybrid equilibrium exists in this case, which implies that $\overline{x}_L$ cannot be strictly greater than $\underline{x}_H^u$. Therefore, we have $\overline{x}_L = \underline{x}_H^u$.
        
        
        Let $F_L$ denote the strategy that satisfies equation~\eqref{eq:asym-FL}. Consider the profile $(F_L, F_H^u)$ in which type-$L$ creators adopt $F_L$ and type-$H$ creators adopt the pseudo-strategy $F_H^u$. The discussion above then demonstrates that $(F_L, F_H^u)$ is the only possible separated equilibrium (implied by the argument that rules out the case $\overline{x}_L < \underline{x}_H^u$), and there is no hybrid equilibrium (evident from the argument that excludes the case $\overline{x}_L > \underline{x}_H^u$). Hence, $(F_L, F_H^u)$ is indeed the unique equilibrium of game $\mathcal{G}^{(2)}$ with $\overline{x}_L = \underline{x}_H$, $\underline{x}_H = \underline{x}_H^{u_H}$ and $\overline{x}_H = \overline{x}_H^{u_H}$.
    \end{enumerate}

    If neither of the two cases holds, then $u_L^{\overline{u}_H} > u_L$ (thus $u_H < \overline{u}_H$) and there exists a $u$ such that $u_L^u = u_L$ and $u_L^u$ is attained at some $x' \in (\underline{x}_H^u, \overline{x}_H^u]$. We consider the following two sub-cases: (1) if $\overline{x}_L < \underline{x}_H$, then $u_H = \overline{u}_H$, which contradicts the assumption; (2) if $\overline{x}_L = \underline{x}_H$, the indifference condition of type-$L$ creators implies that $u_L = u_L^u$ is attained at $\underline{x}_H^u$, contradicting the assumption. Therefore, if neither of the two cases holds, the equilibrium must be hybrid.
\end{proof}

According to Proposition~\ref{prop:asym-separated-equilibrium}, Figure~\ref{fig:2-type-separated-support} illustrates the support structure of the separated equilibrium in the two cases. The key difference is whether type-$L$ creators exert effort at the infimum of the support of type-$H$ creators. We provide a numerical example of the separated equilibrium in Example~\ref{ex:asym-separated-equilibrium}.

\begin{figure}[t]
    \centering
    \begin{subfigure}[b]{0.36\linewidth}
        \centering
        \includegraphics[width=\linewidth]{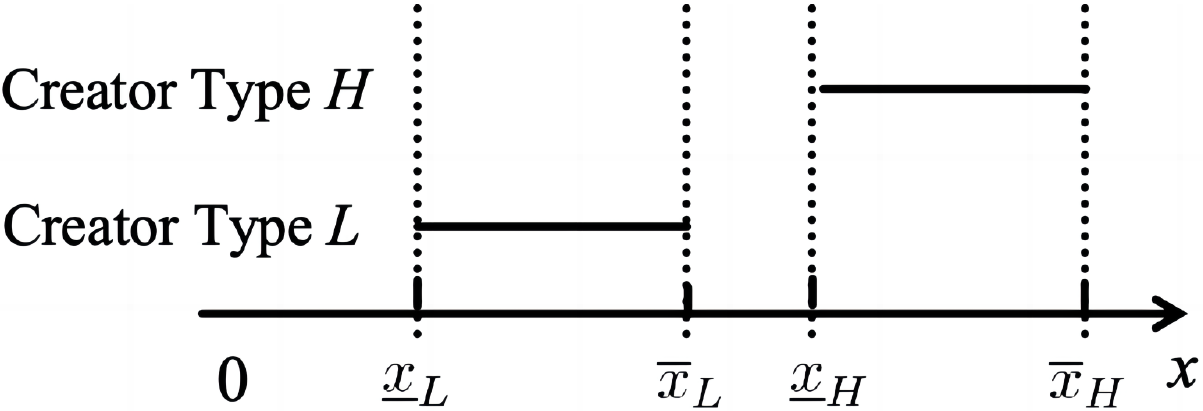}
        \caption{Case~1 of Proposition~\ref{prop:asym-separated-equilibrium}}
        \label{fig:2-type-separated-support-1}
    \end{subfigure}
    \begin{subfigure}[b]{0.36\linewidth}
        \centering
        \includegraphics[width=\linewidth]{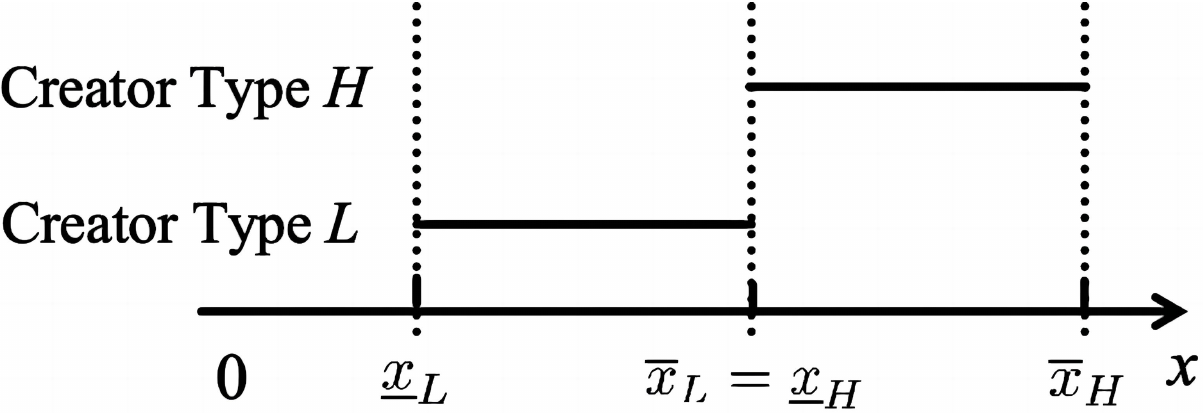}
        \caption{Case~2 of Proposition~\ref{prop:asym-separated-equilibrium}}
        \label{fig:2-type-separated-support-2}
    \end{subfigure}
    \caption{Support structure of separated equilibrium}
    \label{fig:2-type-separated-support}
\end{figure}

\begin{example} \label{ex:asym-separated-equilibrium}
    Consider a setting with $n_H = n_L = 5$, $p_1 = 0.9, p_2 = 0.75, p_3 = 0.63, p_4 = 0.52, p_5 = 0.42, p_6 = 0.34, p_7 = 0.28, p_8 = 0.24, p_9 = 0.20, p_{10} = 0.18$, and $g(x) = x^3$. Let $\gamma_H = 1$, then numerical computation yields:
    \begin{enumerate}
        \item If $\gamma_L = 3$, then the structure of the unique mixed Nash equilibrium corresponds to case~1 in Proposition~\ref{prop:asym-separated-equilibrium}, with the support of type-$H$ creators being approximately $[0.374, 0.884]$ and that of type-$L$ creators being approximately $[0.141, 0.308]$.
        \item If $\gamma_L = 2$, then the structure of the unique mixed Nash equilibrium corresponds to case~2 in Proposition~\ref{prop:asym-separated-equilibrium}, with the support of type-$H$ creators being approximately $[0.377, 0.884]$ and that of type-$L$ creators being approximately $[0.173, 0.377]$.
    \end{enumerate}
\end{example}

\begin{figure}[t]
    \centering
    \begin{subfigure}[b]{0.48\linewidth}
        \centering
        \includegraphics[width=\linewidth]{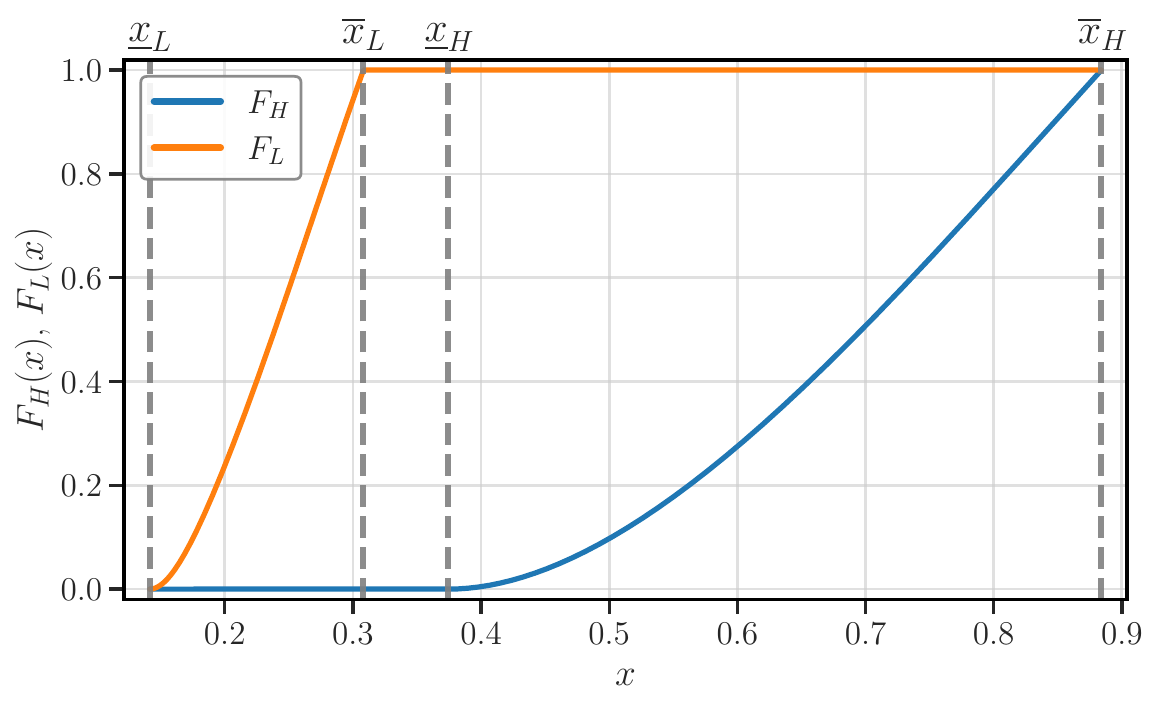}
        \caption{Equilibrium CDFs with $\gamma_L = 3$ in Example~\ref{ex:asym-separated-equilibrium}}
        \label{fig:separated-example-1}
    \end{subfigure}
    \hfill
    \begin{subfigure}[b]{0.48\linewidth}
        \centering
        \includegraphics[width=\linewidth]{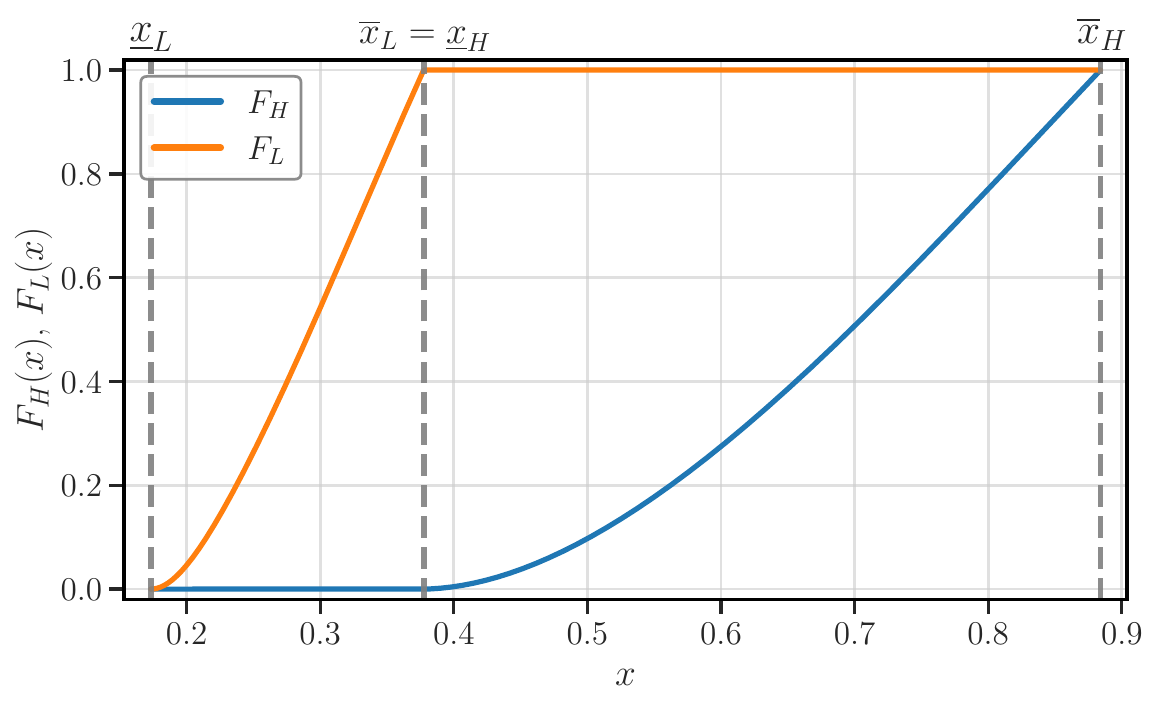}
        \caption{Equilibrium CDFs with $\gamma_L = 2$ in Example~\ref{ex:asym-separated-equilibrium}}
        \label{fig:separated-example-2}
    \end{subfigure}
    \caption{Equilibrium CDFs in Example~\ref{ex:asym-separated-equilibrium}}
    \label{fig:separated-example}
\end{figure}

Next, we turn to the explicit characterization of hybrid equilibria. Unlike separated equilibria, hybrid equilibria exhibit greater complexity, since the structure of the active interval for type-$L$ creators within $[\underline{x}_H, \overline{x}_H]$ can be highly intricate (for instance, it may contain multiple gaps). Therefore, we introduce the following two technical assumptions to make the analysis more tractable. The rationale for these two assumptions will be discussed in detail in Appendix~\ref{subsubsec:assump-1-rationale} and Appendix~\ref{subsubsec:assump-2-rationale}, respectively.

To state the first assumption, let $u_L^u(x)$ denote the utility of a type-$L$ creator at effort level $x \in \left[\underline{x}_H^u, \overline{x}_H^u\right]$ when type-$H$ creators adopt the pseudo strategy $F_H^u$. Based on this, we introduce the first assumption:
\begin{assumption} \label{assump:asym-hybrid-1}
    In game $\mathcal{G}^{(2)}$, $u_L^{u_H}(x)$ is quasiconvex on $\left[\underline{x}_H^{u_H}, \overline{x}_H^{u_H}\right]$, where $u_H$ is the equilibrium utility of type-$H$ creators determined in Proposition~\ref{prop:asym-unique-uH}.
\end{assumption}

This assumption implies that the maximum of $u_L^{u_H}(x)$ over $\left[\underline{x}_H^{u_H}, \overline{x}_H^{u_H}\right]$ must be attained at one of the endpoints of the interval. Given this assumption, we claim that the supremum of type-$L$ creators' equilibrium strategy is either $\underline{x}_H$ or $\overline{x}_H$. Indeed, according to Proposition~\ref{prop:asym-separated-equilibrium}, if $u_L^{u_H}(x)$ attains its maximum at $\underline{x}_H^{u_H}$, then the equilibrium is separated, and thus the supremum of the equilibrium strategy of type-$L$ creators is $\underline{x}_H$; if $u_L^{u_H}(x)$ reaches its maximum at $\overline{x}_H^{u_H}$, suppose by contradiction that the supremum of the equilibrium strategy of type-$L$ creators is some interior point $x'$. In this case, the indifference condition of equilibrium implies that the equilibrium strategy of type-$H$ creators over $[x', \overline{x}_H]$ coincides with the pseudo strategy. Since $u_L^{u_H}(x)$ attains its maximum at $\overline{x}_H^{u_H}$, the utility of type-$L$ creators at $\overline{x}_H$ would be strictly greater than that at $x'$, leading to a contradiction.

Under the preceding assumption, the supremum of the support of the equilibrium strategy of type-$L$ creators is restricted to either $\underline{x}_H$ or $\overline{x}_H$. However, even if the supremum equals $\overline{x}_H$, the structure of the active interval of type-$L$ creators over $[\underline{x}_H, \overline{x}_H]$ may still be quite complicated. To further simplify the analysis, we introduce the second assumption:
\begin{assumption} \label{assump:asym-hybrid-2}
    In game $\mathcal{G}^{(2)}$, if the supremum of the equilibrium strategy of type-$L$ creators is $\overline{x}_H$, then type-$L$ creators are active on the entire interval $[x^*, \overline{x}_H]$ for some $x^* \in [\underline{x}_H, \overline{x}_H)$.
\end{assumption}

The assumption implies that the active interval of type-$L$ creators on $[\underline{x}_H, \overline{x}_H]$ contains no gaps. This is a natural simplification, as the no-gap property already holds in symmetric settings (Corollary~\ref{cor:symmetric-no-gap}) and significantly simplifies the subsequent analysis.  

\begin{proposition} \label{prop:asym-hybrid-equilibrium}
    In game $\mathcal{G}^{(2)}$, under Assumption~\ref{assump:asym-hybrid-1} and \ref{assump:asym-hybrid-2}, if there exists a $u$ such that $u_L^u = u_L$ and $u_L^u$ is attained at $\overline{x}_H^u$, then the equilibrium is hybrid with equilibrium utility of type-$H$ creators $u_H = u$. The support of a hybrid equilibrium belongs to one of the two configurations displayed in Figure~\ref{fig:2-type-hybrid-support}; these two configurations are mutually exclusive for a given parameter tuple $(n_H, n_L, \gamma_H, \gamma_L, \bm{p}, g(x))$.
    
    Moreover, if the sufficient condition given in Proposition~\ref{prop:assumption-2-existence-uniqueness} holds, then the hybrid equilibrium with the support structure illustrated in Figure~\ref{fig:2-type-hybrid-support} is unique.
\end{proposition}

\begin{figure}[t]
    \centering
    \begin{subfigure}[b]{0.36\linewidth}
        \centering
        \includegraphics[width=\linewidth]{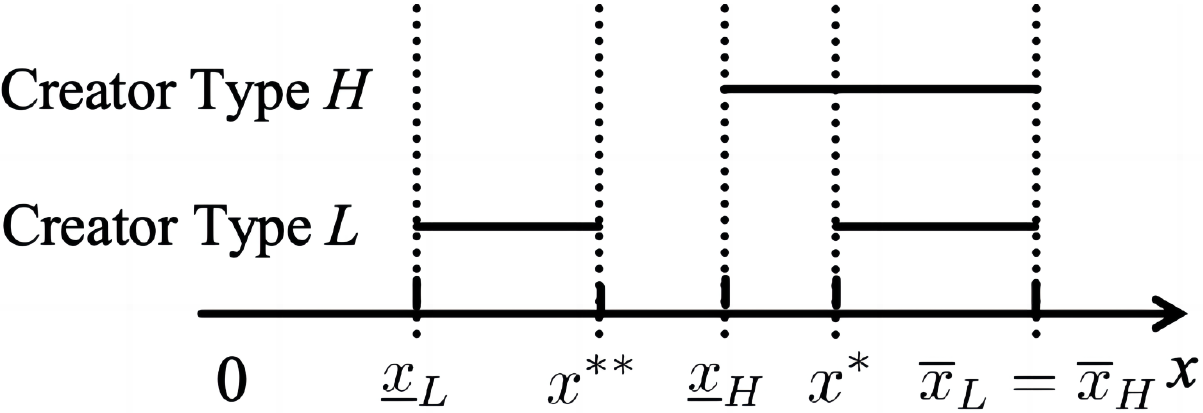}
        \caption{Case~1}
        \label{fig:2-type-hybrid-support-1}
    \end{subfigure}
    \begin{subfigure}[b]{0.36\linewidth}
        \centering
        \includegraphics[width=\linewidth]{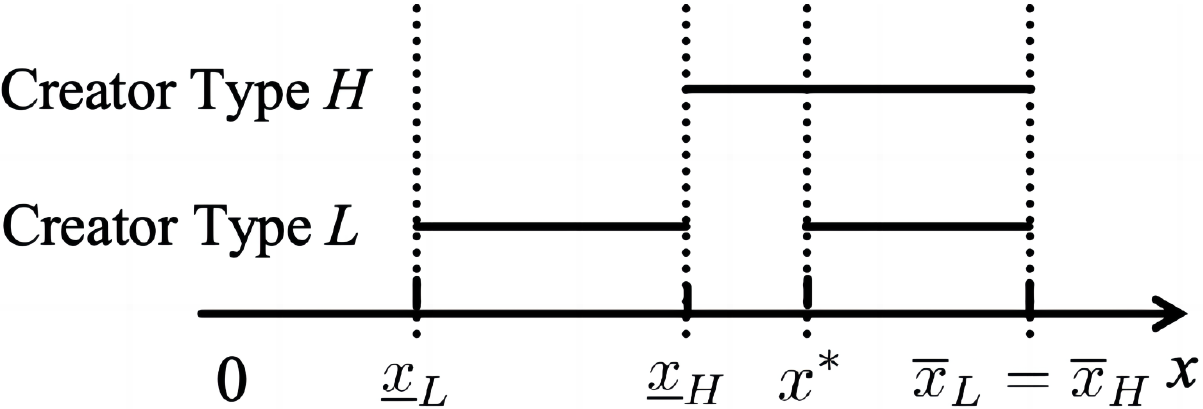}
        \caption{Case~2}
        \label{fig:2-type-hybrid-support-2}
    \end{subfigure}
    \caption{Support structure of hybrid equilibrium}
    \label{fig:2-type-hybrid-support}
\end{figure}

As can be observed, the two possible support structures for a hybrid equilibrium shown in Figure~\ref{fig:2-type-hybrid-support} differ in whether type-$L$ creators are active on the entire interval  $[\underline{x}_L, \underline{x}_H]$. We illustrate the hybrid equilibrium with a numerical example in Example~\ref{ex:asym-hybrid-equilibrium}.

\begin{proof}
    We first show that the the support structure of a hybrid equilibrium is one of the two configurations shown in Figure~\ref{fig:2-type-hybrid-support}. According to Assumption~\ref{assump:asym-hybrid-2}, type-$L$ creators are active on the entire interval $[x^*, \overline{x}_H]$ for some $x^* \in [\underline{x}_H, \overline{x}_H)$. Thus, the support structure on $[\underline{x}_H, \overline{x}_H]$ must align with the configurations in Figure~\ref{fig:2-type-hybrid-support}. Note that if $n_H = 1$, then analogous to Step~1 in the proof of Proposition~\ref{prop:2-player-equ}, type-$L$ creators must be active on the entire interval $(\underline{x}_H, \overline{x}_H]$. Moreover, if the support structure of Figure~\ref{fig:2-type-hybrid-support-1} occurs, then the equilibrium strategy of type-$H$ creators has an atom at $\underline{x}_H$, otherwise the utility of type-$L$ creators at $\underline{x}_H$ would be strictly lower than that at $x^{**}$, contradicting the indifference condition.
    
    Now consider the support structure on $[\underline{x}_L, \underline{x}_H]$. If $n_L = 1$, then Step~1 in the proof of Proposition~\ref{prop:2-player-equ} implies that the strategy of type-$L$ creators on $[\underline{x}_H, \overline{x}_H]$ reduces to a point mass, thus aligning with the configuration in Figure~\ref{fig:2-type-hybrid-support-1} with $x_L = x^{**}$. If $n_L \geqslant 2$, then Theorem~\ref{thm:sym-equ-property} implies that type-$L$ creators must be active on the entire interval $[\underline{x}_L, x^{**}]$ for some $x^{**} \in [\underline{x}_L, \underline{x}_H]$. Therefore, if $x^{**} < \underline{x}_H$, then the support structure on $[\underline{x}_L, \underline{x}_H]$ aligns with that in Figure~\ref{fig:2-type-hybrid-support-1}; otherwise, if $x^{**} = \underline{x}_H$, it aligns with that in Figure~\ref{fig:2-type-hybrid-support-2} with $x_L = \underline{x}_H$.

    Then we prove that the two configurations are mutually exclusive for a given parameter tuple $(n_H, n_L, \gamma_H, \gamma_L, \bm{p}, g(x))$. Suppose, for contradiction, that both configurations constitute equilibria. To distinguish the two configurations, we add superscripts $1$ and $2$ to the notations in Figure~\ref{fig:2-type-hybrid-support-1} and Figure~\ref{fig:2-type-hybrid-support-2}, respectively. Consider the configuration in Figure~\ref{fig:2-type-hybrid-support-1}, then Step~2 in the proof of Lemma~\ref{lem:symmetric-equilibrium} implies that
    \begin{equation} \label{eq:asym-hybrid-1}
        \underline{x}_H^1 = \arg\max_x x W(\bm{F}_{-h}^1(x^{**}), \bm{p}) - \gamma_H g(x), \quad \forall h \in \{1, \dots, n_H\}.
    \end{equation}
    
    Let $q_1$ denote the mass of type-$L$ creators on $[\underline{x}_L^1, x^{**}]$ in Figure~\ref{fig:2-type-hybrid-support-1}, $q_2$ the mass of type-$L$ creators on $[\underline{x}_L^2, \underline{x}_H^2]$ in Figure~\ref{fig:2-type-hybrid-support-2}. Recall that the equilibrium utility of type-$H$ creators $u_H$ is unique, the equation~\eqref{eq:asym-hybrid-1} yields $q_1 \leqslant q_2$. Then the uniqueness of symmetric equilibrium implies that $x^{**} \leqslant \underline{x}_H^2$. If $q_1 = q_2$, then $x^{**} = \underline{x}_H^2$, which aligns with the configuration in Figure~\ref{fig:2-type-hybrid-support-2} rather than that in Figure~\ref{fig:2-type-hybrid-support-1}. Therefore, we assume $q_1 < q_2$ (hence $x^{**} < \underline{x}_H^2$). We thus have two possible cases:
    \begin{enumerate}
        \item $\underline{x}_H^2 \leqslant \underline{x}_H^1$: Since $\underline{x}_H^2 > x^{**}$, we have $x' := \arg\max_x x W(\bm{F}_{-h}^2(\underline{x}_H^2), \bm{p}) - \gamma_H g(x) > \underline{x}_H^1 \geqslant \underline{x}_H^2$ for all $h \in \{1, \dots, n_H\}$. Thus, if any type-$H$ creator follows the strategy in Figure~\ref{fig:2-type-hybrid-support-2}, we have
        \begin{align*}
            x' W(\bm{F}_{-h}^2(x'), \bm{p}) - \gamma_H g(x') &> x' W(\bm{F}_{-h}^1(\underline{x}_H^2), \bm{p}) - \gamma_H g(x') \\
            &> \underline{x}_H^2 W(\bm{F}_{-h}^1(\underline{x}_H^2), \bm{p}) - \gamma_H g(\underline{x}_H^2),
        \end{align*}
        which contradicts the indifference condition of equilibrium.
        \item $\underline{x}_H^2 > \underline{x}_H^1$: In this case, we have $W(\bm{F}_{-h}^2(\underline{x}_H^1), \bm{p}) > W(\bm{F}_{-h}^1(\underline{x}_H^1), \bm{p})$ for all $h \in \{1, \dots, n_H\}$. Thus, if any type-$H$ creator follows the strategy in Figure~\ref{fig:2-type-hybrid-support-2}, a deviation to the effort level $\underline{x}_H^1$ would yield a strictly higher utility than $u_H$, which contradicts the indifference condition of equilibrium.
    \end{enumerate}
    
    Finally, we establish the uniqueness of the hybrid equilibrium when the sufficient condition stated in Proposition~\ref{prop:assumption-2-existence-uniqueness} holds. If the support structure follows Figure~\ref{fig:2-type-hybrid-support-1}, Theorem~\ref{thm:sym-equ-property} guarantees that the equilibrium strategy of type-$L$ creators is uniquely determined on $[\underline{x}_L^1, x^{**}]$. From equation~\eqref{eq:asym-hybrid-1}, we have $u_H = \max_x x W(\bm{F}_{-h}^1(x^{**}), \bm{p}) - \gamma_H g(x)$. Since $W(\bm{F}_{-h}^1(x^{**}), \bm{p})$ is strictly increasing in $x^{**}$, the value of $x^{**}$ (and thus $\underline{x}_H^1$) is uniquely determined. On $[\underline{x}_H^1, x^*]$, the equilibrium strategy of type-$H$ creators is also uniquely determined by Theorem~\ref{thm:sym-equ-property}. Due to the quasiconvexity required by Assumption~\ref{assump:asym-hybrid-1}, there exists a unique $x^*$ satisfying $u_L = x^* W(\bm{F}_{-l}^1(x^*), \bm{p}) - \gamma_L g(x^*)$. Then if the sufficient condition given in Proposition~\ref{prop:assumption-2-existence-uniqueness} holds, the equilibrium strategy of type-$L$ and type-$H$ creators are uniquely determined on $[x^*, \overline{x}_H^1]$. Hence, the equilibrium is unique in this configuration.

    A parallel argument applies when the support structure corresponds to Figure~\ref{fig:2-type-hybrid-support-2}, except the $\underline{x}_H^2$ is determined by $u_H = \underline{x}_H^2 W(\bm{F}_{-h}^2(\underline{x}_H^2), \bm{p}) - \gamma_H g(\underline{x}_H^2)$.
\end{proof}

\begin{remark}
    As shown in the proof, when $n_L = 1$, the strategy of type-$L$ creators on $[\underline{x}_H, \overline{x}_H]$ reduces to a point mass, so only the case depicted in Figure~\ref{fig:2-type-hybrid-support-1} can occur. When $n_H = 1$, type-$L$ creators must be active on the entire interval $(\underline{x}_H, \overline{x}_H]$; moreover, if the support structure of Figure~\ref{fig:2-type-hybrid-support-1} occurs, then the equilibrium strategy of type-$H$ creators has an atom at $\underline{x}_H$. Finally, when $n_H = n_L = 1$, we provide a complete characterization of the equilibrium configuration in Section~\ref{subsec:ap-asym-n2}, where we explicitly describe the equilibrium strategies and prove the uniqueness of the equilibrium without relying on Assumption~\ref{assump:asym-hybrid-1} and Assumption~\ref{assump:asym-hybrid-2}.
\end{remark}

\begin{example} \label{ex:asym-hybrid-equilibrium}
    Consider the following two cases:
    \begin{enumerate}
        \item Case~1: Consider a setting with $n_H = 6, n_L = 4$, $p_1 = 0.9, p_2 = 0.75, p_3 = 0.63, p_4 = 0.52, p_5 = 0.42, p_6 = 0.34, p_7 = 0.2, p_8 = 0.198, p_9 = 0.196, p_{10} = 0.195$, and $g(x) = x^3$. Let $\gamma_H = 2$ and $\gamma_L = 2.1$. The support structure of the hybrid equilibrium corresponds to Figure~\ref{fig:2-type-hybrid-support-1}, with $\underline{x}_L \approx 0.176$, $x^{**} \approx 0.205$, $\underline{x}_H \approx 0.231$, $x^* \approx 0.607$, and $\overline{x}_H = \overline{x}_L \approx 0.642$.
        \item Case~2: Consider a setting with $n_H = 5, n_L = 5$, $p_1 = 0.9, p_2 = 0.75, p_3 = 0.63, p_4 = 0.52, p_5 = 0.42, p_6 = 0.34, p_7 = 0.28, p_8 = 0.24, p_9 = 0.20, p_{10} = 0.18$, and $g(x) = x^3$. Let $\gamma_H = 1$ and $\gamma_L = 1.5$. The support structure of the hybrid equilibrium corresponds to Figure~\ref{fig:2-type-hybrid-support-2}, with $\underline{x}_L \approx 0.169$, $\underline{x}_H \approx 0.350$, $x^* \approx 0.549$, and $\overline{x}_H = \overline{x}_L \approx 0.643$.
    \end{enumerate}
\end{example}

Indeed, most parameter settings lead to the support structure depicted in Figure~\ref{fig:2-type-hybrid-support-2}, rather than the one in Figure~\ref{fig:2-type-hybrid-support-1}. This is because the condition $x^{**} < \underline{x}_H$ requires a large probability mass of $F_L$ to be concentrated on a small interval. This can only happen when the position biases $p_{n_H + 1}$ and $p_n$ are sufficiently close to each other.

\begin{figure}[t]
    \centering
    \begin{subfigure}[b]{0.48\linewidth}
        \centering
        \includegraphics[width=\linewidth]{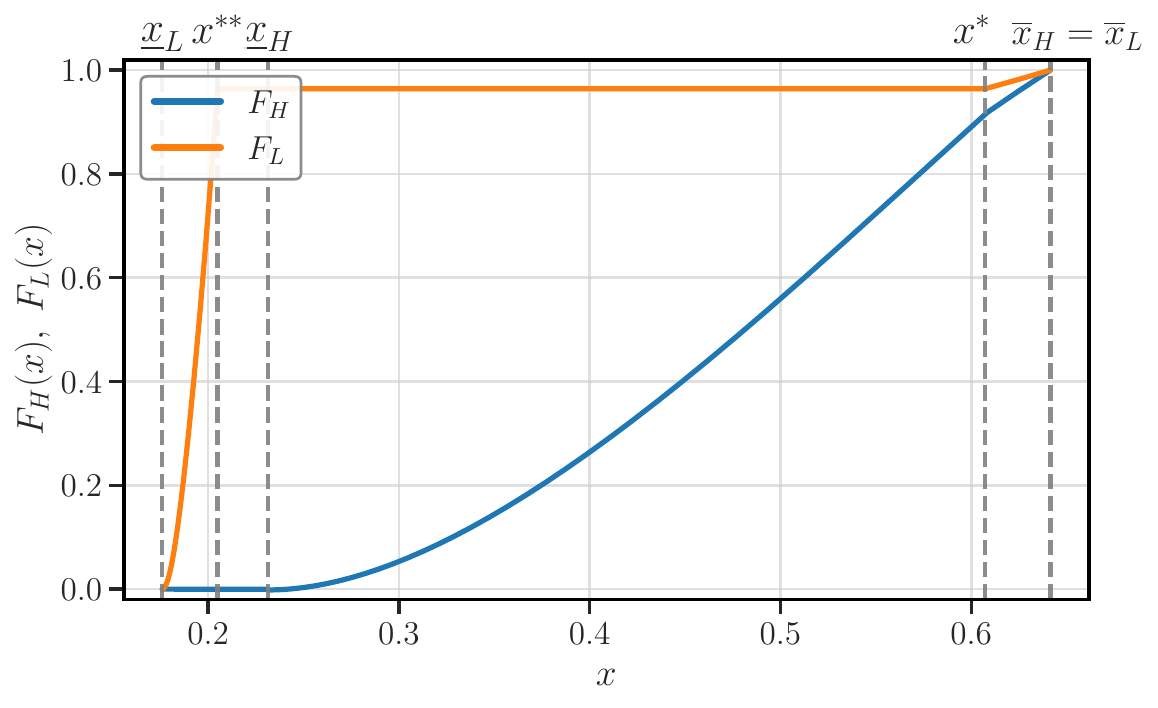}
        \caption{Equilibrium CDFs of Case~1 in Example~\ref{ex:asym-hybrid-equilibrium}}
        \label{fig:hybrid-example-1}
    \end{subfigure}
    \hfill
    \begin{subfigure}[b]{0.48\linewidth}
        \centering
        \includegraphics[width=\linewidth]{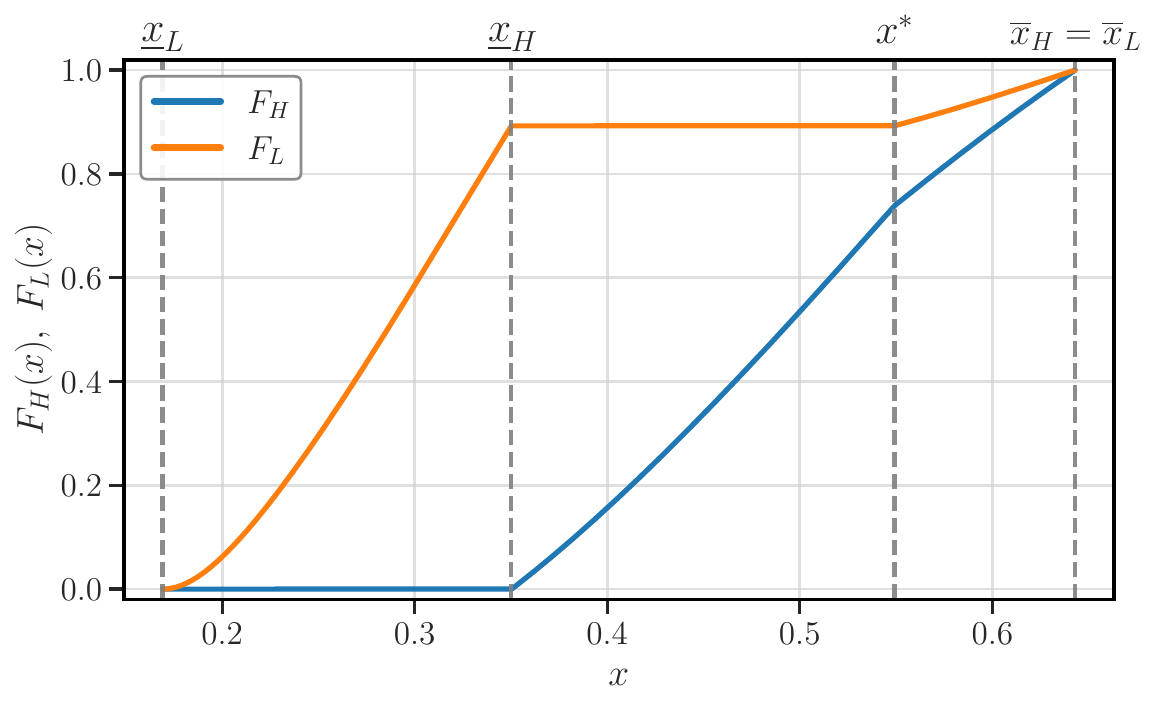}
        \caption{Equilibrium CDFs of Case~2 in Example~\ref{ex:asym-hybrid-equilibrium}}
        \label{fig:hybrid-example-2}
    \end{subfigure}
    \caption{Equilibrium CDFs in Example~\ref{ex:asym-hybrid-equilibrium}}
    \label{fig:hybrid-example}
\end{figure}

Building on the preceding discussion, we present Algorithm~\ref{alg:binary-type-equilibrium-computation} in Appendix~\ref{subsubsec:binary-type-algorithm} for computing the mixed Nash equilibrium for the binary-type setting under plausible parameter configurations. The algorithm incorporates the method for computing $u_H$ given in Proposition~\ref{prop:asym-unique-uH} and integrates the equilibrium structure characterized in Proposition~\ref{prop:asym-separated-equilibrium} and Proposition~\ref{prop:asym-hybrid-equilibrium}.

\subsubsection{Rationale for Assumption~\ref{assump:asym-hybrid-1}} \label{subsubsec:assump-1-rationale}

In this subsection, we first prove that the assumption holds in the special case where positional biases form an arithmetic sequence. Subsequently, we present a counterexample to show that this assumption may not hold universally and analyze the underlying reasons for its failure. Finally, we provide numerical evidence demonstrating that the assumption is robust across a broad range of parameter configurations commonly encountered in practice.

\begin{proposition} \label{prop:assumption-1-arithmetic}
    Suppose that the positional biases form an arithmetic sequence, i.e., there exists a constant $d > 0$ such that $p_i - p_{i + 1} = d$ for all $i = 1, 2, \dots, n - 1$, then Assumption~\ref{assump:asym-hybrid-1} holds.
\end{proposition}

\begin{proof}
    First, for any given $u$, the pseudo-strategy $F_H^u$ of type-$H$ creators satisfies
    \[ x \left(\sum_{i = 1}^{n_H} \binom{n_H - 1}{i - 1} [F_H^u(x)]^{n_H - i} [1 - F_H^u(x)]^{i - 1} p_i\right) - \gamma_H g(x) = u. \]
    Using properties of binomial summation, the above equation can be simplified to
    \[ x (d(n_H - 1)F_H^u(x) + p_{n_H}) - \gamma_H g(x) = u. \]
    Thus, we have
    \begin{align*}
        u_L^{u_H}(x) &= x \left(\sum_{i = 1}^{n_H + 1} \binom{n_H}{i - 1} F_H^u(x)^{n_H - i +1} (1 - F_H^u(x))^{i - 1} p_i\right) - \gamma_L g(x) \\
        &= x (d n_H F_H^u(x) + p_{n_H + 1}) - \gamma_L g(x) \\
        &= \left(\frac{n_H}{n_H - 1}\gamma_H - \gamma_L\right) g(x) - \left(\frac{n_H}{n_H - 1} p_{n_H} - p_{n_H + 1}\right) x + \frac{n_H}{n_H - 1} u.
    \end{align*}
    If $\frac{n_H}{n_H - 1}\gamma_H - \gamma_L \geqslant 0$, then $u_L^{u_H}(x)$ is convex in $x$, hence Assumption~\ref{assump:asym-hybrid-1} holds. If $\frac{n_H}{n_H - 1}\gamma_H - \gamma_L < 0$, then $u_L^{u_H}(x)$ is decreasing in $x$, which also guarantees that Assumption~\ref{assump:asym-hybrid-1} is satisfied.
\end{proof}

The result for arithmetic progressions, however, cannot be extended to arbitrary position biases. We now present a counterexample to show that this assumption may not hold universally.

\begin{example} \label{ex:not-quasiconvex}
    Consider a setting with $n_H = n_L = 2$, $p_1 = 0.8, p_2 = 0.75, p_3 = 0.1, p_4 = 0.05$, and $g(x) = x^8$. Let $\gamma_H = 1$ and $\gamma_L = 3$. Figure~\ref{fig:not-quasiconvex} illustrates the function $u_L^{u_H}(x)$ where $u_H$ is the utility of type-$H$ creators in the mixed Nash equilibrium. It is evident that $u_L^{u_H}(x)$ is not quasi-convex, thus violating Assumption~\ref{assump:asym-hybrid-1}.
\end{example}

\begin{figure}[t]
    \centering
    \includegraphics[width=0.4\textwidth]{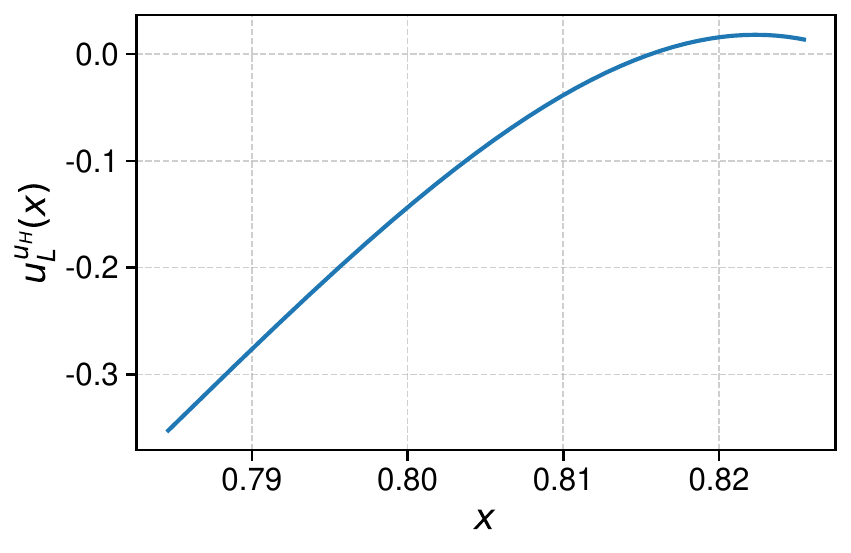} 
    \caption{A counterexample where Assumption~\ref{assump:asym-hybrid-1} fails to hold}
    \label{fig:not-quasiconvex}
\end{figure}

Two key factors underlie the non-quasiconvex pattern observed in Example~\ref{ex:not-quasiconvex}. First, following a derivation analogous to equation~\eqref{eq:sym-equ-cdf-n2}, we obtain
\[ F_H^{u_H}(x) = \frac{u_H + \gamma_H g(x)}{x (p_1 - p_2)} - \frac{p_2}{p_1 - p_2}. \]
Thus, $F_H^u$ is an increasing convex function. This property, coupled with the substantial gap between $p_2$ and $p_3$ in the position bias vector, naturally accounts for the initial monotonic increasing of $u_L^{u_H}(x)$. Second, the strongly convex cost function $g(x) = x^8$ and the large magnitudes of $p_1$ and $p_2$ together induce a sharp increase in cost near the supremum effort level,  causing the utility of type-$L$ creators to eventually decline after an initial ascent. The combination of these two factors explains the observed non-quasiconvex behavior of $u_L^{u_H}(x)$ in Example~\ref{ex:not-quasiconvex}.

We note, however, that the parameters adopted in Example~\ref{ex:not-quasiconvex} represent an extreme scenario. We then conduct numerical experiments to demonstrate that the assumption holds across a broad range of realistic settings. We take line~2 of Table~\ref{tab:position-bias-no-ai-overview} as the base position bias vector, and vary the cost parameter $\gamma_L$ together with the cost function $g(x)$. Specifically, we test four cost functions: $g(x) = x^2, x^4, x^6, x^8$. For each cost function, we vary $\gamma_L$ from $1.05$ to $1.2$ in steps of $0.05$, while fixing $\gamma_H = 1$ and $n_H = n_L = 5$. The results are presented in Figure~\ref{fig:quasiconvex-experiments-1}. It is clear that in all these scenarios, $u_L^{u_H}(x)$ exhibits quasiconvexity, thereby satisfying Assumption~\ref{assump:asym-hybrid-1}.

\begin{figure}[t]
    \centering
    \captionsetup[subfigure]{justification=centering}
    \begin{subfigure}[b]{0.24\linewidth}
        \centering
        \includegraphics[width=\linewidth]{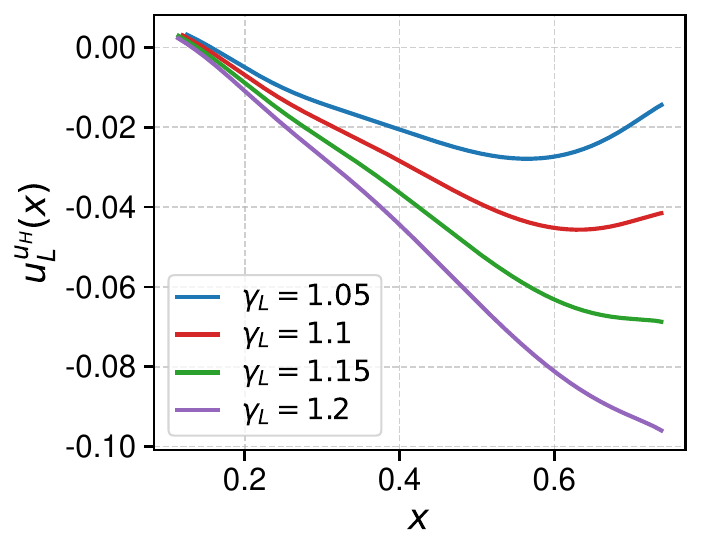}
        \caption{The effect of $\gamma_L$ for $g(x) = x^2$}
        \label{fig:quasiconvex-g2}
    \end{subfigure}
    \begin{subfigure}[b]{0.24\linewidth}
        \centering
        \includegraphics[width=\linewidth]{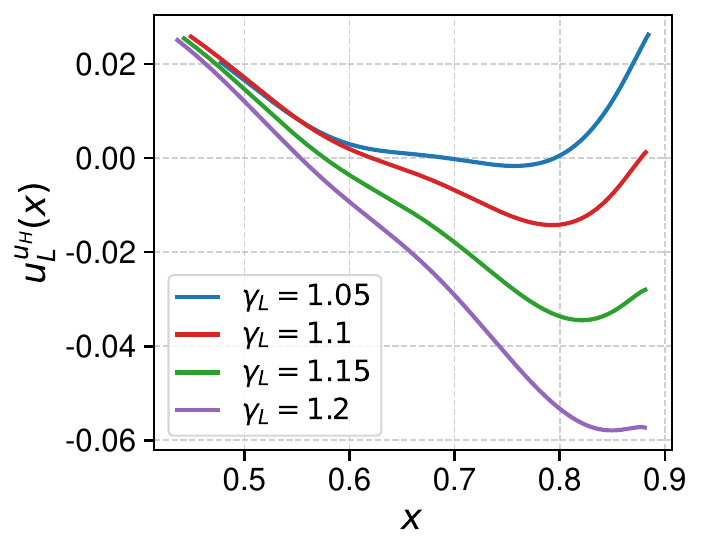}
        \caption{The effect of $\gamma_L$ for $g(x) = x^4$}
        \label{fig:quasiconvex-g4}
    \end{subfigure}
    \begin{subfigure}[b]{0.24\linewidth}
        \centering
        \includegraphics[width=\linewidth]{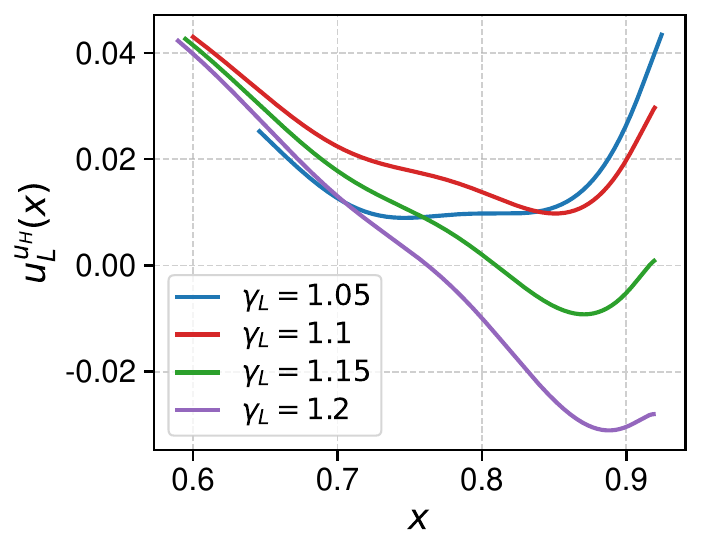}
        \caption{The effect of $\gamma_L$ for $g(x) = x^6$}
        \label{fig:quasiconvex-g6}
    \end{subfigure}
    \begin{subfigure}[b]{0.24\linewidth}
        \centering
        \includegraphics[width=\linewidth]{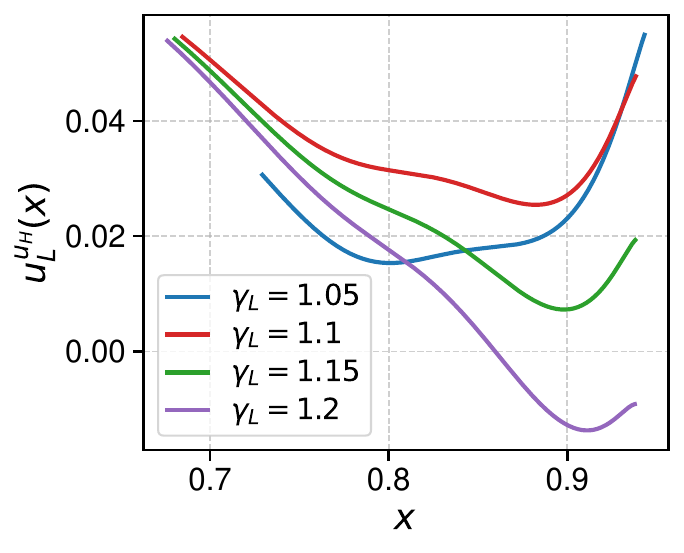}
        \caption{The effect of $\gamma_L$ for $g(x) = x^8$}
        \label{fig:quasiconvex-g8}
    \end{subfigure}
    \caption{Plots of $u_L^{u_H}(x)$ under different cost functions and $\gamma_L$ values}
    \label{fig:quasiconvex-experiments-1}
\end{figure}

Then we fix $\gamma_H = 1$ and $\gamma_L = 1.1$, and consider two different values of $n_H$: $n_H = 3$ and $n_H = 7$ (the case $n_H = 5$ has already been tested in previous experiments). For each $n_H$ value, we test the same four cost functions as the ones used above. The results are presented in Figure~\ref{fig:quasiconvex-experiments-2}. Again, in all these scenarios, $u_L^{u_H}(x)$ exhibits quasiconvexity, thereby satisfying Assumption~\ref{assump:asym-hybrid-1}.

\begin{figure}[t]
    \centering
    \captionsetup[subfigure]{justification=centering}
    \begin{subfigure}[b]{0.32\linewidth}
        \centering
        \includegraphics[width=\linewidth]{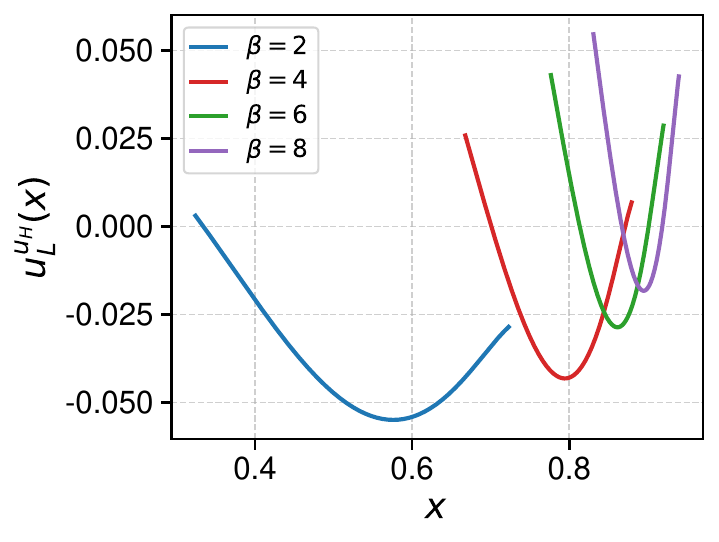}
        \caption{The effect of $g$ for $n_H = 3$}
        \label{fig:quasiconvex-m1}
    \end{subfigure}
    \begin{subfigure}[b]{0.32\linewidth}
        \centering
        \includegraphics[width=\linewidth]{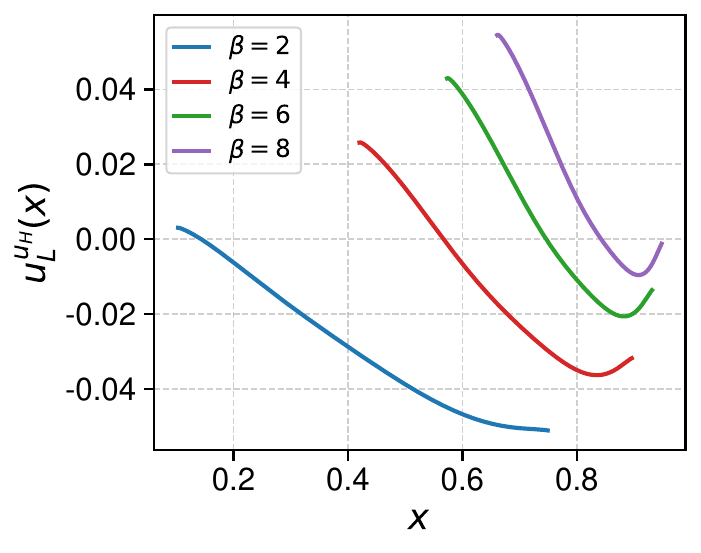}
        \caption{The effect of $g$ for $n_H = 7$}
        \label{fig:quasiconvex-m3}
    \end{subfigure}
    \caption{Plots of $u_L^{u_H}(x)$ under different cost functions and $n_H$ values}
    \label{fig:quasiconvex-experiments-2}
\end{figure}

\subsubsection{Rationale for Assumption~\ref{assump:asym-hybrid-2}} \label{subsubsec:assump-2-rationale}

In this subsection, we first provide sufficient conditions in Proposition~\ref{prop:assumption-2-existence-uniqueness} that guarantee the existence and uniqueness of a no-gap equilibrium. We then conduct numerical experiments to verify that the sufficient condition in Proposition~\ref{prop:assumption-2-existence-uniqueness} holds across a wide range of plausible parameter settings.

To state the following proposition, we first define the Jacobian matrix $J$ as follows:
\begin{equation} \label{eq:asym-hybrid-Jacobian}
    J_{ij} = \frac{\partial W(\bm{F}_{-i}(x), \bm{p})}{\partial F_j}, \quad \forall i, j = 1, 2, \dots, n.
\end{equation}

\begin{proposition} \label{prop:assumption-2-existence-uniqueness}
    Given $x^* \in [\underline{x}_H, \overline{x}_H)$ and the initial condition $\bm{F}(x^*)$, if the Jacobian matrix $J$ defined in equation~\eqref{eq:asym-hybrid-Jacobian} is invertible for any $x \in [x^*, \overline{x}_H]$, then equilibrium of the form described in Assumption~\ref{assump:asym-hybrid-2} exists and is unique.

    Specifically, the invertibility of $J$ is guaranteed under either of the following conditions:
    \begin{itemize}
        \item $n_H = 1$ or $n_L = 1$.
        \item The position biases form a quadratic sequence, i.e., there exists a constant $d$ such that $(p_{i + 1} - p_i) - (p_i - p_{i - 1}) = d$ for all $i = 2, 3, \dots, n - 1$.
    \end{itemize}
\end{proposition}

\begin{proof}
    Consider any $x \in [x^*, \overline{x}_H]$. According to the indifference condition of any type-$H$ creator $h \in \{1, \ldots, n_H\}$ and any type-$L$ creator $l \in \{n_H + 1, \ldots, n\}$, we have
    \begin{align*}
        x W(\bm{F}_{-h}(x), \bm{p}) - \gamma_H g(x) &= u_H, \quad \forall h = 1, \ldots, n_H, \\
        x W(\bm{F}_{-l}(x), \bm{p}) - \gamma_L g(x) &= u_L, \quad \forall l = n_H + 1, \ldots, n,
    \end{align*}
    where $u_H$ and $u_L$ are the equilibrium utilities of type-$H$ and type-$L$ creators, respectively. Differentiating both sides of each equation with respect to $x$ yields
    \begin{align*}
        W(\bm{F}_{-h}(x), \bm{p}) + x \sum_{i = 1}^n \frac{\partial W(\bm{F}_{-h}(x), \bm{p})}{\partial F_i} \frac{\textup{d}F_i}{\textup{d}x} - \gamma_H g'(x) &= 0, \quad \forall h = 1, \ldots, n_H, \\
        W(\bm{F}_{-l}(x), \bm{p}) + x \sum_{i = 1}^n \frac{\partial W(\bm{F}_{-l}(x), \bm{p})}{\partial F_i} \frac{\textup{d}F_i}{\textup{d}x} - \gamma_L g'(x) &= 0, \quad \forall l = n_H + 1, \ldots, n.
    \end{align*}

    Now define $\bm{F} = (F_1, \ldots, F_n)^{\mathrm{T}}$, $\bm{W} = (W(\bm{F}_{-1}(x), \bm{p}), \ldots, W(\bm{F}_{-n}(x), \bm{p}))^{\mathrm{T}}$, 
    $J = (J_{ij})_{n \times n}$ with $J_{ij}$ defined in equation~\eqref{eq:asym-hybrid-Jacobian}, and
    \[ \bm{g} = \bigl(\underbrace{\gamma_H g'(x), \ldots, \gamma_H g'(x)}_{n_H},
           \underbrace{\gamma_L g'(x), \ldots, \gamma_L g'(x)}_{n_L}\bigr)^{\mathrm{T}}. \]
    
    Then the system of equations for all creators can be written in compact matrix form
    \begin{equation} \label{eq:asym-hybrid-indifference-matrix}
        J \cdot \frac{\textup{d}\bm{F}}{\textup{d}x} = \frac{1}{x} \left(\bm{g} - \bm{W}\right).
    \end{equation}
    
    If $J$ is invertible for any $x \in [x^*, \overline{x}_H]$, then we have $\frac{\textup{d}\bm{F}}{\textup{d}x} = \frac{1}{x} J^{-1} \left(\bm{g} - \bm{W}\right)$. According to the Picard-Lindelöf theorem \cite{arnold1992ordinary}, the above ordinary differential equation admits a unique solution given the initial condition $\bm{F}(x^*)$. Consequently, Assumption~\ref{assump:asym-hybrid-2} holds if $J$ is invertible for any $x \in [x^*, \overline{x}_H]$.
    
    Then we prove the sufficient conditions for $J$ to be invertible. By Lemma~\ref{lem:W-monotone}, we know that
    \[ J_{ij} = \begin{cases}
        0, & \text{if } i = j, \\
        W(\bm{F}_{-\{i, j\}}(x), \overline{\bm{p}} - \underline{\bm{p}}) > 0, & \text{if } i \neq j,
    \end{cases} \]
    where $\overline{\bm{p}} = (p_1, \ldots, p_{n - 1})$ and $\underline{\bm{p}} = (p_2, \ldots, p_n)$. Let
    \begin{align*}
        a &= W(\bm{F}_{-\{i, j\}}(x), \overline{\bm{p}} - \underline{\bm{p}}), \quad \forall i, j \in \{1, \ldots, n_H\}, i \neq j, \\
        b &= W(\bm{F}_{-\{i, j\}}(x), \overline{\bm{p}} - \underline{\bm{p}}), \quad \forall i \in \{1, \ldots, n_H\}, j \in \{n_H + 1, \ldots, n\}, \\
        c &= W(\bm{F}_{-\{i, j\}}(x), \overline{\bm{p}} - \underline{\bm{p}}), \quad \forall i, j \in \{n_H + 1, \ldots, n\}, i \neq j.
    \end{align*}
    Then $J$ can be expressed as
    \[ J = \begin{pmatrix}
        A & B \\
        B^{\mathrm{T}} & C
    \end{pmatrix}, \]
    where $A$ is an $n_H \times n_H$ matrix with diagonal entries being $0$ and off-diagonal entries being $a$, $C$ is an $n_L \times n_L$ matrix with diagonal entries being $0$ and off-diagonal entries being $c$, and $B$ is an $n_H \times n_L$ matrix with all entries being $b$. Then the determinant of $J$ can be calculated as
    \[ \det(J) = a^{n_H} b^{n_L} \left(\frac{b}{a}\right)^{n_L} \det(J'), \]
    where
    \[ J' = \begin{pmatrix}
        A' & B' \\
        B'^{\mathrm{T}} & C'
    \end{pmatrix}, \]
    with $A'$ being an $n_H \times n_H$ matrix with diagonal entries being $0$ and off-diagonal entries being $1$, $C'$ being an $n_L \times n_L$ matrix with diagonal entries being $0$ and off-diagonal entries being $ac / b^2$, and $B'$ being an $n_H \times n_L$ matrix with all entries being $1$. Since $a, b, c > 0$, it suffices to determine whether $\det(J') \neq 0$. Let $k = b^2 / ac$, $t = n_L k / (n_L - 1)$,  by the property of block matrices, we have
    \begin{align*}
        \det(J') &= |C'| \cdot |A' - B' C'^{-1} B'^{\mathrm{T}}| \\
        &= |C'| \cdot \det\left(\begin{pmatrix}
            0 & 1 & \cdots & 1 \\
            1 & 0 & \cdots & 1 \\
            \vdots & \vdots & \ddots & \vdots \\
            1 & 1 & \cdots & 0
        \end{pmatrix} - \begin{pmatrix}
            t & t & \cdots & t \\
            t & t & \cdots & t \\
            \vdots & \vdots & \ddots & \vdots \\
            t & t & \cdots & t
        \end{pmatrix}\right) \\
        &= |C'| \cdot \det\left(\begin{pmatrix}
            1 - t & 1 - t & \cdots & 1 - t \\
            1 - t & 1 - t & \cdots & 1 - t \\
            \vdots & \vdots & \ddots & \vdots \\
            1 - t & 1 - t & \cdots & 1 - t
        \end{pmatrix} - E\right) \\ 
        &= |C'| \cdot \det\left(\begin{pmatrix}
            1 \\ 1 \\ \vdots \\ 1
        \end{pmatrix} \begin{pmatrix}
            1 - t & 1 - t & \cdots & 1 - t
        \end{pmatrix} - E\right) = (-1)^{n_H} |C'| \cdot (1 - (1 - t) n_H).
    \end{align*}
    Note that $|C'| = (-1)^{n_L + 1}(n_L - 1) (ac / b^2)^{n_L} \neq 0$, thus $\det(J') \neq 0$ if and only if $1 - (1 - t) n_H \neq 0$, i.e., 
    \begin{equation} \label{eq:asym-hybrid-assump-2-proof-1}
        n_L - 1 \neq n_H (n_L - n_L k - 1)
    \end{equation}
    Since $k > 0$, then if $n_L = 1$ and $n_H \neq 1$, or $n_H = 1$ and $n_L \neq 1$, we have $\det(J') \neq 0$. If $n_L = n_H = 1$, the discussion in Section~\ref{subsec:ap-asym-n2} shows that the equilibrium is unique and has the structure described in Assumption~\ref{assump:asym-hybrid-2}. Thus, if $n_H = 1$ or $n_L = 1$, the existence and uniqueness of equilibrium are guaranteed.
    
    Assume now that $n_H \geqslant 2$ and $n_L \geqslant 2$, then equation~\eqref{eq:asym-hybrid-assump-2-proof-1} holds if 
    \begin{equation} \label{eq:asym-hybrid-assump-2-proof-2}
        k \neq \frac{(n_H - 1)(n_L - 1)}{n_H n_L} < 1.
    \end{equation}
    Next, we prove that if the position biases form a quadratic sequence or a convex sequence, then $\det(J') \neq 0$. Let $\bm{q} = \overline{\bm{p}} - \underline{\bm{p}}$, $\bm{r} = \overline{\bm{q}} - \underline{\bm{q}}$. Let $h_1, h_2 \in \{1, \ldots, n_H\}$ and $l_1, l_2 \in \{n_H + 1, \ldots, n\}$ be any two distinct creators, $F_H$ and $F_L$ be the CDFs of the strategies of type-$H$ and type-$L$ creators, respectively. Then, using the linearity of $W(\cdot, \cdot)$ repeatedly, we have
    \begin{equation} \label{eq:asym-hybrid-assump-2-proof-3}
        \begin{aligned}
            a &= W(F_{-\{h_1, h_2, l_1\}}(x), \underline{\bm{q}}) + F_L(x) W(F_{-\{h_1, h_2, l_1, l_2\}}(x), \underline{\bm{r}}) + F_L^2(x) W(F_{-\{h_1, h_2, l_1, l_2\}}(x), \overline{\bm{r}} - \underline{\bm{r}}), \\
            b &= W(F_{-\{h_1, h_2, l_1\}}(x), \underline{\bm{q}}) + F_H(x) W(F_{-\{h_1, h_2, l_1, l_2\}}(x), \underline{\bm{r}}) + F_H(x) F_L(x) W(F_{-\{h_1, h_2, l_1, l_2\}}(x), \overline{\bm{r}} - \underline{\bm{r}}), \\
            &= W(F_{-\{h_1, l_1, l_2\}}(x), \underline{\bm{q}}) + F_L(x) W(F_{-\{h_1, h_2, l_1, l_2\}}(x), \underline{\bm{r}}) + F_H(x) F_L(x) W(F_{-\{h_1, h_2, l_1, l_2\}}(x), \overline{\bm{r}} - \underline{\bm{r}}), \\
            c &= W(F_{-\{h_1, l_1, l_2\}}(x), \underline{\bm{q}}) + F_H(x) W(F_{-\{h_1, h_2, l_1, l_2\}}(x), \underline{\bm{r}}) + F_H^2(x) W(F_{-\{h_1, h_2, l_1, l_2\}}(x), \overline{\bm{r}} - \underline{\bm{r}}).
        \end{aligned}
    \end{equation}
    Then we have
    \begin{equation} \label{eq:asym-hybrid-assump-2-proof-4}
        2b - (a + c) = -(F_H(x) - F_L(x))^2 W(F_{-\{h_1, h_2, l_1, l_2\}}(x), \overline{\bm{r}} - \underline{\bm{r}}).
    \end{equation}
    If the position biases form a quadratic sequence, then $\overline{\bm{r}} - \underline{\bm{r}}$ is a zero vector, hence $2b - (a + c) = 0$, which implies $b \geqslant \sqrt{ac}$, and thus $k = b^2 / ac \geqslant 1$. Therefore, equation~\eqref{eq:asym-hybrid-assump-2-proof-2} holds. 
\end{proof}

Indeed, the proof of the proposition shows that whenever $W(F_{-\{h_1, h_2, l_1, l_2\}}(x), \overline{\bm{r}} - \underline{\bm{r}}) \leqslant 0$, equation~\eqref{eq:asym-hybrid-assump-2-proof-4} implies $2b - (a + c) \geqslant 0$, and hence $k = b^2 / ac \geqslant 1$. Consequently, the Jacobian matrix $J$ is invertible at such a point $x$. When $W(F_{-\{h_1, h_2, l_1, l_2\}}(x), \overline{\bm{r}} - \underline{\bm{r}}) \geqslant 0$, we first simplify the notation, rewrite equation~\eqref{eq:asym-hybrid-assump-2-proof-3} as
\[ a = A + F_L^2 D, \quad b = B_1 + F_H F_L D = B_2 + F_H F_L D, \quad c = C + F_H^2 D, \]
where $D = W(F_{-\{h_1, h_2, l_1, l_2\}}(x), \overline{\bm{r}} - \underline{\bm{r}}) \geqslant 0$, $A, B_1, B_2, C$ denote the first two terms in the expressions of $a, b, c$ in equation~\eqref{eq:asym-hybrid-assump-2-proof-3}, respectively. Next, we consider the case where $A, B_1, B_2, C > 0$ (this holds, for instance, when the position biases form a convex sequence, i.e., for all $i = 2, 3, \dots, n - 1$, we have $(p_{i + 1} - p_i) - (p_i - p_{i - 1}) \geqslant 0$). Note that $A + C = B_1 + B_2$, then we have
\begin{align*}
    \frac{2b^2}{ac} = \frac{(2b)^2}{2ac} &= \frac{(B_1 + B_2 + 2F_H F_L D)^2}{2 (A + F_L^2 D)(C + F_H^2 D)} = \frac{(A + C + 2F_H F_L D)^2}{2 (A + F_L^2 D)(C + F_H^2 D)} \\
    &= \frac{(A + C)^2 + 4 F_H F_L D (A + C) + 4 F_H^2 F_L^2 D^2}{2 AC + 2 F_H^2 A D + 2 F_L^2 C D + 2 F_H^2 F_L^2 D^2}.
\end{align*}

It is straightforward to verify that if $F_L \leqslant 2 F_H$ (i.e., the difference between $F_H$ and $F_L$ is not too large), then the above value should be greater than or equal to $1$, hence $k = b^2 / ac \geqslant 1$, and consequently the Jacobian matrix $J$ is invertible at this point $x$.

In summary, the preceding discussion indicates that the conditions for $J$ to be invertible are not very restrictive. Next, we will validate the applicability of this proposition under a broader range of realistic parameter settings through numerical experiments. Similar to the experiments conducted in Section~\ref{subsubsec:assump-1-rationale}, we take line~2 of Table~\ref{tab:position-bias-no-ai-overview} as the base position bias vector, and vary the cost parameter $\gamma_L$ together with the cost function $g(x)$. Again, we test four cost functions: $g(x) = x^2, x^4, x^6, x^8$. Unlike the experiments in Section~\ref{subsubsec:assump-1-rationale}, here we choose $\gamma_H$ and $\gamma_L$ such that the hybrid equilibrium exists, which requires the gap between $\gamma_H$ and $\gamma_L$ to be small. Specifically, we fix $n_H = n_L = 5$, $\gamma_H = 2$. For $g(x) = x^2$, we vary $\gamma_L$ from $2.005$ to $2.02$ in steps of $0.005$; for $g(x) = x^4, x^6, x^8$, we vary $\gamma_L$ from $2.02$ to $2.08$ in steps of $0.02$. For each parameter setting, we plot the curve of $k = b^2 / (ac)$ over the hybrid equilibrium interval. The results are presented in Figure~\ref{fig:invertible-experiments-1}. As the figure shows, we can see that under all tested parameter settings, $k$ is always greater than $1$, thus ensuring the invertibility of the Jacobian matrix $J$.

\begin{figure}[t]
    \centering
    \captionsetup[subfigure]{justification=centering}
    \begin{subfigure}[b]{0.24\linewidth}
        \centering
        \includegraphics[width=\linewidth]{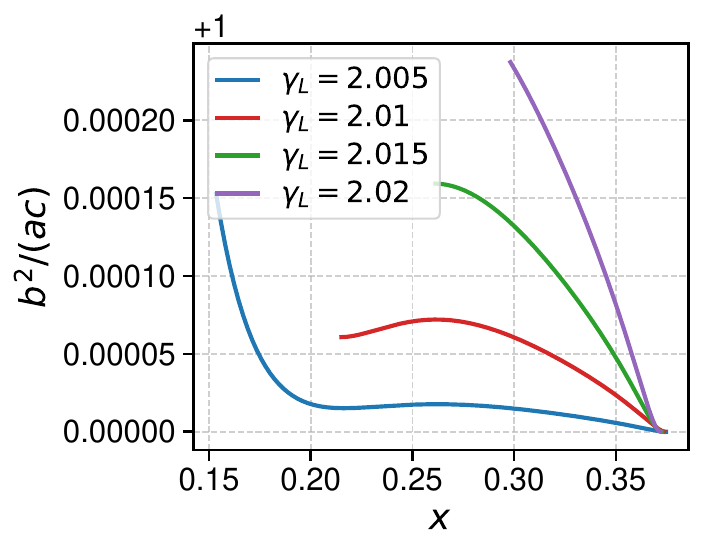}
        \caption{The effect of $\gamma_L$ for $g(x) = x^2$}
        \label{fig:invertible-g2}
    \end{subfigure}
    \begin{subfigure}[b]{0.24\linewidth}
        \centering
        \includegraphics[width=\linewidth]{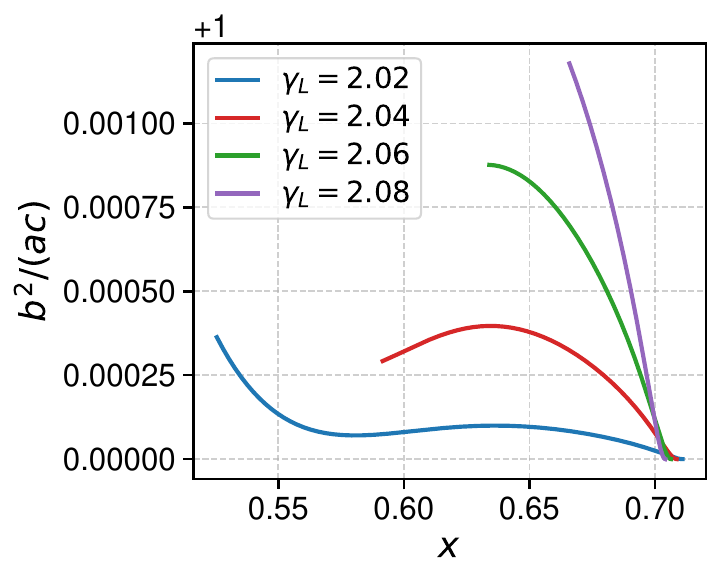}
        \caption{The effect of $\gamma_L$ for $g(x) = x^4$}
        \label{fig:invertible-g4}
    \end{subfigure}
    \begin{subfigure}[b]{0.24\linewidth}
        \centering
        \includegraphics[width=\linewidth]{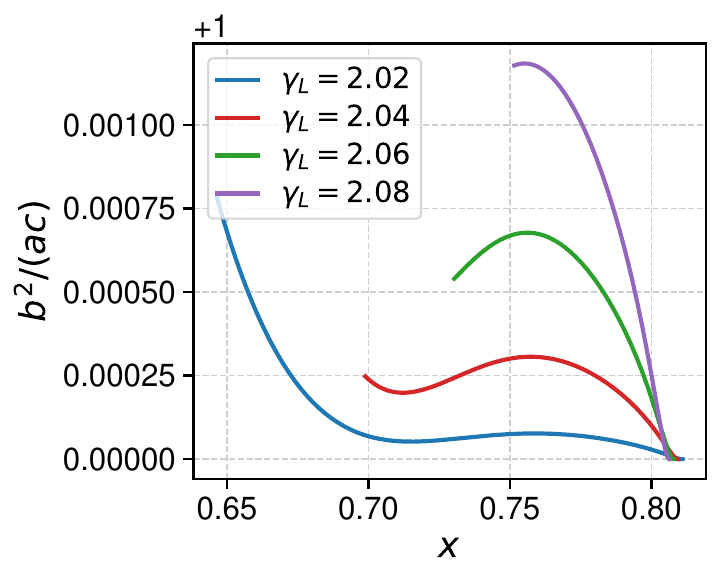}
        \caption{The effect of $\gamma_L$ for $g(x) = x^6$}
        \label{fig:invertible-g6}
    \end{subfigure}
    \begin{subfigure}[b]{0.24\linewidth}
        \centering
        \includegraphics[width=\linewidth]{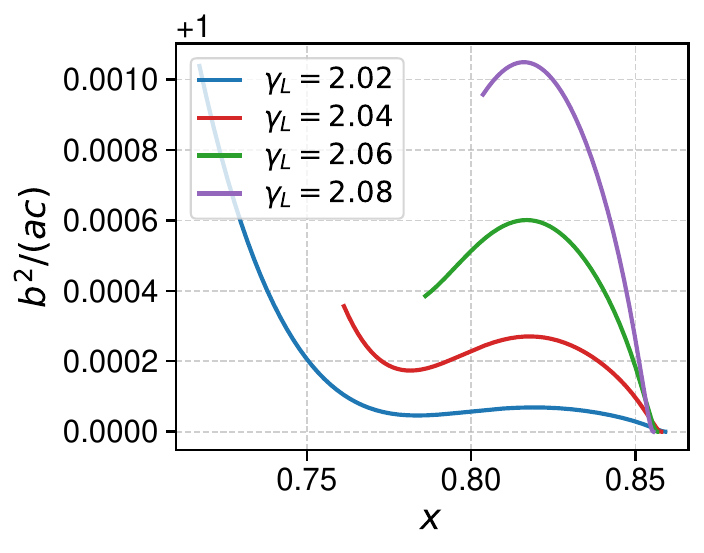}
        \caption{The effect of $\gamma_L$ for $g(x) = x^8$}
        \label{fig:invertible-g8}
    \end{subfigure}
    \caption{Plots of $b^2 / (ac)$ under different cost functions and $\gamma_L$ values}
    \label{fig:invertible-experiments-1}
\end{figure}

Next, we examine two different values of $n_H$: $n_H = 3$ and $n_H = 7$. For $n_H = 3$, we fix $\gamma_H = 2$ and $\gamma_L = 2.02$; for $n_H = 7$, we fix $\gamma_H = 2$ and choose $\gamma_L = 2.001$ to guarantee the existence of hybrid equilibrium. For each $n_H$ value, we test the same four cost functions as those used above, except that for $n_H = 7$ we do not plot the case $g(x) = x^2$ because a hybrid equilibrium would require an even smaller $\gamma_L$; at such close cost parameters the equilibrium is almost symmetric, yielding $k \approx 1$. In practice, $n_H$ and $n_L$ are often smaller than $10$, so $(n_H - 1)(n_L - 1) / (n_H n_L)$ is typically smaller than $0.81$, thus ensuring the invertibility of the Jacobian matrix $J$.


\begin{figure}[t]
    \centering
    \captionsetup[subfigure]{justification=centering}
    \begin{subfigure}[b]{0.32\linewidth}
        \centering
        \includegraphics[width=\linewidth]{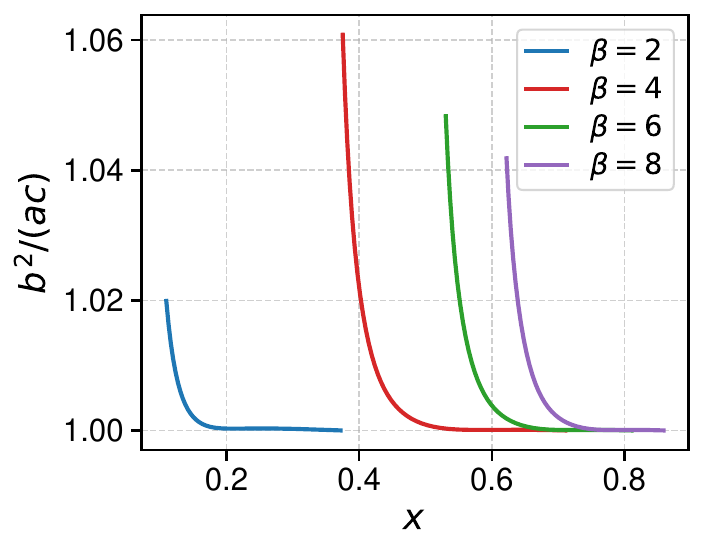}
        \caption{The effect of $g$ for $n_H = 3$}
        \label{fig:invertible-m1}
    \end{subfigure}
    \begin{subfigure}[b]{0.32\linewidth}
        \centering
        \includegraphics[width=\linewidth]{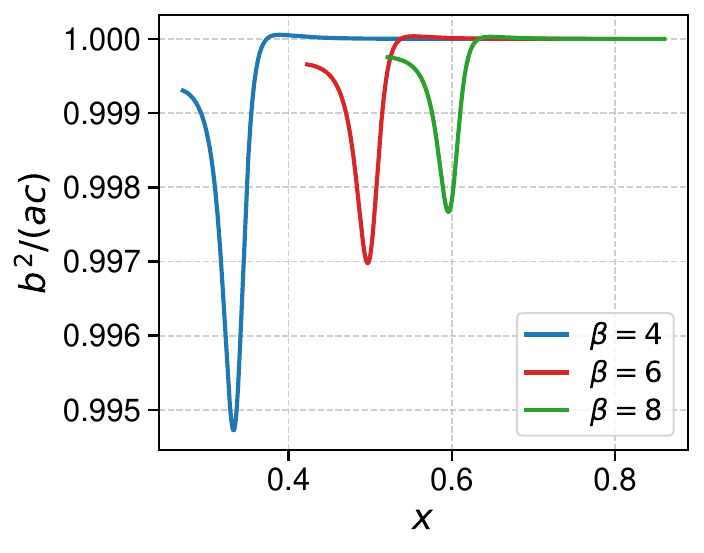}
        \caption{The effect of $g$ for $n_H = 7$}
        \label{fig:invertible-m3}
    \end{subfigure}
    \caption{Plots of $b^2 / (ac)$ under different cost functions and $n_H$ values}
    \label{fig:invertible-experiments-2}
\end{figure}

\subsubsection{Algorithm for Computing the Equilibrium for Binary-type Setting} \label{subsubsec:binary-type-algorithm}

We first present algorithms for computing the equilibrium utility of type-$H$ creators via bisection search in Algorithm~\ref{alg:uH-computation}. The key idea of the algorithm has been discussed in Section~\ref{subsec:asym-two}. Note that the algorithm returns two values: the equilibrium utility of type-$H$ creators $u_H$ and the equilibrium type, where the support structure of equilibrium type $0$ aligns with that in Figure~\ref{fig:2-type-separated-support-1}, type $1$ aligns with that in Figure~\ref{fig:2-type-separated-support-2}, and type $2$ aligns with those in Figure~\ref{fig:2-type-hybrid-support}.

For the special case where $n_H = 1$, the idea for this situation is similar to that of Section~\ref{subsec:ap-asym-n2}, but because $n_L$ can be larger than 1, an extra step is needed: we must check whether the supremum of the support of the low-cost creators’ equilibrium strategy does not exceed the pure strategy equilibrium effort level of the high-cost creator.

We also need to explain how the pseudo strategy $F_H^u$ is computed in Algorithm~\ref{alg:uH-computation}(lines 3, 10, 19). The pseudo strategy $F_H^u$ can be computed by solving the equation~\eqref{eq:pseudo-strategy}. Rewrite that equation as
\[ \sum_{i = 1}^{n_H} \binom{n_H - 1}{i - 1} [F_H^u(x)]^{n_H - i} [1 - F_H^u(x)]^{i - 1} p_i = \frac{u + \gamma_H g(x)}{x}. \]

According to Lemma~\ref{lem:KJ-monotone}, the left-hand side of the equation is strictly increasing in $F_H^u(x)$. Hence, for any fixed $x$, we can solve for $F_H^u(x)$ using a bisection search. In practical implementation, standard numerical libraries often provide root-finding routines that can be used directly to solve this equation. Because the root is unique, any well-chosen initial guess typically allows such a routine to converge efficiently. A useful heuristic is to use the value of $F_H^u$ obtained at the previous point $x - \Delta x$ as the initial guess for the next point $x$. Since $F_H^u$ usually varies smoothly with $x$, this strategy significantly accelerates the root-finding process.

The underlying idea discussed above can be applied directly to compute all CDFs in Algorithm~\ref{alg:binary-type-equilibrium-computation}; hence the algorithms for the remaining CDFs are not described in detail.

\begin{algorithm}[htbp]
    \caption{Computation of the Equilibrium Utility of Type-$H$ Creators}
    \label{alg:uH-computation}
    \begin{algorithmic}[1]
        \REQUIRE Parameters $n_H, \bm{p}, \gamma_H, g(\cdot)$, $\varepsilon$
        \ENSURE Equilibrium utility of type-$H$ creators $u_H$
        \STATE Compute $\underline{x}_L = \arg\max_x x p_n - \gamma_L g(x)$, $u_L = \underline{x}_L p_n - \gamma_L g(\underline{x}_L)$
        \IF{$n_H = 1$}
            \STATE Compute $x_H = \arg\max_x x p_{n_H} - \gamma_H g(x)$, $u_H = x_H p_{n_H} - \gamma_H g(x_H)$
            \STATE Compute $u_L^{u_H} = x_H p_1 - \gamma_L g(x_H)$, and $\overline{x}_L$ as the largest root of $u_L = x p_2 - \gamma_L g(x)$
            \IF{$\overline{x}_L \leqslant x_H$ \textbf{and} $u_L^{u_H} \leqslant u_L$}
                \STATE $equilibrium\_type \leftarrow 0$
                \RETURN $u_H, equilibrium\_type$
            \ELSE
                \STATE $u_H = \overline{x}_L p_1 - \gamma_H g(\overline{x}_L)$
                \STATE $equilibrium\_type \leftarrow 2$
                \RETURN $u_H, equilibrium\_type$
            \ENDIF
        \ENDIF
        \STATE Compute $\underline{u}_H = \underline{x}_L p_n - \gamma_H g(\underline{x}_L)$, $\overline{u}_H = \max_x x p_{n_H} - \gamma_H g(x)$
        \STATE Compute the pseudo strategy $F_H^{\overline{u}_H}$
        \STATE Compute $u_L^{\overline{u}_H}(x)$ using equation~\eqref{eq:asym-uLu}, find $u_L^{\overline{u}_H} = \max_x u_L^{\overline{u}_H}(x)$
        \IF{$u_L^{\overline{u}_H} < u_L$}
            \STATE $equilibrium\_type \leftarrow 0$
            \RETURN $\overline{u}_H, equilibrium\_type$
        \ENDIF
        \STATE $u_H = (\underline{u}_H + \overline{u}_H) / 2$
        \STATE Compute the pseudo strategy $F_H^{u_H}$ and its support $[\underline{x}_H^{u_H}, \overline{x}_H^{u_H}]$
        \STATE Compute $u_L^{u_H}(x)$ using equation~\eqref{eq:asym-uLu}, find $x_m = \arg\max_x u_L^{u_H}(x)$ and $u_L^{u_H} = \max_x u_L^{u_H}(x)$
        \WHILE{$|u_L^{u_H} - u_L| > \varepsilon$}
            \IF{$u_L^{u_H} < u_L$}
                \STATE $\underline{u}_H = u_H$
            \ELSE
                \STATE $\overline{u}_H = u_H$
            \ENDIF
            \STATE $u_H = (\underline{u}_H + \overline{u}_H) / 2$
            \STATE Compute the pseudo strategy $F_H^{u_H}$ and its support $[\underline{x}_H^{u_H}, \overline{x}_H^{u_H}]$
            \STATE Compute $u_L^{u_H}(x)$ using equation~\eqref{eq:asym-uLu}, find $x_m = \arg\max_x u_L^{u_H}(x)$ and $u_L^{u_H} = \max_x u_L^{u_H}(x)$
        \ENDWHILE
        \IF{$x_m = \underline{x}_H^{u_H}$}
            \STATE $equilibrium\_type \leftarrow 1$
        \ELSE
            \STATE $equilibrium\_type \leftarrow 2$
        \ENDIF
        \RETURN $u_H, equilibrium\_type$
    \end{algorithmic}
\end{algorithm}

We then present the complete algorithm for computing the equilibrium strategies of both types of creators in Algorithm~\ref{alg:binary-type-equilibrium-computation}. Note that the algorithm assumes $n_H \geqslant 2$ and $n_L \geqslant 2$. When the number of creators of one type is equal to $1$, the algorithm requires minor adjustments; we will explain these modifications after describing the main procedure. 

In Algorithm~\ref{alg:binary-type-equilibrium-computation}, lines 3-8 handle the cases in which the equilibrium support structure corresponds to type $0$ or type $1$, i.e., the separated equilibrium depicted in Figures~\ref{fig:2-type-separated-support}. In these cases, the equilibrium strategies can be computed directly since $F_H$ and $F_L$ do not overlap on any interval.

Then because the support structure of Figure~\ref{fig:2-type-hybrid-support-2} occurs far more frequently than that of Figure~\ref{fig:2-type-hybrid-support-1}, we first address the former in lines 9-14, if the computed strategies satisfies the monotonicity and indifference conditions, we return them directly (lines 15-17). Otherwise, we proceed to the latter case in lines 18-21. The basic idea of the algorithm is to start from $\underline{x}_L$, then iteratively use the indifference condition to compute the values of $F_L$ and $F_H$ based on the support structure, continuing until the upper bound $\overline{x}_L = \overline{x}_H$ is reached. Several details require clarification:
\begin{enumerate}
    \item The explicit forms of $W(\bm{F}_{-h}(x), \bm{p})$ and $W(\bm{F}_{-l}(x), \bm{p})$ are as follows:
    \begin{align*}
        W(\bm{F}_{-h}(x), \bm{p}) &= \sum_{k = 0}^{n - 1} p_{k + 1} \sum_{j = \max\{0, k - n_L\}}^{\min\{k, n_H - 1\}} \binom{n_H - 1}{j} (1 - F_H(x))^j F_H(x)^{n_H - 1 - j} \\ 
        &\quad\quad\quad\quad\quad\quad\quad\quad\quad\quad\cdot \binom{n_L}{k - j} (1 - F_L(x))^{k - j} F_L(x)^{n_L - (k - j)},
    \end{align*}
    \begin{align*}
        W(\bm{F}_{-l}(x), \bm{p}) &= \sum_{k = 0}^{n - 1} p_{k + 1} \sum_{j = \max\{0, k - (n_L - 1)\}}^{\min\{k, n_H\}} \binom{n_H}{j} (1 - F_H(x))^j F_H(x)^{n_H - j} \\
        &\quad\quad\quad\quad\quad\quad\quad\quad\quad\quad\cdot \binom{n_L - 1}{k - j} (1 - F_L(x))^{k - j} F_L(x)^{n_L - 1 - (k - j)}.
    \end{align*}
    \item In line 11 it suffices to start from $\underline{x}_L$ and increase $x$  until a value $\underline{x}_H$ satisfying the equation is found. This is because on the interval $[\underline{x}_L, \underline{x}_H]$, $F_L$ keeps type-$L$ creators indifferent. For type-$H$ creators, the reward includes an extra term involving $F_L$, while the cost coefficient satisfies $\gamma_H < \gamma_L$. Hence, the utility of type-$H$ creators is strictly increasing in $x$ over this interval, guaranteeing a unique $\underline{x}_H$ that satisfies the equation in line 11.
    \item Solving the system of equations for the CDFs in line 13 is implemented by standard numerical libraries. To accelerate convergence, we employ the same strategy as used to compute $F_H^u$: the solution obtained at the previous point is used as the initial guess for the next point.
    \item In line 18, denote $v(x) := x W(\bm{F}_{-h}(x^{**}), \bm{p}) - \gamma_H g(x)$. Clearly, $\max_x v(x)$ increases monotonically with $x^{**}$. Therefore, we can start from $\underline{x}_L$ and gradually increase $x^{**}$ until the system of equations in line 18 is satisfied.
    \item If $n_L = 1$, for the separated equilibrium cases (equilibrium type $0$ or $1$), lines 5-6 are not needed, and $F_L$ in line 7 is a degenerate distribution that places all mass at $\underline{x}_L$. For hybrid equilibrium cases (equilibrium type $2$), we first combine lines 10-11 to compute the point mass of $F_L$ at $\underline{x}_L$. Since the support structure aligns with Figure~\ref{fig:2-type-hybrid-support-1}, we directly proceed to line 18 after line 11. If $n_H = 1$, the analysis in Section~\ref{subsec:ap-asym-n2} implies that a single type-$H$ creator cannot take the mixed strategy alone. Consequently, line 12, which computes the lower bound of the support of hybrid region, becomes unnecessary. Instead, before line 19, we should compute the point mass of the type-$H$ creator at $\underline{x}_H$ to ensure that $u_L = \underline{x}_H W(\bm{F}_{-l}(\underline{x}_H), \bm{p}) - \gamma_L g(\underline{x}_H) \ (l \in \{n_H + 1, \ldots, n\})$.
\end{enumerate}

\begin{algorithm}[htbp]
    \caption{Computation of the Equilibrium for Binary-type Setting}
    \label{alg:binary-type-equilibrium-computation}
    \begin{algorithmic}[1]
        \REQUIRE Parameters $n_H \geqslant 2, n_L \geqslant 2, \bm{p}, \gamma_H, \gamma_L, g(\cdot)$
        \ENSURE Equilibrium strategies $F_H, F_L$
        \STATE Compute $\underline{x}_L = \arg\max_x x p_n - \gamma_L g(x)$, $u_L = \underline{x}_L p_n - \gamma_L g(\underline{x}_L)$
        \STATE Use Algorithm~\ref{alg:uH-computation} to compute $u_H$ and $equilibrium\_type$
        \IF{$equilibrium\_type = 0$ or $1$}
            \STATE Compute the pseudo strategy $F_H^{u_H}$ and its support $[\underline{x}_H^{u_H}, \overline{x}_H^{u_H}]$
            \STATE Compute $F_L$ using equation~\eqref{eq:asym-FL}
            \STATE Compute $\overline{x}_L$ as the largest root of $u_L = x p_{n_H + 1} - \gamma_L g(x)$
            \RETURN $F_H^{u_H}, F_L$
        \ENDIF
        \STATE $\overline{x}_H, \overline{x}_L \leftarrow \overline{x}_H^{u_H}$
        \STATE Compute $F_L$ using equation~\eqref{eq:asym-FL}
        \STATE Find $\underline{x}_H$ such that $u_H = \underline{x}_H W(\bm{F}_{-h}(\underline{x}_H), \bm{p}) - \gamma_H g(\underline{x}_H)$ ($h \in \{1, \ldots, n_H\}$)
        \STATE Compute $F_H$ from $\underline{x}_H$ using $u_H = x W(\bm{F}_{-h}(x), \bm{p}) - \gamma_H g(x)$ until $x^*$, where $x^*$ satisfies $u_L = x^* W(\bm{F}_{-l}(x^*), \bm{p}) - \gamma_L g(x^*) \ (l \in \{n_H + 1, \ldots, n\})$
        \STATE Compute $F_H$ and $F_L$ on $[x^*, \overline{x}_H]$ using $\begin{cases}
            u_H = x W(\bm{F}_{-h}(x), \bm{p}) - \gamma_H g(x), \\
            u_L = x W(\bm{F}_{-l}(x), \bm{p}) - \gamma_L g(x).
        \end{cases}$
        \STATE $F_L = \begin{cases}
            0, & x \leqslant \underline{x}_L, \\
            \text{result from line 10}, & \underline{x}_L < x \leqslant \underline{x}_H, \\
            F_L(\underline{x}_H), & \underline{x}_H < x \leqslant x^*, \\
            \text{result from line 13}, & x^* < x \leqslant \overline{x}_L, \\
            1, & x > \overline{x}_L,
        \end{cases} \ F_H = \begin{cases}
            0, & x \leqslant \underline{x}_H, \\
            \text{result from line 12}, & \underline{x}_H < x \leqslant x^*, \\
            \text{result from line 13}, & x^* < x \leqslant \overline{x}_H, \\
            1, & x > \overline{x}_H.
        \end{cases}$
        \IF {$(F_H, F_L)$ satisfy monotonicity and indifference conditions}
            \RETURN $F_H, F_L$
        \ENDIF
        \STATE Find $x^{**}$ satisfying $\begin{cases}
            \underline{x}_H = \arg\max_x x W(\bm{F}_{-h}(x^{**}), \bm{p}) - \gamma_H g(x), \\
            u_H = \underline{x}_H W(\bm{F}_{-h}(x^{**}), \bm{p}) - \gamma_H g(\underline{x}_H).
        \end{cases}$
        \STATE Repeat lines 12-13 with $\underline{x}_H$ computed in line 18
        \STATE $F_L = \begin{cases}
            0, & x \leqslant \underline{x}_L, \\
            \text{result from line 10}, & \underline{x}_L < x \leqslant x^{**}, \\
            F_L(x^{**}), & x^{**} < x \leqslant x^*, \\
            \text{result from line 19}, & x^* < x \leqslant \overline{x}_L, \\
            1, & x > \overline{x}_L,
        \end{cases} \ F_H = \begin{cases}
            0, & x \leqslant \underline{x}_H, \\
            \text{result from line 19}, & \underline{x}_H < x \leqslant x^*, \\
            \text{result from line 19}, & x^* < x \leqslant \overline{x}_H, \\
            1, & x > \overline{x}_H.
        \end{cases}$
        \RETURN $F_H, F_L$
    \end{algorithmic}
\end{algorithm}
\section{Omitted Proofs in Section~\ref{subsec:asym-mechanism}} \label{sec:ap-asym-proofs}

\begin{proofof}{Proposition~\ref{prop:asym-separated-compensation-md}}
    In a separated equilibrium, type-$H$ and type-$L$ creators choose their strategies from disjoint supports. Consequently, in a separated equilibrium with a compensation vector $\bm{c}$, the equilibrium strategies $F_H$ and $F_L$ must satisfy the following conditions:
    \begin{align*}
        & x \left(\sum_{i = 1}^{n_H} \binom{n_H - 1}{i - 1} [F_H(x)]^{n_H - i} [1 - F_H(x)]^{i - 1} (p_i + c_i)\right) - \gamma_H g(x) = u_H, \quad \forall x \in [\underline{x}_H, \overline{x}_H], \\
        & x \left(\sum_{i = 1}^{n_L} \binom{n_L - 1}{i - 1} [F_L(x)]^{n_L - i} [1 - F_L(x)]^{i - 1} (p_{n_H + i} + c_{n_H + i})\right) - \gamma_L g(x) = u_L, \quad \forall x \in [\underline{x}_L, \overline{x}_L],
    \end{align*}
    where $u_H$ and $u_L$ are the equilibrium utilities of type-$H$ and type-$L$ creators, respectively, and $[\underline{x}_H, \overline{x}_H]$ and $[\underline{x}_L, \overline{x}_L]$ denote the supports of their equilibrium strategies. 
    
    Let
    \begin{equation} \label{eq:KHL-def}
        K_H(x) = \frac{u_H + \gamma_H g(x)}{x}, \quad K_L(x) = \frac{u_L + \gamma_L g(x)}{x}.
    \end{equation}
    Similar to the proof of Lemma~\ref{lem:K-monotone}, we can show that $K_H(x)$ and $K_L(x)$ are both strictly increasing on their respective domains. Denote $\bm{p}_H = (p_1, p_2, \dots, p_{n_H})$ and $\bm{p}_L = (p_{n_H + 1}, p_{n_H + 2}, \dots, p_n)$, and denote $\bm{c}_H$ and $\bm{c}_L$ similarly. Then we can define
    \begin{equation} \label{JHL-def}
        \begin{aligned}
            x = K_H^{-1}\left(\sum_{i = 1}^{n_H} \binom{n_H - 1}{i - 1} [F_H(x)]^{n_H - i} [1 - F_H(x)]^{i - 1} (p_i + c_i)\right) &:= J_H(F_H(x), \bm{p}_H + \bm{c}_H), \\
            x = K_L^{-1}\left(\sum_{i = 1}^{n_L} \binom{n_L - 1}{i - 1} [F_L(x)]^{n_L - i} [1 - F_L(x)]^{i - 1} (p_{n_H + i} + c_{n_H + i})\right) &:= J_L(F_L(x), \bm{p}_L + \bm{c}_L).
        \end{aligned}
    \end{equation}

    Based on Lemma~\ref{lem:compensation-md-reform}, we can reformulate the Problem~\eqref{eq:mechanism-design} as follows:
    \begin{align*}
        \max_{\bm{c}} \ & n_H \int_{0}^{1} K_H(J_H(y, \alpha\bm{p}_H - \bm{c}_H)) J_H(y, \bm{p}_H + \bm{c}_H) \, \textup{d}y \\
        &+ n_L \int_{0}^{1} K_L(J_L(y, \alpha\bm{p}_L - \bm{c}_L)) J_L(y, \bm{p}_L + \bm{c}_L) \, \textup{d}y \\
        &+ \alpha h\left(\frac{n_H}{n} \int_{0}^{1} J_H(y, \bm{p}_H + \bm{c}_H) \, \textup{d}y + \frac{n_L}{n} \int_{0}^{1} J_L(y, \bm{p}_L + \bm{c}_L) \, \textup{d}y\right) p_0 \\
            \textup{s.t.} \quad & c_1 \geqslant c_2 \geqslant \dots \geqslant c_n \geqslant 0.
    \end{align*}

    Analogous to the proof of Theorem~\ref{thm:compensation-md-sym}, we can show that the symmetric form of the objective function is concave with respect to $(c_1, \dots, c_{n_H - 1})$ and $(c_{n_H + 1}, \dots, c_{n - 1})$. Therefore, the optimal compensation mechanism must satisfy $c_1 = \cdots = c_{n_H - 1}$ and $c_{n_H + 1} = \cdots = c_{n - 1}$.
\end{proofof}

Recall the notation in equation \eqref{eq:asym-optimal-compensation-vector}. A natural question is whether $c_H > c_L$ can hold (otherwise the result collapses to the symmetric structure of Theorem~\ref{thm:compensation-md-sym}). Experiments in Appendix~\ref{subsec:ap-exp-cHcL} show it occurs across a range of plausible settings.

\begin{proofof}{Proposition~\ref{prop:TPBL-inc-asym}}
    The proof is similar to that of Proposition~\ref{prop:UBL-inc-sym}, thus we use the same notations as in that proof. First, we derive the partial derivatives of $W(\bm{p}, \bm{c})$ with respect to $p_i$ and $c_i$ for $i \in [n]$. Then, we compare these two derivatives to establish equation~\eqref{eq:UBL-inc-sym-1}. Then we note that the definition of TPBL citation mechanism ensures that
    \[ \frac{s_1}{c_1} = \cdots = \frac{s_{n - 1}}{c_{n - 1}}, \quad \text{and} \quad s_n = c_n = 0. \]
    Therefore, following an argument analogous to the proof of Proposition~\ref{prop:UBL-inc-sym}, we can consider the two cases $\bm{r} = \bm{s}$ and $\bm{r} = \bm{c}$ respectively to prove the conclusion.
\end{proofof}
\section{Details on Data Collection} \label{sec:ap-data-collection}

\subsection{Ethical Considerations} \label{subsec:ap-ethical}

Our user click experiment data were collected solely for academic research with participants’ informed consent regarding data use. We only release aggregated click statistics and do not share raw or personally identifiable data, mitigating privacy risks. Consequently, we anticipate no adverse societal effects arising from our study. We have obtained approval from our institution’s ethics review board for this research. Participants were properly compensated for their time (above the minimum wage of their country), and we ensured that the experiment posed minimal risk to them.

\subsection{The 24 Search Queries} \label{subsec:ap-search-queries}

The 24 search queries used in our user click experiment are listed below, which cover a range of topics (e.g., mathematics, history, health and travel) and reflected distinct search intents \cite{rowley2018organizing}: the first 12 were fact-based and the last 12 were open-ended.

\begin{enumerate}
    \item What is the capital city of Kenya?
    \item What is the population of Russia?
    \item What was the contribution that earned the 1994 Nobel Prize in Economics?
    \item When is the deadline of NeurIPS 2026?
    \item Calculate the integral $\int_0^{\pi/2} \sin^3 x \ \textup{d}x$.
    \item How can I fix the image size in Markdown?
    \item What is the fine for crossing a double yellow line?
    \item What is defined as the normal blood pressure range for adults?
    \item Detailed overview of the Falklands War.
    \item Prove the Cauchy's Mean Value Theorem.
    \item Algorithm implementations for the various variants of the Knapsack Problem.
    \item The distinct periods in Vincent van Gogh's paintings.
    \item Impact of Zohran Mamdani's election as Mayor of New York City.
    \item How to plan a three-day trip for Beijing?
    \item How can I improve my sleep quality?
    \item Key points to note for a thesis defense.
    \item What are some effective methods and resources for self-learning a new language, such as Japanese?
    \item What books should I read to understand the Chinese economy?
    \item How to design a home study room?
    \item How to appreciate the painting "Mona Lisa"?
    \item What are the fundamental differences between Marxist economics and Western economics?
    \item What key factors should I consider when buying my first car?
    \item How to prepare a creative birthday gift for my mother?
    \item Suggestions for beginners learning to play the piano.
\end{enumerate}

\subsection{Examples of the Three Search Interfaces} \label{subsec:ap-search-interfaces}

The following figures illustrate examples of the three types of search engine interfaces used in our user click experiment. In our experiment, the query order was randomized. Participants clicked the most suitable pages based on their own understanding and search habits; each task additionally required answering a question to encourage careful examination.


The Type A interface is shown in Figure~\ref{fig:interface-A}. The name Baigle is a fictional search engine brand formed by combining the first three letters of ``Baidu'' and the last three letters of ``Google''. The layout follows a conventional search engine design, with all advertisements and other elements that might affect clicking behavior removed. In the Type A interface, users see only the traditional list of organic results; no AI Overview appears. A ``Notice'' box in the upper-right corner reminds participants of the interface type. An answer box in the bottom-right corner promotes careful engagement: users are required to enter their answer after finding a satisfactory result.

\begin{figure}[t]
    \centering
    \includegraphics[width=0.7\textwidth]{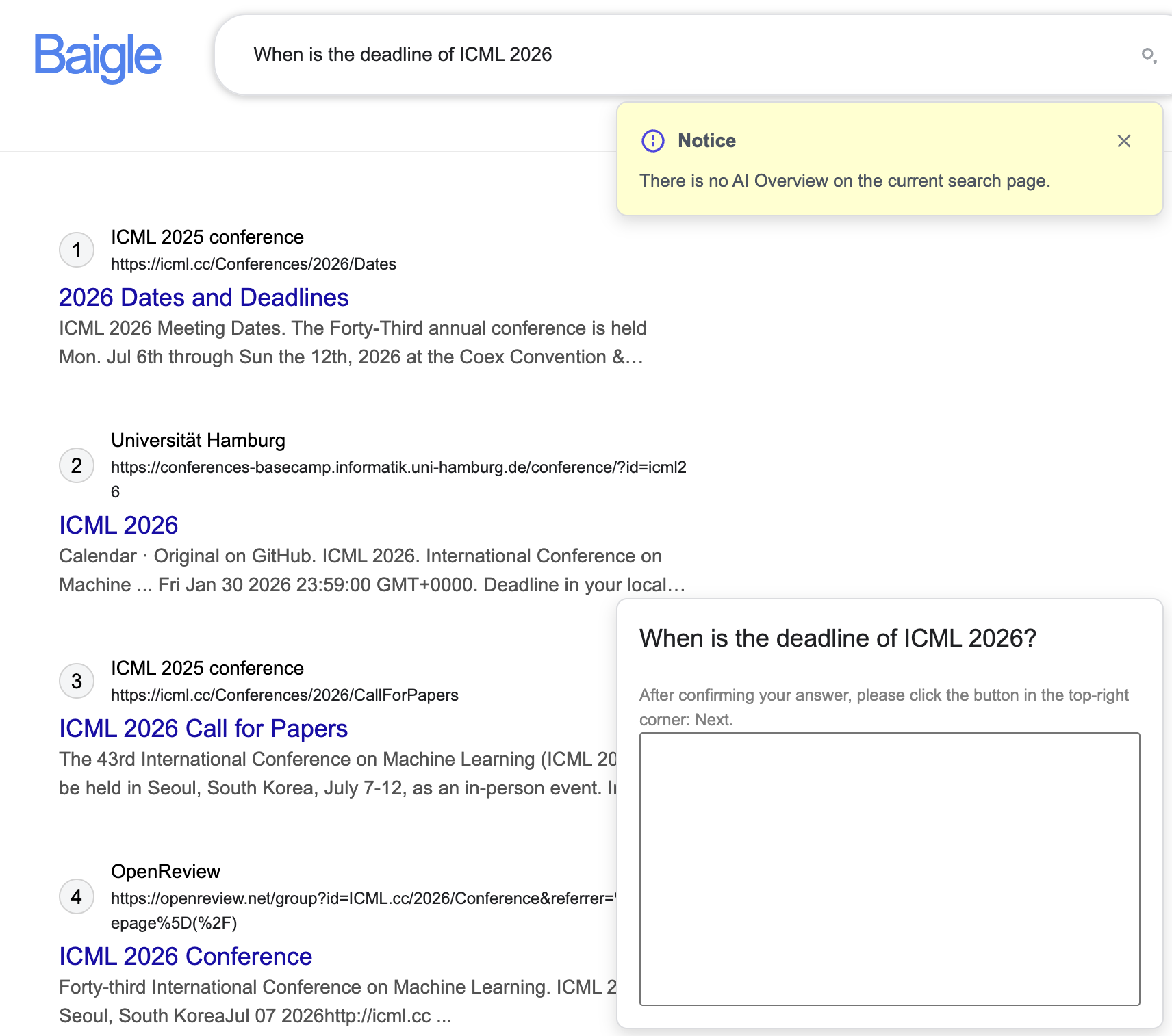} 
    \caption{Example of Type A interface}
    \label{fig:interface-A}
\end{figure}

Figure~\ref{fig:interface-B} presents the Type B interface, which includes an AI Overview at the top of the search results but does not cite any specific web pages. Figure~\ref{fig:interface-C} shows the AI Overview with citations in the Type C interface; here, four relevant web pages are cited below the AI Overview.

\begin{figure}[t]
    \centering
    \includegraphics[width=0.7\textwidth]{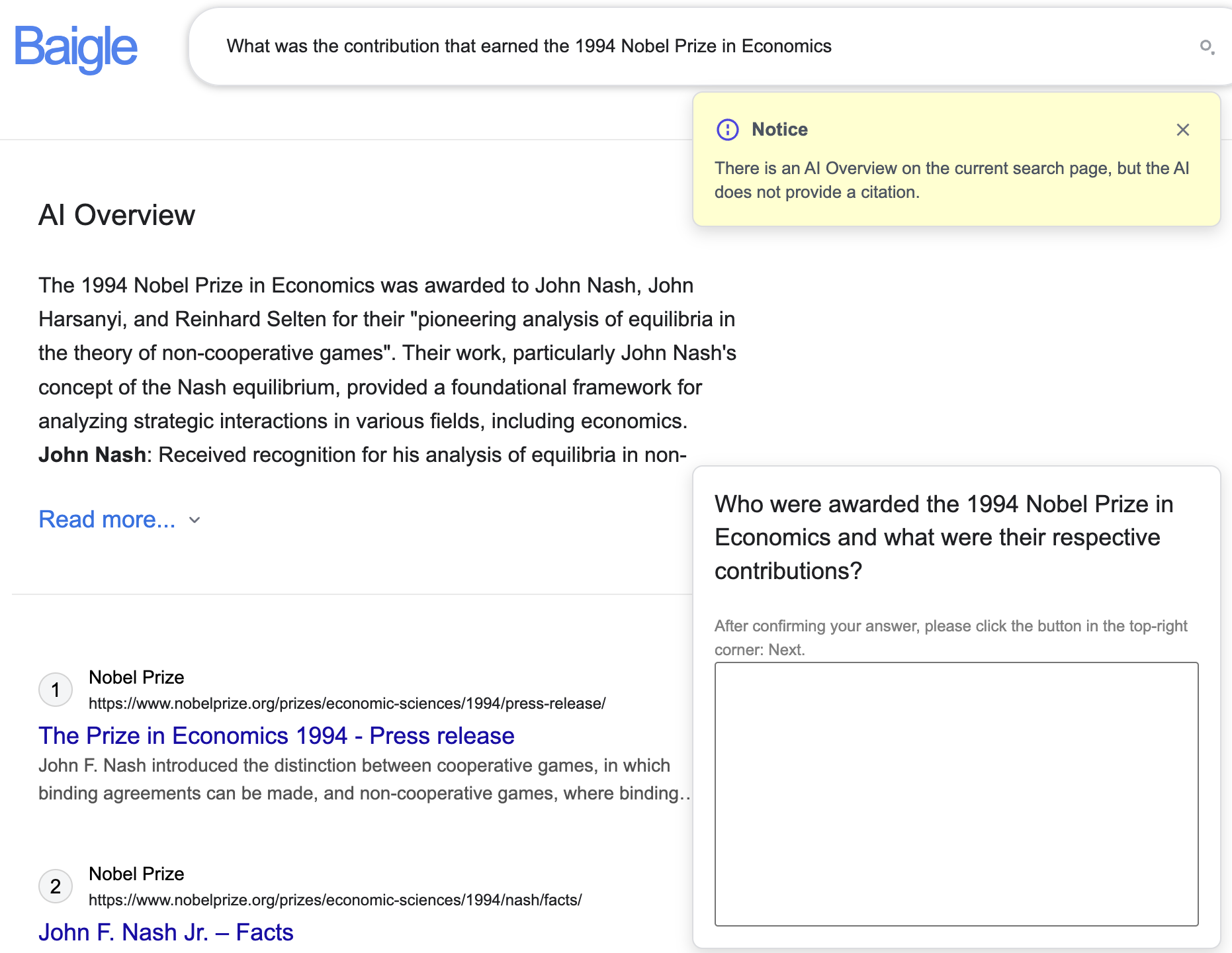} 
    \caption{Example of Type B interface}
    \label{fig:interface-B}
\end{figure}

\begin{figure}[t]
    \centering
    \includegraphics[width=0.8\textwidth]{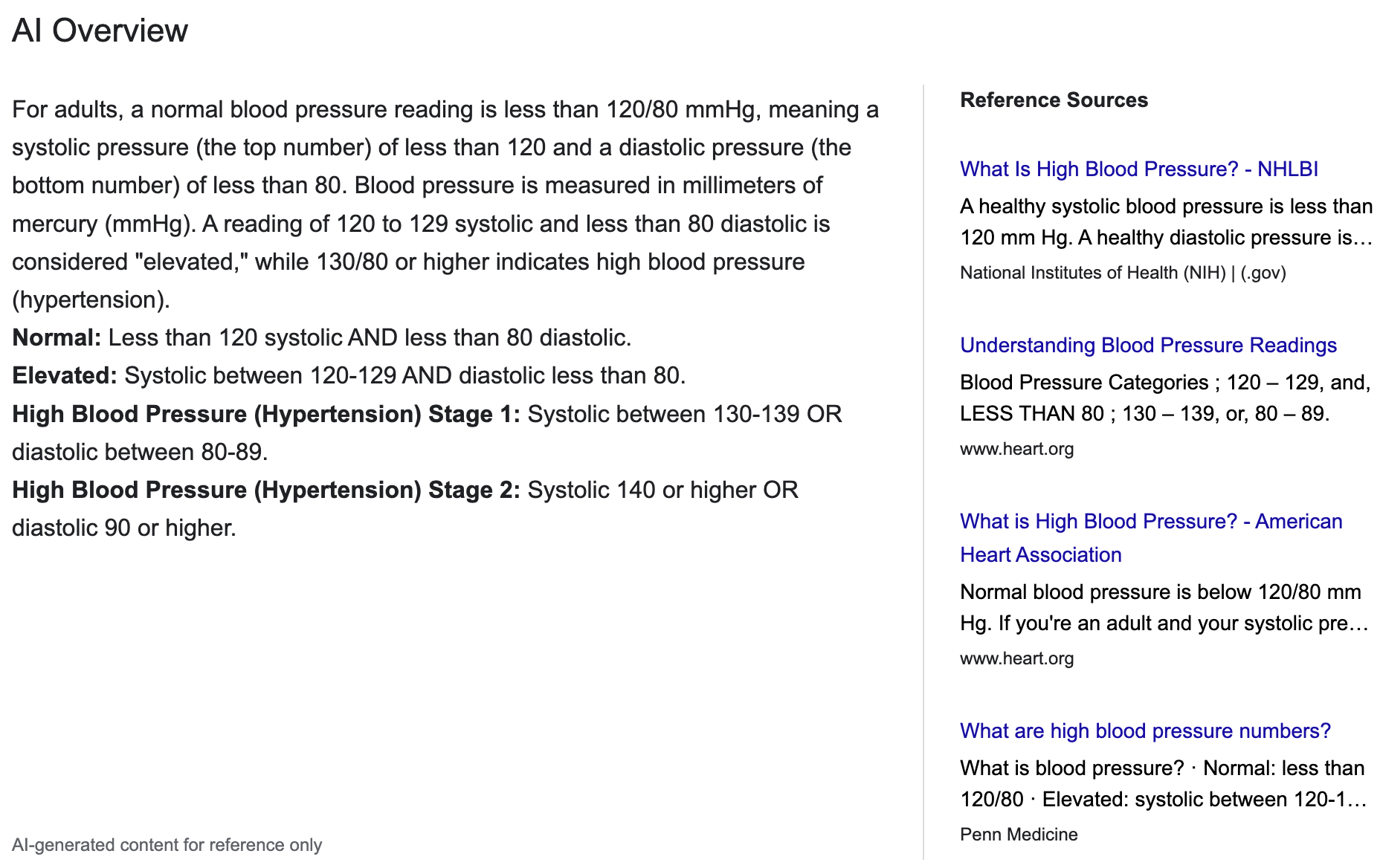} 
    \caption{Example of AI Overview with citations in Type C interface}
    \label{fig:interface-C}
\end{figure}

\subsection{Position Bias Computation Details} \label{subsec:ap-position-bias-calculation}

The code for computing position bias is based on open source resources provided in \cite{chuklin2015click}, which use the EM algorithm to estimate position bias from user click data. Here we explain the adjustments we made to their method in order to obtain the results in our paper. Recall that in the PBM model we adopt, the click-through rate depends on two factors: the attractiveness of the page (the effort of creator) and the position bias of its rank.

\begin{enumerate}
    \item Because the search results are identical across the three interface types, the attractiveness of each page should be the same in all three cases. Consequently, when estimating position bias, we ensure that the attractiveness values used for the three interfaces are consistent.
    \item The position bias of the AI Overview itself requires a separate estimation. Our experiment recorded dwell times for every page, including the AI Overview. We consider the AI Overview to have been examined if the dwell time is at least 2 seconds. For dwell times shorter than 2 seconds, we further check whether any other page was clicked. If no other page was clicked, we assume the user obtained the answer from the AI Overview and therefore treat it as examined. Otherwise, we inspect whether the participant’s answer indicates that they referred to the AI Overview. Applying these criteria to the 800 search tasks that contained an AI Overview, we identified only 9 instances where the AI Overview was not examined. This yields an estimated position bias for the AI Overview of $1 - 9/800 \approx 0.9888$. This value is very close to the position bias of the top-ranked page in the absence of an AI Overview, and because the AI Overview effectively occupies the first position, the estimate is reasonable.
    \item We also need to estimate the position bias of pages cited in the AI Overview. Since these cited pages also appear in the traditional list of organic results, we ensure that their attractiveness is kept consistent with the attractiveness of the corresponding pages in the organic list.
\end{enumerate}

Taking these three points into account, we adapted the code from \cite{chuklin2015click} to compute the position bias estimates presented in Table~\ref{tab:position-bias-no-ai-overview} and Table~\ref{tab:position-bias-ai-overview}.

\subsection{Additional Data and Analyses} \label{subsec:ap-more-data-analysis}

Tables~\ref{tab:ctr-q1-12} and \ref{tab:ctr-ai-overview-q1-12} report the click-through rates (CTRs) for the first 12 fact-based queries, while Tables~\ref{tab:ctr-q13-24} and \ref{tab:ctr-ai-overview-q13-24} present the CTRs for the remaining 12 open-ended queries. The tables show that the changes in CTR following the introduction of the AI Overview and the citation mechanism do not differ significantly between fact-based and open-ended queries. The principal difference between the two query types is the overall CTR level, which is higher for open-ended queries, an intuitive result because users typically click more when addressing this type of questions. Since our discussion centers on trends in position bias, we do not distinguish between the two query types and instead combine them for a consolidated analysis.

\begin{table}[htpb]
    \centering
    \small
    \caption{CTR of organic search results for queries 1-12}
    \begin{tabular}{ccccccccccc}
    \hline
    Type & page1 & page2 & page3 & page4 & page5 & page6 & page7 & page8 & page9 & page10 \\
    \hline
    A & 0.7108 & 0.3775 & 0.3137 & 0.2353 & 0.0686 & 0.0637 & 0.0392 & 0.0294 & 0.0294 & 0.0343 \\
    B & 0.5500 & 0.3400 & 0.1550 & 0.1000 & 0.0450 & 0.0450 & 0.0500 & 0.0500 & 0.0450 & 0.0300 \\
    C & 0.3400 & 0.2050 & 0.1350 & 0.1200 & 0.0350 & 0.0450 & 0.0400 & 0.0350 & 0.0250 & 0.0100 \\
    \hline
    \end{tabular}
    \label{tab:ctr-q1-12}
\end{table}

\begin{table}[htbp]
    \centering
    \caption{CTR of pages included in the AI Overview for queries 1-12}
    \begin{tabular}{cccc}
    \hline
    page1 & page2 & page3 & page4 \\
    \hline
    0.3700 & 0.1800 & 0.1300 & 0.0400 \\
    \hline
    \end{tabular}
    \label{tab:ctr-ai-overview-q1-12}
\end{table}

\begin{table}[htpb]
    \centering
    \small
    \caption{CTR of organic search results for queries 13-24}
    \begin{tabular}{ccccccccccc}
    \hline
    Task & page1 & page2 & page3 & page4 & page5 & page6 & page7 & page8 & page9 & page10 \\
    \hline
    A & 0.8578 & 0.6176 & 0.2990 & 0.2794 & 0.2157 & 0.1176 & 0.1127 & 0.1078 & 0.0882 & 0.0686 \\
    B & 0.5980 & 0.3668 & 0.1859 & 0.1809 & 0.1256 & 0.0854 & 0.0704 & 0.0653 & 0.0704 & 0.0603 \\
    C & 0.3850 & 0.2500 & 0.1500 & 0.1400 & 0.0750 & 0.0650 & 0.0550 & 0.0300 & 0.0350 & 0.0300 \\
    \hline
    \end{tabular}
    \label{tab:ctr-q13-24}
\end{table}

\begin{table}[htbp]
    \centering
    \caption{CTR of pages included in the AI Overview for queries 13-24}
    \begin{tabular}{cccc}
    \hline
    page1 & page2 & page3 & page4 \\
    \hline
    0.4250 & 0.2850 & 0.1800 & 0.0950 \\
    \hline
    \end{tabular}
    \label{tab:ctr-ai-overview-q13-24}
\end{table}

We provide the complete click data for all 24 queries in supplementary material, which includes the number of clicks for each page at each query across the three interface types.
\section{Details and Additional Results of Experiment } \label{sec:ap-experiment}

\subsection{Parameters in Experiments} \label{subsec:ap-exp-para}

We first outline the adjustable parameters in our experiments: the profitability $\alpha$ of search engine, the quality function $h$ of the AI Overview, the exponent $\beta \geqslant 2$ of the creator cost function $g$, the cost coefficient $\gamma_L$ for high-cost creators (with the low-cost creator coefficient $\gamma_H$ fixed due to Proposition~\ref{prop:asym-optimal-ch-cl-2}), and the number of low-cost creators $n_H$ (the total number is fixed at $10$, so $n_L = 10 - n_H$).

The parameters $\alpha$ for each query and $h$ can be easily estimated by the search engine, whereas $\beta$, $\gamma_L$, and $n_H$ require the search engine to infer the equilibrium by aggregating page attraction data ($x_i$ in our model) across many queries from the same category. Using this estimate together with Propositions~\ref{prop:asym-separated-equilibrium} and~\ref{prop:asym-hybrid-equilibrium}, the search engine can infer these parameters. Such inference, however, necessitates large-scale, real-world data on creator behavior, to which we do not have access. Therefore, in our experiments we simulate a range of plausible values for these parameters to examine how they affect the optimal compensation structure.

In practice, these parameters typically differ across search topics. For specialized queries such as academic searches, content creation requires more specialized expertise, resulting in a lower value of $\beta$ (note that $g$ is defined on $[0, 1]$). Moreover, in such contexts, the gap between $\gamma_L$ and $\gamma_H$ tends to be larger. In contrast, for general-interest topics like entertainment or daily life, the barrier to creation is lower, leading to a relatively higher $\beta$ and a narrower difference between $\gamma_L$ and $\gamma_H$. Therefore, a search engine needs to estimate separate sets of parameters for different types of queries and tailor its mechanism accordingly.

\subsection{Numerical Evaluation of Conjecture~\ref{conj:asym-hybrid-compensation-md}} \label{subsec:ap-exp-c}

We evaluate Conjecture~\ref{conj:asym-hybrid-compensation-md} by testing Schur-concavity of the objective with respect to each compensation block, $(c_1, \ldots, c_{n_H-1})$ and $(c_{n_H+1}, \ldots, c_{n-1})$, while holding the remaining components fixed. We fix the linear revenue parameter $\alpha = 1$ and use $h(x) = x$, which introduces no additional concavity through the AI Overview quality function.

For each tested block, we enumerate compensation vectors with a fixed sum $c_{\text{sum}}$, components in $[0,1]$, and step size $0.05$. We vary $q_1$ and $q_{n_H}$ in TPBL to obtain different position bias vectors and use three values of $c_{\text{sum}}$ per configuration. When testing the first block, we set $c_{n_H+1}, \ldots, c_{n-1}$ to zero; when testing the second, we fix $c_1, \ldots, c_{n_H}$ at $0.6$ or $0.8$ to preserve compensation monotonicity. Table~\ref{tab:exp-c-parameters} summarizes the experimental settings. We exclude $n_H=2$ from the first-block tests because a single component admits no nontrivial fixed-sum comparison, and omit $n_H=8$ to favor hybrid equilibria.

\begin{table}[t]
\centering
\footnotesize
\setlength{\tabcolsep}{4pt}
\renewcommand{\arraystretch}{1.15}
\caption{Parameter settings for the Schur-concavity experiments\label{tab:exp-c-parameters}}
\begin{tabular}{@{}l@{\hspace{2em}}c@{}}
\toprule
Parameter & Setting \\
\midrule
Total number of creators $n$ & $10$ \\
Revenue parameter $\alpha$ & $1$ \\
Quality function $h(x)$ & $x$ \\
Cost exponent $\beta$ & $\{4,5,6,7,8\}$ \\
Cost coefficient $\gamma_L$ & $\{2.01,2.05,2.1\}$ \\
$n_H$, testing $(c_1,\ldots,c_{n_H-1})$ & $\{3,4,5,6,7\}$ \\
$n_H$, testing $(c_{n_H+1},\ldots,c_{n-1})$ & $\{2,3,4,5,6,7\}$ \\
Compensation component range & $[0,1]$ \\
Number of $c_{\text{sum}}$ values per configuration & $3$ \\
Grid step & $0.05$ \\
\bottomrule
\end{tabular}
\end{table}

We compare objective values for all enumerated vector pairs satisfying the majorization ordering. For fixed $\gamma_L$, $\beta$, and $c_{\text{sum}}$, the number of comparisons across the tested $n_H$ values is on the order of $10^6$; first-block tests involve roughly $10^2$--$10^3$ pairs for small $n_H$ and reach the order of $10^5$ for larger $n_H$. The experimental results demonstrate that, under all tested parameter configurations, the objective function exhibits Schur-concavity with respect to $c_1, c_2, \ldots, c_{n_H - 1}$ and $c_{n_H + 1}, c_{n_H + 2}, \ldots, c_{n - 1}$, thereby supporting Conjecture~\ref{conj:asym-hybrid-compensation-md}.

\subsection{The Effectiveness of the UBL and TPBL Mechanism} \label{subsec:ap-exp-p}

In this subsection, we use numerical experiments to demonstrate that the objective value achieved by the UBL (Section~\ref{subsec:md-sym}) and TPBL (Section~\ref{subsec:asym-mechanism}) mechanism is close to that of the optimal citation mechanism under symmetric case across a wide range of parameter configurations. This validates the effectiveness of these mechanisms and thus provides practically feasible and effective citation mechanisms.

To quantify the gap between the UBL/TPBL mechanisms and the optimal citation mechanism, we first define $W^*(\bm{p}) = \max_{\bm{c}} W(\bm{p}, \bm{c})$ as the optimal objective value given a position bias vector $\bm{p}$. Let $\bm{p}_{\text{opt}}$ denote the position bias vector induced by the optimal citation mechanism, $\bm{p}_{\text{UBL}}$ and $\bm{p}_{\text{TPBL}}$ those induced by the UBL and TPBL mechanisms, respectively, and $\bm{p}_C$ the position bias vector corresponding to Type C in Table~\ref{tab:position-bias-no-ai-overview}. We define the approximation ratio as
\begin{equation} \label{eq:approx-ratio}
    \rho = \frac{W^*(\bm{p}_{\text{mechanism}}) - W^*(\bm{p}_C)}{W^*(\bm{p}_{\text{opt}}) - W^*(\bm{p}_C)}, \quad \text{mechanism} \in \{\text{UBL}, \text{TPBL}\}.
\end{equation}

This ratio measures the relative improvement of UBL or TPBL over the baseline ($\bm{p}_C$) compared to the improvement achieved by the optimal citation mechanism. A value of $\rho$ closer to 1 indicates that the mechanism performs nearly optimally.

However, computing this approximation ratio poses a challenge. According to Section~\ref{sec:model}, under our experimental configuration ($n = 10$, $m = 4$), directly determining the optimal citation mechanism requires solving an optimization problem with 40 variables subject to 45 constraints. Such a problem is computationally intractable without further structural assumptions. To address this, we utilize the following simplified model: let $p_i^C$ be the position biases of the $i$-th ranked organic search results when they are not cited, $\Delta p_i^C$ be the increase in position bias for the $i$-th ranked organic results due to the citation mechanism, so that $\bm{p}$ satisfies:
\begin{equation} \label{eq:position-bias-structure-simplified}
    \begin{aligned}
        &\qquad \quad p_i = p_i^C + \Delta p_i^C, \quad \forall i \in [n], \\
        &\text{where } \sum_{i = 1}^n \Delta p_i^C = \sum_{i = 1}^m q_i, \\
        &\qquad \quad \Delta p_i^C \in [0, q_1], \quad \forall i \in [n], \\
        &\qquad \quad p_1 \geqslant p_2 \geqslant \dots \geqslant p_n > 0.
    \end{aligned}
\end{equation}
The constraint $\Delta p_i^C \leqslant q_1$ arises because the maximum possible value of $\Delta p_i^C$ occurs when the page is cited with probability $1$ and placed in the first position within the AI Overview. Denote $\mathcal{P}'$ as the feasible set defined in equation~\eqref{eq:position-bias-structure-simplified}, then $\mathcal{P} \subsetneq \mathcal{P}'$: The configuration $\Delta p_1 = 0.6$, $\Delta p_2 = 0.6$, $\Delta p_3 = 0.4407$, and $\Delta p_4 = \cdots = \Delta p_{10} = 0$ is in $\mathcal{P}'$ but is not in $\mathcal{P}$. Therefore, the objective value derived under this simplified model (denoted as $W^*(\tilde{\bm{p}}_{\text{opt}})$) serves as an upper bound on the true optimal value. Consequently, the approximation ratio $\rho$ computed using this approach provides a conservative lower bound on the actual approximation ratio.

We first evaluate the UBL mechanism under symmetric equilibrium conditions. We conduct experiments across various parameter settings, varying $\alpha \in \{1, 3, 5\}$, $h(x) \in \{x, \sqrt{x}, \sqrt[3]{x}\}$, and $\beta \in \{2, 4, 6\}$. We compute $W^*(\bm{p}_C)$, $W^*(\bm{p}_{\text{UBL}})$, and $W^*(\tilde{\bm{p}}_{\text{opt}})$ using a grid search with a step size of $0.1$ for position biases and $0.02$ for compensation values. The resulting approximation ratios are summarized in Table~\ref{tab:approx-ratio-sym}. The results indicate that the UBL mechanism consistently achieves approximation ratios above $0.6$ across all tested configurations, with many instances exceeding $0.8$. Note that these ratios are conservative estimates due to the use of the simplified model for the optimal mechanism, thus the UBL mechanism effectively approximates the optimal citation mechanism in symmetric scenarios.

\begin{table}[t]
\centering
\footnotesize
\setlength{\tabcolsep}{1pt}
\caption{Approximate ratios under symmetric equilibrium}
\begin{tabular}{cc|cccc|cccc|cccc}
\toprule
$\alpha$ & $h$
& \multicolumn{4}{c|}{$\beta=2$}
& \multicolumn{4}{c|}{$\beta=4$}
& \multicolumn{4}{c}{$\beta=6$} \\
& &
$W^*(\bm{p}_C)$ & $W^*(\bm{p}_{\text{UBL}})$ & $W^*(\tilde{\bm{p}}_{\text{opt}})$ & $\rho$
& $W^*(\bm{p}_C)$ & $W^*(\bm{p}_{\text{UBL}})$ & $W^*(\tilde{\bm{p}}_{\text{opt}})$ & $\rho$
& $W^*(\bm{p}_C)$ & $W^*(\bm{p}_{\text{UBL}})$ & $W^*(\tilde{\bm{p}}_{\text{opt}})$ & $\rho$ \\
\midrule
1 & $x$              & 0.885 & 2.225 & 3.054 & 0.618 & 1.932 & 3.503 & 3.812 & 0.836 & 2.321 & 3.922 & 4.105 & 0.897 \\
1 & $\sqrt{x}$       & 1.135 & 2.451 & 3.282 & 0.613 & 2.124 & 3.652 & 3.975 & 0.825 & 2.464 & 4.028 & 4.227 & 0.887 \\
1 & $\sqrt[3]{x}$    & 1.263 & 2.554 & 3.384 & 0.609 & 2.202 & 3.709 & 4.039 & 0.820 & 2.519 & 4.067 & 4.271 & 0.884 \\
3 & $x$              & 4.005 & 9.615 & 12.112 & 0.692 & 5.987 & 10.706 & 11.608 & 0.839 & 6.981 & 11.765 & 12.323 & 0.896 \\
3 & $\sqrt{x}$       & 4.571 & 9.794 & 12.288 & 0.677 & 6.476 & 11.066 & 12.013 & 0.829 & 7.393 & 12.085 & 12.680 & 0.888 \\
3 & $\sqrt[3]{x}$    & 4.816 & 9.878 & 12.371 & 0.670 & 6.677 & 11.207 & 12.175 & 0.824 & 7.556 & 12.202 & 12.814 & 0.884 \\
5 & $x$              & 9.253 & 21.875 & 26.000 & 0.754 & 10.699 & 18.966 & 20.259 & 0.865 & 11.923 & 19.926 & 20.813 & 0.900 \\
5 & $\sqrt{x}$       & 9.617 & 21.006 & 25.266 & 0.728 & 11.322 & 19.321 & 20.693 & 0.854 & 12.486 & 20.324 & 21.276 & 0.892 \\
5 & $\sqrt[3]{x}$    & 9.806 & 20.809 & 25.063 & 0.721 & 11.573 & 19.461 & 20.866 & 0.849 & 12.705 & 20.475 & 21.453 & 0.888 \\
\bottomrule
\end{tabular}
\label{tab:approx-ratio-sym}
\end{table}

Moreover, as $\beta$ increases, creators tend to exert higher effort levels, which in turn diminishes the influence of mechanisms on their behavior. On the one hand, this reduces the differences among various citation mechanisms, leading to an increase in the approximation ratio. On the other hand, while the impact of the compensation mechanism weakens, its associated cost rises, thereby lowering the optimal compensation and lessening the effect of $\alpha$ on the approximation ratio. These trends are consistently reflected in Table~\ref{tab:approx-ratio-sym}.

Then we evaluate the TPBL mechanism under asymmetric equilibrium conditions. We first consider the separated equilibrium case. Specifically, we consider three different values of $n_H$ (i.e., $n_H = 3, 5, 7$) while varying $\alpha \in \{1, 3, 5\}$, $\beta \in \{2, 4, 6\}$ and $\gamma_L \in \{1.5, 2.0, 2.5\}$. We fix $\gamma_H = 1$, and fix $h = \sqrt{x}$ since the impact of $h$ is not significant as observed in symmetric cases. 

The approximate ratios under these configurations are summarized in Tables~\ref{tab:approx-ratio-separated-nH3}, \ref{tab:approx-ratio-separated-nH5} and \ref{tab:approx-ratio-separated-nH7}.  The results indicate that the TPBL mechanism consistently achieves approximation ratios above $0.6$ across all tested configurations, with many scenarios achieving ratios above $0.8$ or even $0.9$. This demonstrates the effectiveness of the TPBL mechanism in enhancing search engine profit in asymmetric settings under separated equilibrium. The effect of $\gamma_L$ on the approximation ratio is not significant, and the effects of $\alpha$ and $\beta$ align with the observations in symmetric cases.

\begin{table}[t]
\centering
\footnotesize
\setlength{\tabcolsep}{1pt}
\caption{Approximate ratios under separated equilibrium with $n_H = 3$}
\begin{tabular}{cc|cccc|cccc|cccc}
\toprule
$\alpha$ & $\gamma_L$
& \multicolumn{4}{c|}{$\beta=2$}
& \multicolumn{4}{c|}{$\beta=4$}
& \multicolumn{4}{c}{$\beta=6$} \\
& &
$W^*(\bm{p}_C)$ & $W^*(\bm{p}_{\text{TPBL}})$ & $W^*(\tilde{\bm{p}}_{\text{opt}})$ & $\rho$
& $W^*(\bm{p}_C)$ & $W^*(\bm{p}_{\text{TPBL}})$ & $W^*(\tilde{\bm{p}}_{\text{opt}})$ & $\rho$
& $W^*(\bm{p}_C)$ & $W^*(\bm{p}_{\text{TPBL}})$ & $W^*(\tilde{\bm{p}}_{\text{opt}})$ & $\rho$ \\
\midrule
1 & 1.5 & 0.973 & 2.871 & 3.267 & 0.827 & 1.950 & 3.636 & 3.743 & 0.940 & 2.334 & 3.975 & 4.035 & 0.965 \\
1 & 2.0 & 0.921 & 2.808 & 3.229 & 0.817 & 1.850 & 3.534 & 3.669 & 0.926 & 2.253 & 3.891 & 3.967 & 0.956 \\
1 & 2.5 & 0.896 & 2.774 & 3.207 & 0.813 & 1.786 & 3.469 & 3.627 & 0.914 & 2.198 & 3.835 & 3.923 & 0.949 \\
3 & 1.5 & 2.985 & 8.844 & 10.406 & 0.789 & 5.899 & 10.969 & 11.265 & 0.945 & 7.002 & 11.924 & 12.104 & 0.965 \\
3 & 2.0 & 2.984 & 8.737 & 11.331 & 0.689 & 5.657 & 10.650 & 11.121 & 0.914 & 6.771 & 11.673 & 11.900 & 0.956 \\
3 & 2.5 & 3.059 & 8.714 & 11.890 & 0.640 & 5.512 & 10.451 & 11.018 & 0.897 & 6.620 & 11.505 & 11.769 & 0.949 \\
5 & 1.5 & 5.022 & 14.829 & 17.180 & 0.807 & 9.895 & 18.355 & 18.847 & 0.945 & 11.744 & 19.974 & 20.239 & 0.969 \\
5 & 2.0 & 4.956 & 14.689 & 17.356 & 0.785 & 9.376 & 17.819 & 18.457 & 0.930 & 11.332 & 19.544 & 19.886 & 0.960 \\
5 & 2.5 & 5.027 & 14.692 & 18.363 & 0.725 & 9.138 & 17.484 & 18.314 & 0.910 & 11.119 & 19.255 & 19.662 & 0.952 \\
\bottomrule
\end{tabular}
\label{tab:approx-ratio-separated-nH3}
\end{table}

\begin{table}[t]
\centering
\footnotesize
\setlength{\tabcolsep}{1pt}
\caption{Approximate ratios under separated equilibrium with $n_H = 5$}
\begin{tabular}{cc|cccc|cccc|cccc}
\toprule
$\alpha$ & $\gamma_L$
& \multicolumn{4}{c|}{$\beta=2$}
& \multicolumn{4}{c|}{$\beta=4$}
& \multicolumn{4}{c}{$\beta=6$} \\
& &
$W^*(\bm{p}_C)$ & $W^*(\bm{p}_{\text{TPBL}})$ & $W^*(\tilde{\bm{p}}_{\text{opt}})$ & $\rho$
& $W^*(\bm{p}_C)$ & $W^*(\bm{p}_{\text{TPBL}})$ & $W^*(\tilde{\bm{p}}_{\text{opt}})$ & $\rho$
& $W^*(\bm{p}_C)$ & $W^*(\bm{p}_{\text{TPBL}})$ & $W^*(\tilde{\bm{p}}_{\text{opt}})$ & $\rho$ \\
\midrule
1 & 1.5 & 1.084 & 2.718 & 3.355 & 0.719 & 2.027 & 3.628 & 3.925 & 0.843 & 2.382 & 3.969 & 4.173 & 0.886 \\
1 & 2.0 & 1.070 & 2.697 & 3.343 & 0.716 & 1.990 & 3.592 & 3.891 & 0.843 & 2.348 & 3.935 & 4.141 & 0.885 \\
1 & 2.5 & 1.062 & 2.685 & 3.335 & 0.714 & 1.969 & 3.567 & 3.870 & 0.840 & 2.326 & 3.913 & 4.120 & 0.884 \\
3 & 1.5 & 4.098 & 10.265 & 13.273 & 0.672 & 6.156 & 10.930 & 11.804 & 0.845 & 7.145 & 11.908 & 12.518 & 0.887 \\
3 & 2.0 & 4.285 & 10.707 & 13.401 & 0.704 & 6.062 & 10.841 & 11.727 & 0.844 & 7.044 & 11.806 & 12.423 & 0.885 \\
3 & 2.5 & 4.331 & 10.872 & 13.391 & 0.722 & 6.000 & 10.774 & 11.664 & 0.843 & 6.978 & 11.739 & 12.361 & 0.884 \\
5 & 1.5 & 6.983 & 17.691 & 24.167 & 0.623 & 10.362 & 18.409 & 19.846 & 0.848 & 11.972 & 19.933 & 20.915 & 0.890 \\
5 & 2.0 & 7.847 & 19.771 & 27.533 & 0.606 & 10.388 & 18.252 & 19.932 & 0.824 & 11.863 & 19.743 & 20.743 & 0.887 \\
5 & 2.5 & 8.365 & 21.321 & 28.755 & 0.635 & 10.355 & 18.247 & 19.989 & 0.819 & 11.767 & 19.619 & 20.677 & 0.881 \\
\bottomrule
\end{tabular}
\label{tab:approx-ratio-separated-nH5}
\end{table}

\begin{table}[t]
\centering
\footnotesize
\setlength{\tabcolsep}{1pt}
\caption{Approximate ratios under separated equilibrium with $n_H = 7$}
\begin{tabular}{cc|cccc|cccc|cccc}
\toprule
$\alpha$ & $\gamma_L$
& \multicolumn{4}{c|}{$\beta=2$}
& \multicolumn{4}{c|}{$\beta=4$}
& \multicolumn{4}{c}{$\beta=6$} \\
& &
$W^*(\bm{p}_C)$ & $W^*(\bm{p}_{\text{TPBL}})$ & $W^*(\tilde{\bm{p}}_{\text{opt}})$ & $\rho$
& $W^*(\bm{p}_C)$ & $W^*(\bm{p}_{\text{TPBL}})$ & $W^*(\tilde{\bm{p}}_{\text{opt}})$ & $\rho$
& $W^*(\bm{p}_C)$ & $W^*(\bm{p}_{\text{TPBL}})$ & $W^*(\tilde{\bm{p}}_{\text{opt}})$ & $\rho$ \\
\midrule
1 & 1.5 & 1.103 & 2.571 & 3.302 & 0.668 & 2.033 & 3.639 & 3.908 & 0.857 & 2.384 & 3.998 & 4.161 & 0.908 \\
1 & 2.0 & 1.098 & 2.560 & 3.298 & 0.665 & 2.020 & 3.615 & 3.895 & 0.851 & 2.365 & 3.975 & 4.143 & 0.906 \\
1 & 2.5 & 1.095 & 2.554 & 3.295 & 0.663 & 2.011 & 3.601 & 3.887 & 0.848 & 2.356 & 3.961 & 4.135 & 0.902 \\
3 & 1.5 & 4.548 & 10.370 & 12.861 & 0.700 & 6.233 & 10.996 & 11.824 & 0.852 & 7.151 & 11.993 & 12.483 & 0.908 \\
3 & 2.0 & 4.534 & 10.349 & 12.850 & 0.699 & 6.197 & 10.930 & 11.789 & 0.846 & 7.104 & 11.926 & 12.428 & 0.906 \\
3 & 2.5 & 4.525 & 10.331 & 12.844 & 0.698 & 6.172 & 10.889 & 11.765 & 0.843 & 7.079 & 11.884 & 12.404 & 0.902 \\
5 & 1.5 & 9.680 & 22.207 & 27.368 & 0.708 & 10.944 & 19.127 & 20.478 & 0.858 & 12.112 & 20.120 & 20.952 & 0.906 \\
5 & 2.0 & 9.702 & 22.514 & 27.352 & 0.726 & 10.890 & 19.056 & 20.433 & 0.856 & 12.043 & 20.011 & 20.886 & 0.901 \\
5 & 2.5 & 9.694 & 22.486 & 27.342 & 0.725 & 10.847 & 18.987 & 20.391 & 0.853 & 12.004 & 19.945 & 20.847 & 0.898 \\
\bottomrule
\end{tabular}
\label{tab:approx-ratio-separated-nH7}
\end{table}

We then consider the hybrid equilibrium case. We again consider three different values of $n_H$ (i.e., $n_H = 3, 5, 7$) while varying $\alpha \in \{1, 3, 5\}$ and $\beta \in \{4, 6, 8\}$, with $h = \sqrt{x}$ and $\gamma_L = 2.01$ fixed. The approximate ratios under these configurations are summarized in Table~\ref{tab:approx-ratio-hybrid}.  Note that the range of $\beta$ differs from earlier experiments, $\gamma_L$ is chosen close to $\gamma_H$, and the case $n_H = 7$ with $\beta = 4$ is omitted—all these choices are intended to ensure the existence of a hybrid equilibrium in the simulations.

\begin{table}[t]
\centering
\footnotesize
\setlength{\tabcolsep}{1pt}
\caption{Approximate ratios under hybrid equilibrium}
\begin{tabular}{cc|cccc|cccc|cccc}
\toprule
$\beta$ & $\alpha$
& \multicolumn{4}{c|}{$n_H = 3$}
& \multicolumn{4}{c|}{$n_H = 5$}
& \multicolumn{4}{c}{$n_H = 7$} \\
& &
$W^*(\bm{p}_C)$ & $W^*(\bm{p}_{\text{TPBL}})$ & $W^*(\tilde{\bm{p}}_{\text{opt}})$ & $\rho$
& $W^*(\bm{p}_C)$ & $W^*(\bm{p}_{\text{TPBL}})$ & $W^*(\tilde{\bm{p}}_{\text{opt}})$ & $\rho$
& $W^*(\bm{p}_C)$ & $W^*(\bm{p}_{\text{TPBL}})$ & $W^*(\tilde{\bm{p}}_{\text{opt}})$ & $\rho$ \\
\midrule
4 & 1 & 1.578 & 2.932 & 3.066 & 0.910 & 1.576 & 2.886 & 3.062 & 0.882 & --- & --- & --- & --- \\
4 & 3 & 4.733 & 8.797 & 9.197 & 0.910 & 4.729 & 8.657 & 9.187 & 0.881 & --- & --- & --- & --- \\
4 & 5 & 7.888 & 14.661 & 15.328 & 0.910 & 7.881 & 14.428 & 15.311 & 0.881 & --- & --- & --- & --- \\
6 & 1 & 2.060 & 3.510 & 3.609 & 0.936 & 2.059 & 3.481 & 3.606 & 0.919 & 2.096 & 3.501 & 3.642 & 0.909 \\
6 & 3 & 6.180 & 10.529 & 10.827 & 0.936 & 6.177 & 10.443 & 10.818 & 0.919 & 6.288 & 10.504 & 10.925 & 0.909 \\
6 & 5 & 10.300 & 17.549 & 18.046 & 0.936 & 10.295 & 17.405 & 18.030 & 0.919 & 10.480 & 17.507 & 18.208 & 0.909 \\
8 & 1 & 2.319 & 3.812 & 3.888 & 0.952 & 2.318 & 3.792 & 3.886 & 0.940 & 2.349 & 3.777 & 3.916 & 0.911 \\
8 & 3 & 6.956 & 11.436 & 11.665 & 0.951 & 6.953 & 11.375 & 11.658 & 0.940 & 7.047 & 11.331 & 11.747 & 0.912 \\
8 & 5 & 11.593 & 19.060 & 19.442 & 0.951 & 11.589 & 18.959 & 19.430 & 0.940 & 11.745 & 18.884 & 19.578 & 0.911 \\
\bottomrule
\end{tabular}
\label{tab:approx-ratio-hybrid}
\end{table}

\subsection{Experiment Results on the Optimal \texorpdfstring{$\bm{c_H}$}{c\_H} and \texorpdfstring{$\bm{c_L}$}{c\_L}} \label{subsec:ap-exp-cHcL}

In this subsection, we present the property $c_H > c_L$ is valid across various parameter configurations. We fix $\gamma_H = 1$, $h = \sqrt{x}$ and $\alpha = 3$, and consider three different values of $n_H$ (i.e., $n_H = 3, 5, 7$) while varying $\beta \in \{2, 4\}$ and $\gamma_L \in \{1.5, 2.0, 2.5\}$. Due to the step size of grid search cannot be too small in practical implementation, we fix $\alpha = 3$ and do not take $\beta = 6$ to avoid the case where the $c_H$ and $c_L$ are too small to be distinguished. For each parameter configuration, we take the position bias vector $\tilde{\bm{p}}_{\text{opt}}$. The optimal $c_H$ and $c_L$ under these configurations are summarized in Tables~\ref{tab:optimal-cHcL}.

\begin{table}[t]
\centering
\small
\setlength{\tabcolsep}{5pt}
\caption{Optimal $c_H$ and $c_L$ under various parameter configurations}
\begin{tabular}{cc|cc|cc|cc}
\toprule
$\beta$ & $\gamma_L$
& \multicolumn{2}{c|}{$n_H = 3$}
& \multicolumn{2}{c|}{$n_H = 5$}
& \multicolumn{2}{c}{$n_H = 7$} \\
& &
$c_H$ & $c_L$
& $c_H$ & $c_L$
& $c_H$ & $c_L$ \\
\midrule
2 & 1.5 & 0.12 & 0.12 & 0.72 & 0.08 & 0.72 & 0.00 \\
4 & 1.5 & 0.02 & 0.02 & 0.02 & 0.02 & 0.10 & 0.00 \\
2 & 2.0 & 0.36 & 0.12 & 0.88 & 0.06 & 0.72 & 0.00 \\
4 & 2.0 & 0.14 & 0.02 & 0.10 & 0.02 & 0.10 & 0.00 \\
2 & 2.5 & 0.54 & 0.16 & 0.92 & 0.04 & 0.72 & 0.00 \\
4 & 2.5 & 0.20 & 0.02 & 0.12 & 0.00 & 0.10 & 0.00 \\
\bottomrule
\end{tabular}
\label{tab:optimal-cHcL}
\end{table}

The results indicate that in all configurations except for some cases with $\gamma_L = 1.5$, we observe $c_H > c_L$ (in fact, if a finer grid search were employed, the case $\gamma_L = 1.5$ might also satisfy $c_H > c_L$). Moreover, as $\gamma_L$ increases, the gap between $c_H$ and $c_L$ generally widens, which is intuitive since a higher $\gamma_L$ implies that low-type creators are less sensitive to the increase of position bias.


\subsection{Additional Results of Search Engine Profit Analysis} \label{subsec:ap-exp-profit}

\begin{figure}[t]
    \centering
    \begin{subfigure}[b]{0.36\linewidth}
        \centering
        \includegraphics[width=\linewidth]{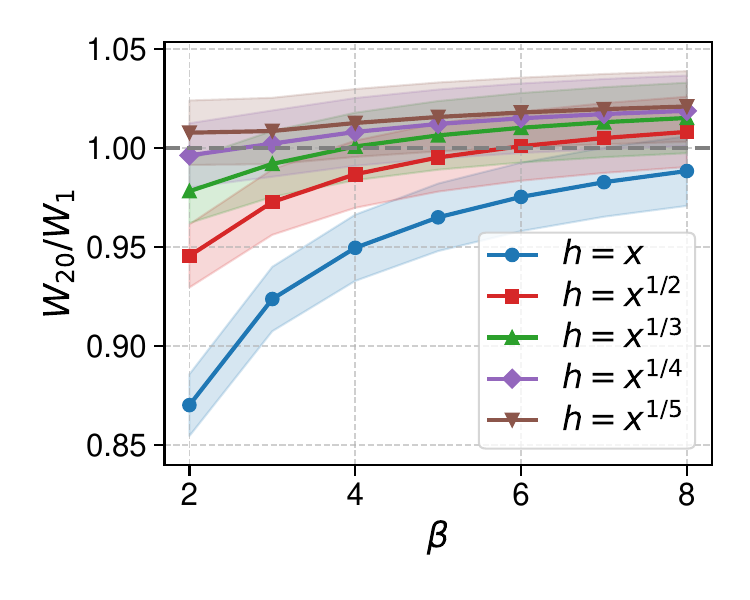}
        \caption{Short term effect}
        \label{fig:short-sym}
    \end{subfigure}
    \begin{subfigure}[b]{0.36\linewidth}
        \centering
        \includegraphics[width=\linewidth]{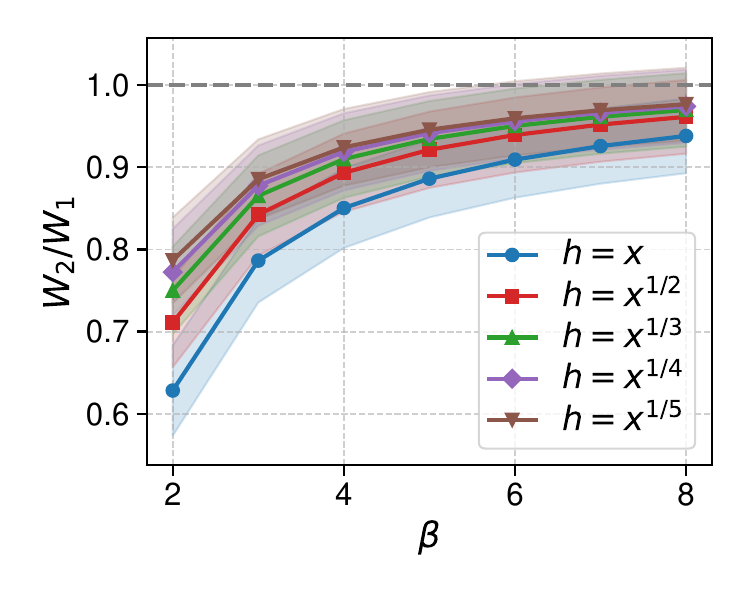}
        \caption{Long term effect}
        \label{fig:long-sym}
    \end{subfigure}
    \caption{Short- and long-term effects of AI Overview in symmetric case}
    \label{fig:short-long-sym}
\end{figure}

In this subsection, we present additional results of search engine profit in Section~\ref{subsec:experiment-profit}. We first demonstrate the results for symmetric case. Compared to Section~\ref{subsec:experiment-profit}, we also examine the short-term effect of introducing the AI Overview without mechanism design (denoted as $W_{20} / W_1$), which demonstrates the change in search engine profit immediately after the introduction of the AI Overview, before creators have adjusted their behavior in response to the new environment. We sweep over $\beta \in \{2, 3, 4, 5, 6, 7, 8\}$ and $h(x) \in \{x, \sqrt{x}, \sqrt[3]{x}, \sqrt[4]{x}, \sqrt[5]{x}\}$. Figure~\ref{fig:short-long-sym} plots $W_{20} / W_1$ and $W_2 / W_1$ (defined as in Section~\ref{subsec:experiment-profit}) across the tested parameters. The observed patterns align with those reported in Section~\ref{subsec:experiment-profit}: a more capable AI Overview (reflected by a more concave $h$) and a larger $\beta$ yields a higher $W_{20} / W_1$ and $W_2 / W_1$, and $W_2 / W_1$ remains below 1 under all parameter configurations in this setting, indicating that introducing an AI Overview without mechanism design would lower long-term search engine profit.

We then examine long-term search engine profit when both the citation mechanism (UBL) and the compensation mechanism are introduced. We sweep over $\beta \in \{2, 3, 4, 5, 6, 7, 8\}$, $\alpha \in \{1, 3, 5, 7, 9\}$, and $h(x) \in \{x, \sqrt[3]{x}, \sqrt[5]{x}\}$. Figure~\ref{fig:UBL-fig-ratio} plots the ratio $W_3 / (\alpha W_1)$ across the tested parameters. The observed patterns align with those reported in Section~\ref{subsec:experiment-profit}: after introducing the mechanisms, long-term search engine profit improves for all parameter settings, and $W_3 / (\alpha W_1) > 1$ is achieved under most configurations. $h$ has a limited effect on the overall trend, higher values of $\alpha$ correlate with increased $W_3 / (\alpha W_1)$, and $W_3 / (\alpha W_1)$ remains close to $W_2 / W_1$ for large $\beta$.

\begin{figure}[t]
    \centering
    \begin{subfigure}[b]{0.32\textwidth}
        \centering
        \includegraphics[width=\linewidth]{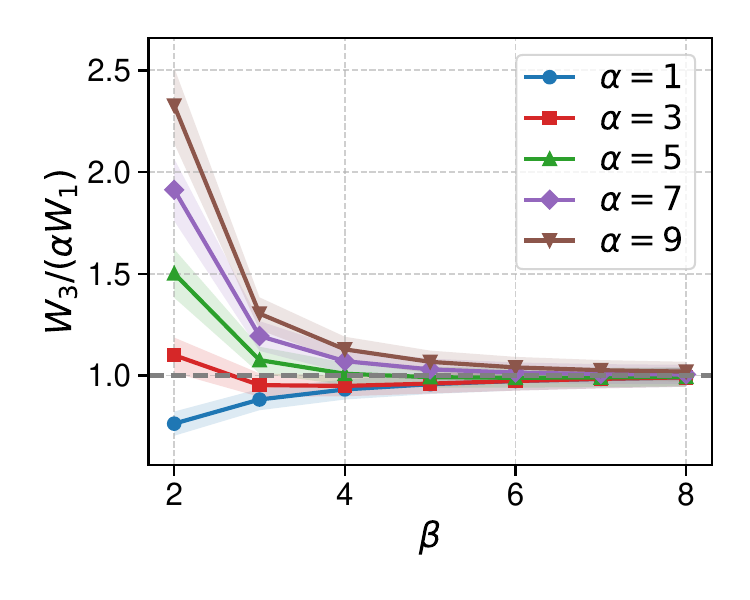}
        \caption{$h = x$}
        \label{fig:UBL-fig-1}
    \end{subfigure}
    \begin{subfigure}[b]{0.32\textwidth}
        \centering
        \includegraphics[width=\linewidth]{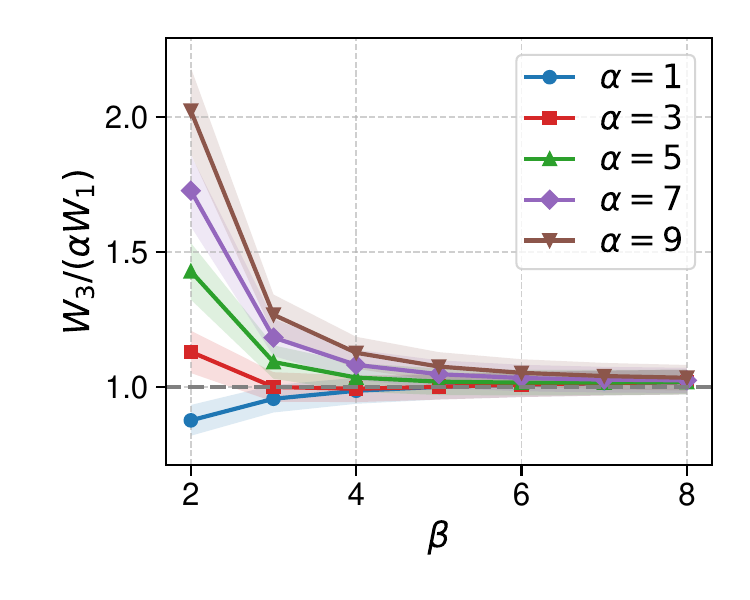}
        \caption{$h = \sqrt[3]{x}$}
        \label{fig:UBL-fig-1:3}
    \end{subfigure}
    \begin{subfigure}[b]{0.32\textwidth}
        \centering
        \includegraphics[width=\linewidth]{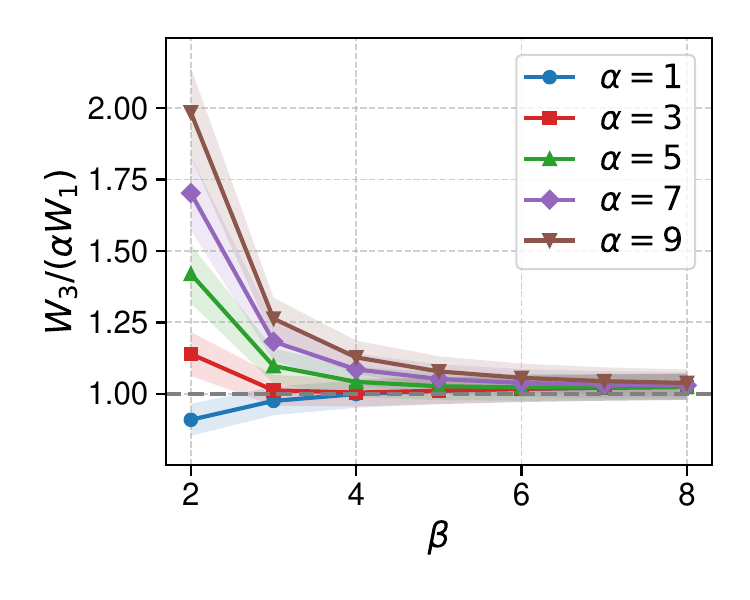}
        \caption{$h = \sqrt[5]{x}$}
        \label{fig:UBL-fig-1:5}
    \end{subfigure}
    \caption{Long-term search engine profit under mechanism design in symmetric case}
    \label{fig:UBL-fig-ratio}
\end{figure}

We then demonstrate the results for asymmetric case. We first consider the separated equilibrium case. For the short-term effect of introducing the AI Overview without mechanism design, we consider $n_H \in \{3, 5, 7, 9\}$ (we do not consider $n_H = 1$ as the equilibrium is always hybrid in this case) while varying $\beta \in \{2, 3, 4, 5, 6, 7, 8\}$, $\gamma_L \in \{1.5, 2.0, 2.5\}$ (we fix $\gamma_H = 1$), and $h(x) \in \{x, \sqrt{x}, \sqrt[3]{x}, \sqrt[4]{x}, \sqrt[5]{x}\}$. Figures~\ref{fig:short-asym-m5-ratio-c2-1p5}–\ref{fig:short-asym-m5-ratio-c2-2p5} report results for $n_H = 5$ across varying values of $\gamma_L$; $\gamma_L$ exhibits negligible influence on the outcome. Consequently, for $n_H = 3$ and $n_H = 9$ (the result for $\gamma = 7$ is in the main text), we present only results with $\gamma_L = 2.0$ (see Figures~\ref{fig:short-asym-m3-ratio-c2-2p0} and \ref{fig:short-asym-m9-ratio-c2-2p0}).

In general, the observed patterns align with those reported in Section~\ref{subsec:experiment-profit}, although in a few cases where $h$ is highly concave, the ratio first declines and then rises. This non-monotonic behavior arises because, when $h$ is strongly concave, the increment in the $h$-dependent term within $W_{20}$ due to a larger $\beta$ is relatively modest, whereas the linear term in $W_1$ grows more substantially. Finally, the overarching trend persists across different values of $n_H$.

\begin{figure}[t]
    \centering
    \begin{subfigure}[b]{0.32\textwidth}
        \centering
        \includegraphics[width=\linewidth]{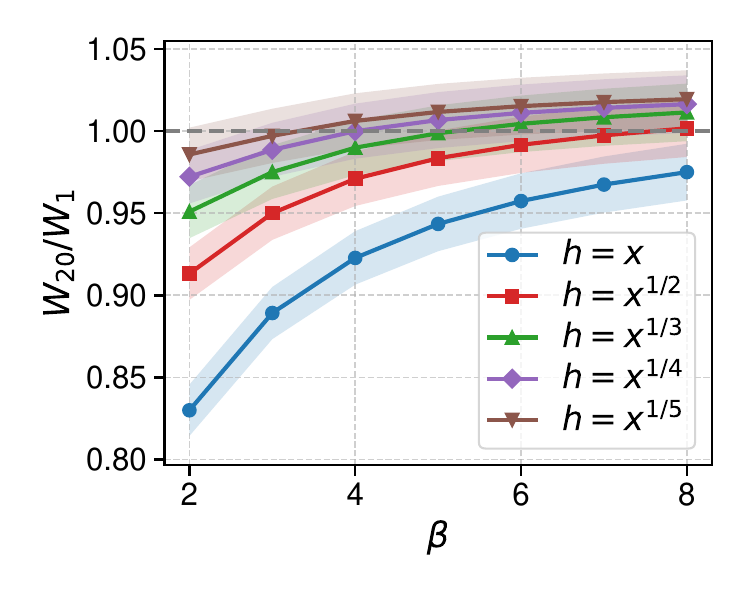}
        \caption{$n_H = 5$, $\gamma_L = 1.5$}
        \label{fig:short-asym-m5-ratio-c2-1p5}
    \end{subfigure}\hfill
    \begin{subfigure}[b]{0.32\textwidth}
        \centering
        \includegraphics[width=\linewidth]{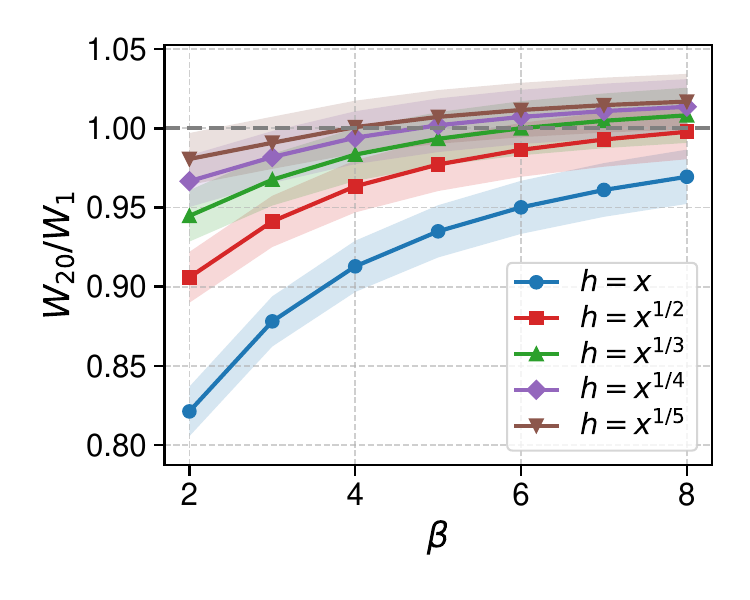}
        \caption{$n_H = 5$, $\gamma_L = 2.0$}
        \label{fig:short-asym-m5-ratio-c2-2p0}
    \end{subfigure}\hfill
    \begin{subfigure}[b]{0.32\textwidth}
        \centering
        \includegraphics[width=\linewidth]{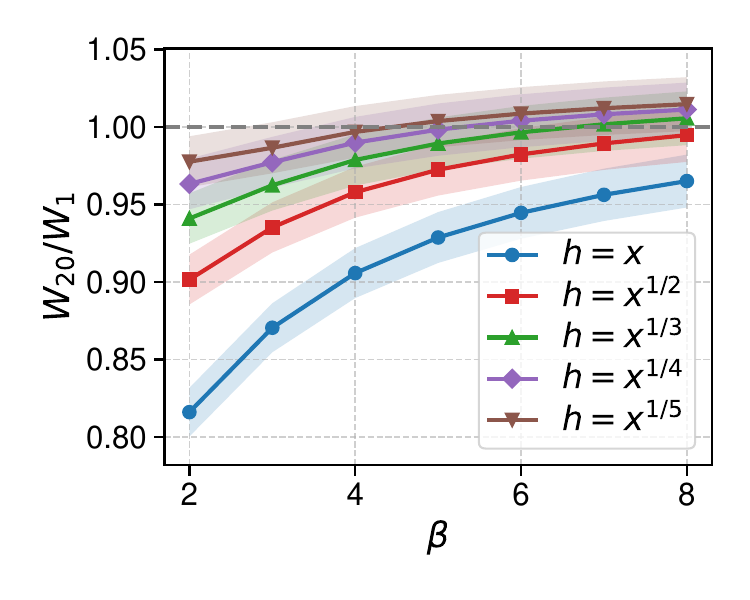}
        \caption{$n_H = 5$, $\gamma_L = 2.5$}
        \label{fig:short-asym-m5-ratio-c2-2p5}
    \end{subfigure}

    \vspace{0.8ex}

    \begin{subfigure}[b]{0.32\textwidth}
        \centering
        \includegraphics[width=\linewidth]{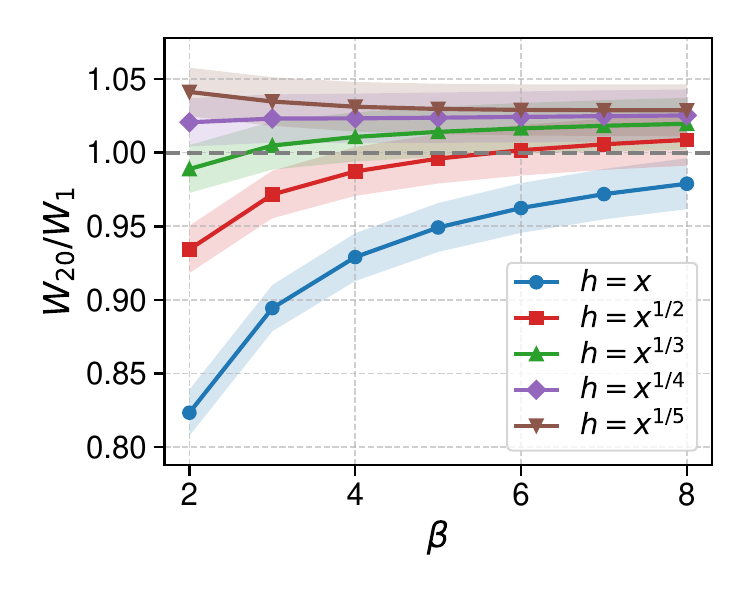}
        \caption{$n_H = 3$, $\gamma_L = 2.0$}
        \label{fig:short-asym-m3-ratio-c2-2p0}
    \end{subfigure}
    \begin{subfigure}[b]{0.32\textwidth}
        \centering
        \includegraphics[width=\linewidth]{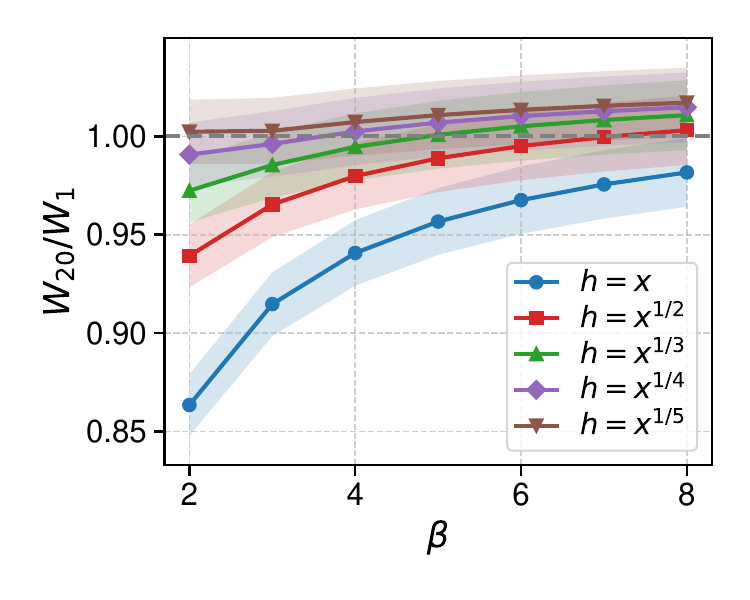}
        \caption{$n_H = 9$, $\gamma_L = 2.0$}
        \label{fig:short-asym-m9-ratio-c2-2p0}
    \end{subfigure}

    \caption{Short-term effect of AI Overview under separated equilibrium}
    \label{fig:short-term-profit-ratio-separated-asymmetric}
\end{figure}

Next, we analyze the long-term effect of introducing the AI Overview without mechanism design. We fix $\gamma_H = 1, \gamma_L = 2$, and vary $n_H \in \{3, 5, 7, 9\}$, $\beta \in \{2, 3, 4, 5, 6, 7, 8\}$, and $h(x) \in \{x, \sqrt{x}, \sqrt[3]{x}, \sqrt[4]{x}, \sqrt[5]{x}\}$. Figure~\ref{fig:long-term-profit-ratio-separated-asymmetric} illustrates the results. Consistent with the results in Section~\ref{subsec:experiment-profit}, $W_2 / W_1$ remains below 1 under all parameter configurations in this setting, the effects of $h$ and $\beta$ on $W_2 / W_1$ follow trends similar to those observed in the short-term analysis, and varying $n_H$ does not substantially alter these patterns. Moreover, the overarching trend persists across different values of $n_H$.

\begin{figure}[t]
    \centering
    \captionsetup[subfigure]{justification=centering}
    \begin{subfigure}[b]{0.32\linewidth}
        \centering
        \includegraphics[width=\linewidth]{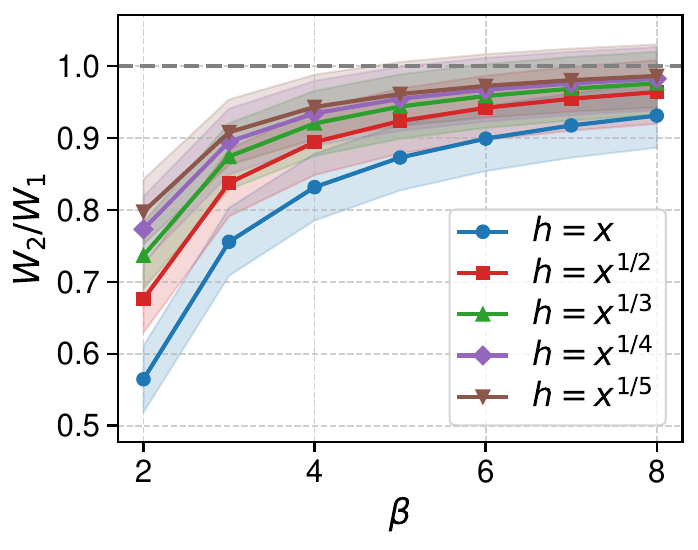}
        \caption{$n_H = 3$}
        \label{fig:long-asym-m3-ratio-c2-2p0}
    \end{subfigure}
    \begin{subfigure}[b]{0.32\linewidth}
        \centering
        \includegraphics[width=\linewidth]{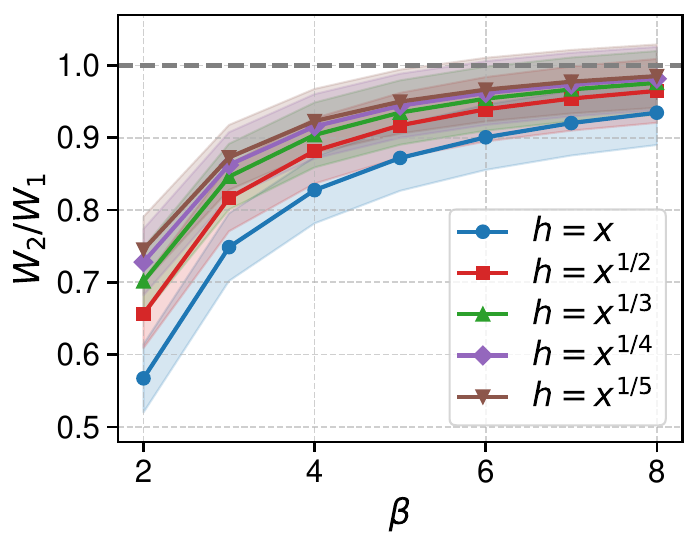}
        \caption{$n_H = 5$}
        \label{fig:long-asym-m5-ratio-c2-2p0}
    \end{subfigure}
    \begin{subfigure}[b]{0.32\linewidth}
        \centering
        \includegraphics[width=\linewidth]{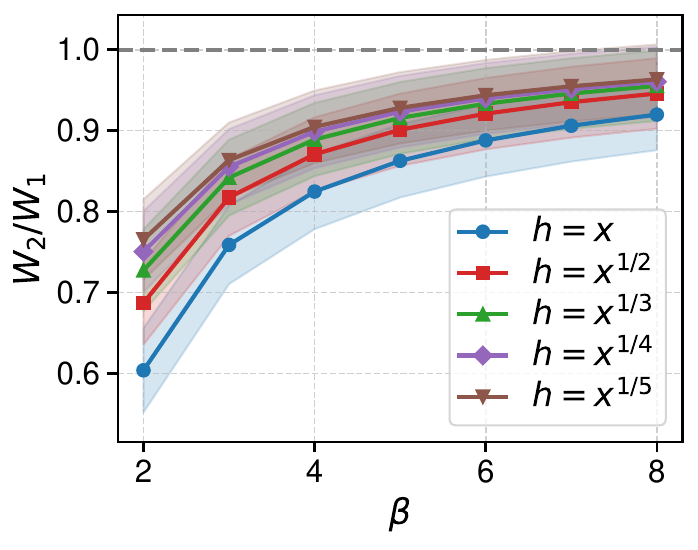}
        \caption{$n_H = 9$}
        \label{fig:long-asym-m9-ratio-c2-2p0}
    \end{subfigure}
    \caption{Long-term effect of AI Overview under separated equilibrium}
    \label{fig:long-term-profit-ratio-separated-asymmetric}
\end{figure}

We then analyze long-term search engine profit when both the citation mechanism (TPBL) and the compensation mechanism are introduced. We fix $\gamma_H = 1, \gamma_L = 2$, and vary $n_H \in \{3, 5, 7, 9\}$, $\beta \in \{2, 3, 4, 5, 6, 7, 8\}$, $h(x) \in \{x, \sqrt[3]{x}, \sqrt[5]{x}\}$, and $\alpha \in \{1, 3, 5, 7, 9\}$. The results are shown in Figure~\ref{fig:long-term-profit-tpbl-separated-asymmetric}. In addition to the general trends described in Section~\ref{subsec:experiment-profit}, we observe that the profit gain from mechanism design becomes more pronounced as $n_H$ increases (tested for $n_H = 3, 5, 7, 9, 10$; the case $n_H = 10$ corresponds to the symmetric setting). This highlights the importance of providing stronger incentives to high‑ability creators.

\begin{figure}[tb]
    \centering
    \begin{subfigure}[b]{0.32\textwidth}
        \centering
        \includegraphics[width=\linewidth]{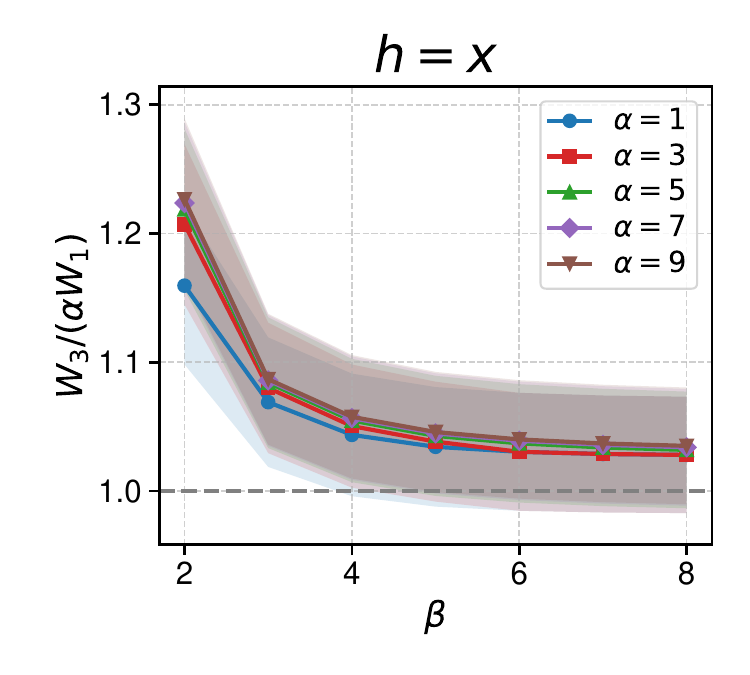}
        \caption{$n_H = 3, h = x$}
        \label{fig:TPBL-m3-h-alpha-c2-2p0}
    \end{subfigure}\hfill
    \begin{subfigure}[b]{0.32\textwidth}
        \centering
        \includegraphics[width=\linewidth]{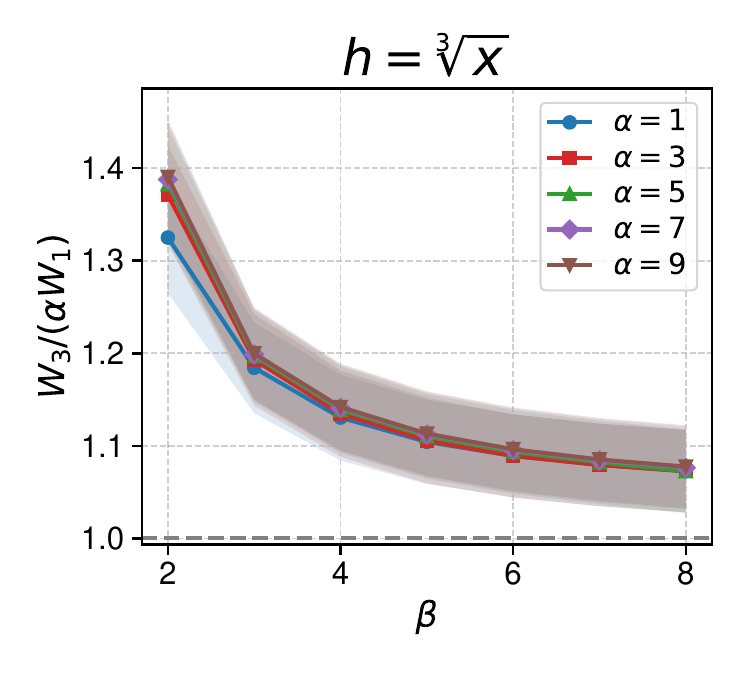}
        \caption{$n_H = 3, h = \sqrt[3]{x}$}
        \label{fig:TPBL-m3-h-alpha-1-3-c2-2p0}
    \end{subfigure}\hfill
    \begin{subfigure}[b]{0.32\textwidth}
        \centering
        \includegraphics[width=\linewidth]{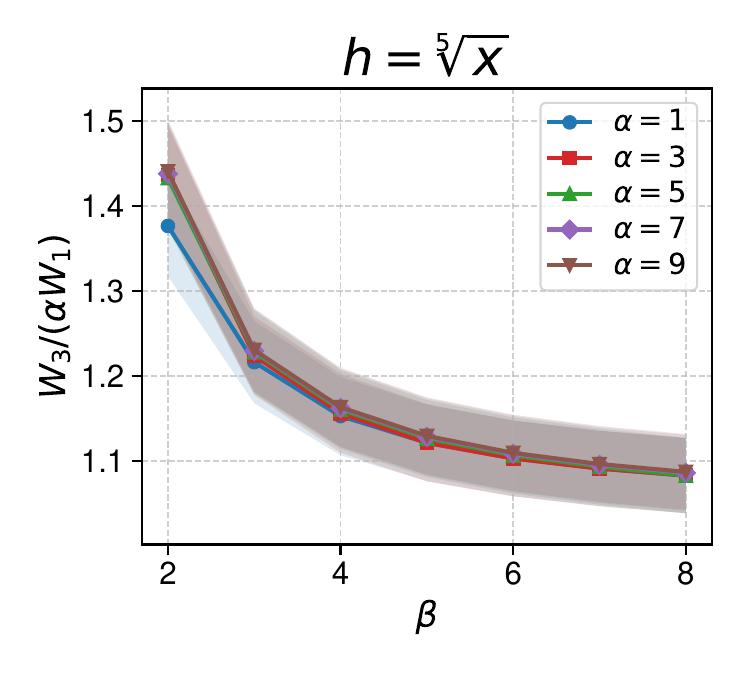}
        \caption{$n_H = 3, h = \sqrt[5]{x}$}
        \label{fig:TPBL-m3-h-alpha-1-5-c2-2p0}
    \end{subfigure}

    \vspace{0.8ex}
    
    \begin{subfigure}[b]{0.32\textwidth}
        \centering
        \includegraphics[width=\linewidth]{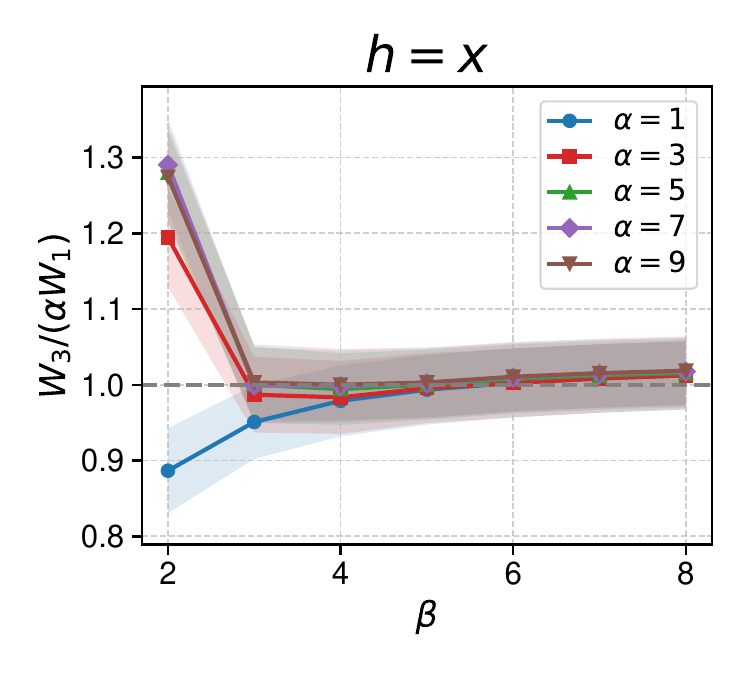}
        \caption{$n_H = 5, h = x$}
        \label{fig:TPBL-m5-h-alpha-c2-2p0}
    \end{subfigure}\hfill
    \begin{subfigure}[b]{0.32\textwidth}
        \centering
        \includegraphics[width=\linewidth]{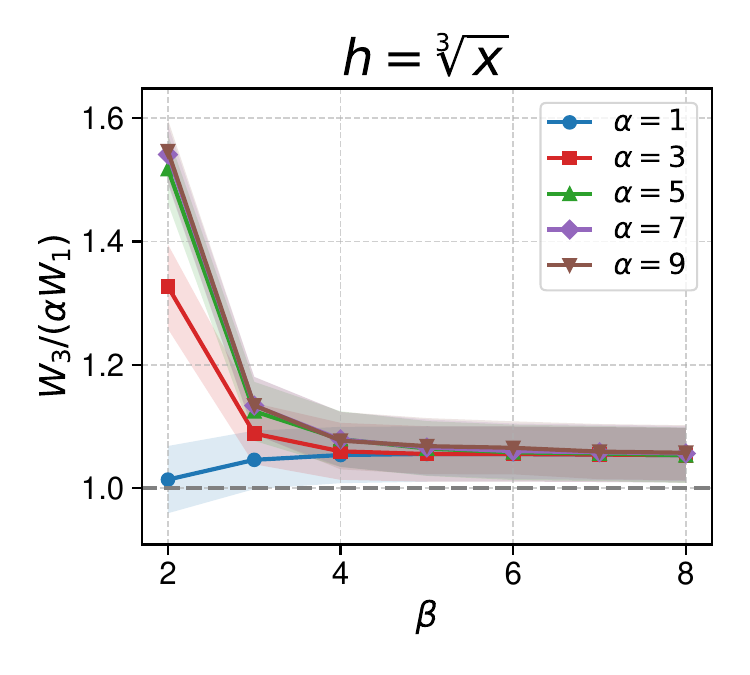}
        \caption{$n_H = 5, h = \sqrt[3]{x}$}
        \label{fig:TPBL-m5-h-alpha-1-3-c2-2p0}
    \end{subfigure}\hfill
    \begin{subfigure}[b]{0.32\textwidth}
        \centering
        \includegraphics[width=\linewidth]{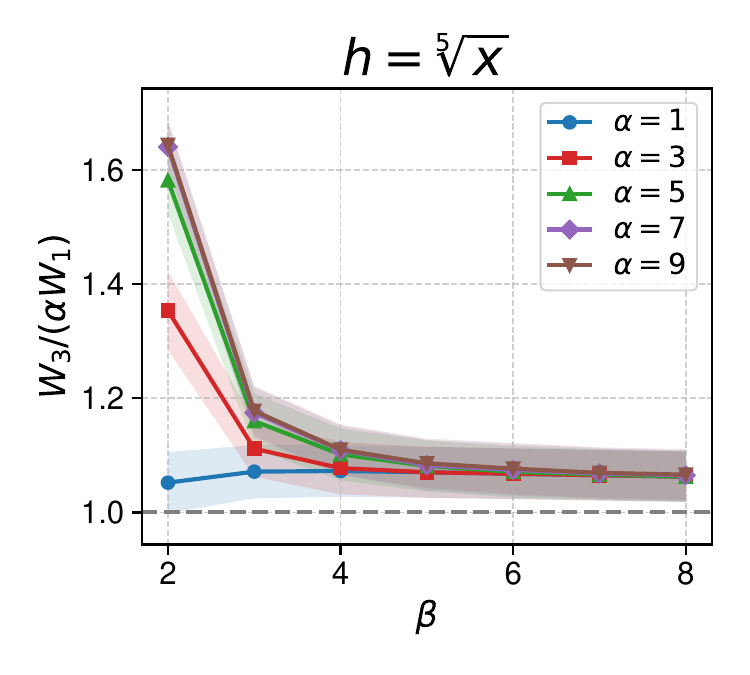}
        \caption{$n_H = 5, h = \sqrt[5]{x}$}
        \label{fig:TPBL-m5-h-alpha-1-5-c2-2p0}
    \end{subfigure}

    \vspace{0.8ex}

    \begin{subfigure}[b]{0.32\textwidth}
        \centering
        \includegraphics[width=\linewidth]{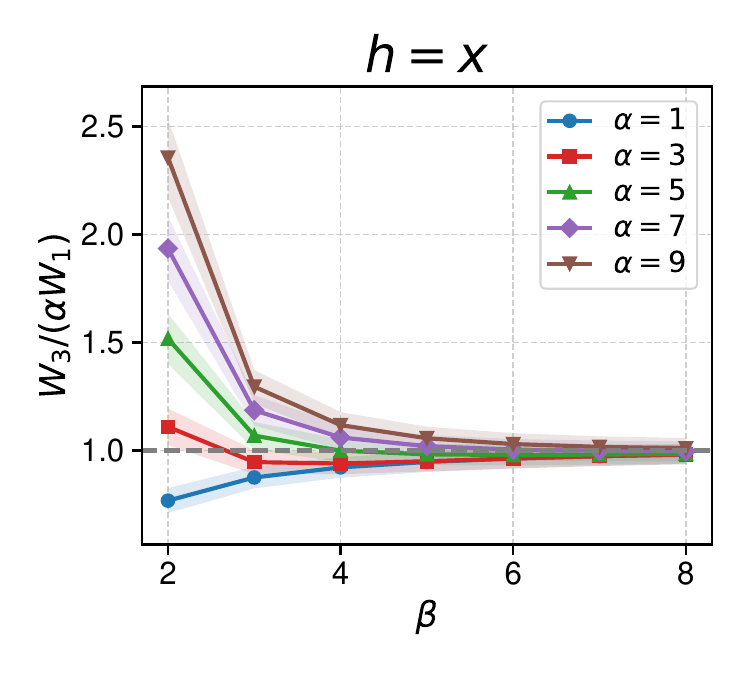}
        \caption{$n_H = 9, h = x$}
        \label{fig:TPBL-m9-h-alpha-c2-2p0}
    \end{subfigure}\hfill
    \begin{subfigure}[b]{0.32\textwidth}
        \centering
        \includegraphics[width=\linewidth]{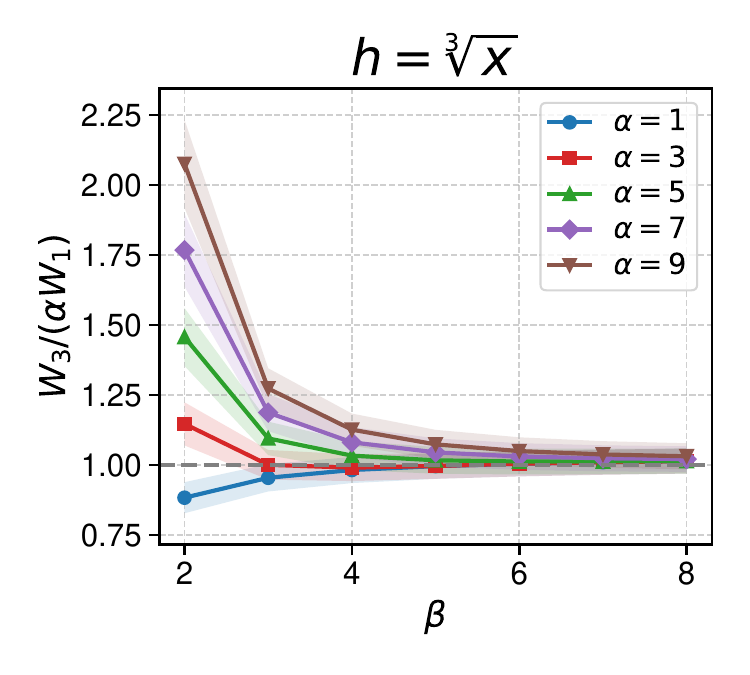}
        \caption{$n_H = 9, h = \sqrt[3]{x}$}
        \label{fig:TPBL-m9-h-alpha-1-3-c2-2p0}
    \end{subfigure}\hfill
    \begin{subfigure}[b]{0.32\textwidth}
        \centering
        \includegraphics[width=\linewidth]{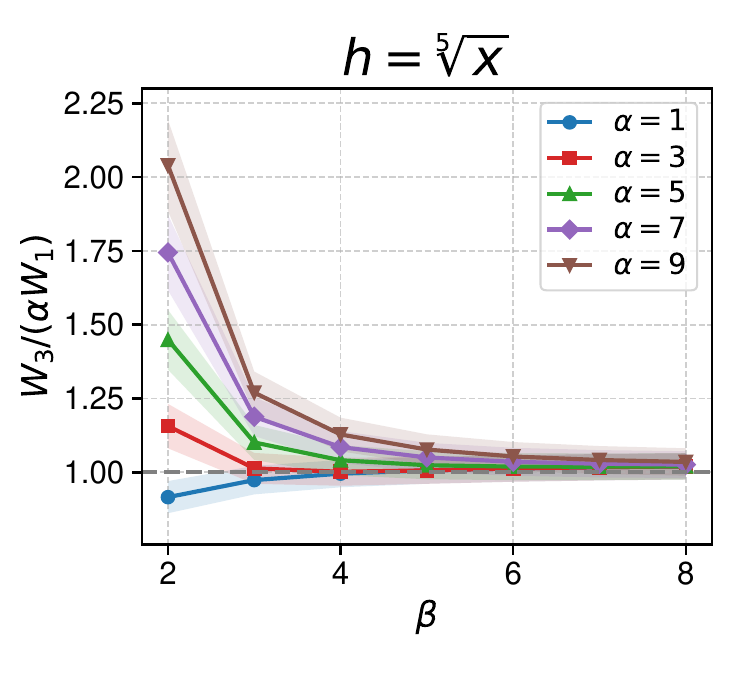}
        \caption{$n_H = 9, h = \sqrt[5]{x}$}
        \label{fig:TPBL-m9-h-alpha-1-5-c2-2p0}
    \end{subfigure}

    \caption{Long-term search engine profit under mechanism design with separated equilibrium}
    \label{fig:long-term-profit-tpbl-separated-asymmetric}
\end{figure}

Then we analyze the hybrid equilibrium case. We first examine the short-term and long-term effect of introducing the AI Overview without mechanism design. Specifically, we consider $n_H \in \{1, 3, 5, 7\}$ (we do not consider $n_H = 9$ as the equilibrium is usually separated in this case) while varying $\beta \in \{4, 5, 6, 7, 8\}$, and $h(x) \in \{x, \sqrt{x}, \sqrt[3]{x}, \sqrt[4]{x}, \sqrt[5]{x}\}$. To ensure the existence of hybrid equilibrium, we fix $\gamma_H = 2$ and $\gamma_L = 2.01$ for $n_H \in \{1, 3, 5\}$, and $\gamma_L = 2.001$ for $n_H = 7$. Figures~\ref{fig:short-term-profit-ratio-hybrid-asymmetric} and Figure~\ref{fig:long-term-profit-ratio-hybrid-asymmetric} illustrate the ratio $W_{20} / W_1$ and $W_2 / W_1$ under the specified configurations, respectively. Overall, the observed patterns align with those reported in Section~\ref{subsec:experiment-profit} and the separated equilibrium case in this subsection, thus we omit further discussion for brevity.

\begin{figure}[tb]
    \centering
    \captionsetup[subfigure]{justification=centering}
    \begin{subfigure}[b]{0.24\linewidth}
        \centering
        \includegraphics[width=\linewidth]{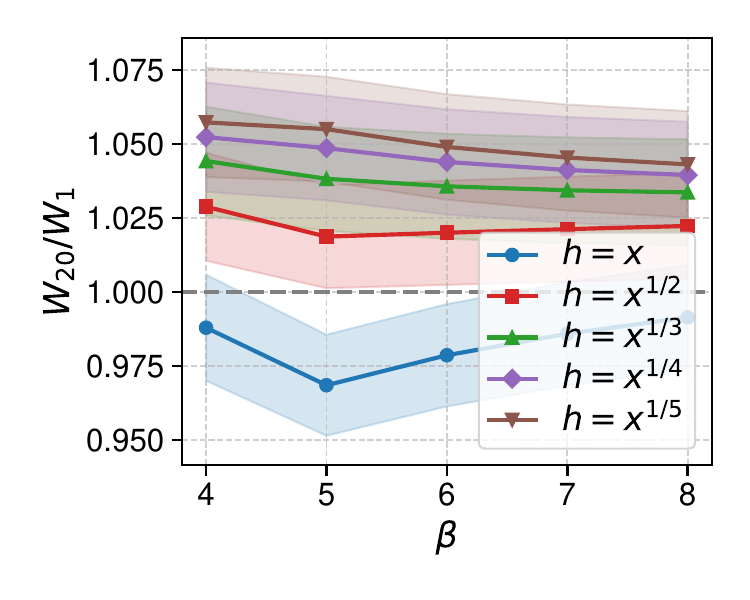}
        \caption{$n_H = 1$}
        \label{fig:short-hybrid-m1-ratio-c2-2p01}
    \end{subfigure}
    \begin{subfigure}[b]{0.24\linewidth}
        \centering
        \includegraphics[width=\linewidth]{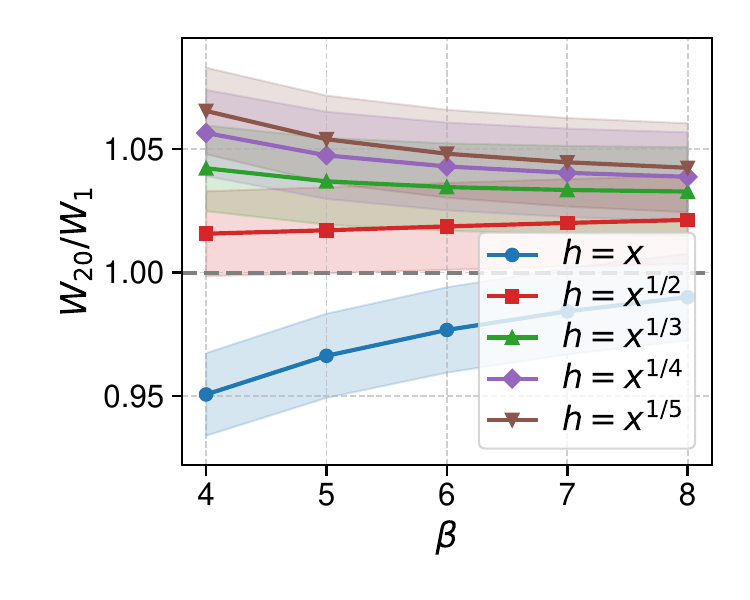}
        \caption{$n_H = 3$}
        \label{fig:short-hybrid-m3-ratio-c2-2p01}
    \end{subfigure}
    \begin{subfigure}[b]{0.24\linewidth}
        \centering
        \includegraphics[width=\linewidth]{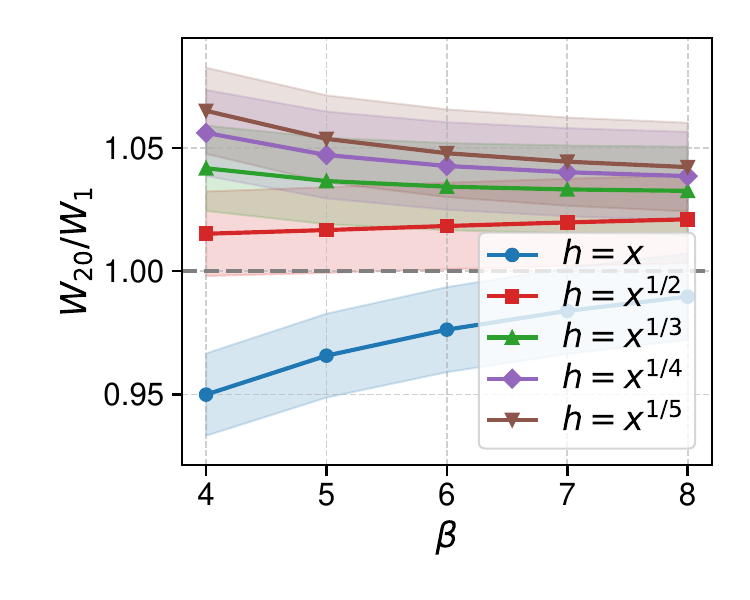}
        \caption{$n_H = 5$}
        \label{fig:short-hybrid-m5-ratio-c2-2p01}
    \end{subfigure}
    \begin{subfigure}[b]{0.24\linewidth}
        \centering
        \includegraphics[width=\linewidth]{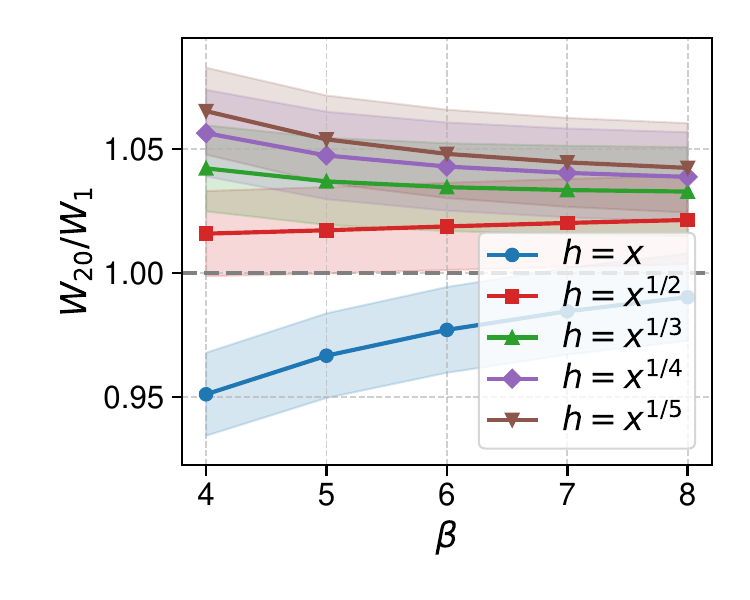}
        \caption{$n_H = 7$}
        \label{fig:short-hybrid-m7-ratio-c2-2p001}
    \end{subfigure}
    \caption{Short-term effect of AI Overview under hybrid equilibrium}
    \label{fig:short-term-profit-ratio-hybrid-asymmetric}
\end{figure}

\begin{figure}[tb]
    \centering
    \captionsetup[subfigure]{justification=centering}
    \begin{subfigure}[b]{0.24\linewidth}
        \centering
        \includegraphics[width=\linewidth]{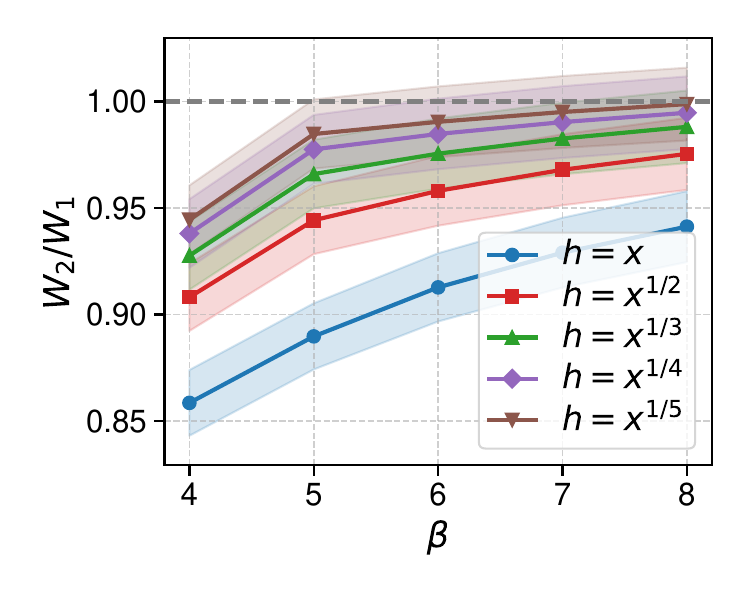}
        \caption{$n_H = 1$}
        \label{fig:long-hybrid-m1-ratio-c2-2p01}
    \end{subfigure}
    \begin{subfigure}[b]{0.24\linewidth}
        \centering
        \includegraphics[width=\linewidth]{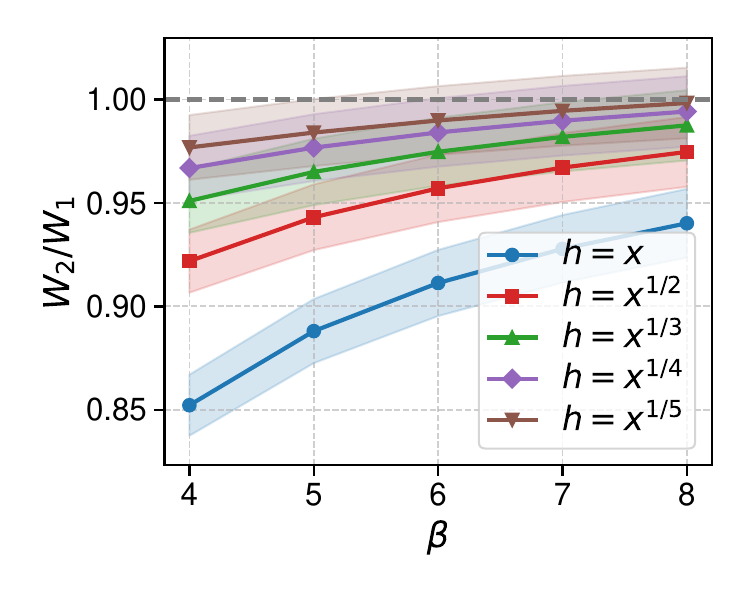}
        \caption{$n_H = 3$}
        \label{fig:long-hybrid-m3-ratio-c2-2p01}
    \end{subfigure}
    \begin{subfigure}[b]{0.24\linewidth}
        \centering
        \includegraphics[width=\linewidth]{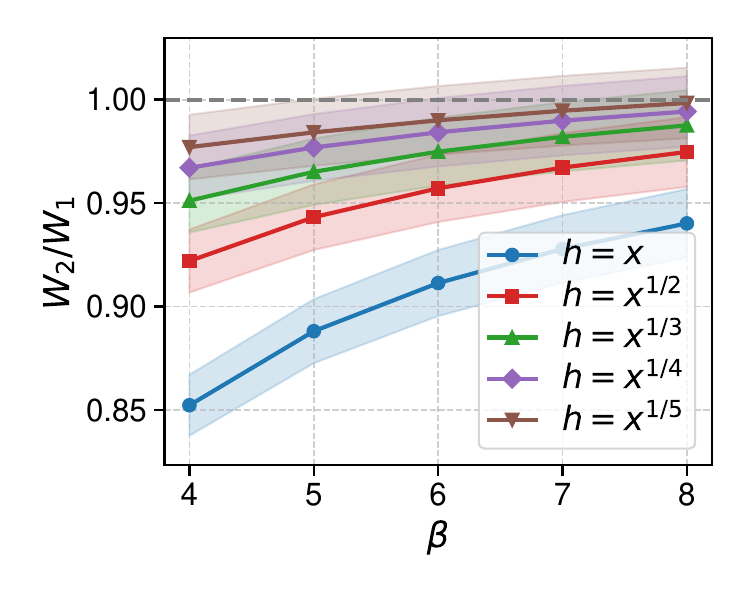}
        \caption{$n_H = 5$}
        \label{fig:long-hybrid-m5-ratio-c2-2p01}
    \end{subfigure}
    \begin{subfigure}[b]{0.24\linewidth}
        \centering
        \includegraphics[width=\linewidth]{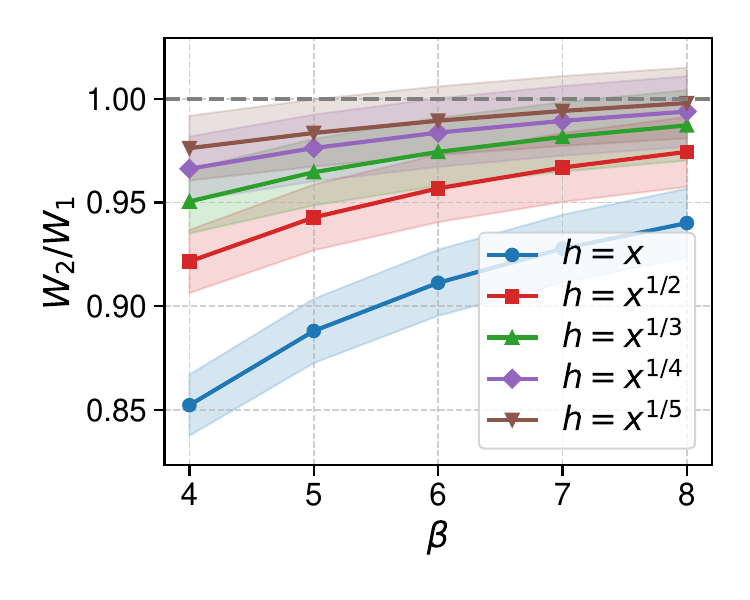}
        \caption{$n_H = 7$}
        \label{fig:long-hybrid-m7-ratio-c2-2p001}
    \end{subfigure}
    \caption{Long-term effect of AI Overview under hybrid equilibrium}
    \label{fig:long-term-profit-ratio-hybrid-asymmetric}
\end{figure}

Finally, we analyze long-term search engine profit when both the citation mechanism (TPBL) and the compensation mechanism are introduced. We fix $\gamma_H = 2, \gamma_L = 2.01$ for $n_H \in \{1, 3, 5\}$, and $\gamma_L = 2.001$ for $n_H = 7$. We then vary $\beta \in \{4, 5, 6, 7, 8\}$, $h(x) \in \{x, \sqrt[3]{x}, \sqrt[5]{x}\}$, and $\alpha \in \{1, 3, 5, 7, 9\}$. The results are shown in Figure~\ref{fig:long-term-profit-tpbl-hybrid-asymmetric}. In fact, the overall patterns align with those reported in Section~\ref{subsec:experiment-profit} and the separated equilibrium case in this subsection, thus we omit further discussion for brevity.

\begin{figure}[tb]
    \centering
    \begin{subfigure}[b]{0.32\textwidth}
        \centering
        \includegraphics[width=\linewidth]{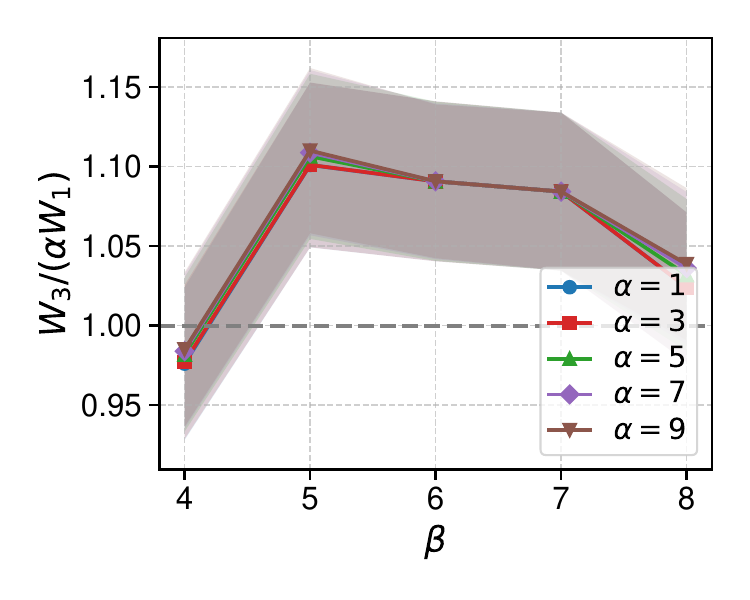}
        \caption{$n_H = 1, h = x$}
        \label{fig:md-hybrid-m1-h-alpha-c2-2p01}
    \end{subfigure}\hfill
    \begin{subfigure}[b]{0.32\textwidth}
        \centering
        \includegraphics[width=\linewidth]{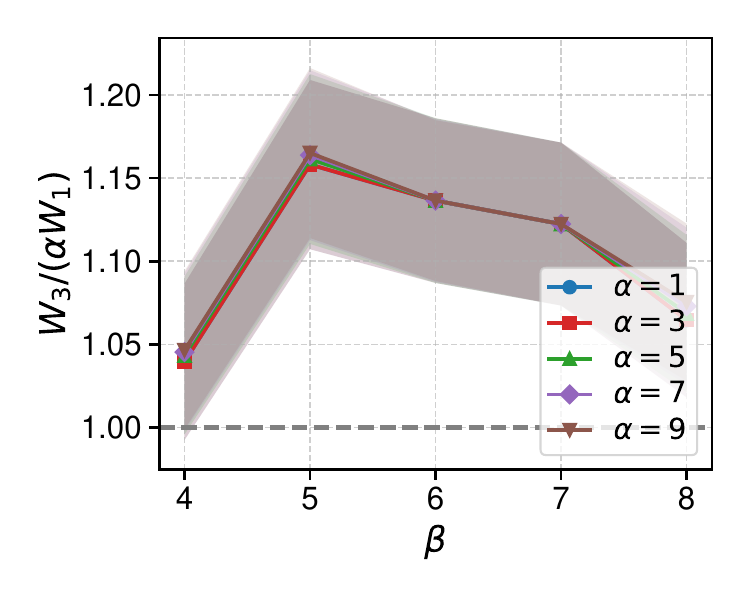}
        \caption{$n_H = 1, h = \sqrt[3]{x}$}
        \label{fig:md-hybrid-m1-h-alpha-1-3-c2-2p01}
    \end{subfigure}\hfill
    \begin{subfigure}[b]{0.32\textwidth}
        \centering
        \includegraphics[width=\linewidth]{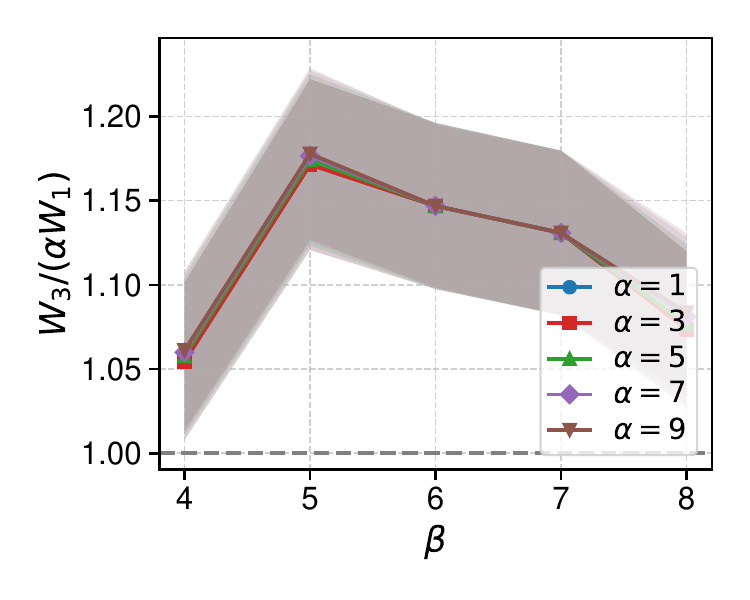}
        \caption{$n_H = 1, h = \sqrt[5]{x}$}
        \label{fig:md-hybrid-m1-h-alpha-1-5-c2-2p01}
    \end{subfigure}

    \vspace{1ex}

    \begin{subfigure}[b]{0.32\textwidth}
        \centering
        \includegraphics[width=\linewidth]{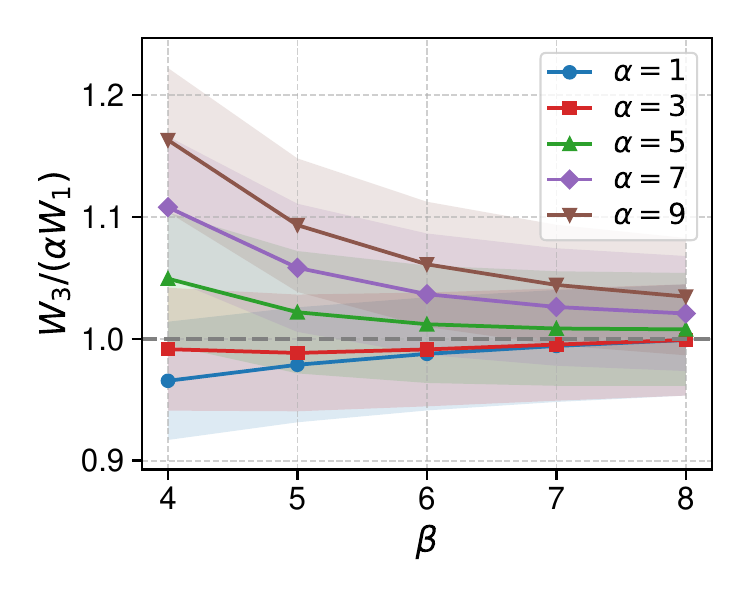}
        \caption{$n_H = 3, h = x$}
        \label{fig:md-hybrid-m3-h-alpha-c2-2p01}
    \end{subfigure}\hfill
    \begin{subfigure}[b]{0.32\textwidth}
        \centering
        \includegraphics[width=\linewidth]{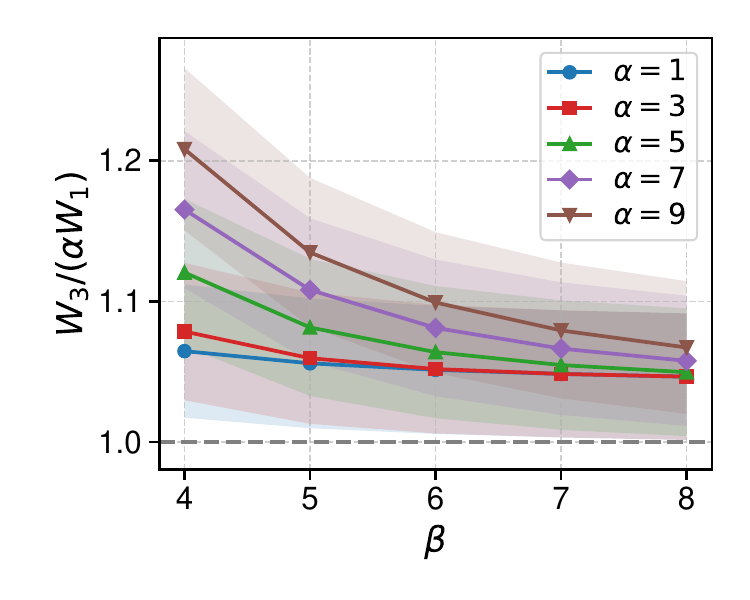}
        \caption{$n_H = 3, h = \sqrt[3]{x}$}
        \label{fig:md-hybrid-m3-h-alpha-1-3-c2-2p01}
    \end{subfigure}\hfill
    \begin{subfigure}[b]{0.32\textwidth}
        \centering
        \includegraphics[width=\linewidth]{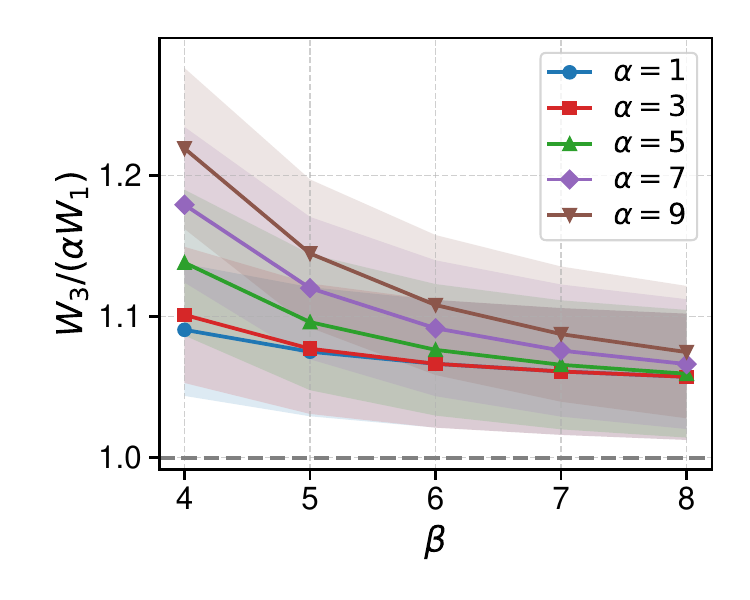}
        \caption{$n_H = 3, h = \sqrt[5]{x}$}
        \label{fig:md-hybrid-m3-h-alpha-1-5-c2-2p01}
    \end{subfigure}

    \vspace{1ex}

    \begin{subfigure}[b]{0.32\textwidth}
        \centering
        \includegraphics[width=\linewidth]{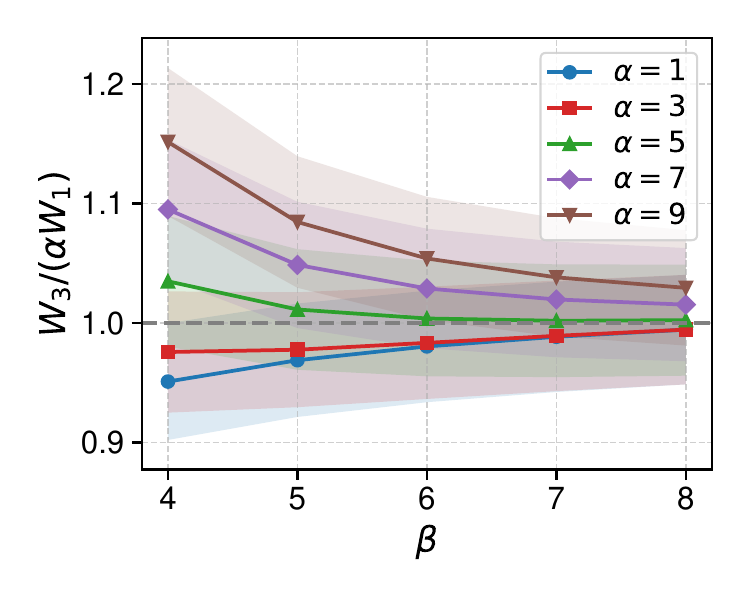}
        \caption{$n_H = 5, h = x$}
        \label{fig:md-hybrid-m5-h-alpha-c2-2p01}
    \end{subfigure}\hfill
    \begin{subfigure}[b]{0.32\textwidth}
        \centering
        \includegraphics[width=\linewidth]{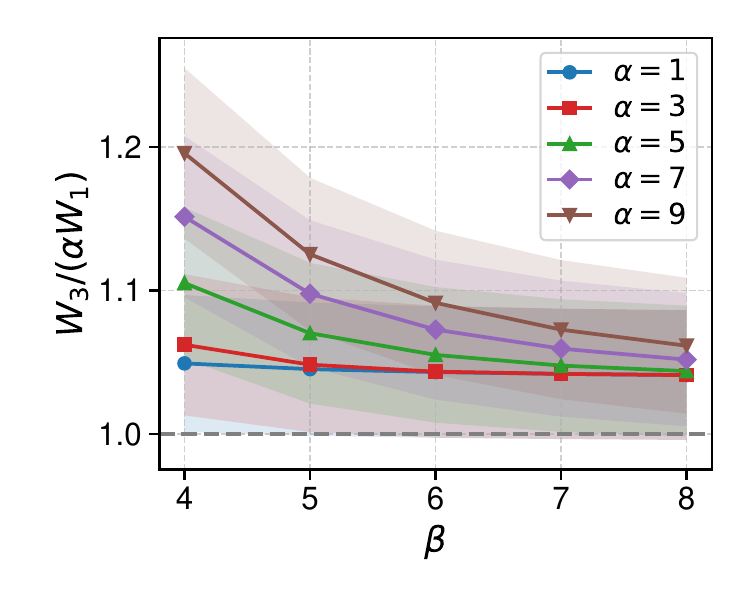}
        \caption{$n_H = 5, h = \sqrt[3]{x}$}
        \label{fig:md-hybrid-m5-h-alpha-1-3-c2-2p01}
    \end{subfigure}\hfill
    \begin{subfigure}[b]{0.32\textwidth}
        \centering
        \includegraphics[width=\linewidth]{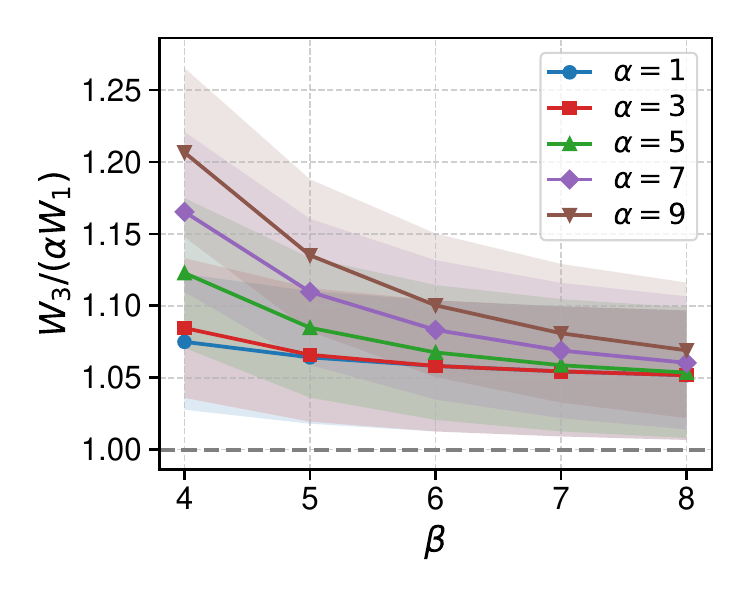}
        \caption{$n_H = 5, h = \sqrt[5]{x}$}
        \label{fig:md-hybrid-m5-h-alpha-1-5-c2-2p01}
    \end{subfigure}

    \vspace{1ex}

    \begin{subfigure}[b]{0.32\textwidth}
        \centering
        \includegraphics[width=\linewidth]{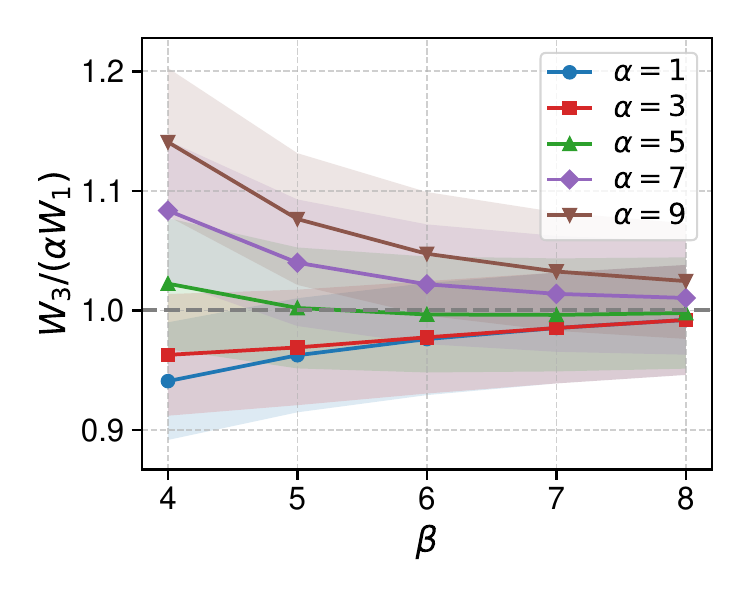}
        \caption{$n_H = 7, h = x$}
        \label{fig:md-hybrid-m7-h-alpha-c2-2p01}
    \end{subfigure}\hfill
    \begin{subfigure}[b]{0.32\textwidth}
        \centering
        \includegraphics[width=\linewidth]{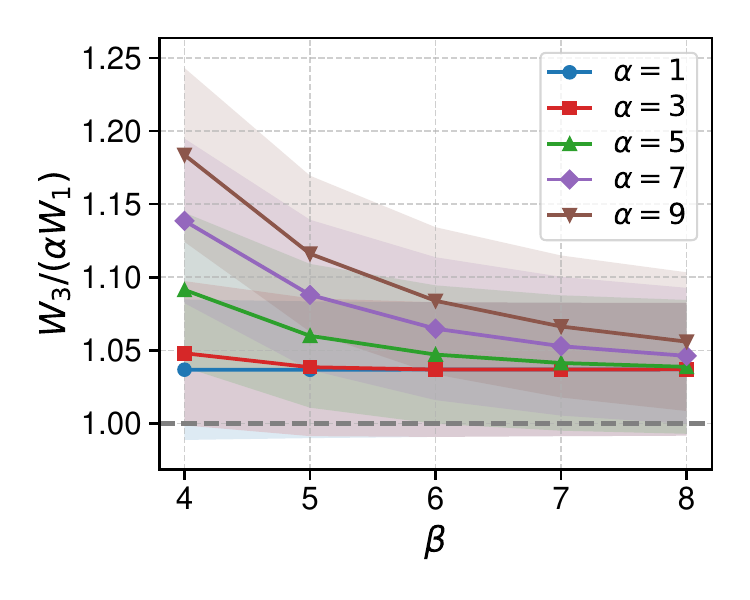}
        \caption{$n_H = 7, h = \sqrt[3]{x}$}
        \label{fig:md-hybrid-m7-h-alpha-1-3-c2-2p01}
    \end{subfigure}\hfill
    \begin{subfigure}[b]{0.32\textwidth}
        \centering
        \includegraphics[width=\linewidth]{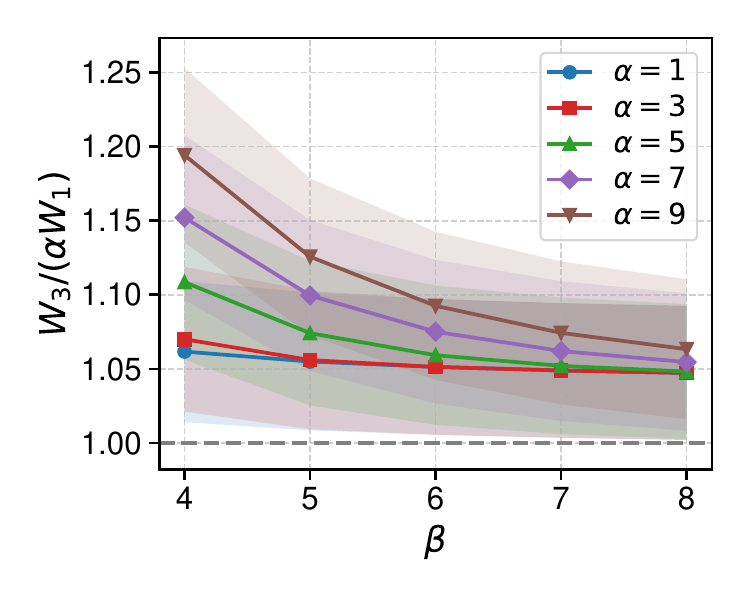}
        \caption{$n_H = 7, h = \sqrt[5]{x}$}
        \label{fig:md-hybrid-m7-h-alpha-1-5-c2-2p01}
    \end{subfigure}
    \caption{Long-term search engine profit under mechanism design with hybrid equilibrium}
    \label{fig:long-term-profit-tpbl-hybrid-asymmetric}
\end{figure}

\end{document}